\documentclass[12pt,english]{article}
\usepackage[T1]{fontenc}
\usepackage[latin9]{inputenc}
\usepackage{geometry}
\usepackage{color}
\usepackage{babel}
\usepackage{prettyref}
\usepackage{float}
\usepackage{amsmath}
\usepackage{amsthm}
\usepackage{amssymb}
\usepackage{setspace}
\usepackage{xargs}[2008/03/08]
\usepackage[unicode=true,pdfusetitle,
 bookmarks=true,bookmarksnumbered=false,bookmarksopen=false,
 breaklinks=false,pdfborder={0 0 1},backref=false,colorlinks=true]
 {hyperref}
\hypersetup{
 {urlcolor=blue},{linkcolor=black},{citecolor=blue}}

\makeatletter
  \theoremstyle{remark}
  \newtheorem{rem}{\protect\remarkname}
\theoremstyle{plain}
\newtheorem{thm}{\protect\theoremname}
  \theoremstyle{plain}
  \newtheorem{prop}{\protect\propositionname}
  \theoremstyle{plain}
  \newtheorem{lem}{\protect\lemmaname}

\@ifundefined{date}{}{\date{}}

\usepackage{geometry}
\usepackage{longtable} 
\PassOptionsToPackage{normalem}{ulem}
\usepackage{ulem}

\usepackage{babel}
\usepackage{array}
\usepackage{varioref}
\usepackage{multirow}
\usepackage{xargs}[2008/03/08]
\PassOptionsToPackage{normalem}{ulem}

\providecommand{\tabularnewline}{\\}

\usepackage{braket}
\usepackage{slashed}
\usepackage{scalerel}
\usepackage{stackengine}

\stackMath 
\newcommand{\tildeXL}[1]{%
\savestack{\tmpbox}{\stretchto{%
\scaleto{%
\scalerel*[\widthof{\ensuremath{#1}}]{\kern-.6pt\sim\kern-.6pt}%
{\rule[-\textheight/2]{1ex}{\textheight}}
}{\textheight}%
}{0.5ex}}%
\stackon[1pt]{#1}{\tmpbox}%
}

\usepackage{tikz}
\usetikzlibrary{calc,decorations.markings}
\usetikzlibrary{decorations.pathmorphing}

\renewcommand{\[}{\begin{equation}}
\renewcommand{\]}{\end{equation}}

\DeclareMathOperator{\im}{im}

\DeclareMathOperator{\haf}{haf}

\usepackage{cite}

\numberwithin{equation}{section}

\usepackage{titlesec}

\usepackage{changepage}

\makeatother

  \providecommand{\lemmaname}{Lemma}
  \providecommand{\propositionname}{Proposition}
  \providecommand{\remarkname}{Remark}
\providecommand{\theoremname}{Theorem}

\begin{document}

\title{Local and Covariant Flow Relations for OPE Coefficients in Lorentzian
Spacetimes}

\author{\texorpdfstring{Mark G. Klehfoth\thanks{mklehfoth@uchicago.edu} $\,$and Robert M. Wald\thanks{rmwa@uchicago.edu}
\\
 \em{Enrico Fermi Institute and Department of Physics} \\
 \em{University of Chicago} \\
 \em{Chicago, Illinois, USA}}{Mark G. Klehfoth and Robert M. Wald}}
\maketitle
\begin{abstract}
For Euclidean-signature quantum field theories with renormalizable
self-interactions, Holland and Hollands have shown the operator product
expansion (OPE) coefficients satisfy ``flow equations'': For a (renormalized)
self-coupling parameter $\lambda$, the partial derivative of any
OPE coefficient with respect to $\lambda$ is given by an integral
over Euclidean space of a sum of products of other OPE coefficients
evaluated at $\lambda$. These Euclidean flow equations were proven
to hold order-by-order in perturbation theory, but they are well defined
non-perturbatively and thus provide a possible route towards giving
a non-perturbative construction of the interacting field theory. The
purpose of this paper is to generalize the Holland and Hollands results
for flat Euclidean space to curved Lorentzian spacetimes in the context
of the solvable ``toy model'' of massive Klein-Gordon scalar field
theory on globally-hyperbolic curved spacetimes, with the squared
mass, $m^{2}$, viewed as the ``self-interaction parameter''. There
are a number of difficulties that must be overcome to carry out this
program. Even in Minkowski spacetime, a serious difficulty arises
from the fact that all integrals must be done over a compact region
of spacetime to ensure convergence. However, there does not exist
any Lorentz-invariant function of compact support, so any flow relations
that involve only integration over a compact region cannot be Lorentz
covariant. We show how covariant flow relations can be obtained by
the addition of ``counterterms'' that cancel the non-covariant dependence
on the cutoff function in a manner similar to that used in the Epstein-Glaser
renormalization scheme. The necessity of integration over a finite
region also effectively introduces an ``infrared cutoff scale''
$L$ into the flow relations, which gives rise to undesirable behavior
of the OPE coefficients under scaling of the metric and $m^{2}$.
(This behavior also occurs in the Euclidean case.) We show how to
modify the flow relations so that the dependence on $L$ is systematically
removed, thereby yielding flow relations compatible with almost homogeneous
scaling of the fields. A potentially even more serious difficulty
arises in curved spacetime simply due to the fact that the flow relations
involve integration over a spacetime region. Such an integration will
cause the OPE coefficients to depend non-locally on the spacetime
metric, in violation of the requirement that the quantum fields should
depend locally and covariantly on the spacetime metric. We show how
this difficulty can be overcome by replacing the metric with a local
polynomial approximation carried to suitable order about the OPE expansion
point. We thereby obtain local and covariant flow relations for the
OPE coefficients of Klein-Gordon theory in an arbitrary curved Lorentzian
spacetime. As a byproduct of our analysis, we prove that the field
redefinition freedom in the Wick fields (i.e. monomials of the scalar
field and its covariant derivatives) can be characterized by the freedom
to add a smooth, covariant, and symmetric function $F_{n}(x_{1},\dots,x_{n};z)$
to the identity OPE coefficients, $C_{\phi\cdots\phi}^{I}(x_{1},\dots,x_{n};z)$,
for the elementary $n$-point products. We thereby obtain an explicit
construction of any renormalization prescription for the nonlinear
Wick fields in terms of the OPE coefficients $C_{\phi\cdots\phi}^{I}$.
The ambiguities inherent in our procedure for modifying the flow relations
are shown to be in precise correspondence with the field redefinition
freedom of the Klein-Gordon OPE coefficients. In an appendix, we develop
an algorithm for constructing local and covariant flow relations in
Lorentzian spacetimes beyond our ``toy model'' based upon the associativity
properties of the OPE coefficients. We illustrate our method by applying
it to the flow relations of $\lambda\phi^{4}$-theory. 
\end{abstract}
\tableofcontents{}

\newpage{}

\section{Introduction and Overview}

A quantum field theory is said to possess an operator product expansion
(OPE) for all observables if, in any physically acceptable state $\Psi$,
the expectation value of any product of local quantum field observables
can be approximated near event $z$ as 
\begin{equation}
\braket{\Phi_{A_{1}}(x_{1})\cdots\Phi_{A_{n}}(x_{n})}_{\Psi}\sim\sum_{B}C_{A_{1}\cdots A_{n}}^{B}(x_{1},\dots,x_{n};z)\braket{\Phi_{B}(z)}_{\Psi},\label{eq:OPE}
\end{equation}
Here $A_{1}\dots,A_{n},B$ label the renormalized field observables
of the theory, and the sum over $B$ extends over all observables.
The coefficients $C_{A_{1}\cdots A_{n}}^{B}(x_{1},\dots,x_{n};z)$
of this expansion are ordinary $c$-valued distributions that are
independent of the state $\Psi$ (within the class of allowed states).
The ``$\sim$'' in eq.~(\ref{eq:OPE}) denotes that this relation
holds asymptotically in the coincidence limit $x_{1},\dots,x_{n}\to z$;
a precise statement of this asymptotic relationship will be given
in formula (\ref{eq:OPE precise asymp rel}) below. OPEs are expected
to exist for any renormalizable local quantum field theory under very
general assumptions \cite{Wilson_nonlagrangian_current_algebra,Wilson_Zimmermann,Zimmermann_normal_products_perturabative_OPE,Hollands_perturbative_OPE_CS,Bostelmann_OPE_via_phase_space,Bostelmann_phase_space_short_distance_QFT,Fredenhagen_conformal_Haag_Kastler_nets_pointlike_fields_existence_OPE,Fredenhagen_local_algebra_and_pointlike_fields}.

OPEs play a very important role both conceptually and practically
in quantum field theory in Euclidean and Minkowski spacetime. Furthermore,
Hollands and Wald \cite{HW_Axiomatic_QFTCS,HW_OPE_Dark_Energy} have
argued that OPEs play an essential role in the formulation of quantum
field theory in a curved Lorentzian spacetime. In a general, curved
Lorentzian spacetime there is no notion of Poincare invariance and
no preferred vacuum state, so properties of the quantum field normally
formulated in terms of vacuum expectation values in Minkowski spacetime
must now be formulated in terms of OPE coefficients. Hollands and
Wald have argued that the key relations satisfied by the quantum field
observables can be expressed via the OPE, so that, in essence, a quantum
field theory in curved spacetime may be viewed as being specified
by providing all of its OPE coefficients $C_{A_{1}\cdots A_{n}}^{B}$.
Thus, it is of considerable interest to determine the OPE coefficients
of an interacting quantum field theory. It would be especially of
interest to determine the OPE coefficients of an interacting theory
by methods that do not rely on perturbation theory, since this would
have the potential for providing a non-perturbative definition of
the interacting theory.

For the case of a Euclidean quantum field theory with power-counting
renormalizable self interactions, Hollands \cite{Hollands_Action_Principle}
has argued that the OPE coefficients must satisfy a ``flow'' relation
under changes of the coupling parameters. Such flow equations have
been proven to hold order-by-order in perturbation theory for several
interacting models, including $\lambda\phi^{4}$-theory \cite{HH_Recursive},
Yang-Mills theories \cite{Frob_Holland_Yang-Mills}, and CFTs with
strictly marginal interactions \cite{Hollands_Action_Principle}.
In particular, Holland and Hollands have proven that, by making using
of the freedom to redefine the quantum field observables, the OPE
coefficients of $\lambda\phi^{4}$-theory in $D=4$ dimensional (flat)
Euclidean space satisfy the following\footnote{In \cite{HH_Recursive}, Holland and Hollands set the expansion point
$z=x_{n}$. We prefer to define the coefficients more symmetrically
in $x_{1},\dots,x_{n}$ by using an independent expansion point $z$. } flow equations to any (finite) perturbative order in $\lambda$,
\begin{align}
\frac{\partial}{\partial\lambda}C_{A_{1}\cdots A_{n}}^{B}(x_{1},\dots,x_{n};z)= & -\frac{1}{4!}\int_{|y-z|\le L}d^{4}y\left[\vphantom{\int_{\epsilon}}C_{\phi^{4}A_{1}\cdots A_{n}}^{B}(y,x_{1},\dots,x_{n};z)+\right.\nonumber \\
 & -\sum_{i=1}^{n}\sum_{[C]\le[A_{i}]}C_{\phi^{4}A_{i}}^{C}(y,x_{i};x_{i})C_{A_{1}\cdots\widehat{A_{i}}C\cdots A_{n}}^{B}(x_{1},\dots,x_{n};z)+\nonumber \\
 & -\sum_{[C]<[B]}\left.\vphantom{\int_{\epsilon}}C_{A_{1}\cdots A_{n}}^{C}(x_{1},\dots,x_{n};z)C_{\phi^{4}C}^{B}(y,z;z)\right],\label{eq:HH fe}
\end{align}
Here $\lambda$ is the renormalized coupling parameter; $L$ is positive
constant with units of length; $\widehat{A_{i}}C$ indicates the replacement
of the label $A_{i}$ with the label $C$; and $[A]$ denotes the
``engineering dimension'' of the renormalized field $\Phi_{A}$
(see our ``Notation and conventions'' at the end of this section).
For the spatial integral over $y$, it is understood that the integration
is initially done over the region bounded by $\epsilon\le|y-x_{i}|$
and $\epsilon\le|y-z|\le L$, the subtractions appearing in the integrand
are performed, and the limit as $\epsilon\to0^{+}$ is then taken.
Holland and Hollands have shown that all ultraviolet divergences that
may arise in individual terms as $\epsilon\to0^{+}$ precisely cancel
between terms\footnote{In the discussion below, the cancellation of non-integrable divergences
at $y=x_{i}$ (for $i=1,\dots,n$) is equivalent to the statement
that the distribution appearing in the integrand of \eqref{eq:HH fe}
is uniquely ``extendible'' to the ``partial diagonals'' involving
$y$ and any single $x_{i}$-point. }, so the $\epsilon\to0^{+}$ limit is well-defined without any additional
regulators or renormalization.

Although the flow equations (\ref{eq:HH fe}) were rigorously derived
in a perturbative setting, these equations make sense mathematically
under very general model-independent assumptions\textemdash specifically,
if the OPE coefficients satisfy the ``associativity'' and ``scaling
degree'' axioms postulated in \cite{HW_Axiomatic_QFTCS}. Thus, it
seems reasonable to assume that eq.~(\ref{eq:HH fe}) holds in the
non-perturbative theory. If it were possible to integrate eq.~(\ref{eq:HH fe})
from $\lambda=0$ (where the field is free and the OPE coefficients
may be computed directly) up to some nonzero $\lambda$, we would
obtain a non-perturbative construction of the interacting OPE coefficients.
Of course, it is not known if there exist solutions to an infinite
system of ordinary differential equations like (\ref{eq:HH fe}).
Nevertheless, flow relations like eq.~(\ref{eq:HH fe}) have the
potential to provide a new approach to the formulation of interacting
quantum field theory, and may be of considerable ``practical'' use
as well.

The OPE flow relations (\ref{eq:HH fe}) and their generalization
to other interacting theories apply for the case of flat Euclidean
space. Recently, Fröb \cite{Frob_Riemannian} has generalized these
relations to quantum fields on curved Riemannian spaces, without,
however, imposing the condition that the OPE coefficients be locally
and covariantly defined. Since the physical world is Lorentzian, it
would be of interest to generalize the flow relations to Lorentzian
spacetimes. Furthermore, the requirement that the OPE coefficients
be locally and covariantly defined in curved spacetime is the natural
generalization of the requirement of Poincare invariance in Minkowski
spacetime \cite{HW_Axiomatic_QFTCS} and it thereby provides an important
requirement on the flow relations. Thus, it is of interest to determine
if the flow relations can be formulated for Lorentzian spacetimes
in a local and covariant manner.

There are two major obstacles to generalizing flow relations such
as eq.~(\ref{eq:HH fe}) to the Lorentzian case: (i) In the Euclidean
case, the infrared cutoff, $L$, appearing in the flow relations (\ref{eq:HH fe})
is fully compatible with rotational invariance, and the resulting
flow relations are automatically Euclidean invariant. However, in
Minkowski spacetime, no bounded region of spacetime can be invariant
under Lorentz boosts. Thus, in Minkowski spacetime, either the corresponding
integral must be taken over an unbounded region\textemdash resulting
in serious problems with convergence of the integral in Minkowski
spacetime as well as with the definition of the OPE coefficients throughout
the region in the generalization to curved spacetime\textemdash or
the corresponding integral will not be Lorentz invariant, leading
to flow relations that are not Poincare invariant. (ii) There is a
fundamental difficulty with obtaining local and covariant results
by performing an integral over a spacetime region. If the curved spacetime
flow relations take a form similar to eq.~(\ref{eq:HH fe}) where
the integral is performed over some neighborhood $U_{z}$ of $z\in M$,
this integral would depend on the spacetime metric in all of $U_{z}$,
not just in an arbitrarily small neighborhood of $z$. Thus, for a
flow relation of the form of eq.~(\ref{eq:HH fe}) with an integral
performed over a finite spacetime region $U_{z}$, the flow of OPE
coefficients will necessarily depend non-locally on the metric.

The purpose of this paper is to show how the above difficulties can
be overcome, thereby showing that local and covariant OPE flow relations
can be defined in curved Lorentzian spacetimes. We will also show
how to modify the flow relations so as to eliminate any dependence
on the infrared cutoff scale $L$. We will restrict consideration
in this paper to the ``toy model'' of massive, non-minimally-coupled
Klein-Gordon theory, with $m^{2}$ and the curvature coupling parameter,
$\xi$, viewed as interaction parameters. Of course, this model is
a free field for all values of the parameters. Nevertheless, we may
treat $m^{2}$ and $\xi$ as coupling constants in an interaction
Lagrangian, in parallel with the treatment of $\lambda$ above. The
resulting flow relations have a form that is very similar in its essential
features to that of a nonlinearly interacting theory, so this toy
model provides a good testing ground for confronting the issues needed
to generalize the flow relations to curved Lorentzian spacetimes.
For this toy model, in Euclidean space of any dimension $D\ge2$,
the direct analog of eq.~(\ref{eq:HH fe}) above is the following
flow relation in $m^{2}$ for the coefficients\footnote{As we shall see in Section \ref{subsec:gen Wick OPE coef}, all other
OPE coefficients are determined by $C_{\phi\cdots\phi}^{I}(x_{1},\dots,x_{n};z)$,
so it suffices to consider only the flow relations for these coefficients.} $C_{\phi\cdots\phi}^{I}(x_{1},\dots,x_{n};z)$ appearing in the OPE
of the $n$-point product of linear field observables, $\braket{\phi(x_{1})\cdots\phi(x_{n})}_{\Psi}$:
\begin{equation}
\frac{\partial}{\partial m^{2}}C_{\phi\cdots\phi}^{I}(x_{1},\dots,x_{n};z)=-\frac{1}{2}\int_{|y-z|^{2}\le L^{2}}d^{D}yC_{\phi^{2}\phi\cdots\phi}^{I}(y,x_{1},\dots,x_{n};z).\label{eq:naive flow eq for Wick id coef}
\end{equation}
Note that in this case the $y$-integral yields a well-defined distribution
in $(x_{1},\dots,x_{n})$ with no need for an ultraviolet cutoff $\epsilon$.
Our goal is to obtain an analogous flow relation in the Lorentzian
case.

The first issue we must address is the ``type'' of products of fields
that must be considered in order for the OPE coefficients to satisfy
flow relations. In the Euclidean case, there is a unique notion of
the $n$-point ($=$correlation$=$Green's$=$Schwinger) distributions
and their corresponding OPE coefficients. However, in the Lorentzian
case, one can consider Wightman products, time-ordered products, retarded
products, etc. Any of these products could be put on the left side
of eq.~(\ref{eq:OPE}) and used to define OPE coefficients. The resulting
OPE coefficients will possess distinct singular behavior (i.e., ``wavefront
sets''), and it is not obvious, a priori, which\textemdash if any\textemdash of
these Lorentzian objects are viable candidates for satisfying flow
relations. Our analysis of this issue in Section \ref{sec:flat Eucl fe}
reveals that the Green's function properties of the $n$-point distributions
play an essential role in the derivation of flow relations. Consequently,
as we discuss in Section \ref{sec:Minkowski-flow-relations}, the
usual Wightman $n$-point OPE coefficients as written in eq.~(\ref{eq:OPE})
are not suitable candidates for satisfying flow relations in the Lorentzian
case. On the other hand, time-ordered products do possess the requisite
Green's function properties for flow relations\footnote{Retarded and advanced products also satisfy the Green's function properties.}.
The Lorentzian flow relations we shall obtain will thus apply to the
OPE coefficients arising from the asymptotic expansion of the time
ordered products $\braket{T\{\Phi_{A_{1}}(x_{1})\cdots\Phi_{A_{n}}(x_{n})\}}_{\Psi}$
rather than the Wightman products $\braket{\Phi_{A_{1}}(x_{1})\cdots\Phi_{A_{n}}(x_{n})}_{\Psi}$.

However, working with time-ordered products has the potential to lead
to significant additional complications, since time-ordered products
possess substantial additional renormalization ambiguities beyond
those associated with the definition of Wick powers and their corresponding
Wightman functions. Time-ordered products of $n$ field observables
are well defined by naive time ordering only when no two points in
the $n$-point distribution coincide, i.e., away from all ``diagonals.''
We denote this well defined, ``unextended'' time-ordered product
by $T_{0}\{\Phi_{A_{1}}(x_{1})\cdots\Phi_{A_{n}}(x_{n})\}$. Any procedure
for extending $T_{0}$ to any of the diagonals (i.e. renormalization)
is generally non-unique and, therefore, must unavoidably introduce
new ambiguities proportional to $\delta$-distributions (i.e. ``contact
terms''). This will result in corresponding ambiguities on the diagonals
of the OPE coefficients defined using time-ordered products. Thus,
if we formulate the flow relations in terms of these OPE coefficients,
it might appear that we will have to deal with substantial additional
renormalization ambiguities on the diagonals.

Fortunately, however, we find that this is not the case. In the OPE,
eq.~(\ref{eq:OPE}), we may keep all of the $x_{i}$ distinct, so
that the unextended time ordered products and corresponding OPE coefficients
are well defined. However, flow relations such as eq.~(\ref{eq:HH fe})
involve an integration over a variable $y$, so we cannot avoid the
coincidence of $y$ with the various $x_{i}$. Thus, it might appear
that the flow relations require us to evaluate the the OPE coefficients
at points where they are not defined. However, it turns out that the
integrand of the OPE flow relations contains a very special combination
of OPE coefficients that has sufficiently mild divergences (i.e.,
``low scaling degrees'') on the ``partial diagonals'' involving
only $y$ and one other spacetime point. Consequently, the integrand
\emph{can }be uniquely extended to these\textemdash and typically
only these\textemdash partial diagonals, and the flow relations are
well defined for the unextended time-ordered products $T_{0}$. Thus,
no new renormalization ambiguities arise beyond those occurring for
the Wick monomials in the flow relations of the OPE coefficients of
unextended time ordered products.

We now explain how the two major obstacles described above to obtaining
Lorentzian flow relations are overcome. The first obstacle originates
from the fact that no bounded neighborhood of $z$ in Minkowski spacetime
can be invariant under Lorentz boosts. To ensure that the integrals
appearing in the flow relations are well defined and convergent, we
introduce into the integrand a smooth function\footnote{It is preferable to work with a smooth function $\chi$ than a step
function as in \eqref{eq:HH fe} and \eqref{eq:naive flow eq for Wick id coef}
since in the Lorentzian case the singular behavior of a step function
will overlap the singular behavior of the OPE coefficients in the
integrand.} $\chi(y-z;L)$ such that $\chi=1$ in a coordinate ball of radius
$L$ and $\chi=0$ outside a coordinate ball of radius $2L$. The
presence of $\chi$ ensures that the integral extends over only a
compact spacetime region, but it also necessarily breaks the Lorentz
covariance of the flow relations. Nevertheless, we prove in Section
\ref{sec:Minkowski-flow-relations} that Lorentz covariance can be
restored in Minkowski spacetime\textemdash to any desired ``scaling
degree''\textemdash by subtracting off finitely many terms in the
flow relations with a compensating failure of Lorentz invariance.
For the OPE coefficients $C_{T_{0}\{\phi\cdots\phi\}}^{I}$, this
results in a Minkowski spacetime flow relation of the form, 
\begin{align}
 & \frac{\partial}{\partial m^{2}}C_{T_{0}\{\phi\cdots\phi\}}^{I}(x_{1},\dots,x_{n};z)\label{eq:Lorentz inv id fe}\\
 & \qquad\sim-\frac{i}{2}\int_{\mathbb{R}^{D}}d^{D}y\chi(y-z;L)C_{T_{0}\{\phi^{2}\phi\cdots\phi\}}^{I}(y,x_{1},\dots,x_{n};z)-\sum_{C}a_{C}[\chi]C_{T_{0}\{\phi\cdots\phi\}}^{C}(x_{1},\dots,x_{n};z),\nonumber 
\end{align}
where $a_{C}$ are spacetime constant tensors which depend on $\chi$.
As described in Appendix \ref{sec:K-G Lorentz cts}, the existence
of such $a_{C}$ is guaranteed by the same kind of cohomological argument
\cite{pedagogical_remark_main_thm_ren_theory} that ensures the Lorentz-covariance
of the Epstein-Glaser renormalization scheme. In Appendix \ref{sec:K-G Lorentz cts},
we also obtain a recursive construction\footnote{The inductive formula for $a_{C}$ is given in eq.~(\ref{eq:a soln for r>2})
with $\boldsymbol{B}^{\kappa\rho}$ given by eq.~(\ref{eq:(B^kappa rho)_C}).} of the coefficients $a_{C}[\chi]$ required for the Lorentz-covariant
flow relations (\ref{eq:Lorentz inv id fe}) in Minkowski spacetime,
in parallel with the analysis given in \cite{Bresser_Pinter_Prange,Prange}
of the covariance-restoring Epstein-Glaser counterterms.

The flow relations (\ref{eq:Lorentz inv id fe}) are Lorentz covariant.
However, they contain an infrared cutoff scale $L$ and the presence
of $L$ in this formula will spoil the required almost homogeneous
scaling of $C_{T_{0}\{\phi\cdots\phi\}}^{I}(x_{1},\dots,x_{n};z)$
under the scalings $g_{ab}\to\lambda^{-2}g_{ab}$, $m^{2}\to\lambda^{2}m^{2}$
of the metric and the mass. This issue also arises for the Euclidean
flow relation eq.~(\ref{eq:naive flow eq for Wick id coef}). Thus,
we must further modify these flow relations so as to eliminate its
$L$ dependence up to any desired scaling degree. This can be accomplished
in the following manner. As shown in Subsection \ref{subsec:Eucl Had NO fe},
the partial derivative with respect to $L$ of the right side of the
Euclidean flow relation eq.~(\ref{eq:naive flow eq for Wick id coef})
is of the form, 
\begin{equation}
\frac{\partial}{\partial L}\left[\text{rhs of }\eqref{eq:naive flow eq for Wick id coef}\right]\sim\sum_{C}\beta_{C}(L)C_{\phi\cdots\phi}^{C}(x_{1},\dots,x_{n};z),\label{eq:partial L of naive Euclidean fe}
\end{equation}
where $\beta_{C}=\beta_{\gamma_{1}\cdots\gamma_{k}}$ denote tensors
that are computed from the OPE coefficients and depend on the infrared
length scale $L$. If the divergences in $\beta_{C}(L)$ were integrable
in a neighborhood containing $L=0$, then the problematic $L$-dependence
of the Euclidean flow relation (\ref{eq:naive flow eq for Wick id coef})
could be removed by simply subtracting the definite integral, 
\begin{equation}
\sum_{C}C_{\phi\cdots\phi}^{C}(x_{1},\dots,x_{n};z)\int_{0}^{L}dL'\beta_{C}(L'),\label{eq:naive L def integral}
\end{equation}
from the right-hand side of (\ref{eq:naive flow eq for Wick id coef}).
However, the divergences in $\beta_{C}(L)$ are not, in general, integrable.
Nevertheless, we show that, for any finite engineering dimension $[C]$,
all divergences in $\beta_{C}(L)$ as $L\to0^{+}$ can be expressed
as a finite linear combination of terms proportional to $L^{-k}\log^{N}L$
for positive integers $k,N$. Such non-integrable terms are in the
kernel of differential operators of the form $(1+k^{-1}L\partial_{L})^{N}$,
and these differential operators simply act like the identity operator
on any $L$-independent terms. Making use of these facts, we construct
a linear differential operator $\mathfrak{L}[L]$ which, when applied
to the right-hand side of (\ref{eq:naive flow eq for Wick id coef}),
effectively removes the $L$-dependent terms which lead to non-integrabilities
in $\beta_{C}(L)$, while perfectly preserving all of its $L$-independent
behavior. Once the operator $\mathfrak{L}[L]$ has been applied to
the right-hand side of (\ref{eq:naive flow eq for Wick id coef}),
any remaining $L$-dependence is guaranteed to be integrable and,
thus, can be eliminated via simple subtraction of a definite integral
as described above. In the Euclidean case, this yields the following
$L$-independent flow relations for OPE coefficients defined by Hadamard
normal ordering: 
\begin{align}
 & \frac{\partial}{\partial m^{2}}C_{\phi\cdots\phi}^{I}(x_{1},\dots,x_{n};z)\label{eq:L-indep Euclidean fe}\\
 & \qquad\sim-\frac{1}{2}\mathfrak{L}[L]\int d^{D}y\chi(y,z;L)C_{\phi^{2}\phi\cdots\phi}^{I}(y,x_{1},\dots,x_{n};z)-\sum_{C}b_{C}(L)C_{\phi\cdots\phi}^{C}(x_{1},\dots,x_{n};z),\nonumber 
\end{align}
with $\mathfrak{L}[L]$ given by eq.(\ref{eq:full L op}) and the
explicit dependence of $b_{C}$ on the OPE coefficients given in formula
(\ref{eq:b_C(L)}) of Theorem \ref{thm:Eucl fe for H-NO id coef}.
(For comparison with the Euclidean flow relations \eqref{eq:HH fe}
and \eqref{eq:naive flow eq for Wick id coef}, one should take $\chi$
to be a step function cutoff, $\chi(y,z)=\theta(L^{-2}|y-z|^{2})$.)
In the Minkowski case, the flow relations for the case where the Wick
powers are defined by Hadamard normal ordering\footnote{A similar formula holds for the case of a general definition of Wick
powers, with the only difference being the presence of additional
terms containing factors of the smooth functions $F_{k}$ that parameterize
the field-redefinition freedom of Wick fields.} become (see Theorem \ref{thm:Minkowski flow rel}) 
\begin{align}
 & \frac{\partial}{\partial m^{2}}C{}_{T_{0}\{\phi\cdots\phi\}}^{I}(x_{1},\dots,x_{n};z)\\
 & \qquad\sim-\frac{i}{2}\int d^{D}y\mathfrak{L}[L]\chi(y,z;L)C{}_{T_{0}\{\phi^{2}\phi\cdots\phi\}}^{I}(y,x_{1},\dots,x_{n};z)-\sum_{C}c_{C}(L)C{}_{T_{0}\{\phi\cdots\phi\}}^{C}(x_{1},\dots,x_{n};z),\nonumber 
\end{align}
where $c_{C}$ is given by formula (\ref{eq:c_C Minkowski}). The
ambiguities in the choice of $c_{C}$ correspond to the inherent renormalization
ambiguities in the OPE coefficients of Hadamard normal-ordered Wick
monomials.

The second major obstacle to obtaining Lorentzian flow relations arises
in curved spacetimes as a result of the nonlocal dependence on the
metric caused by integrating over a region of finite size. We overcome
this obstacle by replacing the true spacetime metric, $g_{\mu\nu}$,
with its Taylor polynomial, $g_{\mu\nu}^{(N)}$, in Riemannian normal
coordinates about $z$, carried to sufficiently high order, $N$,
to achieve equivalence in the flow relations up to the desired scaling
degree. This replacement is made prior to evaluating the spacetime
integral, so the resulting flow relations will be suitably ``local''
in the sense that they depend only on finitely-many derivatives of
the metric evaluated at the event $z$. However, we still need to
introduce a cutoff function, $\chi$, with an associated length scale
$L$ and, thus, these local flow relations will fail to be covariant
on account of the presence of $\chi$ and fail to scale almost homogeneously
due to the presence of $L$. Nevertheless, we can again introduce
compensating local counterterms to render the flow relation covariant
and we can construct an operator $\mathfrak{L}$ to eliminate the
dependence on $L$ to any desired asymptotic scaling degree. In any
Riemannian normal coordinate system with origin at $z$, the resulting
flow relations take the form, 
\begin{align}
\frac{\partial}{\partial m^{2}}C_{T_{0}\{\phi\cdots\phi\}}^{I}(x_{1},\dots,x_{n};\vec{0})\sim & -\frac{i}{2}\int_{\mathbb{R}^{D}}d^{D}y\sqrt{-g^{(N)}(y)}\mathfrak{L}[L]\chi(y,\vec{0};L)C_{T_{0}\{\phi^{2}\phi\cdots\phi\}}^{I}(y,x_{1},\dots,x_{n};\vec{0})+\nonumber \\
 & -\sum_{C}c_{C}C_{T_{0}\{\phi\cdots\phi\}}^{C}(x_{1},\dots,x_{n};\vec{0}),\label{eq:local covariant curved fe}
\end{align}
where the OPE coefficients on both sides and the counterterm coefficients
$c_{C}$ are functionals of the polynomial metric $g_{\mu\nu}^{(N)}$.
All dependence of $c_{C}$ on the polynomial metric at event $z$
can be expressed entirely in terms of totally-symmetric covariant
derivatives of the Riemannian curvature tensor. The explicit form
of $c_{C}$ is given in terms of the OPE coefficients in formula (\ref{eq:c_C CS}).
Overall, the key new aspects of the curved spacetime flow relations
(\ref{eq:local covariant curved fe}) are the replacement of the metric
by a polynomial approximation and the presence of additional counterterms
involving the curvature.

Finally, we note that our derivations of the flow relations for flat
Euclidean space given in Section \ref{sec:flat Eucl fe}, the flow
relations for Minkowski spacetime given in Section \ref{sec:Minkowski-flow-relations},
and the flow relations for general curved Lorentzian spacetimes given
in Section \ref{sec:Wick fe in CS} were based upon formulas for OPE
coefficients that we obtained explicitly in Section \ref{sec:K-G OPE-coefficients}.
However, for nonlinear models, such explicit non-perturbative formulas
for the OPE coefficients are not available. However, in Appendix \ref{sec:model-indep counterterms},
we show that for the integrals which appear in the flow relations,
one can derive covariance-restoring counterterms using only the associativity
property of OPE coefficients, without explicit knowledge of the coefficients.
When specialized to Klein-Gordon theory, this general algorithm reproduces
the results we derived in Sections \ref{sec:Minkowski-flow-relations}-\ref{sec:Wick fe in CS}.
When applied to $\lambda\phi^{4}$-theory in a curved Lorentzian spacetime
$(M,g_{ab})$, the algorithm developed in Appendix \ref{sec:model-indep counterterms}
yields 
\begin{align}
 & \frac{\partial}{\partial\lambda}C_{T_{0}\{A_{1}\cdots A_{n}\}}^{B}(x_{1},\dots,x_{n};\vec{0})\sim\nonumber \\
 & -\frac{1}{4!}\int d^{4}y\sqrt{-g^{(N)}(y)}\chi(y,\vec{0};L)\left[\vphantom{\int_{\epsilon}}C_{T_{0}\{\phi^{4}A_{1}\cdots A_{n}\}}^{B}(y,x_{1},\dots,x_{n};\vec{0})+\right.\nonumber \\
 & -\sum_{i=1}^{n}\sum_{[C]\le[A_{i}]}\left[C_{T_{0}\{\phi^{4}A_{i}\}}^{C}(y,x_{i};x_{i})-\right.\sum_{[D]}\left.c_{D}^{C}C_{T_{0}\{A_{i}\}}^{D}(x_{i};\vec{0})\right]C_{T_{0}\{A_{1}\cdots\widehat{A_{i}}C\cdots A_{n}\}}^{B}(x_{1},\dots,x_{n};\vec{0})+\nonumber \\
 & -\left[\vphantom{C_{T_{0}\{\phi^{4}\}}^{B}}\right.\sum_{[C]<[B]}C_{T_{0}\{\phi^{4}C\}}^{B}(y,\vec{0};\vec{0})-\sum_{[C]\ge[B]}\left.\vphantom{C_{T_{0}\{\phi^{4}\}}^{B}}c_{C}^{B}\right]\left.\vphantom{\int_{\epsilon}}C_{T_{0}\{A_{1}\cdots A_{n}\}}^{C}(x_{1},\dots,x_{n};\vec{0})\right],\label{eq:lambda phi 4 Mink fe}
\end{align}
where the $[D]$-sum in the second line and the $[C]\ge[B]$ sum in
the third line are carried out to sufficiently-large but finite engineering
dimensions\footnote{The coefficient $C_{T_{0}\{A_{i}\}}^{D}=C_{A_{i}}^{D}$ involving
a single field factor is given by the geometric factors that appear
in an ordinary Taylor expansion (see \eqref{eq:1 point OPE}).}. The form of the counterterm coefficients $c_{C}^{B}$ is given in
Appendix \ref{sec:model-indep counterterms} for flat Minkowski spacetime.
It would be natural to associate the inherent local and covariant
ambiguities in $c_{C}^{B}$ with the field-redefinition freedom of
$\lambda\phi^{4}$-theory, but we have not investigated this issue\footnote{This analysis would require an understanding of what field-redefinition
freedom is allowed for the non-perturbative interacting theory. Note
also that we have not attempted to eliminate the $L$-dependence of
the flow relations \eqref{eq:lambda phi 4 Mink fe}. The techniques
described in Subsection \ref{subsec:Eucl Had NO fe} can be used to
eliminate the $L$-dependence of \eqref{eq:lambda phi 4 Mink fe}
to any finite order in perturbation theory, but it is not obvious
how to remove the $L$-dependence non-perturbatively. }. \bigskip{}

The structure of our paper is as follows. In Section \ref{sec:K-G theory and Wick fields},
we review the theory of a free Klein-Gordon field on a curved Lorentzian
spacetime. The ambiguities in the definition of arbitrary Wick monomials\textcolor{blue}{{}
}$\Phi_{A}\equiv\nabla_{\alpha_{1}}\phi\cdots\nabla_{\alpha_{n}}\phi$
(where $\alpha_{i}$ denote spacetime multi-indices) is fully analyzed.
The precise form of the ``mixing matrix'' $\mathcal{Z}{}_{A}^{B}$
describing allowed field redefinitions is given in Theorem \ref{thm:Wick uniquenes Z},
and it is shown in Proposition \ref{prop:Wick uniqueness Fn} that
the field redefinition freedom is fully characterized by a sequence
of smooth, real-valued functions $F_{n}(x_{1},\dots,x_{n};z)$ that
are symmetric in $(x_{1},\dots,x_{n})$.

In Section \ref{sec:K-G OPE-coefficients}, we show that the Klein-Gordon
field admits an OPE of the form eq.(\ref{eq:OPE}) for Hadamard states
$\Psi$. In Theorem \ref{thm: Had NO OPE coef}, we obtain an explicit
formula for the OPE coefficients for the case where Wick monomials
are defined by Hadamard normal ordering. For a general prescription
for Wick monomials, we show that the OPE coefficients\textcolor{red}{{}
}$C_{A_{1}\cdots A_{n}}^{B}$ for products of general Wick monomials
are completely determined by the OPE coefficients $C_{\phi\cdots\phi}^{I}$
of the identity operator, $I$, for the $n$-point products of the
linear field observable, $\phi(x_{1})\cdots\phi(x_{n})$. Furthermore,
$C_{\phi\cdots\phi}^{I}(x_{1},\dots,x_{n};z)$ is uniquely determined
by the coefficients $C_{\phi\cdots\phi}^{I}$ with smaller $n$ up
to the addition of the function $F_{n}(x_{1},\dots,x_{n};z)$ appearing
in Proposition \ref{prop:Wick uniqueness Fn}. The existence and properties
of the OPE for a general definition of Wick monomials is summarized
in Theorem \ref{thm:existence Wick coef and associativity}. An inductive
construction of the Wick monomials in terms of $C_{\phi\cdots\phi}^{I}$
is given in Proposition \ref{prop:induct construction of Wick fields via id OPE coefs}.
As discussed in Subsection \ref{subsec:TOP}, all these statements
carry over to the OPE for unextended time-ordered products, since
the formulas for their OPE coefficients may be obtained in a simple
and direct manner from the formulas for $C_{A_{1}\cdots A_{n}}^{B}$.

In Section \ref{sec:flat Eucl fe}, we derive the flow relations for
the OPE coefficients of the Euclidean version of the Klein-Gordon
field. The modification of the flow relations needed to remove the
$L$-dependence is given in Subsection \ref{subsec:Eucl Had NO fe}.

In Section \ref{sec:Minkowski-flow-relations}, we analyze the flow
relations for the OPE coefficients of the Klein-Gordon field in Minkowski
spacetime. The counterterms in the flow relations needed to restore
Lorentz covariance are obtained, with the technical details given
in Appendix \ref{sec:K-G Lorentz cts}.

The generalization to curved spacetimes is given in Section \ref{sec:Wick fe in CS}.
To any specified scaling degree, we replace the spacetime metric by
a Taylor approximation in a Riemannian normal coordinate system defined
relative to the expansion point $z$. We then show that suitable counterterms
can be introduced to yield local and covariant flow relations that
are independent of $L$.

Finally, although our analysis in this paper is restricted to the
toy model of the free Klein-Gordon field, we show in Appendix \ref{sec:model-indep counterterms}
that our construction of the covariance-restoring counterterms requires
only the associativity property of the OPE coefficients and thus should
be applicable to nonlinearly interacting theories. The algorithm for
constructing counterterms given in Appendix \ref{sec:model-indep counterterms}
reproduces the results we derived in Sections \ref{sec:Minkowski-flow-relations}-\ref{sec:Wick fe in CS}
when applied to Klein-Gordon theory. When applied to $\lambda\phi^{4}$-theory
in Lorentzian spacetime $(M,g_{ab})$, we obtain the local and covariant
Lorentzian analogue (\ref{eq:lambda phi 4 Mink fe}) of the Holland
and Hollands Euclidean flow relations (\ref{eq:HH fe}).

\global\long\def\local{\text{W1}}
 \global\long\def\spectrum{\text{W2}}
 \global\long\def\commutator{\text{W3}}
 \global\long\def\Leibniz{\text{W4}}
 \global\long\def\Hermiticity{\text{W5}}
 \global\long\def\symmetry{\text{W6}}
 \global\long\def\scaling{\text{W7}}
 \global\long\def\conservation{\text{W8}}

\newcommandx\T[1][usedefault, addprefix=\global, 1=]{S^{#1}}
 \global\long\def\H#1#2{\Phi_{#1_{#2}}^{H}}
 \global\long\def\W#1{\Phi_{#1}}
 \global\long\def\na#1#2{\nabla_{\alpha_{(#1,#2)}}^{(x_{#1})}}
 \global\long\def\Z#1#2{\mathcal{Z}_{#1}^{#2}}
 \global\long\def\Zinv#1#2{(\mathcal{Z}^{-1})_{#1}^{#2}}

\paragraph*{Notation and conventions: \addcontentsline{toc}{subsection}{Notation and conventions}\label{par:Notation-and-conventions}}

We use letters from the beginning of the Latin alphabet to denote
abstract indices and our spacetime geometry conventions coincide with
those of \cite{Wald_GR_text}. Tensors are often abbreviated with
multi-indices chosen from the beginning of the Greek alphabet ($\alpha,\beta,\gamma,\dots$)\textemdash e.g.,
we denote a tensor $T_{\phantom{a_{1}\cdots a_{n}}b_{1}\cdots b_{m}}^{a_{1}\cdots a_{n}}$
of type $(n,m)$ simply as $T_{\hphantom{\alpha}\beta}^{\alpha}$.
In combinatorial formulas involving abstract multi-indices, we use
the obvious analogues of the standard multi-index conventions: e.g.,
for $T^{\alpha}\equiv T^{a_{1}\cdots a_{n}}$, we have $|\alpha|\equiv n$
and $\alpha!\equiv|\alpha|!$. When coordinate components of a tensor
are needed, we denote ordinary spacetime indices with letters from
the middle of the Greek alphabet $(\mu,\nu,\kappa,\rho,\dots)$ but
continue to denote multi-indices with ($\alpha,\beta,\gamma,\dots$).
Throughout, $\mathbb{N}$ denotes the natural numbers (positive integers,
excluding $0$) and $\mathbb{N}_{0}\equiv\{0\}\cup\mathbb{N}$. We
use ``smooth'' to mean infinitely differentiable, i.e. $C^{\infty}$,
and the ``Taylor coefficients of $f$ evaluated at $z$'' will refer
to the set, $\nabla_{\alpha_{1}}\cdots\nabla_{\alpha_{n}}f(x_{1},\dots,x_{n})|_{x_{1},\dots x_{n}=z}$,
of covariant derivatives of a multivariate smooth function $f$ evaluated
at $z$ \emph{without} the numerical factor $1/(\alpha_{1}!\cdots\alpha_{n}!)$.
The set of smooth functions of compact support is denoted by $C_{0}^{\infty}$
and the dual space of distributions is denoted by $\mathcal{D}':C_{0}^{\infty}\to\mathbb{R}$.

Some notation in the paper may not always be redefined with each use.
For the convenience of the reader, we include here a list of frequently-employed
non-standard symbols and their definitions or, in cases where the
definition is too lengthy, we reference the equation where the symbol
is defined. \setlength{\LTleft}{-1cm}\setlength{\LTright}{0cm}\begin{longtable}{c>{\raggedright}p{0.95\textwidth}} \multicolumn{2}{c}{\uline{field notation}}\tabularnewline  & \tabularnewline $\Phi_{A}$ & the differentiated scalar field monomial, $\nabla_{\alpha_{1}}\phi\nabla_{\alpha_{2}}\phi\cdots\nabla_{\alpha_{p}}\phi$ \tabularnewline $\Phi_{A}^{H}$ & monomial, $(\nabla_{\alpha_{1}}\phi\cdots\nabla_{\alpha_{p}}\phi)_{H}$, defined via ``Hadamard normal ordering'', see eq.~\eqref{eq:H normal ordered}\tabularnewline $\mathcal{Z}_{A}^{B}$ & field redefinition ``mixing matrix'' defined in eq.~\eqref{def:mixing matrix Z}\tabularnewline $[A]_{\phi}$ & the number of $\phi$-factors appearing in $\Phi_{A}$ (i.e., $p$, in this case)\tabularnewline $[A]_{\nabla}$ & the number of covariant derivatives appearing in $\Phi_{A}$ (i.e., $\text{\ensuremath{\sum}}_{i=1}^{p}|\alpha_{i}|$, in this case)\tabularnewline $[A]$ & ``engineering dimension'' of $\Phi_{A}$ given by (rational) number ${\displaystyle (D/2-1)\times}[A]_{\phi}+[A]_{\nabla}$ \tabularnewline $C_{A_{1}\cdots A_{n}}^{B}$ & OPE coefficients defined in relation \eqref{eq:OPE}\tabularnewline $(C_{H})_{A_{1}\cdots A_{n}}^{B}$ & OPE coefficients defined in relation \eqref{eq:Had OPE def} for Hadamard normal-ordered fields\tabularnewline $C_{T_{0}\{A_{1}\cdots A_{n}\}}^{B}$ & OPE coefficients of unextended time-ordered products defined in \eqref{eq:def T0 OPE coefs}\tabularnewline $(C_{H}){}_{T_{0}\{A_{1}\cdots A_{n}\}}^{B}$ & Hadamard normal-ordered version of $C_{T_{0}\{A_{1}\cdots A_{n}\}}^{B}$\tabularnewline  & \tabularnewline \multicolumn{2}{c}{\uline{differential operators, parametrices and Greens functions}}\tabularnewline  & \tabularnewline $K$ & Klein-Gordon operator, $K\equiv-g^{ab}\nabla_{a}\nabla_{b}+m^{2}+\xi R$\tabularnewline $H$ & Hadamard parametrix defined in eq.~\eqref{eq:H}\tabularnewline $H_{F}$ & Feynman parametrix, $H_{F}\equiv H-i\Delta^{\text{adv}}$, see also Footnote \ref{fn:i epsilon for H_F}\tabularnewline $\Delta$ & causal propagator, $\Delta\equiv\Delta^{\text{adv}}-\Delta^{\text{ret}}$\tabularnewline $\Delta^{\text{adv}},\text{\ensuremath{\Delta}}^{\text{ret}}$ & advanced and retarded, resp., Greens function of $K$ \tabularnewline $\mathfrak{L}$ & operator defined in terms of infrared length scale $L$ and $\partial/\partial L$ in eq.~\eqref{eq:full L op}\tabularnewline  & \tabularnewline \multicolumn{2}{c}{\uline{geometric notation}}\tabularnewline  & \tabularnewline $D$ & the spacetime dimension, i.e., $(\text{\# spatial dimensions})+(\text{\# temporal dimensions})$\tabularnewline $d\mu_{g}(x)$ & covariant volume element, $d^{D}x\sqrt{-g(x)}$, on spacetime $(M,g_{ab})$\tabularnewline $S^{\beta}(x,z)$ & bi-tensor defined with respect to the geodesic distance function in eq.~\eqref{eq:covariant Taylor coef}\tabularnewline $Z^{*}M$ & zero section of the cotangent bundle $T^{*}M$ \tabularnewline $V_{x}^{\pm}$ & future/past lightcone of the cotangent space $T_{x}^{*}M$\tabularnewline $\dot{V}_{x}^{\pm}$ & boundary of future/past lightcone of cotangent space $T_{x}^{*}M$\tabularnewline $(x_{1},k_{1})\sim(x_{2},k_{2})$ & equivalence relation defined below eq.~\eqref{eq:Had WF} for $(x_{1},k_{1}),(x_{2},k_{2})\in T^{*}M$\tabularnewline  & \tabularnewline \multicolumn{2}{c}{\uline{asymptotic equivalence relations}}\tabularnewline  & \tabularnewline $\sim_{\mathcal{T},\delta}$ & \multirow{2}{0.95\textwidth}{asymptotic equivalence to scaling degree $\delta$ for merger tree $\mathcal{T}$, defined in the paragraph surrounding eq.~\eqref{eq:OPE precise asymp rel}}\tabularnewline &  \tabularnewline $\sim_{\delta}$ & shorthand for ``$\sim_{\mathcal{T},\delta}$'' when $\mathcal{T}$ is the trivial merger tree, i.e., all spacetime points\\  merge at the same rate to $z$\tabularnewline $\approx$ & asymptotic equivalence for all $\delta$ and $\mathcal{T}$, defined in the paragraph surrounding eq.~\eqref{eq:OPE precise asymp rel}\tabularnewline \end{longtable}

\section{Klein-Gordon theory and local Wick fields\label{sec:K-G theory and Wick fields}}

The theory of a Klein-Gordon scalar field on a $D$-dimensional spacetime
$(M,g_{ab})$ with mass $m$ and curvature coupling $\xi$ is given
by the action, 
\begin{equation}
S_{\text{KG}}\equiv-\frac{1}{2}\int_{M}d^{D}x\sqrt{-g(x)}\left[g^{ab}(x)\nabla_{a}\phi(x)\nabla_{b}\phi(x)+\left(m^{2}+\xi R(x)\right)\phi^{2}(x)\right].\label{eq:K-G action}
\end{equation}
The equation of motion arising from this action is 
\begin{equation}
K\phi=0,
\end{equation}
where the Klein-Gordon operator $K$ is given by 
\begin{equation}
K\equiv-g^{ab}\nabla_{a}\nabla_{b}+m^{2}+\xi R.\label{eq:K-G op}
\end{equation}
To guarantee well-defined dynamics and to avoid causal pathologies,
we will restrict consideration throughout to globally-hyperbolic spacetimes,
$(M,g_{ab})$. Any globally-hyperbolic spacetime admits unique advanced,
$\Delta^{\text{adv}}$, and retarded, $\Delta^{\text{ret}}$, Green's
distributions of the Klein-Gordon operator, $K$ \cite{Friedlander}.

In this section, we consider the quantum field theory of the Klein-Gordon
field. Our main concern is the ambiguities in the definition of arbitrary
Wick monomials, i.e., quantum field observables of the form 
\begin{equation}
\Phi_{A}\equiv\nabla_{\alpha_{1}}\phi\cdots\nabla_{\alpha_{p}}\phi,\label{eq:multi index notation for monomials}
\end{equation}
Here $\alpha_{i}$ denotes an abstract multi-index, i.e., $\alpha_{i}=a_{i,1}\dots a_{i,|\alpha_{i}|}$
where each $a_{i,j}$ is a spacetime index. Thus, $\Phi_{A}$ corresponds
to a tensor constructed from $p$-factors of $\phi$, with $|\alpha_{i}|$-number
of derivatives on the $i$-th factor. The ambiguities in $\Phi_{A}$
will give rise to corresponding ambiguities in the $n$-point distributions,
\begin{equation}
\braket{\Phi_{A_{1}}(f_{1})\cdots\Phi_{A_{n}}(f_{n})}_{\Psi},\label{eq:general n point Wightman distributions}
\end{equation}
as well as the $n$-point distributions for the un-extended time-ordered
products. This will, in turn, give rise to corresponding ambiguities
in the OPE coefficients. The main result of this section will be to
obtain a simple characterization of the ambiguities in the definition
of Wick monomials which will be extremely useful for characterizing
the corresponding ambiguities in the OPE coefficients derived in the
next section.

In Subsection \ref{subsec:algebra, axioms, state space Wick poly},
we review the construction of the abstract algebra\footnote{The algebraic approach to quantum field theory was initiated in \cite{Haag_Kastler}.
A comprehensive review may be found in \cite[Chapter III]{Haag_book}.} containing Wick polynomials and the requirements (``axioms'') imposed
on the Wick monomials. The known uniqueness theorem for Wick monomials
implied by these axioms (see Theorem \ref{thm:Wick uniquenes Z})
is then reformulated in Subsection \ref{subsec:uniq Wick poly and novel reparm}
in terms in terms of a choice of smooth functions $F_{n}$ (see Proposition
\ref{prop:Wick uniqueness Fn}).

\subsection{Wick algebra and state space: axioms and existence of Wick polynomials
\label{subsec:algebra, axioms, state space Wick poly}}

In this subsection, we review the definition of the algebra of observables
$\mathcal{W}(M,g_{ab})$ for the Klein-Gordon field and the axioms
that determine the Wick monomials\textemdash up to the uniqueness
discussed in in the following subsection. Our discussion closely follows
\cite{HW_local_Wick_poly} which built on the earlier work of \cite{BFK_muSC_Wick_polynomials,BF_muA_interacting_QFTs,DF_pert_algebraicQFT_deform_quant,DF_algebraicQFT_loop_expansion}.

The construction of $\mathcal{W}(M,g_{ab})$ begins with the standard
CCR (canonical commutation relation) algebra $\mathcal{A}(M,g_{ab})$
generated by observables that are linear in $\phi$. To define $\mathcal{A}$,
we start with the free $*$-algebra $\mathcal{A}_{0}$ generated by
the identity $I$ and the fundamental (smeared) field $\phi(f)$ with
$f\in C_{0}^{\infty}(M)$. We then factor $\mathcal{A}_{0}$ by all
of the relations we wish to impose. To do so, we let $\mathcal{I}\subset\mathcal{A}_{0}$
be the two-sided ideal consisting of all elements in $\mathcal{A}_{0}$
that contain at least one factor that can be put into any of the following
forms: 
\begin{enumerate}
\item[i)] $\phi(c_{1}f_{1}+c_{2}f_{2})-c_{1}\phi(f_{1})-c_{2}\phi(f_{2}),$
with $c_{1},c_{2}\in\mathbb{C}$ 
\item[ii)] $\phi(f)^{*}-\phi(\overline{f})$ 
\item[iii)] $\phi(Kf),$ with the Klein-Gordon operator $K$ given by eq.~(\ref{eq:K-G op}). 
\item[iv)] $\phi(f_{1})\phi(f_{2})-\phi(f_{2})\phi(f_{1})-i\Delta(f_{1},f_{2})I,$
where $\Delta[M,g_{ab}]$ denotes the advanced minus retarded Green's
distribution for $K[g_{ab},m^{2},\xi]$ on $M$ 
\end{enumerate}
The algebra $\mathcal{A}$ is then defined to be the free algebra
factored by this ideal, 
\[
\mathcal{A}(M,g_{ab})\equiv\mathcal{A}_{0}\backslash\mathcal{I}(M,g_{ab}).
\]
Thus, the CCR algebra effectively incorporates (i) the distributional
nature of quantum fields, (ii) the Hermiticity of real-valued fields,
(iii) the Klein-Gordon field equation, and (iv) the canonical commutation
relations. It contains all elements that are finite linear combinations
of products of the (smeared) fundamental field. Quantum states of
the CCR algebra are then just linear maps $\braket{\cdot}_{\Psi}:\mathcal{A}(M,g_{ab})\to\mathbb{C}$
which are normalized, $\braket{I}_{\Psi}=1$, and positive, $\braket{A^{*}A}_{\Psi}\ge0$
for all $A\in\mathcal{A}$.

The first step towards enlarging $\mathcal{A}(M,g_{ab})$ to the full
algebra of observables $\mathcal{W}(M,g_{ab})$ is to define the normal-ordered
product relative to a state $\braket{\cdot}_{\Psi}$ by the formula
\begin{equation}
:\phi(f_{1})\cdots\phi(f_{n}):_{\Psi}\;\equiv\sum_{P}(-1)^{|P|}\prod_{(i,j)\in P}\braket{\phi(f_{i})\phi(f_{j})}_{\Psi}\prod_{k\in\{1,\dots,n\}\backslash P}\phi(f_{k}),\label{eq:def N-O in A}
\end{equation}
where the $P$ are sets containing disjoint, ordered pairs taken from
$\{1,\dots,n\}$ such that $i<j$, and $|P|$ denotes the number of
pairs in $P$. Note normal-ordered elements (\ref{eq:def N-O in A})
of $\mathcal{A}$ are symmetric under interchange of test functions,
i.e., $:\phi(f_{1})\cdots\phi(f_{n}):_{\Psi}=:\phi(f_{\pi(1)})\cdots\phi(f_{\pi(n)}):_{\Psi}$
for any permutation $\pi$. Products of normal-ordered elements also
satisfy the following important identity (``Wick's theorem''), 
\begin{equation}
\negthickspace\negthickspace\negthickspace:\phi(f_{1})\cdots\phi(f_{n}):_{\Psi}:\phi(f_{n+1})\cdots\phi(f_{n+m}):_{\Psi}\;=\sum_{p\le\min(n,m)}\prod_{(i,j)\in P_{p}}\braket{\phi(f_{i})\phi(f_{j})}_{\Psi}\;:\negthickspace\negthickspace\negthickspace\prod_{k\in\{1,\dots,n\}\backslash P_{p}}\negthickspace\negthickspace\negthickspace\phi(f_{k}):_{\Psi},\label{eq:Wicks thm CCR}
\end{equation}
where $P_{p}$ denote a set containing $p$ disjoint, ordered pairs
$(i,j)$ such that $i\in\{1,2,\dots,n\}$ and $j\in\{n+1,n+2,\dots,n+m\}$.
Noting that $:\phi(f):_{\Psi}=\phi(f)$, it follows from this identity
that normal-ordered elements, in fact, comprise a basis of the CCR
algebra in the sense that any element of $\mathcal{A}(M,g_{ab})$
can be expressed via (\ref{eq:Wicks thm CCR}) as a linear combination
of terms of the form (\ref{eq:def N-O in A}) (see \eqref{eq:Wick expansion of product of linear fields}
for an explicit formula).

It is useful to view $:\phi(f_{1})\cdots\phi(f_{n}):_{\Psi}$ as mapping
$t_{n}=f_{(1}\otimes f_{2}\otimes\cdots\otimes f_{n)}$ into $\mathcal{A}(M,g_{ab})$.
We write 
\begin{equation}
W_{n}(t_{n})=:\phi(f_{1})\cdots\phi(f_{n}):_{\Psi}=\int_{\times^{n}M}d\mu_{g}(x_{1})\cdots d\mu_{g}(x_{n}):\phi(x_{1})\cdots\phi(x_{n}):_{\Psi}t_{n}(x_{1},\dots,x_{n})\label{wtn}
\end{equation}
where $d\mu_{g}(x)\equiv d^{D}x\sqrt{-g(x)}$. Similarly, denote by
$u_{m}\equiv f_{(n+1}\otimes f_{n+2}\otimes\cdots\otimes f_{n+m)}$
another symmetrized tensor product of smooth test functions. In this
notation, we may write eq.~(\ref{eq:Wicks thm CCR}) as 
\begin{equation}
W_{n}(t_{n})W_{m}(u_{m})=\sum_{k\le\min(n,m)}W_{n+m-2k}(t_{n}\otimes_{k}u_{m}),\label{eq:Wick product rule}
\end{equation}
where we define, for $n,m\ge k$, 
\begin{align}
(t_{n} & \otimes_{k}u_{m})(x_{1},\dots x_{n+m-2k})\nonumber \\
 & \equiv\sum_{\pi\in\Pi_{k}}\frac{n!m!}{k!((n-k)!)^{2}((m-k)!)^{2}}\int_{\times^{2k}M}d\mu_{g}(y_{1})\cdots d\mu_{g}(y_{2k})\braket{\phi(y_{1})\phi(y_{2})}_{\Psi}\cdots\braket{\phi(y_{2k-1})\phi(y_{2k)})}_{\Psi}\times\nonumber \\
 & \times t_{n}(y_{1},y_{3},\dots,y_{2k-1},x_{\pi(1)},\dots,x_{\pi(n-k)})u_{m}(y_{2},y_{4},\dots,y_{2k},x_{\pi(n-k+1)},\dots,x_{\pi(n+m-2k)}),\label{prodrule}
\end{align}
where $\Pi_{k}$ denotes any permutation of $\{1,\dots,n+m-2k\}$
such that $\pi(1)<\pi(2)<\cdots<\pi(n-k)$ and $\pi(n-k+1)<\pi(n-k+2)<\cdots<\pi(n+m-2k)$.
Note \eqref{prodrule} is symmetric in $(x_{1},\dots,x_{n+m-2k})$.

We now require $\Psi$ to be a Hadamard state, i.e a state whose two-point
distribution $\Psi_{2}(f_{1},f_{2})\equiv\braket{\phi(f_{1})\phi(f_{2})}_{\Psi}$
has a wavefront set of the form: 
\begin{equation}
\text{WF}[\Psi_{2}]=\left\{ (x_{1},k_{1};x_{2},k_{2})\in\times^{2}(T^{*}M\backslash Z^{*}M)|(x_{1},k_{1})\sim(x_{2},-k_{2}),k_{1}\in\dot{V}_{x_{1}}^{+}\right\} .\label{eq:Had WF}
\end{equation}
Here $\dot{V}_{x}^{\pm}$ denotes the boundary of the future/past
lightcone of $x$ and $(x_{1},k_{1})\sim(x_{2},k_{2})$ if $x_{1}$
and $x_{2}$ can be joined by a null-geodesic with respect to which
the covectors $k_{1},k_{2}$ are cotangent and coparallel. In any
convex normal neighborhood, the two-point distribution of a Hadamard
state takes the form\footnote{For states of the CCR algebra $\mathcal{A}$, the equivalence of the
microlocal spectral version (\ref{eq:Had WF}) of the Hadamard condition
and the position-space version (\ref{eq:Hstate}) was established
by Radzikowski in \cite[Theorem 5.1]{Rad_microlocal_Hadamard_cond}.}: 
\begin{align}
\Psi_{2}(x_{1},x_{2}) & =\frac{U(x_{1},x_{2})}{\left[\sigma(x_{1},x_{2})+2i0^{+}(T(x_{1})-T(x_{2}))+(0^{+})^{2}\right]^{D/2-1}}+\label{eq:Hstate}\\
 & \qquad+V(x_{1},x_{2})\log\ell^{-2}\left[\sigma(x_{1},x_{2})+2i0^{+}(T(x_{1})-T(x_{2}))+(0^{+})^{2}\right]+W_{\Psi}(x_{1},x_{2}),\nonumber 
\end{align}
where $T$ is any local time function; $\sigma$ is the (signed) squared
geodesic distance\footnote{i.e., $\sigma$ is equal to twice the ``Synge bi-scalar/world function''. }
between points $x_{1}$ and $x_{2}$; $\ell$ is an arbitrary length
scale; and $U$, $V$ and $W_{\Psi}$ are smooth symmetric functions.
If $D$ is odd, then $V=0$. Moreover, $U$ and $V$ are independent
of the Hadamard state $\Psi$ and are locally and covariantly determined
by the Hadamard recursion relations\footnote{More precisely, all of the derivatives of $U$ and $V$ at coincidence
$x_{1}=x_{2}$ are uniquely as well as locally and covariantly determined
by the fact that $K\Psi_{2}=\text{smooth}$, with the Klein-Gordon
operator $K$, see eq.~(\ref{eq:K-G op}), acting on either variable.}. It is known that there exist Hadamard states on $\mathcal{A}(M,g_{ab})$
for any globally-hyperbolic spacetime $(M,g_{ab})$.

Thus far, we have merely rewritten the product rules of $\mathcal{A}(M,g_{ab})$
in terms of normal-ordered products. The enlargement of the algebra
$\mathcal{A}(M,g_{ab})$ to the desired algebra $\mathcal{W}(M,g_{ab})$
is accomplished by recognizing that for Hadamard states, eq.~(\ref{prodrule})
makes sense not merely when $t_{n}$ and $u_{m}$ are products of
test functions but also when they are distributions of the following
type: Denote by $\mathcal{V}_{n}(M,g_{ab})$ the set of all elements
of the (product) cotangent bundle $\times^{n}T^{*}M$ that are entirely
contained within either the future or past lightcones, 
\[
\mathcal{V}_{n}(M,g_{ab})\equiv\left\{ (x_{1},k_{1};x_{2},k_{2};\dots;x_{n}k_{n})\in\times^{n}T^{*}M|(k_{i}\in V_{x_{i}}^{+},\forall i\in n)\text{ or }(k_{i}\in V_{x_{i}}^{-},\forall i\in n)\right\} ,
\]
Let $\mathcal{E}'(M,g_{ab})$ denote the space of compactly-supported
symmetric distributions $\mathcal{D}'_{0}(\times^{n}M)$ whose wavefront
sets do not intersect $\mathcal{V}_{n}(g_{ab},M)$, 
\[
\mathcal{E}'(\times^{n}M,g_{ab})\equiv\left\{ t\in\mathcal{D}'_{0}(\times^{n}M)|n\in\mathbb{N}\text{ and }\text{WF}(t)\cap\mathcal{V}_{n}(g_{ab},M)=\emptyset\right\} .
\]
Then formula (\ref{prodrule}) is well defined whenever $t_{n}$ and
$u_{m}$ are distributions in $\mathcal{E}'$. This means that we
can extend the algebra $\mathcal{A}(M,g_{ab})$ to an algebra $\mathcal{W}(M,g_{ab})$
generated by quantities of the form $W(t_{n})$ for all $t_{n}\in\mathcal{E}'$,
with product rule given by eq.~(\ref{eq:Wick product rule}). An
example of such a distribution in $\mathcal{E}'$ is $t_{n}=f(x_{1})\delta(x_{1},\dots,x_{n})$.
By eq.~(\ref{wtn}), $W(t_{n})$ corresponds to $:\phi^{n}:_{\Psi}(f)$.
Thus, $\mathcal{W}(M,g_{ab})$ includes elements corresponding to
the normal-ordered powers of the field. More generally, it includes
all normal-ordered monomials, $:\nabla_{\alpha_{1}}\phi\cdots\nabla_{\alpha_{n}}\phi:_{\Psi}(f^{\alpha_{1}\cdots\alpha_{n}})$,
where the $\alpha_{i}$ denote multi-spacetime-indices and $f^{\alpha_{1}\cdots\alpha_{n}}$
denotes a test tensor field. For notational convenience, we will typically
suppress the multi-indices of $f^{\alpha_{1}\cdots\alpha_{n}}$ and
write $:\nabla_{\alpha_{1}}\phi\cdots\nabla_{\alpha_{n}}\phi:_{\Psi}(f)$,
with it always being understood that $f$ is a tensor field dual to
the tensor $:\nabla_{\alpha_{1}}\phi\cdots\nabla_{\alpha_{n}}\phi:_{\Psi}$.
Note that all Hadamard states on $\mathcal{A}(M,g_{ab})$ can be naturally
extended to states on $\mathcal{W}(M,g_{ab})$. Furthermore, it can
be shown that the only continuous states on $\mathcal{W}(M,g_{ab})$
are Hadamard states \cite{Hollands_Ruan}.

The above construction of $\mathcal{W}(M,g_{ab})$ made use of a choice
of Hadamard state $\Psi$. However, it is not difficult to show that,
as an abstract algebra, $\mathcal{W}(M,g_{ab})$ does not depend on
the choice of $\Psi$ \cite[see Lemma 2.1]{HW_local_Wick_poly}. Nevertheless,
normal-ordered quantities such as $:\phi^{n}:_{\Psi}$ do depend on
the choice of $\Psi$ for any $n>1$, i.e., $:\phi^{n}:_{\Psi'}\neq:\phi^{n}:_{\Psi}$
if $\Psi'\neq\Psi$. Which quantity should represent the true field
observable $\phi^{n}$ and other Wick monomials? In fact, when $n>1$,
$:\phi^{n}:_{\Psi}$ for any choice of Hadamard state $\Psi$ is not
a suitable candidate to represent $\phi^{n}$ since it does not satisfy
the requirement of being locally and covariantly defined. Following
\cite{HW_local_Wick_poly,HW_Conservation_Stress-energy}, we determine
the Wick monomials by imposing the requirements (``axioms'') on
their definition. Existence of a definition of Wick monomials satisfying
these axioms can then be proven. We will consider the allowed freedom
(i.e., non-uniqueness) in the definition of the Wick monomials in
the next section.

The following are our axioms\footnote{These axioms differ from the ones originally given in \cite{HW_local_Wick_poly}
in that the Leibniz rule W4 and the conservation of stress-energy
W8 have been added as in \cite{HW_Conservation_Stress-energy}. In
addition, the analytic dependence condition of \cite{HW_local_Wick_poly,HW_Conservation_Stress-energy}
has been replaced by the joint smoothness condition of \cite{KM_analytic_dep_unnecessary,KM_Wick_poly}.} for Wick monomials:

\subparagraph{$\protect\local$ Local and covariant.}

The Wick monomials are required to be ``local and covariant'' in
the following sense. Let $(M,g_{ab})$ and $(M',g_{ab}')$ denote
two globally-hyperbolic spacetimes. Suppose $\psi:M\to M'$ is an
isometric embedding (i.e., $g_{ab}=\psi^{*}g_{ab}'$, where $\psi^{*}$
denotes the pullback by $\psi$) that also is causality-preserving:
i.e., $\psi(x_{1}),\psi(x_{2})\in M'$ is connected by a causal curve
 only if $x_{1},x_{2}\in M$ is connected by a causal curve. Then,
as shown in \cite[Lemma 3.1]{HW_local_Wick_poly}, there is a canonical
injective unital $*$-homomorphism $\alpha_{\psi}:\mathcal{W}(M,g_{ab})\to\mathcal{W}(M',g_{ab}')$.
We demand that the definition of any Wick monomial $\Phi_{A}(f)=(\nabla_{\alpha_{1}}\phi\cdots\nabla_{\alpha_{n}}\phi)(f)$
be such that, under this homomorphism, we have $\alpha_{\psi}\left[\text{\ensuremath{\Phi_{A}}}(f)\right]\mapsto\Phi_{A}(\psi_{*}f)$,
where $f$ is a test tensor field on $M$ dual to $\Phi_{A}$ and
$\psi_{*}f$ is the push-forward of $f$ via $\psi$.

\subparagraph{$\protect\spectrum$ Smoothness and joint smoothness.}

For any Wick monomial $\Phi_{A}$ and for any Hadamard state $\braket{\cdot}_{\Psi}$,
we require that $\text{WF}[\braket{\Phi_{A}}_{\Psi}]=\emptyset$,
i.e., that $\braket{\Phi_{A}(x)}_{\Psi}$ is smooth. Furthermore,
we require that this quantity be jointly smooth in $x$, the spacetime
metric, and the parameters $m^{2}$ and $\xi$. To define this notion,
we must first allow $m^{2}$ and $\xi$ to have spacetime dependence.
We then consider one parameter variations $g_{ab}(s_{1})$, $m^{2}(s_{2})$,
and $\xi(s_{3})$ in a compact spacetime region $\mathcal{R}$, such
that $(M,g_{ab}(s_{1}))$ is globally hyperbolic for all $s_{1}$.
As shown in \cite[Lemma 4.1]{HW_local_Wick_poly}, we may naturally
identify the algebra $\mathcal{W}$ associated with $(g_{ab}(s_{1}),m^{2}(s_{2}),\xi(s_{3}))$
with the algebra associated with $(g_{ab}(0),m^{2}(0),\xi(0))$ by
identifying these algebras on a Cauchy surface lying outside the future
of $\mathcal{R}$. Consequently, we may identify a Hadamard state
$\braket{\cdot}_{\Psi}$ on the algebra for $(g_{ab}(0),m^{2}(0),\xi(0))$
with a Hadamard state on the algebra associated with $(g_{ab}(s_{1}),m^{2}(s_{2}),\xi(s_{3}))$.
For any Hadamard state $\braket{\cdot}_{\Psi}$, for any Wick monomial
$\Phi$, and for any family $(g_{ab}(s_{1}),m^{2}(s_{2}),\xi(s_{3}))$
as above, we require that $\braket{\Phi_{A}[g_{ab}(s_{1}),m^{2}(s_{2}),\xi(s_{3})](x)}_{\Psi}$
be jointly smooth in $(x,s_{1},s_{2},s_{3})$.

\subparagraph{$\protect\commutator$ Commutator. }

The commutator of any Wick monomial $\W A=\nabla_{\alpha_{1}}\phi\cdots\nabla_{\alpha_{n}}\phi$
with the fundamental field $\phi$ is given by, 
\[
\left[(\nabla_{\alpha_{1}}\phi\cdots\nabla_{\alpha_{n}}\phi)(f_{1}),\phi(f_{2})\right]=i\sum_{i=1}^{n}(\nabla_{\alpha_{1}}\phi\cdots\widehat{\nabla_{\alpha_{i}}\phi}\cdots\nabla_{\alpha_{n}}\phi)(f_{1})\Delta\left((-1)^{|\alpha_{i}|}\nabla_{\alpha_{i}^{T}}f_{1},f_{2}\right),
\]
where $\Delta=\Delta^{\text{adv}}-\Delta^{\text{ret}}$ is the advanced
minus retarded Green's function, $\widehat{\nabla_{\alpha_{i}}\phi}$
denotes the omission of the $\nabla_{\alpha_{i}}\phi$ factor and
for the multi-index $\alpha\equiv a_{1}a_{2}\cdots a_{|\alpha|}$,
we use the notation $\alpha^{T}\equiv a_{|\alpha|}a_{|\alpha|-1}\cdots a_{1}$.

\subparagraph{$\protect\Leibniz$ Leibniz rule.}

Any Wick monomial $\W A=\nabla_{\alpha_{1}}\phi\cdots\nabla_{\alpha_{n}}\phi$
must satisfy the Leibniz rule in the sense that 
\begin{equation}
(\nabla_{\alpha_{1}}\phi\cdots\nabla_{\alpha_{n}}\phi)\left(-\nabla_{a}f\right)=\left((\nabla_{a}\nabla_{\alpha_{1}})\phi\cdots\nabla_{\alpha_{n}}\phi\right)(f)+\dots+\left(\nabla_{\alpha_{1}}\phi\cdots(\nabla_{a}\nabla_{\alpha_{n}})\phi\right)(f).\label{eq:symmetrized Leibniz rule}
\end{equation}
Here, the left side of this equation is the distributional derivative
of $\Phi$ whereas the right side is what one would obtain by applying
the Leibniz rule to the classical expression $\Phi_{A}=\nabla_{\alpha_{1}}\phi\cdots\nabla_{\alpha_{n}}\phi$.

\subparagraph{$\protect\Hermiticity$ Hermiticity.}

All Wick monomials are required to be Hermitian in the sense that,
\[
(\nabla_{\alpha_{1}}\phi\cdots\nabla_{\alpha_{n}}\phi)(f)^{*}=(\nabla_{\alpha_{1}}\phi\cdots\nabla_{\alpha_{n}}\phi)(\overline{f}).
\]

\subparagraph{$\protect\symmetry$ Symmetry. }

Any Wick monomial is required to be symmetric under interchange of
the fields\textemdash i.e., 
\[
(\nabla_{\alpha_{\pi(1)}}\phi\cdots\nabla_{\alpha_{\pi(n)}}\phi)(f)=(\nabla_{\alpha_{1}}\phi\cdots\nabla_{\alpha_{n}}\phi)(f),
\]
for all permutations $\pi$ of $\{1,\dots,n\}$.

\subparagraph{$\protect\scaling$ Scaling. }

For $\lambda>0$, let $\sigma_{\lambda}:\mathcal{W}(M,\lambda^{-2}g_{ab},\lambda^{2}m^{2},\xi)\to\mathcal{W}(M,g_{ab},m^{2},\xi)$
be the canonical $*$-isomorphism defined in \cite[Lemma 4.2]{HW_local_Wick_poly}.
The ``scaling dimension'' $d_{A}$ of any local, covariant field
$\Phi_{A}$ is defined to be the smallest real number $\delta$ such
that 
\[
\lim_{\lambda\to0^{+}}\lambda^{(D-\delta)}\sigma_{\lambda}\left[\W A[\lambda^{-2}g_{ab},\lambda^{2}m^{2},\xi]\right](f)=0,
\]
for all $(g_{ab},m^{2},\xi)$. The factor of $\lambda^{D}$ accounts
for the fact that the volume element scales as $d\mu_{\lambda^{-2}g}=\lambda^{D}d\mu_{g}$.
Define: 
\begin{equation}
[A]\equiv\frac{(D-2)}{2}\times[A]_{\phi}+[A]_{\nabla},\label{eq:[A]}
\end{equation}
where $[A]_{\phi}$ and $[A]_{\nabla}$ denote, respectively, the
number of $\phi$-factors and the number of covariant derivatives
in $\W A$ (i.e., for $\W A\equiv(\nabla_{\alpha_{1}}\phi\cdots\nabla_{\alpha_{p}}\phi)$,
$[A]_{\phi}=p$ and $[A]_{\nabla}=\sum_{i=1}^{p}|\alpha_{i}|$). We
require the Wick monomial $\Phi_{A}$ to have scaling dimension, 
\begin{align}
d_{A}=[A] & +2\times(\text{\# of curvature factors})+(\text{\# of ``up'' indices})-(\text{\# of ``down'' indices})+\label{eq:d_A}\\
 & +2\times(\text{\# of \ensuremath{m^{2}}-factors\ensuremath{)}}\nonumber 
\end{align}
For example, $(\nabla_{a}\phi\nabla_{b}\nabla_{c}\phi)$ has scaling
dimension $D-2$, whereas $R^{ab}\phi$ has scaling dimension $D/2+3$.
We further require that $\Phi_{A}$ scale homogeneously up to logarithms:
i.e., there must exist finite $N$ such that, 
\[
\frac{\partial^{N}}{\partial(\log\lambda)^{N}}\left[\lambda^{(D-d_{A})}\sigma_{\lambda}\left[\W A[\lambda^{-2}g_{ab},\lambda^{2}m^{2},\xi\right](f)\right]=0.
\]

\subparagraph{$\protect\conservation$ Conservation of stress-energy.}

The stress-energy tensor, $T_{ab}(f)\in\mathcal{W}(M,g_{ab})$, is
given by 
\begin{align}
T_{ab}=\left(1-2\xi\right) & (\nabla_{a}\phi\nabla_{b}\phi)+\left(2\xi-\frac{1}{2}\right)g_{ab}(\nabla^{c}\phi\nabla_{c}\phi)+\label{eq:Tab}\\
+2\xi g_{ab} & (\phi\nabla^{c}\nabla_{c}\phi)-2\xi(\phi\nabla_{a}\nabla_{b}\phi)+\left(\xi G_{ab}-\frac{1}{2}m^{2}g_{ab}\right)\phi^{2},\nonumber 
\end{align}
where $G_{ab}\equiv R_{ab}-\frac{1}{2}g_{ab}R$ is the Einstein tensor.
We require that $T_{ab}$ is divergence free, 
\begin{equation}
0=T_{ab}(-\nabla^{a}f)=-(\nabla_{b}\phi K\phi)(f),\label{eq:conservation of Tab}
\end{equation}
where $K=K[g_{ab},m^{2},\xi]$ is the Klein-Gordon operator, eq.~(\ref{eq:K-G op}),
and the second equality in (\ref{eq:conservation of Tab}) follows
straightforwardly from differentiating (\ref{eq:Tab}) and using the
Leibniz and symmetry axioms. 
\begin{rem}
Note that even in flat spacetime where $G_{ab}=0$, the stress-energy
tensor (\ref{eq:Tab}) has nontrivial dependence on the curvature
coupling $\xi$. However, the conservation constraint (\ref{eq:conservation of Tab})
is independent of $\xi$ in any region with vanishing Ricci scalar
curvature, since $K[g_{ab},m^{2},\xi=0]=K[g_{ab},m^{2},\xi]$ at any
spacetime point $x$ where $R(x)=0$. 
\end{rem}
If we wished to define Wick monomials by normal ordering with respect
to a Hadamard state, we would have to choose a Hadamard state $\Psi(M,g_{ab})$
for each globally hyperbolic spacetime $(M,g_{ab})$. However, as
we have already mentioned above, it can be shown \cite{HW_local_Wick_poly}
that no choice of $\Psi(M,g_{ab})$ can give rise to a prescription
for Wick monomials that satisfies the local and covariant condition,
W1. Nevertheless, a construction of Wick monomials satisfying all
of our requirements W1-W8 can be given by normal ordering with respect
to a locally and covariantly constructed Hadamard parametrix, $H(x_{1},x_{2})$,
rather than a Hadamard state. We define $H(x_{1},x_{2})$ in a sufficiently
small neighborhood of the diagonal $x_{1}=x_{2}$ by, 
\begin{align}
H(x_{1},x_{2})=\  & \frac{U(x_{1},x_{2})}{\left[\sigma(x_{1},x_{2})+2i0^{+}(T(x_{1})-T(x_{2}))+(0^{+})^{2}\right]^{D/2-1}}+\label{eq:H}\\
 & +V(x_{1},x_{2})\log\ell^{-2}\left[\sigma(x_{1},x_{2})+2i0^{+}(T(x_{1})-T(x_{2}))+(0^{+})^{2}\right],\nonumber 
\end{align}
where the quantities appearing in this equation are defined as in
eq.~(\ref{eq:Hstate}). Thus, $H(x_{1},x_{2})$ differs from the
two-point function of any Hadamard state, $\Psi$, by a state-dependent,
smooth, symmetric function $W_{\Psi}(x_{1},x_{2})$. We refer to $H(x_{1},x_{2})$
as a ``parametrix'' because, although it does \emph{not} satisfy
the Klein-Gordon equation in either variable, its failure to satisfy
the Klein-Gordon equation is smooth. We define the normal-ordered
product of field operators with respect to $H$ by, 
\begin{equation}
:\phi(x_{1})\cdots\phi(x_{n}):_{H}\equiv\sum_{P}(-1)^{|P|}\prod_{(i,j)\in P}H(x_{i},x_{j})\prod_{k\in\{1,\dots,n\}\backslash P}\phi(x_{k}),\label{hadnodef}
\end{equation}
i.e., by the same formula as in eq.~(\ref{eq:def N-O in A}) but
with the two-point function, $\braket{\phi(x_{i})\phi(x_{j})}_{\Psi}$,
of a state, $\Psi$, replaced by the Hadamard parametrix $H(x_{i},x_{j})$.
Note that the Hadamard normal-ordered elements satisfy Wick's theorem
(\ref{eq:Wick product rule}) with, again, $\braket{\phi(x_{i})\phi(x_{j})}_{\Psi}$
replaced by $H(x_{i},x_{j})$ in eq.~(\ref{prodrule}). Using $H$,
we define the Wick monomial corresponding to $\nabla_{\alpha_{1}}\phi\cdots\nabla_{\alpha_{n}}\phi$
by, 
\begin{align}
\H A{}(f) & \equiv(\nabla_{\alpha_{1}}\phi\cdots\nabla_{\alpha_{n}}\phi)_{H}(f)\label{eq:H normal ordered}\\
 & \equiv\int_{\times^{n+1}M}d\mu_{g}(y)d\mu_{g}(x_{1})\cdots d\mu_{g}(x_{n}):\phi(x_{1})\cdots\phi(x_{n}):_{H}t_{n+1}[f](y,x_{1},\dots,x_{n}),\nonumber 
\end{align}
and
$t_{n+1}[f]$ given by, 
\begin{equation}
t_{n+1}[f](y,x_{1},\dots,x_{n})=f^{\alpha_{1}\cdots\alpha_{n}}(y)(-1)^{[A]_{\nabla}}\nabla_{\alpha_{1}^{T}}^{(x_{1})}\cdots\nabla_{\alpha_{n}^{T}}^{(x_{n})}\delta(y,x_{1},\dots,x_{n}),\label{eq:t distribution}
\end{equation}
recalling $[A]_{\nabla}\equiv\sum_{i=1}^{n}|\alpha_{i}|$. In contrast
to normal ordering defined with respect to a Hadamard state, the prescription
(\ref{eq:H normal ordered}) for $\Phi_{H}$ given by normal ordering
with respect to the locally and covariantly constructed Hadamard parametrix
eq.~(\ref{eq:H}) satisfies requirement W1. It also satisfies \cite{Hollands_perturbative_OPE_CS,KM_analytic_dep_unnecessary}
requirements W2-W7 for Wick monomials.

However, the failure of $H$ to be an exact solution of the Klein-Gordon
wave equation implies this prescription generally does not satisfy
requirement W8, 
\begin{equation}
(\nabla_{b}\phi K\phi)_{H}(f)=\int d\mu_{g}(y)f(y)\nabla_{b}^{(x_{1})}K_{x_{2}}H(x_{1},x_{2})|_{x_{1},x_{2}=y}\ne0.\label{eq:failure to satisfy W8}
\end{equation}
Odd dimensions are an exception: For $D$ odd, formula (\ref{eq:H})
contains only half-integer powers of $\sigma(x_{1},x_{2})$, so it
follows that for $U(x_{1},x_{2})$ smooth, $H(x_{1},x_{2})$ is a
parametrix of the Klein-Gordon equation only if, 
\begin{equation}
K_{x_{2}}H(x_{1},x_{2})|_{x_{1},x_{2}=y}=0.\label{eq:KH=0 for D odd}
\end{equation}
Furthermore, it can be shown \cite[Lemma 2.1]{Moretti_comments_stress-energy}
that, 
\begin{equation}
\nabla_{b}^{(x_{1})}K_{x_{2}}H(x_{1},x_{2})|_{x_{1},x_{2}=y}=\frac{D}{2(D+2)}\nabla_{b}^{(y)}\left[K_{x_{2}}H(x_{1},x_{2})\right]_{x_{1},x_{2}=y},\label{eq:nabla K H}
\end{equation}
so (\ref{eq:KH=0 for D odd}) implies the left-hand side of
(\ref{eq:failure to satisfy W8}) does, in fact, vanish and, thus,
W8 is satisfied in all odd dimensions.

In even dimensions, however, $K_{x_{2}}H(x_{1},x_{2})|_{x_{1},x_{2}=y}$
yields a curvature scalar which is non-vanishing in general spacetimes
and, thus, normal-ordering with respect to the parametrix (\ref{eq:H})
fails to produce Wick fields satisfying the conservation axiom W8.
Nevertheless, we prove in Appendix \ref{sec:existence of Q } that
for $D>2$, there exists a smooth symmetric function $Q(x_{1},x_{2})$
which is locally and covariantly defined for $x_{1}=x_{2}$ such that
\begin{equation}
\nabla_{b}^{x_{1}}K_{x_{2}}H(x_{1},x_{2})|_{x_{1},x_{2}=y}=-\nabla_{b}^{x_{1}}K_{x_{2}}Q(x_{1},x_{2})|_{x_{1},x_{2}=y},.\label{eq:nabla K H equals -nabla K Q}
\end{equation}
Furthermore, $Q$ is smooth in $(m^{2},\xi)$ and scales as, 
\begin{equation}
Q[\lambda^{-2}g_{ab},\lambda^{2}m^{2},\xi]=\lambda^{(D-2)}Q[g_{ab},m^{2},\xi],\label{eq:Q scaling}
\end{equation}
in a sufficiently small neighborhood of $x_{1},x_{2}=y$. Therefore,
normal-ordering instead with respect to the new Hadamard parametrix,
\[
H'\equiv H+Q,
\]
will give a construction of Wick fields satisfying the axioms W1-W8.

It will be understood below that, unless otherwise stated, we are
always normal-ordering with respect to a Hadamard parametrix $H$
which is smooth in $(m^{2},\xi)$, satisfies 
\begin{equation}
\nabla_{b}^{(x_{1})}K_{x_{2}}H(x_{1},x_{2})|_{x_{1},x_{2}=y}=0,\label{eq:H conservation condition-1}
\end{equation}
and scales homogeneously up to logarithms, 
\begin{equation}
\lambda^{-(D-2)}H[\lambda^{-2}g_{ab},\lambda^{2}m^{2},\xi]=H[g_{ab},m^{2},\xi]+V[g_{ab},m^{2},\xi]\log\lambda^{2}.\label{eq:scaling H'}
\end{equation}
(Recall $V=0$ for $D$ odd, so $H$ scales exactly homogeneously
in odd spacetime dimensions.) Thus, for any $D>2$, Hadamard normal
ordering yields a prescription for defining Wick monomials that satisfies
W1-W8. For $D=2$, no such $Q$ exists, and condition W8 cannot be
satisfied by any prescription that satisfies W1-W7 \cite[see Subsection 3.2]{HW_Conservation_Stress-energy}.
However, Hadamard normal ordering satisfies W1-W7.

We turn our attention now to the characterization of the non-uniqueness
of prescriptions satisfying W1-W8 (or W1-W7 for $D=2$).

\subsection{Uniqueness of Wick monomials \label{subsec:uniq Wick poly and novel reparm}}

In the previous subsection, we imposed conditions W1-W8 on the definition
of Wick monomials and gave a prescription based on ``Hadamard normal-ordering''
which satisfies these requirements (or requirements W1-W7 for $D=2$).
This prescription is not unique. In this subsection, we will show
that the difference between any two prescriptions $\Phi_{A}$ and
$\widetilde{\Phi}_{A}$ for Wick monomials satisfying W1-W8 (or W1-W7
for $D=2$) are described by a ``mixing matrix'' $\mathcal{Z}$
such that 
\begin{equation}
\widetilde{\Phi}_{A}(x)=\sum_{B}\mathcal{Z}{}_{A}^{B}(x)\Phi_{B}(x).\label{def:mixing matrix Z}
\end{equation}
Theorem \ref{thm:Wick uniquenes Z} below explicitly gives the general
form of $\mathcal{Z}$ which, thereby, characterizes the freedom to
modify any prescription, such as the Hadamard prescription of the
previous subsection.

It will be convenient to use the following notation for $\mathcal{Z}_{A}^{B}$.
An arbitrary Wick monomial is of the form $\Phi_{A}=\nabla_{\alpha_{1}}\phi\cdots\nabla_{\alpha_{p}}\phi$
and thus is characterized by the multi-indices $\alpha_{1},\dots,\alpha_{p}$.
For $\widetilde{\Phi}_{A}=\tildeXL{\nabla_{\alpha_{1}}\phi\cdots\nabla_{\alpha_{p}}\phi}$
and $\Phi_{B}=\nabla_{\beta_{1}}\phi\cdots\nabla_{\beta_{q}}\phi$,
we represent $\mathcal{Z}_{A}^{B}$ as 
\begin{equation}
\mathcal{Z}_{A}^{B}=\mathcal{Z}_{\alpha_{1}\cdots\alpha_{p}}^{\beta_{1}\cdots\beta_{q}}.\label{eq:Z multi-index notation}
\end{equation}
Each multi-index, $\alpha$, is itself a product of spacetime indices,
$\alpha=a_{1}\cdots a_{|\alpha|}$, so we may, in turn, write $\mathcal{Z}$
as a spacetime tensor field 
\begin{equation}
\mathcal{Z}_{\alpha_{1}\cdots\alpha_{p}}^{\beta_{1}\cdots\beta_{q}}=\mathcal{Z}_{\{a_{1,1}\cdots a_{1,|\alpha_{1}|}\}\cdots\{a_{n,1}\cdots a_{n,|\alpha_{p}|}\}}^{\{b_{1,1}\cdots b_{1,|\beta_{1}|}\}\cdots\{b_{k,1}\cdots b_{k,|\beta_{q}|}\}}\label{eq:Z spacetime index notation}
\end{equation}
In this notation, we enclose the spacetime indices corresponding to
any given multi-index with a curly bracket. If any multi-index is
``empty''\textemdash i.e., if any factor of $\phi$ in the corresponding
Wick monomial has no derivatives acting on it, then we insert a ``$\{0\}$''
as a place-holder. If $q$ is zero, it is understood $\Phi_B=I$ and we
simply write ``$I$'' in the superscripts of \eqref{eq:Z multi-index notation}
and \eqref{eq:Z spacetime index notation} as in examples \eqref{eq:quad uniqueness}-\eqref{eq:Z not ex 2}
below. Similarly, when $p=0$, it is understood $\Phi_A=I$ and we write
``$I$'' in the subscripts of \eqref{eq:Z multi-index notation}
and \eqref{eq:Z spacetime index notation}.

As an example to illustrate this notation, it will follow from the
theorem below that the difference between any two prescriptions for
Wick monomials that are quadratic in $\phi$ will be given by a multiple
of the identity element, $I$. In our notation, this would be expressed
as 
\begin{equation}
(\tildeXL{\nabla_{\alpha_{1}}\phi\nabla_{\alpha_{2}}\phi})(x)=\sum_{\beta_{1},\beta_{2}}\mathcal{Z}_{\alpha_{1}\alpha_{2}}^{\beta_{1}\beta_{2}}(x)(\nabla_{\beta_{1}}\phi\nabla_{\beta_{2}}\phi)(x)+\mathcal{Z}_{\alpha_{1}\alpha_{2}}^{I}(x)I,\label{eq:quad uniqueness}
\end{equation}
where $\mathcal{Z}_{\alpha_{1}\alpha_{2}}^{\beta_{1}\beta_{2}}=\delta_{(\alpha_{1}}^{\beta_{1}}\delta_{\alpha_{2})}^{\beta_{2}}$
and $\delta_{\alpha}^{\beta}$ is the Kronecker delta for the multi-indices
defined by $\delta_{\alpha}^{\beta}=1$ if the multi-indices $\alpha$
and $\beta$ coincide and zero otherwise. As particular examples of
(\ref{eq:quad uniqueness}), we have 
\begin{equation}
(\tildeXL{\nabla_{a}\phi\nabla_{b}\nabla_{c}\phi})(x)=(\nabla_{a}\phi\nabla_{b}\nabla_{c}\phi)(x)+\mathcal{Z}_{\{a\}\{bc\}}^{I}(x)I,\label{eq:Z not ex 1}
\end{equation}
whereas 
\begin{equation}
(\tildeXL{\phi\nabla_{a}\nabla_{b}\phi})(x)=(\phi\nabla_{a}\nabla_{b}\phi)(x)+\mathcal{Z}_{\{0\}\{ab\}}^{I}(x)I.\label{eq:Z not ex 2}
\end{equation}

With this notation established, we may state our main result in the
following theorem. 
\begin{thm}
\label{thm:Wick uniquenes Z} The Wick mixing matrix $\mathcal{Z}_{A}^{B}$
defined in (\ref{def:mixing matrix Z}) is nonzero only when $[B]_{\phi}\le[A]_{\phi}$,
i.e. $q\le p$, and is given in terms of $\mathcal{Z}_{A}^{I}$ by,
\begin{equation}
\mathcal{Z}_{\alpha_{1}\cdots\alpha_{p}}^{\beta_{1}\cdots\beta_{q}}=\binom{p}{q}\delta_{(\alpha_{1}}^{\beta_{1}}\cdots\delta_{\alpha_{q}}^{\beta_{q}}\mathcal{Z}_{\alpha_{q+1}\cdots\alpha_{p})}^{I},\label{eq:Wick uniqueness Z}
\end{equation}
where $\binom{p}{q}$ denotes the binomial coefficient. Furthermore,
we have $\mathcal{Z}_{I}^{I}=1$ and $\mathcal{Z}_{\alpha_{1}}^{I}=0$.
For $p\ge2$, each $\mathcal{Z}_{\alpha_{1}\cdots\alpha_{p}}^{I}$
is a real-valued, smooth tensor field of type $(0,\sum_{i=1}^{p}|\alpha_{i}|)$
that is symmetric under permutation $\pi$ of multi-indices, 
\[
\mathcal{Z}_{\alpha_{1}\cdots\alpha_{p}}^{I}=\mathcal{Z}_{\alpha_{\pi(1)}\cdots\alpha_{\pi(p)}}^{I},
\]
and is of the form, 
\begin{equation}
\mathcal{Z}_{A}^{I}=\mathcal{Z}_{A}^{I}[g_{ab},R_{abcd},\dots,\nabla_{(e_{1}\cdots}\nabla_{e_{n})}R_{abcd}(x),m^{2},\xi].\label{eq:Z dependence on theory param}
\end{equation}
where the right side is a jointly smooth function of its arguments
with polynomial dependence on $m^{2}$, $R_{abcd}$, and finitely
many (totally-symmetric) covariant derivatives of $R_{abcd}$. The
$\mathcal{Z}_{A}^{I}$ scale as, 
\begin{equation}
\mathcal{Z}{}_{A}^{I}[\lambda^{-2}g_{ab},\lambda^{2}m^{2},\xi]=\lambda^{d_{A}}\mathcal{Z}{}_{A}^{I}[g_{ab},m^{2},\xi],\label{eq:scaling Z_A^I}
\end{equation}
recalling the definition (\ref{eq:d_A}) of the scaling dimension
$d_{A}$. Furthermore, the tensor fields $\mathcal{Z}_{A}^{I}$ satisfy
the Leibniz condition, 
\begin{equation}
\nabla_{b}^{(x)}\mathcal{Z}_{\alpha_{1}\cdots\alpha_{p}}^{I}(x)=\mathcal{Z}{}_{\{b\alpha_{1}\}\alpha_{2}\cdots\alpha_{p}}^{I}(x)+\mathcal{Z}{}_{a_{1}\{b\alpha_{2}\}\cdots\alpha_{p}}^{I}(x)+\cdots+\mathcal{Z}_{\alpha_{1}\cdots\{b\alpha_{p}\}}^{I}(x),\label{eq:Leibniz rule Z}
\end{equation}
where $\{b\alpha\}\equiv ba_{1}a_{2}\cdots a_{|\alpha|}$ for $\alpha\equiv a_{1}\cdots a_{|\alpha|}$.
In addition, on account of W8, for $D>2$, the tensor fields $\mathcal{Z}_{\{b\}\{ac\}}^{I}$
and $\mathcal{Z}_{\{b\}\{0\}}^{I}$ must satisfy, 
\begin{equation}
g^{ac}\mathcal{Z}_{\{b\}\{ac\}}^{I}=(m^{2}+\xi R)\mathcal{Z}_{\{b\}\{0\}}^{I}.\label{eq:Z conservation constraint}
\end{equation}
Conversely, if $\{\Phi_{B}(x)|B=\beta_{1}\cdots\beta_{q}\}_{q\in\mathbb{N}_{0}}$
are any Wick monomials satisfying W1-W8 (or W1-W7 for $D=2$) and
$\mathcal{Z}{}_{A}^{B}$ satisfy all of the above conditions of this
theorem, then the new prescription $\{\tildeXL{\Phi}_{A}|A=\alpha_{1}\cdots\alpha_{p}\}_{p\in\mathbb{N}_{0}}$
defined by eq.~(\ref{def:mixing matrix Z}) will also satisfy W1-W8
(or W1-W7 for $D=2$). Consequently, the inverse mixing matrix $(\mathcal{Z}^{-1})_{A}^{B}$
satisfies the same properties as $\mathcal{Z}_{A}^{B}$. 
\end{thm}
\emph{Sketch of Proof}: The proof follows \cite[Proof of Theorem 5.1]{HW_local_Wick_poly},
with the main difference being that they did not consider Wick powers
involving derivatives and did not impose requirement $\Leibniz$.
The key first step is to note that if, inductively, the prescription
for Wick monomials involving $q$-factors of $\phi$ has been fixed
for all $q<p$, then the prescription for any Wick monomial with $p$-factors
of $\phi$ is uniquely determined by the commutator condition $\commutator$
up to the addition of a multiple of the identity $I$. In our notation,
this $c$-number multiple is denoted by $\mathcal{Z}_{\alpha_{1}\cdots\alpha_{p}}^{I}$.
In particular, eq.~(\ref{eq:quad uniqueness}) holds for $p=2$.
We then can prove eq.~(\ref{eq:Wick uniqueness Z}) for general $p$
by induction. By condition W6, $\mathcal{Z}_{\alpha_{1}\cdots\alpha_{p}}^{I}$
must be totally symmetric in its multi-indices. By condition $\local$,
$\mathcal{Z}_{\alpha_{1}\cdots\alpha_{p}}^{I}$ must be local and
covariant, and thus must be constructed from the metric and the Riemann
tensor and its derivatives as well as from $m^{2}$ and $\xi$. By
the arguments of \cite{KM_analytic_dep_unnecessary,KM_Wick_poly}
the joint smoothness requirement, $\spectrum$, and the scaling requirement,
$\scaling$, imply polynomial dependence\footnote{The corresponding result was obtained in \cite[Theorem 5.1]{HW_local_Wick_poly}
by imposing an additional analytic variation requirement, which we
do not impose here.} on $m^{2}$, $R_{abcd}$, and finitely many derivatives of $R_{abcd}$.
The remaining properties of $\mathcal{Z}_{\alpha_{1}\cdots\alpha_{p}}^{I}$
follow directly from the axioms. The verification of the converse
is straightforward. 
\begin{rem}
\label{rem:existence of tensor fields for (n,D)} The fact that $\mathcal{Z}_{\alpha_{1}\cdots\alpha_{p}}^{I}$
has polynomial dependence on $m^{2}$, $R_{abcd}$, and finitely many
of its derivatives and must have the scaling behavior stated in the
theorem puts significant constraints on $\mathcal{Z}_{\alpha_{1}\cdots\alpha_{p}}^{I}$.
In particular, (\ref{eq:scaling Z_A^I}) can hold non-trivially only
if $p(D-2)/2$ is even. Hence, $\mathcal{Z}_{\alpha_{1}\cdots\alpha_{p}}^{I}=0$
when $p$ is odd and $D\neq2+4k$ for integer $k$. Furthermore, if
$D$ is odd, then we also have $\mathcal{Z}_{\alpha_{1}\cdots\alpha_{p}}^{I}=0$
whenever $p=2+4k$. 
\end{rem}
\begin{rem}
\label{rem:Z recursion relation} For the purpose of proving Theorem
\ref{prop:gen Wick intermed coef} of Subsection \ref{subsec:gen Wick OPE coef},
it is useful to note the Wick mixing matrices $\mathcal{Z}_{A}^{B}$
satisfy the following recursion relation, for any $r\le q$, 
\begin{equation}
\mathcal{Z}{}_{\alpha_{1}\cdots\alpha_{p}}^{\beta_{1}\cdots\beta_{q}}=\binom{p}{r}\binom{q}{r}^{-1}\delta_{(\alpha_{1}}^{\beta_{1}}\cdots\delta_{\alpha_{r}}^{\beta_{r}}\mathcal{Z}_{\alpha_{(r+1)}\cdots\alpha_{p})}^{\beta_{(r+1)}\cdots\beta_{q}}.\label{eq:Z recursion rel-1}
\end{equation}
This identity is immediately established by plugging the expression
(\ref{eq:Wick uniqueness Z}) for $\mathcal{Z}_{A}^{B}$ into both
sides of (\ref{eq:Z recursion rel-1}) and noting, 
\[
\binom{p}{r}\binom{q}{r}^{-1}\binom{p-r}{q-r}=\binom{p}{q}.
\]
\end{rem}
We now prove the following result that will enable us to characterize
in a simple and direct manner the freedom in the prescription for
defining Wick monomials specified by Theorem \ref{thm:Wick uniquenes Z}.
This new characterization will be very useful for characterizing the
freedom of the OPE coefficients for products of Wick monomials. 
\begin{prop}
\label{prop:Wick uniqueness Fn} For each $n\geq2$, there exists
a smooth, real-valued function $F_{n}(x_{1},\dots,x_{n};z)$ on some
neighborhood of $\times^{n+1}M$ containing $(z,\dots,z)$ such that
$F_{n}$ is symmetric in $(x_{1},\dots,x_{n})$ and such that the
coefficients $\mathcal{Z}_{\alpha_{1}\cdots\alpha_{n}}^{I}$ of eq.
(\ref{eq:Wick uniqueness Z}) are given by, 
\begin{equation}
\mathcal{Z}_{\alpha_{1}\cdots\alpha_{n}}^{I}(z)=\nabla_{\alpha_{1}}^{(x_{1})}\cdots\nabla_{\alpha_{n}}^{(x_{n})}F_{n}(x_{1},\dots,x_{n};z)|_{x_{1},\dots,x_{n}=z}.\label{eq:Z=deriv Fn}
\end{equation}
Furthermore, $F_{n}$ satisfy, 
\begin{equation}
\left[\nabla_{\alpha_{1}}^{(x_{1})}\cdots\nabla_{\alpha_{n}}^{(x_{n})}\nabla_{\beta}^{(z)}F_{n}(x_{1},\dots,x_{n};z)\right]_{x_{1},\dots,x_{n}=z}=0.\label{eq:Fm's symmetric independence of x on diagonal}
\end{equation}
\begin{proof}[Sketch of proof.]
Let $x$ be in a normal neighborhood of $z\in M$ and let $\sigma(x,z)$
denote the (signed) squared geodesic distance between $z$ and $x$.
Let 
\begin{equation}
\sigma_{a}(x,z)\equiv\frac{1}{2}\nabla_{a}^{(z)}\sigma(x,z).
\end{equation}
Note that in flat spacetime, in global inertial coordinates, we have
\begin{equation}
\sigma^{\mu}(x,z)=-(x^{\mu}-z^{\mu}).
\end{equation}
Let $f:M\to\mathbb{R}$ be smooth at $z$. Then the covariant Taylor
expansion of $f$ at $z$ is given by \cite[see ``Addendum to chapter 4: derivation of covariant Taylor expansions'']{Barvinsky_Vilkovisky_covariant_Taylor}
\begin{equation}
f(x)\sim\sum_{k}\frac{(-1)^{k}}{k!}\nabla_{a_{1}}\cdots\nabla_{a_{n}}f(x)\big\vert_{x=z}\sigma^{a_{1}}(x,z)\cdots\sigma^{a_{n}}(x,z),\label{eq:f Taylor exp-1}
\end{equation}
where the meaning of this equation is that if the sum on the right
side is taken from $k=0$ to $k=N$, then its difference with the
right side in any coordinates vanishes to order $(x-z)^{N}$. Note
that $\sigma^{a_{1}}\cdots\sigma^{a_{n}}=\sigma^{(a_{1}}\cdots\sigma^{a_{n})}$,
so only the totally-symmetric part of $f$'s covariant derivatives
contribute non-trivially to (\ref{eq:f Taylor exp-1}). We may write
this equation more compactly as, 
\begin{equation}
f(x)\sim\sum_{\beta}\nabla_{\beta}f(x)\big\vert_{x=z}\T[\beta](x,z)\label{eq:covtayser}
\end{equation}
where the sum ranges over all multi-indices $\beta$ and we have written,
\begin{equation}
\T[\{b_{1}\cdots b_{|\beta|}\}](x;z)\equiv\frac{(-1)^{|\beta|}}{|\beta|!}\sigma^{b_{1}}(x,z)\cdots\sigma^{b_{|\beta|}}(x,z).\label{eq:covariant Taylor coef}
\end{equation}
Note that in flat spacetime in global inertial coordinates, we have
\begin{equation}
\T[\{\mu_{1}\cdots\mu_{k}\}](x;z)=\frac{1}{k!}(x^{\mu_{1}}-z^{\mu_{1}})\cdots(x^{\mu_{n}}-z^{\mu_{k}})\label{eq:T^mu1...muk flat spacetime}
\end{equation}
Applying the operator $\nabla_{\alpha}^{(x)}$ to both sides of (\ref{eq:f Taylor exp-1})
and evaluating at $x=z$ should yield the trivial identity, $\nabla_{\alpha}^{(z)}f(z)=\nabla_{\alpha}^{(z)}f(z)$.
This will be the case, in general, if and only if, 
\begin{equation}
\sum_{|\beta|\le|\alpha|}\nabla_{\alpha}^{(x)}\T[\beta](x,z)|_{x=z}\nabla_{\beta}^{(z)}=\nabla_{\alpha}^{(z)},\label{eq:sum nabla_alpha T^beta = nabla_alpha-1}
\end{equation}
when applied to any smooth scalar field\footnote{Of course, for any finite $|\alpha|$, this identity could (with substantial
computational labor) alternatively be directly derived from the values
of the differentiated geodesic distance function $\sigma(x,z)$ at
coincidence $x=z$. In global inertial coordinates in flat spacetime,
the identity (\ref{eq:sum nabla_alpha T^beta = nabla_alpha-1})
holds if and only if, 
\begin{equation}
\nabla_{\alpha}^{(y)}\T[\beta](x;z)|_{x=z}=\delta_{\alpha}^{\beta},\label{eq:Nabla alpha S beta on diag equals Kronecker delta}
\end{equation}
since covariant derivatives commute in this case. Note the identity
(\ref{eq:Nabla alpha S beta on diag equals Kronecker delta}) can
be directly verified using formula (\ref{eq:T^mu1...muk flat spacetime})
for $\T[\beta]$ in flat spacetime. However, in curved spacetime,
the left-hand side of (\ref{eq:sum nabla_alpha T^beta = nabla_alpha-1})
receives non-trivial contributions which depend on the curvature tensor
from $|\beta|<|\alpha|$ .}. It follows that if the multivariable series, 
\begin{equation}
\sum_{\beta_{1}\cdots\beta_{n}}\mathcal{Z}_{\beta_{1}\cdots\beta_{n}}^{I}(z)\T[\beta_{1}](x_{1};z)\cdots\T[\beta_{n}](x_{n};z),\label{eq:divergent Tayl series}
\end{equation}
were to converge to a smooth function of $(x_{1},\dots,x_{n};z)$,
then this function would satisfy eq.~(\ref{eq:Z=deriv Fn})
by construction. However, there is no reason why the series (\ref{eq:divergent Tayl series})
need converge. Nevertheless, it is always possible to modify the series
(\ref{eq:divergent Tayl series}) away from $x_{1},\dots,x_{n}=z$
so as to render it convergent and $C^{\infty}$, while preserving
the desired identity (\ref{eq:Z=deriv Fn}). To see this, fix
$z$, choose a tetrad at $z$, and let $U_{z}\subset M$ be a convex
normal neighborhood of $z$. In Riemannian normal coordinates $x^{\mu}$
centered at $z$ and based on this tetrad, the (non-convergent) series
(\ref{eq:divergent Tayl series}) takes the form, 
\begin{equation}
\sum_{\beta_{1}\cdots\beta_{n}}\mathcal{Z}_{\beta_{1}\cdots\beta_{n}}^{I}(\vec{0})x_{1}^{\beta_{1}}\cdots x_{n}^{\beta_{n}},\label{eq:divergent Tayl series RNC}
\end{equation}
with $x^{\beta}\equiv x^{\mu_{1}}\cdots x^{\mu_{|\beta|}}$. By Borel's
Lemma \cite[see Corollary 1.3.4]{Hormander_book}, every power series
is the Taylor series of some smooth function, so we may always construct
$F_{n}\in C^{\infty}(\times^{n}\mathbb{R}^{D})$ such that, 
\begin{equation}
\partial_{a_{1,1}}^{(x_{1})}\cdots\partial_{a_{1,k_{1}}}^{(x_{1})}\cdots\partial_{\alpha_{n,1}}^{(x_{n})}\cdots\partial_{a_{n,k_{n}}}^{(x_{n})}F_{n}(x_{1},\dots,x_{n})|_{x_{1},\dots,x_{n}=\vec{0}}=\mathcal{Z}_{\{(a_{1,1}\cdots a_{1,k_{1}})\}\cdots\{(a_{n,1}\cdots\alpha_{n,k_{n}})\}}^{I}(\vec{0}),\label{eq:Borel deriv f n at coinc}
\end{equation}
where we note the equality of mixed partials and the index symmetry
of the terms which contribute non-trivially to (\ref{eq:divergent Tayl series RNC}).
Without loss of generality, we may assume that the support of $F_{n}$
is contained in $\times^{n}U_{z}$ since, if necessary, we may multiply
it by smooth function $\chi(x_{1},\dots,x_{n};\vec{0})$ which is
equal to unity in a neighborhood of the origin and has support in
$\times^{n}U_{z}$. However, in any RNC system, the ordinary partial
derivatives of a scalar field evaluated at the origin coincide with
the totally-symmetrized covariant derivatives of the scalar field
evaluated at the origin\footnote{\label{fn:equivalence of partial deriv and cov deriv at origin of RNC}In
any RNC system, it can be deduced from the geodesic equation for geodesics
passing through the origin that 
\begin{equation}
\partial_{(\sigma_{1}}^{\vphantom{\kappa}}\cdots\partial_{\sigma_{n}}^{\vphantom{\kappa}}\Gamma_{\mu\nu)}^{\kappa}(x)|_{x=\vec{0}}=0,\label{eq:sym deriv Christoffel}
\end{equation}
with $\Gamma_{\mu\nu}^{\kappa}$ denoting the Christoffel symbols.
For scalar fields evaluated at the origin, the equivalence between
partial derivatives and totally-symmetrized covariant derivatives
can then be inductively established for all $n$ using (\ref{eq:sym deriv Christoffel}). }. It follows then from the identity (\ref{eq:sum nabla_alpha T^beta = nabla_alpha-1})
that, in fact, 
\[
\nabla_{\alpha_{1}}^{(x_{1})}\cdots\nabla_{\alpha_{n}}^{(x_{n})}F_{n}(x_{1},\dots,x_{n})|_{x_{1},\dots,x_{n}=\vec{0}}=\mathcal{Z}_{\alpha_{1}\cdots\alpha_{n}}^{I}(\vec{0}).
\]
Thus, we have obtained a function $F_{n}$ satisfying (\ref{eq:Z=deriv Fn})
in a neighborhood of one fixed event $z$. However, by choosing a
smooth set of tetrad vector fields and using them to define RNC systems
at each event, $F_{n}$ satisfying (\ref{eq:Z=deriv Fn}) can
be defined as a smooth function of $z$ for any event $z\in M,$ noting
that $\mathcal{Z}_{\alpha_{1}\cdots\alpha_{n}}^{I}(z)$ is smooth
in $z$ by Theorem \ref{thm:Wick uniquenes Z}.

Although this construction of $F_{n}$ will depend on $z$ (and the
arbitrarily-chosen tetrad vector fields) away from total coincidence,
the ``germ'' of $F_{n}$ at $x_{1},\dots,x_{n}=z$ is independent
of $z$ in the sense of (\ref{eq:Fm's symmetric independence of x on diagonal}).
To prove eq.~(\ref{eq:Fm's symmetric independence of x on diagonal})
we use the fact that 
\begin{align}
 & \nabla_{b}^{(z)}\left[\nabla_{\alpha_{1}}^{(x_{1})}\cdots\nabla_{\alpha_{n}}^{(x_{n})}F_{n}(x_{1},\dots,x_{n};z)|_{x_{1},\dots,x_{m}=z}\right]\label{eq:deriv Fm identity}\\
 & =\left[\left((\nabla_{\{b\alpha_{1}\}}^{(x_{1})}\cdots\nabla_{\alpha_{n}}^{(x_{n})})+\cdots+(\nabla_{\alpha_{1}}^{(x_{1})}\cdots\nabla_{\{b\alpha_{n}\}}^{(x_{n})})+(\nabla_{\alpha_{1}}^{(x_{1})}\cdots\nabla_{\alpha_{n}}^{(x_{n})}\nabla_{b}^{(z)})\right)F_{n}(x_{1},\dots,x_{n};z)\right]_{x_{1},\dots,x_{n}=z},\nonumber 
\end{align}
which follows from the ordinary Leibniz rule and the commutativity
of derivatives with respect to different variables. The Leibniz condition,
eq.~(\ref{eq:Leibniz rule Z}), on $\mathcal{Z}_{\alpha_{1}\cdots\alpha_{n}}^{I}$
then implies that the left side of eq.~(\ref{eq:deriv Fm identity})
is equal to the sum of the first $n$-terms on the right side, so
the last term must vanish identically. This establishes the result
(\ref{eq:Fm's symmetric independence of x on diagonal}) for $\beta=\{b\}$.
The general case, $|\beta|>1$ follows via induction. 
\end{proof}
\end{prop}
\begin{rem}
\label{rem:conditions when F_n=0} By Remark \ref{rem:existence of tensor fields for (n,D)}
below Theorem \ref{thm:Wick uniquenes Z}, $F_{n}$ and all its derivatives
on the total diagonal are greatly constrained by the Wick axioms and
will vanish identically unless $n(D-2)/2$ is even. In particular,
$\nabla_{\alpha_{1}}^{(x_{1})}\cdots\nabla_{\alpha_{n}}^{(x_{n})}F_{n}(x_{1},\dots,x_{n};z)|_{x_{1},\dots,x_{n}=z}$
vanish when $n$ is odd and $D\neq2+4k$ for integer $k$, and when
$D$ is odd and $n=2+4k$. 
\end{rem}
\begin{rem}
\label{rem:germ equivalence} Only the germ of $F_{n}(x_{1},\dots,x_{n};z)$
on the total diagonal is relevant to (\ref{eq:Z=deriv Fn})
and (\ref{eq:Fm's symmetric independence of x on diagonal}). Hence,
if $F_{n}$ and $F_{n}'$ have the same germ on the total diagonal,
they are equivalent as far as Proposition \ref{prop:Wick uniqueness Fn}
is concerned. Note that $F_{n}$ is not locally and covariantly defined
away from the total diagonal on account of the coordinate system and
cutoff function used in its construction. However, $F_{n}$ and its
derivatives on the total diagonal are local and covariant. 
\end{rem}
\begin{rem}
The property (\ref{eq:Fm's symmetric independence of x on diagonal})
implies the germ of $F_{n}(x_{1},\dots,x_{n};z)$ on the total diagonal
is independent of its right-most point, $z$. By the previous remark,
$F_{n}(x_{1},\dots,x_{n};z)$ is, therefore, equivalent to, e.g.,
$F_{n}(x_{1},\dots,x_{n};x_{1})$ or $F_{n}(x_{1},\dots,x_{n};x_{n})$.
Therefore, it is possible to write $F_{n}$ as functions of only $n$-spacetime
points rather than $(n+1)$-points. However, in anticipation of the
role they will play in the Wick OPE coefficients of Subsection \ref{subsec:gen Wick OPE coef},
it is more convenient to write $F_{n}$ symmetrically with respect
to $x_{1},\dots,x_{n}$ as we have done here by using the auxiliary
point, $z$. 
\end{rem}
\begin{rem}
\label{rem:all quadratic fields are normal ordered}A notable consequence
of Proposition \ref{prop:Wick uniqueness Fn} is that all prescriptions
for constructing the quadratic Wick fields may be obtained by normal-ordering
with respect to some Hadamard parametrix. Suppose $H$ is any Hadamard
parametrix such that the prescription for Wick monomials satisfies
axioms W1-W8 (or W1-W7 for $D=2$). Then by the above proposition,
any other prescription will satisfy 
\[
(\tildeXL{\nabla_{\alpha_{1}}\phi\nabla_{\alpha_{2}}\phi})(z)=(\nabla_{\alpha_{1}}\phi\nabla_{\alpha_{2}}\phi)_{H}(z)+\nabla_{\alpha_{1}}^{(x_{1})}\nabla_{\alpha_{2}}^{(x_{2})}F_{2}(x_{1},x_{2};z)|_{x_{1},x_{2}=z}I.
\]
This prescription for general quadratic Wick monomials can be reproduced
by Hadamard normal ordering with respect to the new Hadamard parametrix
\[
\widetilde{H}(x_{1},x_{2})=H(x_{1},x_{2})+\frac{1}{2}F_{2}(x_{1},x_{2};x_{1})+\frac{1}{2}F_{2}(x_{1},x_{2};x_{2})
\]
This result is special to the quadratic fields. Prescriptions for
the higher-order Wick monomials are generally \emph{not} equivalent
to Hadamard normal ordering. 
\end{rem}
Thus, we have shown that the ambiguities between any two definitions
of the Wick monomials is completely characterized by a sequence of
functions $\{F_{n}(x_{1},\dots,x_{n};z)\}_{n\ge2}$. As described
in the previous subsection, normal ordering (see eq.~(\ref{eq:H normal ordered}))
with respect to a Hadamard parametrix satisfying (\ref{eq:H conservation condition-1})
provides an explicit construction of the Wick monomials compatible
with axioms W1-W8 (or W1-W7 in $D=2$). Our results, therefore, imply
any Hadamard normal-ordered monomial $\Phi_{A}^{H}$ may be expressed
as 
\begin{equation}
\H A{}=\sum_{B}\mathcal{Z}_{A}^{B}\Phi_{B}=\sum_{q=0}^{p}\binom{p}{q}(\nabla_{(\alpha_{1}}\phi\cdots\nabla_{\alpha_{q}}\phi)[\nabla_{\alpha_{q+1}}\cdots\nabla_{\alpha_{p})}F_{p-q}],\label{eq:Phi^H = Z Phi}
\end{equation}
where $\Phi_{B}$ corresponds to a Wick monomial defined via \emph{any}
renormalization prescription satisfying the axioms, and we have introduced
the shorthand 
\begin{equation}
[\nabla_{\alpha_{1}}\cdots\nabla_{\alpha_{n}}F_{n}]_{z}\equiv\nabla_{\alpha_{1}}^{(x_{1})}\cdots\nabla_{\alpha_{n}}^{(x_{n})}F_{n}(x_{1},\dots,x_{n};z)|_{x_{1},\dots,x_{n}=z}.\label{eq:Taylor coef for F_n at z shorthand}
\end{equation}

The right-most equality in \eqref{eq:Phi^H = Z Phi} follows
directly from plugging (\ref{eq:Z=deriv Fn}) of Proposition
\ref{prop:Wick uniqueness Fn} into the expression (\ref{eq:Wick uniqueness Z})
for $\mathcal{Z}_{A}^{B}$ in Theorem \ref{thm:Wick uniquenes Z}.
Of course, \eqref{eq:Phi^H = Z Phi} can be inverted to express
any monomial $\Phi_{A}$ in a general Wick prescription in terms of
Hadamard normal-ordered fields 
\begin{equation}
\Phi_{A}=\sum_{B}(\mathcal{Z}^{-1})_{A}^{B}\Phi_{B}^{H}.\label{def:Phi=Z Phi^H}
\end{equation}
Note that Theorem \ref{thm:Wick uniquenes Z} and Proposition \ref{prop:Wick uniqueness Fn}
apply also to $\Zinv AB$ and $\Zinv AI$. We can obtain expressions
for $\Zinv AI$ in terms of the functions $F_{n}(x_{1},\dots,x_{n};z)$
by using $\sum_{C}\Zinv CI\Z AC=\delta_{A}^{I}$ together with the
expression for $\Z AC$ in terms of $F_{n}$ implied by eqs. (\ref{eq:Wick uniqueness Z})
and (\ref{eq:Z=deriv Fn}). For $A\neq I$, this yields 
\begin{equation}
(\mathcal{Z}^{-1})_{\alpha_{1}\cdots\alpha_{n}}^{I}=-\sum_{k=0}^{n-2}\binom{n}{k}(\mathcal{Z}^{-1}){}_{(\alpha_{1}\cdots\alpha_{k}}^{I}[\nabla_{\alpha_{k+1}}\cdots\nabla_{\alpha_{n})}F_{n-k}],\label{eq:Z inv (Fn)}
\end{equation}
where we recall the shorthand \eqref{eq:Taylor coef for F_n at z shorthand}
for the Taylor coefficients of $F_{n}$ at $z$. This relation allows
one to recursively solve for $(\mathcal{Z}^{-1})_{\alpha_{1}\cdots\alpha_{n}}^{I}$.
For example, we have 
\begin{align}
(\mathcal{Z}^{-1})_{\alpha_{1}\alpha_{2}}^{I} & =-[\nabla_{\alpha_{1}}\nabla_{\alpha_{2}}F_{2}]\label{eq:Z inv explicit examples}\\
(\mathcal{Z}^{-1})_{\alpha_{1}\alpha_{2}\alpha_{3}\alpha_{4}}^{I} & =-[\nabla_{\alpha_{1}}\nabla_{\alpha_{2}}\nabla_{\alpha_{3}}\nabla_{\alpha_{4}}F_{4}]+\binom{4}{2}[\nabla_{(\alpha_{1}}\nabla_{\alpha_{2}}F_{2}][\nabla_{\alpha_{3}}\nabla_{\alpha_{4})}F_{2}].
\end{align}

In this way, \eqref{def:Phi=Z Phi^H} provides a construction
of the Wick monomials in any prescription satisfying the axioms in
terms of $F_{n}$ and the Hadamard normal-ordered monomials defined
in \eqref{eq:H normal ordered}. In the next section, we will see
that the corresponding ambiguities in the OPE coefficients for products
of Wick fields can be expressed in a simple way in terms of $F_{n}$.

\section{Klein-Gordon OPE coefficients \label{sec:K-G OPE-coefficients}}

A renormalization prescription for the Wick monomials uniquely determines
the Wightman products of Wick fields as well as the unextended time-ordered
products. In Subsection \ref{subsec:H N-O coef}, we obtain the explicit
form of the OPE coefficients of the $n$-point Wightman distributions
invloving Wick monomials defined via a local Hadamard normal-ordering
procedure (see Theorem \ref{thm: Had NO OPE coef}). In Subsection
\ref{subsec:gen Wick OPE coef}, we then give the general form of
the (Wightman) OPE coefficients corresponding to any prescription
for the Wick monomials satisfying axioms W1-W8 in terms of the smooth
functions $F_{n}$ appearing in the Wick monomial uniqueness theorem
(see Theorem \ref{thm:existence Wick coef and associativity}). In
Subsection \ref{subsec:TOP}, we show that the OPE coefficients for
unextended time-ordered products are given by the same expressions
as for the Wightman products with the substitution $H\to H_{F}$,
where $H$ is a locally constructed Hadamard distribution and $H_{F}$
is a locally-constructed Feynman distribution (see Proposition \ref{thm:explicit TO coef}).

\subsection{Local Hadamard normal-ordered OPE coefficients \label{subsec:H N-O coef}}

In this subsection, we show that products of Wick monomials defined
by local Hadamard normal ordering admit an operator product expansion
(OPE), i.e., we will show that for any Wick monomials $\H A1$$,\dots,\H An$
defined via the local Hadamard normal ordering prescription (see eq.
(\ref{eq:H normal ordered})) and in any Hadamard state $\Psi$ we
have, 
\begin{equation}
\braket{\H A1(x_{1})\cdots\H An(x_{n})}_{\Psi}\approx\sum_{B}(C_{H}){}_{A_{1}\cdots A_{n}}^{B}(x_{1},\dots,x_{n};z)\braket{\H B{}(z)}_{\Psi},\label{eq:Had OPE def}
\end{equation}
where the $B$-sum runs over all Wick monomials. In Theorem \ref{thm: Had NO OPE coef}
below, we will also obtain explicit expressions for the local and
covariant OPE coefficients, $(C_{H}){}_{A_{1}\cdots A_{n}}^{B}$.
For products involving more than two Hadamard normal-ordered monomials
(i.e. $n>2$), the OPE coefficients of (\ref{eq:Had OPE def}) are
found to satisfy important relations called ``associativity'' conditions
which are especially useful for analyzing the OPE for Wick monomials
$\Phi_{A}$ that are \emph{not }defined via Hadamard normal ordering.
In the next subsection, we will then show that a general definition
of Wick monomials $\Phi_{A}$ also satisfy an OPE, and we will characterize
how the freedom in the choice of the definition of Wick monomials
affects its OPE coefficients $C_{A_{1}\cdots A_{n}}^{B}$.

The asymptotic equivalence relation ``$\approx$'' used in the definition
of the OPE (\ref{eq:Had OPE def}) was precisely formulated in a local
and covariant manner in \cite{HW_Axiomatic_QFTCS} by introducing
a family of asymptotic equivalence relations ``$\sim_{\mathcal{T},\delta}$''
which are parameterized by a positive real number $\delta$ and a
``merger tree'' $\mathcal{T}$. We introduce here the details relevant
for our analysis and refer the reader to \cite{HW_Axiomatic_QFTCS}
for further discussion. Merger trees classify the different ways in
which the limit $x_{1},\dots,x_{n}\to z$ may be taken. For instance,
when $n=3$, one possible merger tree would correspond to taking all
three points $(x_{1},x_{2},x_{3})$ together to $z$ at the ``same
rate'', while another possible merger tree would correspond to having
two of the points, e.g. $x_{1}$ and $x_{2}$, approach each other
``faster'' than all three points $(x_{1},x_{2},x_{3})$ approach
$z$. For a given merger tree, $\mathcal{T}$, the positive number
$\delta$ in ``$\sim_{\mathcal{T},\delta}$'' indicates how rapidly
the difference between both sides of the equivalence relation goes
to zero as the spacetime points approach $z$ at their various rates.
Altogether, the equivalence relation ``$\approx$'' in (\ref{eq:Had OPE def})
means that, for every $\mathcal{T}$ and $\delta>0$, there exists
a real number $\Delta$ such that, 
\begin{equation}
\braket{\H A1(x_{1})\cdots\H An(x_{n})}_{\Psi}\sim_{\mathcal{T},\delta}\sum_{[B]\le\Delta(\mathcal{T},\delta)}(C_{H}){}_{A_{1}\cdots A_{n}}^{B}(x_{1},\dots,x_{n};z)\braket{\Phi_{B}^{H}(z)}_{\Psi},\label{eq:OPE precise asymp rel}
\end{equation}
where we recall the definition (\ref{eq:[A]}) of $[B]$.

The rate at which a distribution either diverges or converges to zero
in the limit all its spacetime points merge to $z$ at the same rate
(i.e. for the trivial merger tree) is known as its ``scaling degree
at $z$''\footnote{The ``scaling degree'' was introduced by Steinmann in the context
of Minkowski spacetime \cite[Section 5]{Steinmann}. See \cite{BF_muA_interacting_QFTs}
for further discussion in the context of curved manifolds.}. By convention, positive scaling degrees indicate divergence and
negative scaling degrees imply convergence: For example, the geometric
factors $\T[\beta](x;z)$ have scaling degree $-|\beta|$ at $z$,
while the Hadamard parametrix $H$ has scaling degree $D-2$. As we
will see, the engineering dimension $[A]$ of a Wick field $\Phi_{A}$
is related to the scaling degree of the coefficient $C_{AA}^{I}$
as follows\footnote{If the scaling degree varies for different background geometries,
then $[A]$ is equal to the supremum of the right-hand side with respect
to $(M,g_{ab})$. If $\Phi_{A}$ is tensor-valued, then the maximum
scaling degree of the tensor components is used.}: 
\begin{equation}
[A]=\frac{1}{2}\text{sd}_{z}\left[C_{AA}^{I}(x_{1},x_{2};z)\right].\label{eq:[A] via scaling degree}
\end{equation}
Moreover, we will find the scaling degree of all Wick OPE coefficients
are bounded from above by: 
\begin{equation}
\text{sd}_{z}\left[C_{A_{1}\cdots A_{n}}^{B}\right]\le[A_{1}]+\cdots+[A_{n}]-[B].\label{eq:scaling degree OPE coef}
\end{equation}

The key result needed to show the existence of an OPE for Hadamard
normal-ordered Wick monomials is that, in any Hadamard state $\Psi$,
the distribution, 
\begin{equation}
h_{n,\Psi}(x_{1},\dots,x_{n})\equiv\braket{:\phi(x_{1})\cdots\phi(x_{n}):_{H}}_{\Psi}\label{smoothhno}
\end{equation}
is, in fact, a smooth function\footnote{It was proven in \cite[Lemma III.1]{Hollands_Ruan} that (\ref{smoothhno})
is smooth if and only if $\Psi$ is Hadamard and the truncated $n$-point
functions of $\Psi$ are smooth. However, Sanders later proved that
all Hadamard states have smooth truncated $n$-point functions \cite[Proposition 3.1.14]{Sanders_Thesis}
and, therefore, (\ref{smoothhno}) is smooth for \emph{all} Hadamard
$\Psi$ and \emph{only} Hadamard $\Psi$.} of $(x_{1},\dots,x_{n})$. It then follows immediately from the definition,
eq.~(\ref{eq:H normal ordered}), of the Hadamard normal-ordered
Wick power $\phi_{H}^{n}(z)$ that we have, 
\begin{equation}
\braket{\phi_{H}^{n}(z)}_{\Psi}=h_{n,\Psi}(z,\dots,z)\label{phinhno}
\end{equation}
i.e., the expectation value of the Wick power $\phi_{H}^{n}$ evaluated
at $z$ is the total coincidence value at $z$ of the smooth function
$h_{\Psi}(x_{1},\dots,x_{n})$. More generally, we have, 
\begin{equation}
\braket{(\nabla_{\alpha_{1}}\phi\cdots\nabla_{\alpha_{n}}\phi)_{H}(z)}_{\Psi}=\nabla_{\alpha_{1}}^{(x_{1})}\cdots\nabla_{\alpha_{n}}^{(x_{n})}h_{n,\Psi}(x_{1},\dots,x_{n})\big\vert_{x_{1},\dots,x_{n}=z}\label{derphinhno}
\end{equation}

The simplest example of an OPE is the one for the two point function
$\braket{\phi(x_{1})\phi(x_{2})}_{\Psi}$. From the definition eq.~(\ref{hadnodef})
of Hadamard normal ordering, we have for $x_{1}$ and $x_{2}$ in
a common convex normal neighborhood, 
\begin{equation}
\phi(x_{1})\phi(x_{2})=:\phi(x_{1})\phi(x_{2}):_{H}+H(x_{1},x_{2})I\label{had2ptfun}
\end{equation}
We now take the expectation value of this equation in an arbitrary
Hadamard state $\Psi$. Since $\braket{:\phi(x_{1})\phi(x_{2}):_{H}}_{\Psi}$
is smooth, we may take its covariant Taylor expansion (see eq.~(\ref{eq:covtayser})
above) for $x_{1}$ and $x_{2}$ in a normal neighborhood of some
arbitrarily chosen point $z$, thereby obtaining asymptotic relations\footnote{\label{rem:n=2 trivial merger tree}For $n=2$, we omit the
$\mathcal{T}$ symbol since there is only one possible merger tree
in this case.} that hold in the coincidence limit, 
\begin{equation}
\braket{:\phi(x_{1})\phi(x_{2}):_{H}}_{\Psi}\sim_{\delta}\sum_{|\beta_{1}|+|\beta_{2}|\le\delta}\T[\beta_{1}](x_{1};z)\T[\beta_{2}](x_{1};z)\nabla_{\beta_{1}}^{(x_{1})}\nabla_{\beta_{2}}^{(x_{2})}h_{2,\Psi}(x_{1},x_{2})\big\vert_{x_{1}=x_{2}=z},\label{eq:exp Psi(:phi(x1)phi(x2):)}
\end{equation}
using the fact that 
\[
\T[\beta_{1}](x_{1};z)\T[\beta_{2}](x_{2};z)\sim_{\delta}0,\qquad\text{for }|\beta_{1}|+|\beta_{2}|>\delta.
\]
Substituting expression (\ref{eq:exp Psi(:phi(x1)phi(x2):)}) into
eq.~(\ref{had2ptfun}) and using eq.~(\ref{derphinhno}), we find
that for any Hadamard state $\Psi$, we have, 
\begin{equation}
\braket{\phi(x_{1})\phi(x_{2})}_{\Psi}\sim_{\delta}\sum_{|\beta_{1}|+|\beta_{2}|\le\delta}\T[\beta_{1}](x_{1};z)\T[\beta_{2}](x_{2};z)\braket{(\nabla_{\beta_{1}}\phi\nabla_{\beta_{2}}\phi)_{H}(z)}_{\Psi}+H(x_{1},x_{2})\braket{I}_{\Psi},\label{eq:phi phi Had NO OPE}
\end{equation}
Noting this holds for all $\delta>0$, this equation takes the form
of an OPE, from which we can read off the OPE coefficients, 
\begin{equation}
(C_{H}){}_{\phi\phi}^{I}(x_{1},x_{2};z)=H(x_{1},x_{2}),\qquad(C_{H}){}_{\phi\phi}^{(\nabla_{\beta_{1}}\phi\nabla_{\beta_{2}}\phi)}(x_{1},x_{2};z)=\T[(\beta_{1}](x_{1};z)\T[\beta_{2})](x_{2};z),\label{eq:phi phi Had coef}
\end{equation}
where we have symmetrized over $\beta_{1}$ and $\beta_{2}$ on the
right side of the last expression since $(\nabla_{\beta_{1}}\phi\nabla_{\beta_{2}}\phi)_{H}$
is symmetric in $\beta_{1}$ and $\beta_{2}$, so only the symmetric
part of this OPE coefficient contributes. All other OPE coefficients
of the form $C_{\phi\phi}^{B}$ vanish. Given the scaling degree of
$\T[\beta]$ and $H$ stated above, we indeed find (as anticipated
in formula (\ref{eq:[A] via scaling degree})), 
\[
[\phi]=\frac{1}{2}\text{sd}_{z}\left[(C_{H})_{\phi\phi}^{I}(x_{1},x_{2};z)\right]=\frac{1}{2}(D-2),
\]
and (as anticipated in formula (\ref{eq:scaling degree OPE coef}))
the scaling degree of $(C_{H}){}_{\phi\phi}^{B}$ at $z$ is found
to be bounded from above by $[\phi]+[\phi]-[B]$ with the non-trivial
coefficients saturating the bound.

In order to illustrate how more general OPEs are obtained for Hadamard
normal-ordered monomials and to understand the patterns that emerge
in the structure of the general OPE coefficients, it is instructive
to consider another simple example, namely $n=2$ and $\H A1,\H A2=\phi_{H}^{2}$.
Wick's theorem (\ref{eq:Wick product rule}) implies that for $x_{1},x_{2}$
in a common convex normal neighborhood, we have 
\begin{align}
 & \braket{\phi_{H}^{2}(x_{1})\phi_{H}^{2}(x_{2})}_{\Psi}\label{eq:Wick's thm phi^2 phi^2}\\
 & \qquad=\braket{:\phi(x_{1})\phi(x_{1})\phi(x_{2})\phi(x_{2}):_{H}}_{\Psi}+4H(x_{1},x_{2})\braket{:\phi(x_{1})\phi(x_{2}):_{H}}_{\Psi}+2H(x_{1},x_{2})H(x_{1},x_{2}).\nonumber 
\end{align}
Again, all of the ``totally normal-ordered'' quantities appearing
on the right-hand side are smooth functions. Therefore, we may covariantly
Taylor expand these terms about $x_{1},x_{2}=z$, to obtain 
\begin{align}
 & \braket{\phi_{H}^{2}(x_{1})\phi_{H}^{2}(x_{2})}_{\Psi}\nonumber \\
 & \quad\sim_{\delta}\sum_{\beta_{1},\beta_{2},\beta_{3},\beta_{4}}\T[(\beta_{1}](x_{1};z)\T[\beta_{2}](x_{1};z)\T[\beta_{3}](x_{2};z)\T[\beta_{4})](x_{2};z)\braket{(\nabla_{\beta_{1}}\phi\nabla_{\beta_{2}}\phi\nabla_{\beta_{3}}\phi\nabla_{\beta_{4}}\phi)_{H}(z)}_{\Psi}+\label{eq:OPE Had phi^2phi^2}\\
 & \quad+4H(x_{1},x_{2})\sum_{\beta_{1},\beta_{2}}\T[(\beta_{1}](x_{1};z)\T[\beta_{2})](x_{2};z)\braket{(\nabla_{\beta_{1}}\phi\nabla_{\beta_{2}}\phi)_{H}(z)}_{\Psi}+2H(x_{1},x_{2})H(x_{1},x_{2})\braket{I}_{\Psi},\nonumber 
\end{align}
where the respective sums run over $\sum_{i}|\beta_{i}|\le\delta$.
Thus, the nonvanishing OPE coefficients are, 
\begin{equation}
(C_{H}){}_{\phi^{2}\phi^{2}}^{B}(x_{1},x_{2};z)=\begin{cases}
\T[(\beta_{1}](x_{1};z)\T[\beta_{2}](x_{1};z)\T[\beta_{3}](x_{2};z)\T[\beta_{4})](x_{2};z) & \H B{}=(\nabla_{\beta_{1}}\phi\nabla_{\beta_{2}}\phi\nabla_{\beta_{3}}\phi\nabla_{\beta_{4}}\phi)_{H}\\
4\T[(\beta_{1}](x_{1};z)\T[\beta_{2})](x_{2};z)H(x_{1},x_{2}) & \H B{}=(\nabla_{\beta_{1}}\phi\nabla_{\beta_{2}}\phi){}_{H}\\
2H(x_{1},x_{2})H(x_{1},x_{2}) & \H B{}=I
\end{cases}\label{eq:C phi^2 phi^2 cases}
\end{equation}
Thus, we see that all of the nonvanishing OPE coefficients are given
by products of the Hadamard parametrix $H(x_{1},x_{2})$ and the geometrical
factors $S^{\beta}(x;z)$ defined by eq.(\ref{eq:covariant Taylor coef}).

The existence of an OPE for an arbitrary product of $n$ Hadamard
normal-ordered Wick monomials 
\begin{equation}
\braket{(\nabla_{\alpha_{(1,1)}}\phi\cdots\nabla_{\alpha_{(1,k_{1})}}\phi)_{H}(x_{1})(\nabla_{\alpha_{(2,1)}}\phi\cdots\nabla_{\alpha_{(2,k_{2})}}\phi)_{H}(x_{2})\cdots(\nabla_{\alpha_{(n,1)}}\phi\cdots\nabla_{\alpha_{(n,k_{n})}}\phi)_{H}(x_{n})}_{\Psi}\label{genopeexp}
\end{equation}
can be established by paralleling the derivation used in the above
examples. As previously introduced in condition W7 above, we denote
the number of factors of $\phi$ that appear in a Wick monomial $\Phi_{A}$
by $[\Phi_{A}]_{\phi}$. Thus, for the factor $\H Ai=(\nabla_{\alpha_{(i,1)}}\phi\cdots\nabla_{\alpha_{(i,k_{i})}}\phi)_{H}$
in eq.(\ref{genopeexp}), we have $[\H Ai]_{\phi}=k_{i}$. We denote
by $K$ the total number of factors of $\phi$ appearing in the expression
(\ref{genopeexp}), 
\begin{equation}
K=\sum_{i=1}^{n}k_{i}\label{Kdef}
\end{equation}
We write the quantity (\ref{genopeexp}) in terms of products of $H$
and normal ordered products of $\phi$'s. We then obtain an OPE for
(\ref{genopeexp}) by Taylor expanding the normal-ordered products
of $\phi$'s. It follows that the general OPE coefficients are given
by products of $H(x_{1},x_{2})$, $S^{\beta}(x;z)$ and their derivatives.
It also can be seen that the only fields $\H B{}=(\nabla_{\beta_{1}}\phi\nabla_{\beta_{2}}\phi\cdots\nabla_{\beta_{m}}\phi)_{H}$
for which $(C_{H}){}_{A_{1}\cdots A_{n}}^{B}$ can be non-vanishing
are such that $[\H B{}]_{\phi}=m$ takes the values $m=K,K-2,K-4,\dots$
and $m\ge0$.

In order to explain the combinatorics of the formula for the general
OPE coefficients in terms of $H(x_{1},x_{2})$, $S^{\beta}(x;z)$
and their derivatives, it is useful to introduce a uniform notation
for all the multi-indices relevant for $(C_{H}){}_{A_{1}\cdots A_{n}}^{B}$
by pairing each $\beta_{j}$ multi-index with a ``$0$'' and write
$\H B{}=(\nabla_{\beta_{(0,1)}}\phi\nabla_{\beta_{(0,2)}}\phi\cdots\nabla_{\beta_{(0,m)}}\phi)_{H}$.
The multi-indices relevant for $(C_{H}){}_{A_{1}\cdots A_{n}}^{B}$
then comprise the set of ordered pairs, 
\begin{equation}
{\mathcal{S}}=\left\{ (0,1),\dots,(0,m),(1,1),(1,2),\dots,(1,k_{1}),(2,1),\dots,(2,k_{2}),\dots,(n,1),\dots,(n,k_{n})\right\} \label{Sdef}
\end{equation}
This set has $(m+K)$-elements, which is an even number whenever $(C_{H}){}_{A_{1}\cdots A_{n}}^{B}$
is nonvanishing. In order to describe the combinations of $\T[\beta_{j}](x_{u})$
and $H(x_{v},x_{w})$ and their derivatives that appear in the formula
for $(C_{H}){}_{A_{1}\cdots A_{n}}^{B}$, we follow \cite[see Section 4.1]{HH_Associativity}
by employing the notion of ``perfect matchings''\footnote{The terminology is borrowed from graph theory: The elements of $\mathcal{S}$
can be viewed as labeling the vertices of a graph. (See e.g. Figure
\ref{fig:graphs of perfect matchings} below). An arrow connecting
two vertices of this graph corresponds then to a pairing between two
elements of $\mathcal{S}$. A ``perfect matching'' is achieved when
every vertex is connected to exactly one arrow and there are no loops
(connecting a vertex to itself): i.e., every element of $\mathcal{S}$
is paired with precisely one other element of $\mathcal{S}$. } for elements of $\mathcal{S}$. By definition, the set, $\mathcal{M}(\mathcal{S})$,
of perfect matchings is the set of all partitions of $\mathcal{S}$
into subsets each of which contains precisely two elements. Each pair
of distinct elements of $\mathcal{S}$ is of the form $\{(v,i),(w,j)\}$.
It is convenient to require that these pairs be ordered so that $v\leq w$.
(When $v=w$, we may require $i<j$, but the matrix elements of the
matrix $\mathcal{N}$ defined below will vanish in that case, so the
ordering is irrelevant.) Since $\mathcal{S}$ has $(m+K)$-elements
it follows that $\mathcal{M}(\mathcal{S})$ has $(m+K-1)!!\equiv(m+K-1)(m+K-3)(m+K-5)\cdots1$
elements when $m+K$ is even. Thus, for example, if ${\mathcal{S}}=\{(0,1),(1,1),(1,2),(2,1)\}$
corresponding to $n=2$, $K=3$, and $m=1$, then $\mathcal{M}(\mathcal{S})$
consists of the three partitions: 
\begin{equation}
\mathcal{M}(\mathcal{S})=\bigl\{\{(0,1),(1,1);(1,2),(2,1)\},\{(0,1),(1,2);(1,1),(2,1)\},\{(0,1),(2,1);(1,1),(1,2)\}\bigr\},\label{eq:perfect matching ex}
\end{equation}
which are diagrammed in the following figure. 
\begin{figure}[H]
\global\long\def\edge#1#2{\draw[->,thick] (#1) to [out=270,in=270,looseness=1.75] (#2);}
 \begin{tikzpicture}[node distance={12mm},main/.style ]  
\node[main] (1) {$(0,1)$};
\node[main] (2) [right of =1] {$(1,1)$}; 
\node[main] (3) [right of =2] {$(1,2)$};
\node[main] (4) [right of =3] {$(2,1)$};
\edge{1}{2}
\edge{3}{4}

\node[main] (20) [right of =4] {};

\node[main] (5) [right of =20] {$(0,1)$};
\node[main] (6) [right of =5] {$(1,1)$}; 
\node[main] (7) [right of =6] {$(1,2)$};
\node[main] (8) [right of =7] {$(2,1)$};

\edge{5}{7}
\edge{6}{8}

\node[main] (21) [right of =8] {};

\node[main] (9) [right of =21] {$(0,1)$};
\node[main] (10) [right of =9] {$(1,1)$}; 
\node[main] (11) [right of =10] {$(1,2)$};
\node[main] (12) [right of =11] {$(2,1)$};

\edge{9}{12}
\edge{10}{11}

\end{tikzpicture} \caption{\label{fig:graphs of perfect matchings}Directed graphs representing
the three perfect matchings in (\ref{eq:perfect matching ex}). Arrow
direction points from a vertex $(v,i)\in\mathcal{S}$ toward a vertex
$(w,j)\in\mathcal{S}$ such that $v\le w$ and $i<j$.}
\end{figure}
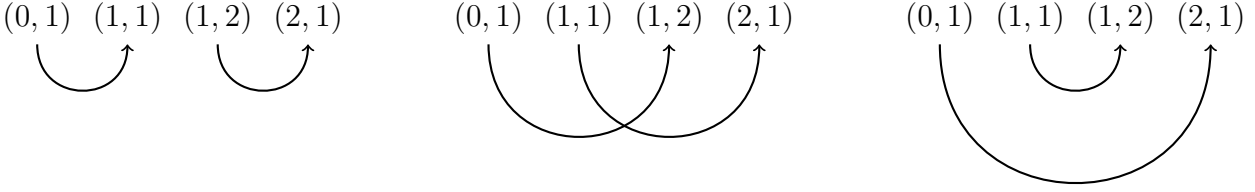

It is useful to combine the relevant multi-index derivatives of $\T[\beta_{j}](x_{u})$
and $H(x_{v},x_{w})$ into a single $(K+m)\times(K+m)$ matrix\textcolor{red}{{}
}$\mathcal{N}$ as follows, 
\begin{equation}
\mathcal{N}_{(v,i)(w,j)}\equiv\begin{cases}
\na vi\na wjH(x_{v},x_{w}) & v\ne w;v,w\ne0\\
\na wj\T[\beta_{i}](x_{w};z) & v=0,w\ne0\\
0 & \text{otherwise}
\end{cases}\label{wmatdef}
\end{equation}
The {\em hafnian} of $\mathcal{N}$ is defined by \cite{Caianiello_textbook}
\begin{equation}
\haf\mathcal{N}\equiv\sum_{M\in\mathcal{M}({\mathcal{S}})}\prod_{\{(v,i),(w,j)\}\in M}\mathcal{N}_{(v,i)(w,j)},\label{haf}
\end{equation}
where the sum is taken over the $(m+K-1)!!$ perfect matchings, $M$,
of the set $\mathcal{S}$, eq.~(\ref{Sdef}), and the product is
taken over all ordered pairs $\{(v,i),(w,j)\}$ occurring in $M$.
With these definitions, the existence of an OPE for Hadamard normal-ordered
products and the general formula for the OPE coefficients may now
be stated as the following theorem: 
\begin{thm}
\label{thm: Had NO OPE coef} For Hadamard normal-ordered fields $\H Ai$,
there exists an OPE of the form eq.~(\ref{eq:Had OPE def}), with
local and covariant OPE coefficients $(C_{H})_{A_{1}\cdots A_{n}}^{B}(x_{1},\dots,x_{n};z)$.
The OPE coefficients $(C_{H}){}_{A_{1}\cdots A_{n}}^{B}$ can be nonvanishing
only when $m=[\H B{}]_{\phi}$ takes the values $m=K,K-2,K-4,\dots$
and $m\ge0$, where $K$ is given by eq.~(\ref{Kdef}). Furthermore,
the OPE coefficients are explicitly given by 
\begin{equation}
(C_{H}){}_{A_{1}\cdots A_{n}}^{B}(x_{1},\dots,x_{n};z)=\frac{1}{m!}\haf\mathcal{N},\label{eq:H NO coefficients}
\end{equation}
with $\haf\mathcal{N}$ given by eq.(\ref{haf}) and they satisfy
the scaling degree properties (\ref{eq:[A] via scaling degree})
and (\ref{eq:scaling degree OPE coef}), saturating the bound whenever
\eqref{eq:H NO coefficients} is nonzero. 
\end{thm}
A formal proof of the existence of an OPE for scalar field theories
with renormalizable interactions on any globally-hyperbolic spacetime
was given (to any finite order in perturbation theory) in \cite[Theorem 1]{Hollands_perturbative_OPE_CS}.
Since at zeroth-order the quantum fields in \cite{Hollands_perturbative_OPE_CS}
were defined via Hadamard normal ordering, this result encompasses
the case considered here. For the case of flat spacetime, we have
$\nabla_{\alpha}^{(x)}S^{\beta}(x;z)=\frac{1}{|\beta|!}\partial_{\alpha}^{(x)}(x-z)^{\beta}$
and our formula (\ref{eq:H NO coefficients}) for the Hadamard normal-ordered
OPE coefficients corresponds\footnote{There is a discrepancy of a factor of $1/m!$ between our formula
(\ref{eq:H NO coefficients}) and the formula given in \cite{HH_Associativity}.} to the formula given in \cite[Section 4.1]{HH_Associativity} for
the vacuum normal-ordered (flat) Euclidean OPE coefficients after
replacing $H$ with the Euclidean Green's function $G_{E}$ (see eq.(\ref{eq:explicit G_E})
below). The scaling degree properties stated in the theorem follow
immediately from eq.~(\ref{eq:H NO coefficients}) and the scaling
behavior of the Hadamard parametrix and the geometric factors $\T[\beta](x;z)$.
\begin{rem}
In the Euclidean case considered in \cite[Section 4.1]{HH_Associativity},
$G_{E}(x_{1},x_{2})$ is symmetric in $(x_{1},x_{2})$ so the ordering
of $(x_{1},x_{2})$ is irrelevant. However, in the Lorentzian case,
the anti-symmetric part of $H$ is proportional to the causal propagator,
$i\Delta\equiv i\Delta^{\text{adv}}-i\Delta^{\text{ret}}$, modulo
$C^{\infty}(M\times M)$, so the ordering of the events matters. 
\end{rem}
\begin{rem}
For $B=I$, we have $m=0$, so we have $(C_{H}){}_{A_{1}\cdots A_{n}}^{I}=0$
if $K$ is odd. If $K$ is even, then since $v=0$ does not arise
on the right side of eq.~(\ref{eq:H NO coefficients}) when $m=0$,
we may replace $\mathcal{N}_{(v,i)(w,j)}$ by $\na vi\na wjH(x_{v},x_{w})$,
so $(C_{H}){}_{A_{1}\cdots A_{n}}^{I}$ is given by 
\begin{equation}
(C_{H}){}_{A_{1}\cdots A_{n}}^{I}(x_{1},\dots,x_{n};z)=\sum_{M\in\mathcal{M}({\mathcal{S}})}\prod_{\{(v,i),(w,j)\}\in M}\na vi\na wjH(x_{v},x_{w}),\label{eq:H id coefs}
\end{equation}
i.e., $(C_{H}){}_{A_{1}\cdots A_{n}}^{I}$ is a sum of products of
derivatives of $H$'s. 
\end{rem}
\begin{rem}
At the other extreme, when $m=K$, then if any product on the right
side of eq.~(\ref{eq:H NO coefficients}) contained a factor with
both $v\neq0$ and $w\neq0$, then it would also have to contain a
factor with $v=w=0$ and thus would vanish. Thus, for $m=K$, the
only elements of $\mathcal{M}(\mathcal{S})$ which may contribute
nontrivially to (\ref{eq:H NO coefficients}) are those such that
$v=0$ and $w\ne0$, and the OPE coefficients $(C_{H}){}_{A_{1}\cdots A_{n}}^{B}$
are given by a sum of terms composed of products of derivatives of
$\T[\beta]$'s. Explicitly, this formula is, 
\begin{equation}
(C_{H}){}_{A_{1}\cdots A_{n}}^{B}(x_{1},\dots,x_{n};z)=\text{sym}_{\beta}\prod_{i=1}^{n}\prod_{j=1}^{k_{i}}\na ij\T[\beta_{p(i,j)}](x_{i};z),\label{eq:T-factor coef}
\end{equation}
where $p(i,j)\equiv j+\sum_{q=1}^{i-1}k_{q}$ and the symmetrization
over the $\beta$-multi-indices (already seen in examples (\ref{eq:phi phi Had coef})
and (\ref{eq:C phi^2 phi^2 cases})) is here denoted using ``$\text{sym}_{\beta}$''
as follows, 
\begin{equation}
\text{sym}_{\beta}\T[\beta_{1}]\cdots\T[\beta_{m}]\equiv\T[(\beta_{1}]\cdots\T[\beta_{m})]\equiv\frac{1}{m!}\sum_{\sigma}\T[\beta_{\sigma(1)}]\cdots\T[\beta_{\sigma(m)}],\label{def:sym_beta}
\end{equation}
where $\sigma$ sums over the permutations of $\{1,\dots,m\}$. 
\end{rem}
For $0<m<K$, the OPE coefficient $(C_{H})_{A_{1}\cdots A_{n}}^{B}$
will be a sum of terms involving products of derivatives of both $H$'s
and $\T[\beta]$'s. In fact, the formula for $(C_{H})_{A_{1}\cdots A_{n}}^{B}$
in this case satisfies very useful recursion relations in terms of
a sum of products of OPE coefficients of smaller $K$. An example
of this structure can be seen from eqs.(\ref{eq:phi phi Had coef})
and (\ref{eq:C phi^2 phi^2 cases}) where, by inspection, we see that
\begin{equation}
(C_{H}){}_{\phi^{2}\phi^{2}}^{(\nabla_{\beta_{1}}\phi\nabla_{\beta_{2}}\phi)}(x_{1},x_{2};z)=4(C_{H}){}_{\phi\phi}^{(\nabla_{\beta_{1}}\phi\nabla_{\beta_{2}}\phi)}(x_{1},x_{2};z)(C_{H}){}_{\phi\phi}^{I}(x_{1},x_{2}).\label{eq:intermed H coef ex}
\end{equation}
To state the general result, let $\mathcal{S}_{A}$ be the set of
the $K$ multi-index labels of the $\H A1$, $\dots$,$\H An$ fields,
i.e., $\mathcal{S}_{A}=\{(1,1),\dots,(n,k_{n})\}$. ($\mathcal{S}_{A}$
differs from $\mathcal{S}$ by not including the labels $\{(0,1),\dots,(0,m)\}$
associated with multi-indices of the operator $\Phi_{H}^{B}$.) Let
$p$ be an integer with $0<p\leq m$. Partition $\mathcal{S}_{A}$
into two subsets $P_{1}$, $P_{2}$, such that $P_{1}$ contains $p$
elements and $P_{2}$ contains $(K-p)$ elements, i.e., $P_{1}$ and
$P_{2}$ are complements of each other with respect to the set $\mathcal{S}_{A}$.
(There are $\binom{K}{p}$ possible ordered partitions of this sort.)
For any such partition, we define, 
\begin{align}
\Phi_{A_{i}'} & \equiv\prod_{(i,j)\in P_{1}}\nabla_{\alpha_{(i,j)}}\phi\label{eq:Phi_P1}\\
\Phi_{A_{i}''} & \equiv\prod_{(i,j)\in P_{2}}\nabla_{\alpha_{(i,j)}}\phi\label{eq:Phi_P2}
\end{align}
For any $i$ such that there exists no $(i,j)\in P_{1}$, then we
set $\W{A_{i}'}=I$ and, similarly, for any $i$ such that there are
no $(i,j)\in P_{2}$, we have $\W{A_{i}''}=I$. Our result on the
Hadamard OPE coeffcients $(C_{H})_{A_{1}\cdots A_{n}}^{B}$ with $0<m<K$
is the following: 
\begin{prop}
\label{Prop:intermed H-coef recursion rel} For $0<m<K$, the Hadamard
normal-ordered OPE coefficients (\ref{eq:H NO coefficients}) of Theorem
\ref{thm: Had NO OPE coef} satisfy, 
\begin{align}
 & (C_{H}){}_{A_{1}\cdots A_{n}}^{B}(x_{1},\dots,x_{n};z)\label{eq:H intermed coef}\\
 & =\binom{m}{p}^{-1}\sum_{\{P_{1},P_{2}\}\in\mathcal{P}_{p}(\mathcal{S}_{A})}(C_{H}){}_{A_{1}'\cdots A_{n}'}^{(\nabla_{\beta_{1}}\phi\cdots\nabla_{\beta_{p}}\phi)}(x_{1},\dots,x_{n};z)(C_{H}){}_{A_{1}''\cdots A_{n}''}^{(\nabla_{\beta_{(p+1)}}\phi\cdots\nabla_{\beta_{m}}\phi)}(x_{1},\dots,x_{n};z),\nonumber 
\end{align}
Here $p$ is any integer with $0<p\leq m$ and the sum is taken over
the $\binom{K}{p}$-ordered partitions $\mathcal{P}_{p}(\mathcal{S}_{A})$
into subsets, $P_{1}$ and $P_{2}$, containing $p$ and $K-p$ elements,
respectively. The fields $\H{A'}i$ and $\H{A''}i$ were defined with
respect to the partition by eqs.~(\ref{eq:Phi_P1}) and (\ref{eq:Phi_P2}),
respectively. 
\end{prop}
\begin{proof}
From the explicit expression for the Hadamard normal-ordered OPE coefficients
(\ref{eq:H NO coefficients}) given in Theorem \ref{thm: Had NO OPE coef},
it can be seen directly that (\ref{eq:H intermed coef}) is equivalent,
for any $0<p\le m$ and $0<m<K$, to the relation 
\begin{align}
 & \sum_{M\in\mathcal{M}({\mathcal{S}})}\prod_{\{(v,i),(w,j)\}\in M}\mathcal{N}_{(v,i)(w,j)}\label{eq:H intermed coef equiv}\\
 & =\sum_{\{P_{1};P_{2}\}\in\mathcal{P}_{p}(\mathcal{S}_{A})}\left\{ \left[\sum_{M_{1}\in\mathcal{M}_{1}[P_{1}]}\prod_{\{(v,i),(w,j)\}\in M_{1}}\mathcal{N}_{(v,i)(w,j)}\right]\left[\sum_{M_{2}\in\mathcal{M}_{2}[P_{2}]}\prod_{\{(v,i),(w,j)\}\in M_{2}}\mathcal{N}_{(v,i)(w,j)}\right]\right\} ,\nonumber 
\end{align}
with $\mathcal{M}_{1}\equiv\mathcal{M}(P_{1}\cup\{(0,1),\dots,(0,p)\})$
and $\mathcal{M}_{2}\equiv\mathcal{M}(P_{2}\cup\{(0,p+1),\dots,(0,m)\})$.
To prove this relation, we note that the left side of eq.~(\ref{eq:H intermed coef equiv})
instructs us to take the product of the matrix elements $\mathcal{N}_{(v,i)(w,j)}$
over a perfect matching of $\mathcal{S}$ and then sum over all perfect
matchings. By eq.~(\ref{wmatdef}), in order for any perfect matching
to contribute nontrivially, any element of the form $(0,j)$ must
be matched with an element of ${\mathcal{S}}_{A}$. Fix any integer
$p$ with $0<p\leq m$. For a given perfect matching that contributes
nontrivially to eq.~(\ref{eq:H intermed coef equiv}), the elements
of ${\mathcal{S}}_{A}$ that are paired with $(0,1),\dots,(0,p)$
define a subset, $P_{1}$, of ${\mathcal{S}}_{A}$ with $p$ elements.
Let $P_{2}={\mathcal{S}}_{A}\setminus P_{1}$ so that $\{P_{1},P_{2}\}$
is a partition of ${\mathcal{S}}_{A}$ into subsets of $p$ and $K-p$
elements, respectively. When we sum over all perfect matchings, we
may first sum over all perfect matchings that respect these partitions.
That sum yields the term in large curly braces on the right side of
eq.~(\ref{eq:H intermed coef equiv}). It then remains to sum over
all partitions $\{P_{1},P_{2}\}\in\mathcal{P}_{p}(\mathcal{S}_{A})$,
which yields eq.~(\ref{eq:H intermed coef equiv}). 
\end{proof}
\begin{rem}
\label{rem:m=p intermed coef}An important case is $m=p$ for
which relation (\ref{eq:H intermed coef}) of Proposition \ref{Prop:intermed H-coef recursion rel}
reduces to, 
\begin{align}
 & (C_{H}){}_{A_{1}\cdots A_{n}}^{B}(x_{1},\dots,x_{n};z)\label{eq:H intermed coef p=m}\\
 & \qquad=\sum_{\{P_{1},P_{2}\}\in\mathcal{P}_{m}(\mathcal{S}_{A})}(C_{H}){}_{A_{1}'\cdots A_{n}'}^{B}(x_{1},\dots,x_{n};z)(C_{H}){}_{A_{1}''\cdots A_{n}''}^{I}(x_{1},\dots,x_{n}).\nonumber 
\end{align}
This implies that every Hadamard normal-ordered OPE coefficient can
be expressed as a sum of products of OPE coefficients with $m=0$
of the form (\ref{eq:H id coefs}) and OPE coefficients with $m=K'$
of the form (\ref{eq:T-factor coef}). While eq.~(\ref{eq:H intermed coef})
was derived here using the particular form (\ref{eq:H NO coefficients})
of the Hadamard normal-ordered coefficients, we will show, in the
next subsection, these identities for the Hadamard normal-ordered
OPE coefficients and the field redefinition relations for Wick fields
obtained in Subsection \ref{subsec:uniq Wick poly and novel reparm}
can be used to prove relation (\ref{eq:H intermed coef}) holds also
for the OPE coefficients corresponding to completely general constructions
of the Wick fields; that is, we will show that (\ref{eq:H intermed coef})
continues to be a valid formula even when the $H$-subscripts are
removed. 
\end{rem}
Above, we have given explicit formulas for all of the OPE coefficients
occurring for products of Wick monomials of the Klein-Gordon field
defined by Hadamard normal ordering. There is an important associativity
property satisfied by these OPE coefficients, which will be seen in
the next subsection to hold for general prescriptions for Wick monomials
and, indeed, is expected to hold for general interacting theories
\cite{Hollands_perturbative_OPE_CS}. As already mentioned at the
beginning of this subsection, for an OPE involving $n>2$ spacetime
points $x_{i}$, we have different possible ``merger trees,'' i.e.,
different possible rates at which the different $x_{i}$'s may approach
$z$. For example, for an OPE involving three spacetime points $(x_{1},x_{2},x_{3})$,
we could let $x_{1}$ and $x_{2}$, approach each other faster than
the remaining point, $x_{3}$. In this case, one might expect that
the OPE and its coefficients could be alternatively computed by first
expanding the expectation value around $x_{1}$ and $x_{2}$ about
an auxiliary point $z'$ and, subsequently, expanding $z'$ and $x_{3}$
about $z$. For this to be self-consistent, the OPE coefficients obtained
via this iterated expansion should be asymptotically equivalent to
the original OPE coefficients. This implies that OPE coefficients
involving $n>2$ spacetime points must factorize into a sum of products
of OPE coefficients involving fewer spacetime points. This property
is referred to as ``associativity.''

The associativity conditions corresponding to the most general possible
merger trees may be found in \cite[Section 3]{HW_Axiomatic_QFTCS}.
For our purposes, it will be useful to have an explicit formula for
the following merger trees: Consider the set of $K=k_{1}+k_{2}+\cdots+k_{n}$
spacelike-separated spacetime points, 
\begin{equation}
\left\{ x_{(1,1)},\dots,x_{(1,k_{1})},x_{(2,1)},\dots,x_{(2,k_{2})},\dots,x_{(n,1)},\dots,x_{(n,k_{n})}\right\} .\label{eq:set of points for merger tree}
\end{equation}
Let $\mathcal{T}$ denote any merger tree where, for all $i\in\{1,\dots,n\}$,
the $k_{i}$-spacetime points $x_{(i,1)},\dots,x_{(i,k_{i})}$ approach
each other faster than the remaining points in (\ref{eq:set of points for merger tree}).
Supposing a Wick field is located at each one of these spacetime points,
the associativity condition for this class of merger trees is, 
\begin{align}
 & (C_{H}){}_{A_{(1,1)}\cdots A_{(n,k_{n})}}^{B}(\vec{x}_{1},\dots,\vec{x}_{n};z)\label{eq:pertinent assoc cond for C_H}\\
 & \sim_{\mathcal{T},\delta}\sum_{C_{1},\dots,C_{n}}(C_{H})_{A_{(1,1)}\cdots A_{(1,k_{1})}}^{C_{1}}(\vec{x}_{1};z_{1})\cdots(C_{H}){}_{A_{(n,1)}\cdots A_{(n,k_{n})}}^{C_{n}}(\vec{x}_{n};z_{n})(C_{H}){}_{C_{1}\cdots C_{n}}^{B}(z_{1},\dots,z_{n};z),\nonumber 
\end{align}
where we have introduced the shorthand $\vec{x}_{i}\equiv x_{(i,1)},\dots,x_{(i,k_{i})}$.
Here the $C_{1},\dots,C_{n}$-sums are carried out to sufficiently
high, but finite, $[C_{i}]$ for all $i$. The associativity condition
and other properties of the OPE coefficients were established in \cite[Section 4]{Hollands_perturbative_OPE_CS}.
We state this result in the following theorem: 
\begin{thm}
\label{thm:existence H NO coef and associativity} The OPE coefficients
$(C_{H})_{A_{1}\cdots A_{n}}^{B}$ satisfy (\ref{eq:pertinent assoc cond for C_H})
and the more general associativity conditions of \cite{HW_Axiomatic_QFTCS,Hollands_perturbative_OPE_CS}. 
\end{thm}

\subsection{OPE coefficients for a general definition of Wick monomials \label{subsec:gen Wick OPE coef}}

We are now in a position to obtain the expression for the coefficients
that arise in the OPE expansion of products of Wick monomials defined
using an arbitrary prescription for Wick monomials that satisfies
the axioms of Subsection \ref{subsec:algebra, axioms, state space Wick poly}.
Let $\H A{}$ denote the Hadamard normal-ordered prescription for
Wick monomials and let $\W A$ be an arbitrary prescription. The key
equations \eqref{eq:Phi^H = Z Phi} and \eqref{def:Phi=Z Phi^H}
relating $\Phi_{A}^{H}$ and $\Phi_{A}$ via $\mathcal{Z}$ and its
inverse, respectively, were obtained in Subsection \ref{subsec:uniq Wick poly and novel reparm}.

To obtain an OPE for $\braket{\W{A_{1}}(x_{1})\cdots\W{A_{n}}(x_{n})}_{\Psi}$
for our arbitrary prescription for Wick monomials, we now use eq.
(\ref{def:Phi=Z Phi^H}) to write 
\[
\braket{\W{A_{1}}(x_{1})\cdots\W{A_{n}}(x_{n})}_{\Psi}=\sum_{C_{1}}\cdots\sum_{C_{n}}\Zinv{A_{1}}{C_{1}}(x_{1})\cdots\Zinv{A_{n}}{C_{n}}(x_{n})\braket{\H C1(x_{1})\cdots\H Cn(x_{n})}_{\Psi}.
\]
It should be noted that the sums on the right-hand side include only
a finite number of terms because $(\mathcal{Z}^{-1})_{A}^{C}=0$ unless
$[C]\le[A]$. Next, we use the OPE, eq.~(\ref{eq:Had OPE def}),
for the Hadamard normal-ordered Wick monomials, with OPE coefficients
given by eq.~(\ref{eq:H NO coefficients}) to obtain 
\begin{align}
 & \braket{\W{A_{1}}(x_{1})\cdots\W{A_{n}}(x_{n})}_{\Psi}\\
 & \qquad\approx\sum_{C_{1}}\cdots\sum_{C_{n}}\Zinv{A_{1}}{C_{1}}(x_{1})\cdots\Zinv{A_{n}}{C_{n}}(x_{n})\left[\sum_{C_{0}}(C_{H}){}_{C_{1}\cdots C_{n}}^{C_{0}}(x_{1},\dots,x_{n};z)\braket{\H C0(z)}_{\Psi}\right]\nonumber 
\end{align}
Finally, we use eq.~\eqref{eq:Phi^H = Z Phi} to write $\braket{\H B{}(z)}_{\Psi}$
in terms of one-point Wick monomials in the prescription that we are
using, 
\begin{align}
 & \braket{\W{A_{1}}(x_{1})\cdots\W{A_{n}}(x_{n})}_{\Psi}\\
 & \approx\sum_{C_{1}}\cdots\sum_{C_{n}}\Zinv{A_{1}}{C_{1}}(x_{1})\cdots\Zinv{A_{n}}{C_{n}}(x_{n})\left[\sum_{C_{0}}(C_{H}){}_{C_{1}\cdots C_{n}}^{C_{0}}(x_{1},\dots,x_{n};z)\left[\sum_{B}\Z{C_{0}}B(z)\braket{\W B(z)}_{\Psi}\right]\right].\nonumber 
\end{align}
This provides an OPE expansion for $\braket{\W{A_{1}}(x_{1})\cdots\W{A_{n}}(x_{n})}_{\Psi}$,
from which we can read off the OPE coefficients 
\begin{align}
 & C_{A_{1}\cdots A_{n}}^{B}(x_{1},\dots,x_{n};z)\label{eq:gen Wick OPE coefficients}\\
 & \approx\sum_{C_{0}}\Z{C_{0}}B(z)\left[\sum_{C_{1}}\cdots\sum_{C_{n}}\Zinv{A_{1}}{C_{1}}(x_{1})\cdots\Zinv{A_{n}}{C_{n}}(x_{n})(C_{H}){}_{C_{1}\cdots C_{n}}^{C_{0}}(x_{1},\dots,x_{n};z)\right].\nonumber 
\end{align}

Expressions for the Hadamard normal-ordered coefficients $(C_{H}){}_{C_{1}\cdots C_{n}}^{C_{0}}$
were given in terms of $\T[\beta]$ and $H$ by eq.~(\ref{eq:H NO coefficients})
of Theorem \ref{thm: Had NO OPE coef}. The mixing matrix $\Z AB$
was given in terms of $F_{n}$ via eq.~(\ref{eq:Wick uniqueness Z})
of Theorem \ref{thm:Wick uniquenes Z} and (\ref{eq:Z=deriv Fn})
of Proposition \ref{prop:Wick uniqueness Fn}. As described in Subsection
\ref{subsec:uniq Wick poly and novel reparm}, $(\mathcal{Z}^{-1})_{A}^{B}$
can also be expressed in terms of $F_{n}$ using eq.~(\ref{eq:Wick uniqueness Z})
and eq.~\eqref{eq:Z inv (Fn)}. Thus, as desired, eq.~(\ref{eq:gen Wick OPE coefficients})
yields a formula for the OPE coefficients $C_{A_{1}\cdots A_{n}}^{B}$
in terms of a Hadamard parametrix $H$, the geometric factors $\T[\beta]$,
and the smooth functions $F_{n}$ (which characterize the difference
between $\W A$ and $\H A{}$). 
\begin{thm}
\label{thm:existence Wick coef and associativity}For any prescription
for the Wick monomials $\{\W A|A\equiv\alpha_{1}\cdots\alpha_{n}\}_{n\in\mathbb{N}_{0}}$
compatible with axioms W1-$\conservation$, there exists an OPE in
the sense of (\ref{eq:OPE precise asymp rel}) with local and covariant
defined OPE coefficients $C_{A_{1}\cdots A_{n}}^{B}(x_{1},\dots,x_{n};z)$
given by (\ref{eq:gen Wick OPE coefficients}). These OPE coefficients
satisfy (\ref{eq:pertinent assoc cond for C_H}) (with the $H$-subscripts
removed) as well as the general associativity conditions of \cite{HW_Axiomatic_QFTCS,Hollands_perturbative_OPE_CS}.
The coefficients are also compatible with the scaling degree properties
(\ref{eq:[A] via scaling degree}) and (\ref{eq:scaling degree OPE coef}). 
\end{thm}
\begin{proof}[Sketch of proof.]
See Appendix \ref{sec:proofs for gen Wick coef subsec}. 
\end{proof}
Equation~(\ref{eq:gen Wick OPE coefficients}) provides a complete
characterization of the OPE coefficients for an arbitrary prescription
for Wick monomials and, thus, achieves the primary goal of this subsection.
However, there are important properties of the general Wick coefficients
which are not immediately apparent from (\ref{eq:gen Wick OPE coefficients})
but will be extremely useful for our analysis of the flow relations
in future sections as well as for illuminating the general qualitative
structure of the Wick coefficients. In particular, as we will show,
the special form of the Wick mixing matrices (\ref{def:mixing matrix Z})
and the factorization properties (\ref{eq:H intermed coef}) of the
Hadamard normal-ordered products together imply knowledge of just
the $C_{\phi\cdots\phi}^{I}$-coefficients is sufficient for one to
determine all other Wick OPE coefficients. This property of the Wick
coefficients will greatly reduce the number of independent flow relations
we must consider in future sections. Moreover, the relative simplicity
of the $C_{\phi\cdots\phi}^{I}$-coefficients permits us to obtain
an explicit formula for these elementary coefficients in terms of
$H$ and $F_{n}$, thereby generalizing the Hadamard normal-ordered
formula (\ref{eq:H id coefs}) to arbitrary prescriptions.

We now outline the steps that allow us to obtain an arbitrary OPE
coefficient $C_{A_{1}\cdots A_{n}}^{B}$ in terms of $\T[\beta]$
and OPE coefficients of the form $C_{\phi\cdots\phi}^{I}$. We will
then give an explicit formula (see eq.~\eqref{eq:gen linear C ident coef-1})
for $C_{\phi\cdots\phi}^{I}$ in terms of the Hadamard parametrix
$H$ and the functions $F_{n}$. Finally, we obtain in Proposition
\ref{prop:induct construction of Wick fields via id OPE coefs} an
explicit (inductive) construction for the Wick monomials in terms
of the OPE coefficients $C_{\phi\cdots\phi}^{I}$.

We first note that eq.~(\ref{eq:gen Wick OPE coefficients}) implies
that $C_{A_{1}\cdots A_{n}}^{B}=0$ whenever $m>K$ for $m\equiv[B]_{\phi}$,
since this property holds for $(C_{H})_{A_{1}\cdots A_{n}}^{B}$ and
the mixing matrices $\Z AB$ and $\Zinv AB$ never increase the number
of powers of $\phi$ appearing in any Wick monomial. For the case
$m=K$, the only terms in $\Z AB$ and $\Zinv AB$ that can contribute
nontrivially to eq.~(\ref{eq:gen Wick OPE coefficients}) are $\delta_{A}^{B}$.
Thus, for $m=K$ we obtain, 
\begin{equation}
C_{A_{1}\cdots A_{n}}^{B}(x_{1},\dots,x_{n};z)=(C_{H}){}_{A_{1}\cdots A_{n}}^{B}(x_{1},\dots,x_{n};z)=\text{\textnormal{sym}}_{\beta}\prod_{i=1}^{n}\prod_{j=1}^{k_{i}}\na ij\T[\beta_{p(i,j)}](x_{i};z),\label{eq:gen Taylor coef}
\end{equation}
with $p(i,j)\equiv j+\sum_{q=1}^{i-1}k_{q}$ and $\text{sym}_{\beta}$
defined as in (\ref{def:sym_beta}). Thus, for $m=K$ the OPE coefficients
for a general prescription are the same as for the Hadamard normal
ordered prescription, and depend only on the geometrical factors $\T[\beta]$.

Next, we show that the OPE coefficients $C_{A_{1}\cdots A_{n}}^{B}$
such that $0<m<K$ are determined by OPE coefficients with $B=I$
together with OPE coefficients of the form eq.~(\ref{eq:gen Taylor coef}).
More precisely, 
\begin{align}
 & C_{A_{1}\cdots A_{n}}^{B}(x_{1},\dots,x_{n};z)=\sum_{\{P_{1},P_{2}\}\in\mathcal{P}_{m}(\mathcal{S}_{A})}C{}_{A_{1}'\cdots A_{n}'}^{B}(x_{1},\dots,x_{n};z)C{}_{A_{1}''\cdots A_{n}''}^{I}(x_{1},\dots,x_{n}),\label{eq:intermed coef id gen-1}
\end{align}
with the notation as in Proposition \ref{Prop:intermed H-coef recursion rel}.
Since we have $[A'_{1}]_{\phi}+\cdots+[A'_{n}]_{\phi}=m\equiv[B]_{\phi}$,
the coefficients $C_{A_{1}'\cdots A_{n}'}^{B}(x_{1},\dots,x_{n};z)$
are of the form (\ref{eq:gen Taylor coef}). Thus, (\ref{eq:intermed coef id gen-1})
expresses a general OPE coefficient with $0<m<K$ in terms of OPE
coefficients with $B=I$. Formula (\ref{eq:intermed coef id gen-1})
is a special case of the following proposition when $p=m$. 
\begin{prop}
\label{prop:gen Wick intermed coef} For $0<p\le m<K$, the OPE coefficients
given by (\ref{eq:gen Wick OPE coefficients}) satisfy the same formula
(\ref{eq:H intermed coef}) as derived in Proposition \ref{Prop:intermed H-coef recursion rel}
for the Hadamard normal-ordered OPE coefficients. i.e., formula (\ref{eq:H intermed coef})
remains a valid formula when the $H$-subscripts are removed. 
\end{prop}
\begin{proof}[Sketch of proof.]
See Appendix \ref{sec:proofs for gen Wick coef subsec}. 
\end{proof}
The following proposition shows that any Wick OPE coefficient $C_{A_{1}\cdots A_{n}}^{B}(x_{1},\dots,x_{n};z)$
is ultimately fixed by OPE coefficients of the form $C_{\phi\cdots\phi}^{C_{i}}(x_{1},\dots,x_{k_{i}})$
for $[C_{i}]_{\phi}\le[A_{i}]_{\phi}=k_{i}$. When combined with the
previous proposition, this implies all Wick OPE coefficients may be
obtained from a finite number of OPE coefficients of the form $C_{\phi\cdots\phi}^{I}$. 
\begin{prop}
\label{prop:C^A1...An_B via point splitting}The Wick OPE coefficients
(\ref{eq:gen Wick OPE coefficients}) satisfy, 
\begin{align}
C_{A_{1}\cdots A_{n}}^{B}(x_{1},\dots,x_{n};z) & =\lim_{\vec{y}_{1}\to x_{1}}\cdots\lim_{\vec{y}_{n}\to x_{n}}\nabla_{\alpha_{(1,1)}}^{y_{(1,1)}}\cdots\nabla_{\alpha_{(n,k_{n})}}^{y_{(n,k_{n})}}\left[\vphantom{\sum}C_{\phi\cdots\phi}^{B}(\vec{y}_{1},\dots,\vec{y}_{n};z)+\right.\label{eq:C^A1...An_I point-splitting formula-1}\\
 & -\sum_{\substack{[C_{1}]<[A_{1}]\\{}
[C_{1}]_{\phi}<[A_{1}]_{\phi}
}
}\cdots\sum_{\substack{[C_{n}]<[A_{n}]\\{}
[C_{n}]_{\phi}<[A_{n}]_{\phi}
}
}\left.C_{\phi\cdots\phi}^{C_{1}}(\vec{y}_{1};x_{1})\cdots C_{\phi\cdots\phi}^{C_{n}}(\vec{y}_{n};x_{n})C_{C_{1}\cdots C_{n}}^{B}(x_{1},\dots,x_{n};z)\vphantom{\sum}\right],\nonumber 
\end{align}
where we define the shorthand $\vec{y}_{i}\equiv y_{(i,1)},\dots,y_{(i,k_{i})}$
and denote $k_{i}\equiv[A_{i}]_{\phi}$. 
\end{prop}
\begin{proof}
See Appendix \ref{sec:proofs for gen Wick coef subsec}. 
\end{proof}
\begin{rem}
Recall the definition (\ref{eq:[A]}) of $[C]\equiv(D-2)/2\times[C]_{\phi}+[C]_{\nabla}$.
Hence, for a fixed $[A]$, there are only finitely-many $[C]_{\phi}$
and $[C]_{\nabla}$ such that $[C]<[A]$ and, thus, the $C$-sums
in (\ref{eq:C^A1...An_I point-splitting formula-1}) are all finite
sums. Note also $[C]<[A]$ iff 
\begin{equation}
[C]_{\nabla}<\frac{D-2}{2}\times\left([A]_{\phi}-[C]_{\phi}\right)+[A]_{\nabla},\label{eq:[C]_nabla bound}
\end{equation}
where the right-hand side is non-negative for $[C]_{\phi}<[A]_{\phi}$
and reduces to $[A]_{\nabla}$ when $D=2$. 
\end{rem}
The preceding proposition enables us to inductively compute all OPE
coefficients using only the OPE coefficients $C_{\phi\cdots\phi}^{B}$
as input. To observe this, first note that, for these elementary OPE
coefficients, 
\[
C_{(\nabla_{\alpha_{1}}\phi)\cdots(\nabla_{\alpha_{n}}\phi)}^{B}(x_{1},\dots,x_{n};z)=\nabla_{\alpha_{1}}^{(x_{1})}\cdots\nabla_{\alpha_{n}}^{(x_{n})}C_{\phi\cdots\phi}^{B}(x_{1},\dots,x_{n};z),
\]
and, thus, knowledge of $C_{\phi\cdots\phi}^{B}$ implies knowledge
of $C_{(\nabla_{\alpha_{1}}\phi)\cdots(\nabla_{\alpha_{n}}\phi)}^{B}$
for all $\alpha_{i}$. Hence, by assumption, we begin with knowledge
of all OPE coefficients $C_{A_{1}\cdots A_{n}}^{B}$ such that $[A_{i}]_{\phi}=1$
and $[A_{i}]_{\nabla}<\infty$ for $i\in\{1,\dots,n\}$. Noting the
bounds on the $C_{i}$-sums, we may, therefore, immediately use formula
(\ref{eq:C^A1...An_I point-splitting formula-1}) to calculate any
$C_{A_{1}\cdots A_{n}}^{B}$ such that, for all $i$, $[A_{i}]_{\phi}\le2$
and $[A_{i}]_{\nabla}<\infty$. Of course, this, in turn, provides
enough data to compute any coefficient such that $[A_{i}]_{\phi}\le3$
and $[A_{i}]_{\nabla}<\infty$ and, in this way, we may obtain any
OPE coefficient $C_{A_{1}\cdots A_{n}}^{B}$ from formula (\ref{eq:C^A1...An_I point-splitting formula-1})
starting from knowledge of just $C_{\phi\cdots\phi}^{B}$. 
\begin{rem}
For any finite $[A_{i}]_{\phi}$ and $[A_{i}]_{\nabla}$, we emphasize
that the coefficient $C_{A_{1}\cdots A_{n}}^{B}$ can be computed
from (\ref{eq:C^A1...An_I point-splitting formula-1}) with only a
finite number of iterations. In particular, it is \emph{not }required
that we compute all $[A_{i}]_{\nabla}<\infty$ for a given $[A_{i}]_{\phi}$
before incrementing to $[A_{i}']_{\phi}=[A_{i}]_{\phi}+1$. By inequality
(\ref{eq:[C]_nabla bound}), computing $C_{A_{1}\cdots A_{n}}^{B}$
for any $[A_{i}]_{\phi}$ and $[A_{i}]_{\nabla}$ only requires knowledge
of coefficients $C_{C_{1}\cdots C_{n}}^{B}$ such that $[C_{i}]_{\phi}<[A_{i}]_{\phi}$
and $[C_{i}]_{\nabla}<(D-2)/2\times[A_{i}]_{\phi}+[A_{i}]_{\nabla}$.

Taken together, the above results allow us to express an arbitrary
Wick OPE coefficient $C_{A_{1}\cdots A_{n}}^{B}$ in terms\footnote{The only method we have provided for computing coefficients of the
form $C_{A_{1}\cdots A_{n}}^{I}$ from $C_{\phi\cdots\phi}^{I}$ is
via formula (\ref{eq:C^A1...An_I point-splitting formula-1}) of Proposition
\ref{prop:C^A1...An_B via point splitting}. However, coefficients
$C_{A_{1}\cdots A_{n}}^{B}$for $B\ne I$ may be computed from $C_{\phi\cdots\phi}^{I}$
either via formula (\ref{eq:C^A1...An_I point-splitting formula-1})
or, alternatively, by plugging the values of $C_{A_{1}\cdots A_{n}}^{I}$
back into formula (\ref{eq:H intermed coef}) of Proposition \ref{prop:gen Wick intermed coef}. } of the OPE coefficients $C_{\phi\cdots\phi}^{I}$ and pure geometrical
factors $\T[\beta]$. Finally, we give an explicit formula for $C_{\phi\cdots\phi}^{I}$.
To see how this formula is obtained, consider first the simplest case
of $C_{\phi\phi}^{I}$. We have, 
\begin{align}
C_{\phi\phi}^{I}(x_{1},x_{2};z) & \approx\mathcal{Z}_{I}^{I}(C_{H})_{\phi\phi}^{I}(x_{1},x_{2})+\sum_{\gamma_{1},\gamma_{2}}\Z{\gamma_{1}\gamma_{2}}I(z)(C_{H})_{\phi\phi}^{(\nabla_{\gamma_{1}}\phi\nabla_{\gamma_{2}}\phi)}(x_{1},x_{2};z)\nonumber \\
 & \approx H(x_{1},x_{2})+\sum_{\gamma_{1},\gamma_{2}}[\nabla_{\gamma_{1}}\nabla_{\gamma_{2}}F_{2}]_{z}\,\T[(\gamma_{1}](x_{1};z)\T[\gamma_{2})](x_{2};z)\nonumber \\
 & \approx H(x_{1},x_{2})+F_{2}(x_{1},x_{2};z),\label{eq:C phi phi I gen}
\end{align}
where in the last line, we used the fact that the series, 
\[
\sum_{\gamma_{1},\gamma_{2}}[\nabla_{\gamma_{1}}\nabla_{\gamma_{2}}F_{2}]_{z}\,\T[(\gamma_{1}](x_{1};z)\T[\gamma_{2})](x_{2};z),
\]
is simply the covariant Taylor expansion of the smooth function $F_{2}(x_{1},x_{2};z)$.
Proceeding in a similar manner and recalling formulas (\ref{eq:gen Wick OPE coefficients})
and (\ref{eq:gen Taylor coef}), we obtain the general formula, 
\begin{align}
C_{\phi\cdots\phi}^{I} & (x_{1},\dots,x_{n};z)\label{eq:gen linear C ident coef-1}\\
\approx F_{n} & (x_{1},\dots,x_{n};z)+\sum_{k=1}^{\lfloor n/2\rfloor}\sum_{\pi_{k}}H(x_{\pi(1)},x_{\pi(2)})\cdots H(x_{\pi(2k-1)},x_{\pi(2k)})F_{(n-2k)}(x_{\pi(2k+1)},\dots,x_{\pi(n)};z),\nonumber 
\end{align}
where the $\pi_{k}$ sums over the permutations of $\{1,\dots,n\}$
such that $\pi(1)<\pi(2),\pi(3)<\pi(4),\cdots,\pi(2k-1)<\pi(2k)$;
$\pi(1)<\pi(3)<\cdots<\pi(2k-1)$; and $\pi(2k+1)<\pi(2k+2)<\cdots<\pi(n)$.
Formula (\ref{eq:gen linear C ident coef-1}) generalizes the normal-ordered
formula (\ref{eq:H id coefs}) obtained in the previous subsection
to arbitrary prescriptions for the Wick monomials. 
\end{rem}
Formula (\ref{eq:gen linear C ident coef-1}) now implies the full
renormalization freedom for the Wick fields may be expressed entirely
in terms of the identity coefficients $\{C_{\phi\cdots\phi}^{I}(x_{1},\dots,x_{n};z)\}_{n\ge2}$
and, thus, the set of these coefficients uniquely determines a prescription
for the Wick fields $\{\W A|A\equiv\alpha_{1}\cdots\alpha_{n}\}_{n\in\mathbb{N}_{0}}$.
To see this, note eq.~(\ref{eq:gen linear C ident coef-1}) implies
$C_{\phi\phi}^{I}$ is itself a Hadamard parametrix $\widetilde{H}\equiv H+F_{2}$
(in accordance with Remark \ref{rem:all quadratic fields are normal ordered}).
If we choose to normal order instead with respect to the parametrix
$\widetilde{H}$, i.e. use $\Phi_{B}^{\widetilde{H}}$ in formula
(\ref{def:Phi=Z Phi^H}) rather than $\H B{}$, then the preceding
manipulations would again yield formula (\ref{eq:gen linear C ident coef-1})
for $C_{\phi\cdots\phi}^{I}$ but now with all $H$'s replaced by
$\widetilde{H}=C_{\phi\phi}^{I}$. Since this formula depends only
on $F_{k\le n}$ and OPE coefficients of the form $C_{\phi\cdots\phi}^{I}$,
it may be iteratively inverted to express $F_{n}$ purely in terms
of $C_{\phi\cdots\phi}^{I}(x_{1},\dots,x_{k};z)$ for $k\le n$. The
claim is then an immediate consequence of Proposition \ref{prop:Wick uniqueness Fn}
and the Wick uniqueness theorem (Theorem \ref{thm:Wick uniquenes Z}).

Using identities \eqref{eq:Wick uniqueness Z} and (\ref{eq:Z inv (Fn)}),
our expression for $F_{n}$ in terms of the OPE coefficients allows
us to similarly express $(\mathcal{Z}^{-1}){}_{B}^{A}$ purely in
terms of $C_{\phi\cdots\phi}^{I}(x_{1},\dots,x_{k};z)$ for $k\le n$.
A Wick monomial $\Phi_{A}$ in any prescription satisfying axioms
W1-W8 can, thus, be expressed via $\W A=\sum_{[B]\le[A]}(\mathcal{Z}^{-1})_{A}^{B}\Phi_{B}^{\widetilde{H}}$
in terms of just $\{C_{\phi\cdots\phi}^{I}(x_{1},\dots,x_{n};z)\}_{n\le[A]}$
and products of the linear field observable $\phi$, noting that the
normal-ordered Wick fields $\Phi_{B}^{\widetilde{H}}$ are themselves
defined in (\ref{eq:H normal ordered}) with respect to only products
of the linear field observable $\phi$ and the OPE coefficient $\widetilde{H}=C_{\phi\phi}^{I}$.
An explicit inductive formula for $\Phi_{A}$ expressed purely in
terms of $\phi$, $C_{\phi\cdots\phi}^{I}$ and the geometric factors
$\T[\beta]$ is obtained in the following proposition. 
\begin{prop}
\label{prop:induct construction of Wick fields via id OPE coefs}For
the OPE coefficients $C_{\phi\cdots\phi}^{I}$ given by the formula
(\ref{eq:gen linear C ident coef-1}), the monomial $\Phi_{A}$ in
any prescription satisfying axioms W1-W8 satisfies: 
\begin{align}
 & (\nabla_{\alpha_{1}}\phi\cdots\nabla_{\alpha_{n}}\phi)(f)=\int_{z,x_{1},\dots,x_{n}}\negthickspace f^{\alpha_{1}\cdots\alpha_{n}}(z)\delta(z,x_{1},\dots,x_{n})\nabla_{\alpha_{1}}^{(x_{1})}\cdots\nabla_{\alpha_{n}}^{(x_{n})}\left[\vphantom{\sum_{B}}\phi(x_{1})\cdots\phi(x_{n})+\right.\label{eq:explicit Wick monomial for gen Rx}\\
 & -\sum_{\substack{m<n\\{}
[B]<[A]
}
}\sum_{\pi\in\Pi_{m}}\left.C_{\phi\cdots\phi}^{I}(x_{\pi(m+1)},\dots,x_{\pi(n)};z)\T[\beta_{1}](x_{\pi(1)};z)\cdots\T[\beta_{m}](x_{\pi(m)};z)(\nabla_{\beta_{1}}\phi\cdots\nabla_{\beta_{m}}\phi)(z)\vphantom{\sum_{B}}\right],\nonumber 
\end{align}
where $\Pi_{m}$ denotes the set of all permutations of $\{1,\dots,n\}$
such that $\pi(1)<\pi(2)<\cdots<\pi(m)$ and $\pi(m+1)<\pi(m+2)<\cdots<\pi(n)$,
and we abbreviate 
\[
\int_{z,x_{1},\dots,x_{n}}\equiv\int_{\times^{(n+1)}M}d\mu_{g}(z)d\mu_{g}(x_{1})\cdots d\mu_{g}(x_{n}).
\]
\end{prop}
\begin{proof}
See Appendix \ref{sec:proofs for gen Wick coef subsec}. 
\end{proof}

\subsection{OPE coefficients of (unextended) time-ordered products \label{subsec:TOP}}

As we shall see in the next section, the flow relations for OPE coefficients
that we shall obtain in Lorentzian spacetimes will involve expansions
of time-ordered products\textemdash rather than ordinary products\textemdash of
Wick monomials. Away from the diagonals\footnote{The ``diagonals'' are the subset of the product manifold, $\left\{ (x_{1},\dots x_{n})\in\times^{n}M|\;x_{i}=x_{j}\text{ for any \ensuremath{i,j\in\{1,\dots,n\}}}\right\} $.
Thus, ``away from the diagonals'' means when all points are distinct.} the ``unextended time-ordered product'' of Wick monomials is defined
by 
\begin{equation}
T_{0}\{\Phi_{A_{1}}(x_{1})\cdots\Phi_{A_{n}}(x_{n})\}=\Phi_{A_{P(1)}}(x_{P(1)})\cdots\Phi_{A_{P(n)}}(x_{P(n)}),\label{naiveto}
\end{equation}
where $P$ is a permutation of $\{1,\dots,n\}$ such that $x_{P(i)}\notin J^{-}(x_{P(i+1)})$,
where $J^{-}$ denotes the causal past. In other words, $T_{0}\{\Phi_{A_{1}}(x_{1})\cdots\Phi_{A_{n}}(x_{n})\}$
re-orders the product $\Phi_{A_{1}}(x_{1})\cdots\Phi_{A_{n}}(x_{n})$
by the ``time'' at which the Wick monomial is being evaluated. The
right side of eq.~(\ref{naiveto}) yields a well-defined (algebra-valued)
distribution on the product manifold $\times^{n}M$ minus all of the
diagonals.

Renormalization theory is primarily concerned with the ``extension
of $T_{0}\{\Phi_{A_{1}}(x_{1})\cdots\Phi_{A_{n}}(x_{n})\}$ to the
diagonals'': i.e., obtaining (algebra-valued) distributions $T\{\Phi_{A_{1}}(x_{1})\cdots\Phi_{A_{n}}(x_{n})\}$
that are well-defined on all of $\times^{n}M$, including the diagonals
and defined such that, 
\[
T\{\Phi_{A_{1}}(x_{1})\cdots\Phi_{A_{n}}(x_{n})\}=T_{0}\{\Phi_{A_{1}}(x_{1})\cdots\Phi_{A_{n}}(x_{n})\},
\]
away from all diagonals. In curved spacetime, it has been proven \cite{HW_Existence_TOP}
that there exist ``extensions'' of (\ref{naiveto}) that are compatible
with a list of axioms that generalize those stated here ($\local$-$\conservation$)
for Wick powers. However, generally, there are additional ``contact
term'' ambiguities in these extensions, corresponding to the freedom
to add finitely-many ``$\delta$-function-type'' terms on the diagonals.
Although these ambiguities can be fully characterized \cite{HW_Conservation_Stress-energy},
they greatly complicate the analysis of time-ordered products. For
the integral on the right-hand side of Lorentzian flow relations such
as eq.(\ref{eq:Lorentz inv id fe}) to be well-defined, it is necessary
that the unextended time-ordered products be extended to, at least,
all partial diagonals involving the integration variable, $y$. Fortunately,
as we shall see in Section \ref{sec:Minkowski-flow-relations}, the
unextended time-ordered-products will satisfy flow relations where
the extension to the requisite partial diagonals is unambiguous and,
thus, independent of contact terms. Therefore, we will only ever need
to consider the OPE of unextended time-ordered-products and the field
redefinition freedom of its coefficients, and we may thereby bypass
all of the usual complications of renormalization theory.

It is clear that the unextended time-ordered products satisfy OPE
relations of the form, 
\begin{equation}
\braket{T_{0}\{\Phi_{A_{1}}(x_{1})\cdots\Phi_{A_{n}}(x_{n})\}}_{\Psi}\approx\sum_{B}C_{T_{0}\{A_{1}\cdots A_{n}\}}^{B}(x_{1},\dots,x_{n};z)\braket{\Phi_{B}(z)}_{\Psi},\label{eq:OPE defT0}
\end{equation}
where 
\begin{align}
C_{T_{0}\{A_{1}\cdots A_{n}\}}^{B}(x_{1},\dots,x_{n};z) & =C_{A_{P(1)}\cdots A_{P(n)}}^{B}(x_{P(1)},\dots,x_{P(n)};z),\label{eq:def T0 OPE coefs}
\end{align}
with the permutations $P$ as defined in (\ref{naiveto}). It is understood
that eq.~(\ref{eq:OPE defT0}) holds only on $\times^{n}M$ minus
all of the diagonals. As described in the following proposition, the
explicit form of the time-ordered OPE coefficients (\ref{eq:def T0 OPE coefs})
is readily obtained from our previously-stated expressions for the
Wick OPE coefficients in terms of a Hadamard parametrix $H$, the
geometrical quantities $\T[\beta]$, and the smooth functions $F_{n}$
that characterize the difference between the Hadamard normal-ordering
prescription for Wick monomials and the prescription being used.
\begin{prop}
\emph{\label{thm:explicit TO coef}For any fixed prescription for
the Wick monomials, the time-ordered OPE coefficients }(\ref{eq:def T0 OPE coefs})\emph{
are simply obtained from the formula for $C_{A_{1}\cdots A_{n}}^{B}(x_{1},\dots,x_{n};z)$
by individually time-ordering all Hadamard parametrices: i.e., replacing
every occurrence of $H$ with its corresponding Feynman parametrix}\footnote{\label{fn:i epsilon for H_F}The precise asymptotic behavior of the
distribution kernel for $H_{F}$ is obtained by replacing the ``$i\epsilon$-prescription\textquotedblright{}
in the expression (\ref{eq:H}) for $H$ with the usual Feynman prescription:
i.e., making the following substitution, 
\[
[2i0^{+}\left(T(x_{1})-T(x_{2})\right)+(0^{+})^{2}]\to i0^{+}
\]
}\emph{, 
\begin{equation}
H_{F}(x_{1},x_{x})=H(x_{1},x_{2})-i\Delta^{\text{adv}}(x_{1},x_{2}),\label{eq:H_F def}
\end{equation}
with $\Delta^{\text{adv}}$ denoting the advanced Green's function.
} 
\end{prop}
\begin{proof}
By applying the definition of time-ordering (\ref{eq:def T0 OPE coefs})
to the formula (\ref{eq:gen Wick OPE coefficients}) for the general
Wick OPE coefficients, it is straightforwardly shown that we have
\begin{align}
 & C_{T_{0}\{A_{1}\cdots A_{n}\}}^{B}(x_{1},\dots,x_{n};z)\\
 & \approx\sum_{C_{0}}\mathcal{Z}_{C_{0}}^{B}(z)\left[\sum_{C_{1}}\cdots\sum_{C_{n}}(\mathcal{Z}^{-1}){}_{A_{1}}^{C_{1}}(x_{1})\cdots(\mathcal{Z}^{-1}){}_{A_{n}}^{C_{n}}(x_{n})(C_{H}){}_{T_{0}\{C_{1}\cdots C_{n}\}}^{C_{0}}(x_{1},\dots,x_{n};z)\right],\nonumber 
\end{align}
where we use the notation $(C_{H}){}_{T_{0}\{A_{1}\cdots A_{n}\}}^{B}$
for the OPE coefficients of the unextended time-ordered products \eqref{eq:def T0 OPE coefs}
when the Wick fields are defined via a Hadamard normal-ordering prescription,
$\Phi_{A}=\Phi_{A}^{H}$. It then follows from the factorization property
(\ref{eq:H intermed coef p=m}) of the Hadamard normal-ordered
OPE coefficients that we have 
\begin{align}
 & (C_{H}){}_{T_{0}\{C_{1}\cdots C_{n}\}}^{C_{0}}(x_{1},\dots,x_{n};z)\\
 & =\sum_{\{P_{1},P_{2}\}\in\mathcal{P}_{m}(\mathcal{S})}(C_{H})_{T_{0}\{A_{1}'\cdots A_{n}'\}}^{C_{0}}(x_{1},\dots,x_{n};z)(C_{H})_{T_{0}\{A_{1}''\cdots A_{n}''\}}^{I}(x_{1},\dots,x_{n}).\nonumber 
\end{align}
The first factor in each product is unaffected by time-ordering because
they depend only on symmetric combinations of $\T[\beta]$-factors,
i.e., we have 
\[
(C_{H})_{T_{0}\{A_{1}'\cdots A_{n}'\}}^{C_{0}}(x_{1},\dots,x_{n};z)=(C_{H})_{A_{1}'\cdots A_{n}'}^{C_{0}}(x_{1},\dots,x_{n};z),
\]
where the explicit form of the right-hand side is given by (\ref{eq:T-factor coef}).
Finally, recalling (\ref{eq:H id coefs}), we have 
\begin{align}
(C_{H})_{T_{0}\{A_{1}''\cdots A_{n}''\}}^{I} & =T_{0}\left\{ \sum_{P\in\mathcal{M}({\mathcal{S}})}\prod_{\{(v,i),(w,j)\}\in P}\nabla_{\alpha_{v,i}}^{(x_{v})}\nabla_{\alpha_{w,j}}^{(x_{w})}H(x_{v},x_{w})\right\} \nonumber \\
 & =\sum_{P\in\mathcal{M}({\mathcal{S}})}\prod_{\{(v,i),(w,j)\}\in P}\nabla_{\alpha_{v,i}}^{(x_{v})}\nabla_{\alpha_{w,j}}^{(x_{w})}T_{0}\left\{ H(x_{v},x_{w})\right\} ,\label{T0[C_H^I_A1...An]}
\end{align}
where the second line follows from the fact that tensor products of
ordinary $c$-number distributions commute. Altogether, we conclude
the time-ordering map acts non-trivially only on the Hadamard parametrices
and in the way specified by the proposition. We note, when $x_{1}\ne x_{2}$,
the time-ordered Hadamard parametrix $T_{0}\{H(x_{1},x_{2})\}$ is
equivalent to the Feynman parametrix (\ref{eq:H_F def}). 
\end{proof}
\begin{rem}
\label{rem:ext of T0 H}Although $T_{0}\{H(x_{1},x_{2})\}$ is a priori
only defined away from $x_{1}=x_{2}$, its extension to its diagonal
$x_{1}=x_{2}$ uniquely yields the Feynman parametrix (\ref{eq:H_F def}),
because the scaling degree of $T_{0}\{H\}$ is $D-2$ which is less
than that of the Dirac delta distribution (and all of its distributional
derivatives) and, thus, there do not exist any possible ``contact
terms'' with the correct scaling degree. 
\end{rem}
As examples, from eqs.~(\ref{eq:phi phi Had coef}) and (\ref{eq:C phi^2 phi^2 cases}),
we see that for the Hadamard normal-ordering prescription, we have,
\begin{equation}
(C_{H}){}_{T_{0}\{\phi\phi\}}^{I}(x_{1},x_{2};z)=H_{F}(x_{1},x_{2}),\quad\quad(C_{H}){}_{T_{0}\{\phi^{2}\phi^{2}\}}^{I}(x_{1},x_{2};z)=H_{F}(x_{1},x_{2})H_{F}(x_{1},x_{2}).
\end{equation}
The wavefront set calculus implies $H^{2}$ is a well-defined distribution
on (a convex normal neighborhood of) the product manifold $M\times M$
and, thus, the ordinary OPE coefficient $(C_{H})_{\phi^{2}\phi^{2}}^{I}$
is similarly well-defined. However, the pointwise product of the Feynman
parametrix, $H_{F}^{2}$, is only well-defined as a distribution on
the product manifold minus the diagonal so the time-ordered coefficient
$(C_{H}){}_{T_{0}\{\phi^{2}\phi^{2}\}}^{I}$ is thus only well-defined
as a distribution for $x_{1}\ne x_{2}$.

The advanced Green's function scales almost homogeneously and, thus,
$H_{F}$ defined via (\ref{eq:H_F def}) will scale almost homogeneously
if and only if $H$ is compatible with axiom $\scaling$. Note (\ref{eq:H_F def})
is symmetric in its spacetime variables and solves the inhomogeneous
Klein-Gordon equation with ``$\delta$-source'' up to a smooth remainder,
\begin{equation}
KH_{F}(x_{1},x_{2})=-i\delta^{D}(x_{1},x_{2})\mod C^{\infty}(M\times M)\label{eq:KH_F}
\end{equation}
Any bi-distribution satisfying (\ref{eq:KH_F}) is referred to as
a parametrix of a fundamental solution for the differential operator
$K$. If $H(x_{1},x_{2})$ is any Hadamard parametrix (of the homogeneous
Klein-Gordon equation) in $D>2$ satisfying the conservation constraint
(\ref{eq:H conservation condition-1}), then the Feynman parametrix
defined via (\ref{eq:H_F def}) will then necessarily satisfy, 
\begin{equation}
\left[\nabla_{b}^{(x_{1})}K_{x_{2}}H_{F}(x_{1},x_{2})+i\nabla_{b}^{(x_{1})}\delta(x_{1},x_{2})\right]_{x_{1},x_{2}=z}=0.\label{eq:conservation constr for H_F}
\end{equation}
Conversely, for any Feynman parametrix satisfying (\ref{eq:conservation constr for H_F}),
the corresponding Hadamard parametrix, $H=H_{F}+i\Delta^{\text{adv}}$,
will satisfy the conservation constraint (\ref{eq:H conservation condition-1}).

\section{Flow relations for OPE coefficients in flat Euclidean space\label{sec:flat Eucl fe}}

\global\long\def\note#1{{\color{blue}[#1]}}
 In this section, we obtain flow equations in $m^{2}$ for the Wick
OPE coefficients in flat Euclidean space. We initially focus our attention
on the flow relation for $C_{\phi\phi}^{I}$, since our analysis of
the preceding section implies the flow relations for all other Wick
OPE coefficients can be readily obtained after the flow relation for
$C_{\phi\phi}^{I}$ is known. We begin in Subsection \ref{subsec:Eucl vac NO fe }
by deriving flow relations for the case where the Euclidean Green's
function $G_{E}(x_{1},x_{2})$ is used to define a Hadamard normal-ordering
prescription, so $(C_{G}){}_{\phi\phi}^{I}(x_{1},x_{2};z)=G_{E}(x_{1},x_{2})$.
These flow equations (see eq.~(\ref{eq:flow eq Eucl id coef}) below)
are the direct analogues of the Holland and Hollands flow equations
(\ref{eq:HH fe}) for Klein-Gordon theory in the limit as the infrared
cutoff $L$ is removed, i.e., $L\to+\infty$. Note that Hadamard normal-ordering
with respect to $G_{E}(x_{1},x_{2})$ corresponds to ordinary normal
ordering with respect to the Euclidean vacuum state. However, $G_{E}(x_{1},x_{2};m^{2})$
does not have smooth dependence in $m^{2}$ at $m^{2}=0$, so it is
not acceptable to use it in a Hadamard normal-ordering prescription
that is valid in an open interval in $m^{2}$ containing $m^{2}=0$.
Nevertheless, it can be used outside any open interval in $m^{2}$
containing $m^{2}=0$, and it is convenient to begin our consideration
of flow relations with it because the flow relation analysis is much
simpler when a Green's function (rather than a parametrix) is used
in the Hadamard normal-ordering prescription.

We turn then in Subsection \ref{subsec:Eucl Had NO fe} to the derivation
of Euclidean flow relations for the case of a Euclidean-invariant
parametrix that has smooth dependence on $m^{2}$ for all $m^{2}$,
including $m^{2}=0$. To avoid infrared divergences, this requires
introducing a cutoff function in the integral over all space appearing
in the flow relation. The cutoff function can be chosen to be Euclidean
invariant, so it will not spoil the Euclidean invariance of the flow
relations. However, it will unavoidably spoil the scaling behavior
of the flow relations. Nevertheless, we develop an algorithm for modifying
the flow relations which restores proper scaling behavior to any desired
scaling degree. We show any ambiguities in our algorithm are in a
1-1 correspondence with the ambiguities of Euclidean OPE coefficients
for Hadamard normal-ordered Wick fields (see Theorem \ref{thm:Eucl fe for H-NO id coef}).

\subsection{Vacuum normal ordering without an infrared cutoff (\texorpdfstring{$m^{2}>0$}{m\unichar{178}>0}) \label{subsec:Eucl vac NO fe }}

The Riemannian version of quantum field theory in curved spacetime
has been formulated by \cite{Riemannian_Wick_algebra} in close parallel
with the axiomatic formulation for the Lorentzian case given in Subsection
\ref{subsec:algebra, axioms, state space Wick poly}. An analogue
of the ``Hadamard normal-ordering'' prescription for defining Wick
monomials can then be given by choosing a local and covariant Green's
parametrix for the (now elliptic) Klein-Gordon operator. OPE coefficients
for the Euclidean Wick OPE coefficients can then be obtained in parallel
with the Lorentzian case away from the diagonals\footnote{Defining products of Euclidean Wick fields on diagonals generally
requires renormalization analogous to extending the Lorentzian unextended
time-ordered products to their diagonals and, thus, is subject to
additional contact-term renormalization ambiguities. }.

In this subsection, we will be concerned only with the case of flat,
Euclidean space $(\mathbb{R}^{D},\delta_{ab})$. In this case, there
is a unique Green's function, $G_{E}(x_{1},x_{2};m^{2})$, for the
operator, 
\begin{equation}
K=-\delta^{ab}\partial_{a}\partial_{b}+m^{2},\label{eq:Euclidean K}
\end{equation}
such that $G_{E}$ vanishes as $|x_{1}-x_{2}|\to\infty$. It would
be extremely convenient to use this Green's function in a Hadamard
normal-ordering prescription for Wick monomials. Indeed, since this
Green's function is the vacuum 2-point function of the Euclidean quantum
field theory, 
\begin{equation}
\langle\phi(x_{1})\phi(x_{2})\rangle_{\text{vac}}=G_{E}(x_{1},x_{2})
\end{equation}
it follows that Hadamard normal ordering with respect to $G_{E}(x_{1},x_{2};m^{2})$
corresponds to ordinary normal ordering with respect to the Euclidean
vacuum state. However, as previously mentioned, $G_{E}(x_{1},x_{2};m^{2})$
does not have smooth dependence on $m^{2}$ at $m^{2}=0$. Thus, it
is not acceptable to use it in a Hadamard normal-ordering prescription
that is valid in an open interval in $m^{2}$ containing $m^{2}=0$,
since the corresponding Wick monomials defined in this way will not
have the required smooth dependence on $m^{2}$. In order to obtain
an acceptable prescription that includes the case $m^{2}=0$, we therefore
must use a Green's function (or parametrix) for $K$ that has smooth
dependence on $m^{2}$. Nevertheless, there are significant simplifications
in the derivation of the flow relations for $G_{E}(x_{1},x_{2};m^{2})$.
Therefore, we will proceed by first obtaining flow relations for normal
ordering with respect to $G_{E}(x_{1},x_{2};m^{2})$ for $m^{2}>0$,
and then derive flow relations for normal ordering with respect to
a parametrix that is smooth in $m^{2}$.

Although we shall not need to make use of its explicit form, we note
that for $m^{2}>0$, $G_{E}(x_{1},x_{2};m^{2})$ is given explicitly
by, 
\begin{equation}
G_{E}(\Delta x;m^{2})=\int_{\mathbb{R}^{D}}\frac{d^{D}p}{(2\pi)^{D}}\frac{e^{ip\cdot\Delta x}}{p^{2}+m^{2}}=\frac{m^{(D-2)/2}}{(2\pi)^{D/2}(|\Delta x|^{2}){}^{(D-2)/4}}K_{(D-2)/2}(m\sqrt{|\Delta x|^{2}}),\label{eq:explicit G_E}
\end{equation}
where $\Delta x=x_{1}-x_{2}$ and $K_{(D-2)/2}$ is a modified Bessel
function of the second kind. It should be noted that $G_{E}(x_{1},x_{2};m^{2})$
is symmetric in $x_{1}$ and $x_{2}$. The wavefront set of $G_{E}$
is the same as the wavefront set of a (two-variable) $\delta$-function,
\begin{equation}
\text{WF}[G_{E}]=\text{WF}[\delta]\equiv\left\{ (x_{1},k;x_{2},-k)\in\times^{2}(T^{*}\mathbb{R}^{D}\backslash Z^{*}\mathbb{R}^{D})|x_{1}=x_{2},k\in T^{*}\mathbb{R}^{D}\backslash Z^{*}\mathbb{R}^{D}\right\} .\label{eq:WF delta}
\end{equation}
In particular, $G_{E}$ is smooth in $x_{1}$ and $x_{2}$ for $\Delta x\ne0$.
Furthermore, it follows from the form of its wavefront set that, when
smeared in one of its variables with any test function $f$, $G_{E}(x_{1},f;m^{2})$
is smooth\footnote{This is established by a straightforward application of \cite[Theorem 8.2.12]{Hormander_book}.
In fact, as explained in the next section, this property holds for
any translation invariant bi-distribution.} in $(x_{1},m^{2})$. In other words, as a Schwartz kernel, $G_{E}(x_{1},x_{2})$
defines a continuous linear map from $C_{0}^{\infty}(\mathbb{R}^{D})$
into $C^{\infty}(\mathbb{R}^{D})$.

For Wick monomials defined by normal ordering with respect to $G_{E}$,
the Euclidean OPE coefficients are given, away from their diagonals\footnote{Whereas the wavefront set calculus implies the pointwise products
of any Lorentzian Hadamard parametrix $H$ appearing in formula (\ref{eq:H NO coefficients})
are guaranteed to be well-defined as distributions, the corresponding
pointwise products of $G_{E}$ are generally ill-defined as distributions
on diagonals. }, by formula (\ref{eq:H NO coefficients}) with the replacement $H\to G_{E}$.
In particular, for the OPE coefficient $(C_{G}){}_{\phi\cdots\phi}^{I}$
with $n$ factors of $\phi$, we have $(C_{G}){}_{\phi\cdots\phi}^{I}=0$
when $n$ is odd, whereas when $n$ is even, we have, 
\begin{align}
(C_{G}){}_{\phi\cdots\phi}^{I}(x_{1},\dots,x_{n}) & =\sum_{\pi}G_{E}(x_{\pi(1)},x_{\pi(2)})\cdots G_{E}(x_{\pi(n-1)},x_{\pi(n)}),\label{eq:Schwinger functions}
\end{align}
where the $\pi$ sums over all permutations such that $\pi(1)<\pi(3)<\cdots<\pi(n-1)$
and $\pi(1)<\pi(2),\pi(3)<\pi(4),\cdots,\pi(n-1)<\pi(n)$. As discussed
in Subsection \ref{subsec:gen Wick OPE coef}, the Wick OPE coefficients
for any prescription are determined by the values of the $C_{\phi\cdots\phi}^{I}$-coefficients.

We first motivate the form of the flow equations for $(C_{G}){}_{\phi\cdots\phi}^{I}$
following Hollands \cite{Hollands_Action_Principle}. Since the OPE
coefficients $(C_{G})_{\phi\cdots\phi}^{I}$ defined by normal ordering
are just $n$-point ``Schwinger functions,'' they are formally given
by the functional integral, 
\begin{equation}
(C_{G})_{\phi\cdots\phi}^{I}(x_{1},\dots,x_{n})=\int_{\mathcal{S}'(\mathbb{R}^{D})}d\mu[\varphi]\,\varphi(x_{1})\cdots\varphi(x_{n}),\label{eq:id coef func int rep}
\end{equation}
with measure, 
\begin{equation}
d\mu[\varphi]\ =\mathcal{D}[\varphi]\frac{1}{Z_{0}}\exp\left(-S_{\text{KG}}[\delta_{ab},\mathbb{R}^{D}]\right),\label{eq:formal measure}
\end{equation}
Formal differentiation of eq.~(\ref{eq:id coef func int rep}) with
respect to $m^{2}$ yields, 
\begin{equation}
\frac{\partial}{\partial m^{2}}(C_{G}){}_{\phi\cdots\phi}^{I}(x_{1},\dots,x_{n})=-\frac{1}{2}\int_{\mathbb{R}^{D}}d^{D}y\int_{\mathcal{S}'(\mathbb{R}^{D})}d\mu\,\varphi^{2}(y)\varphi(x_{1})\cdots\varphi(x_{n}),\label{eq:formal m^2 deriv of id coef}
\end{equation}
This suggests that we should have the flow relation 
\begin{align}
\frac{\partial}{\partial m^{2}}(C_{G}){}_{\phi\cdots\phi}^{I}(x_{1},\dots,x_{n}) & =-\frac{1}{2}\int_{\mathbb{R}^{D}}d^{D}y(C_{G})_{\phi^{2}\phi\cdots\phi}^{I}(y,x_{1},\dots,x_{n}),\label{eq:flow eq Eucl id coef}
\end{align}
That this flow equation, eq.~\eqref{eq:flow eq Eucl id coef}, does
indeed hold will be seen to be a consequence of the following lemma: 
\begin{lem}
\label{lem:G_E flow eq} The Euclidean Green's function $G_{E}$ satisfies
the flow relation 
\begin{equation}
\frac{\partial}{\partial m^{2}}G_{E}(x_{1},x_{2};m^{2})=-\int_{\mathbb{R}^{D}}d^{D}yG_{E}(y,x_{1};m^{2})G_{E}(y,x_{2};m^{2}).\label{eq:G_E flow eq}
\end{equation}
\end{lem}
\begin{proof}
We note first that, by a trivial calculation, the commutator of the
differential operators $K=-\partial^{a}\partial_{a}+m^{2}$ and $\partial_{m^{2}}$
is given by, 
\begin{equation}
[K,\partial_{m^{2}}]=-I
\end{equation}
Thus, in particular, we have, 
\begin{align}
K_{y}\frac{\partial}{\partial m^{2}}G_{E}(y,x;m^{2}) & =-G_{E}(y,x;m^{2})+\frac{\partial}{\partial m^{2}}K_{y}G_{E}(y,x;m^{2})\nonumber \\
 & =-G_{E}(y,x;m^{2})+\frac{\partial}{\partial m^{2}}\delta(y,x)\nonumber \\
 & =-G_{E}(y,x;m^{2})
\end{align}
where we used the Green's function property, $K_{y}G_{E}(y,x;m^{2})=\delta(y,x)$,
to get the second line and we used the fact that the $\delta$-function
has no $m^{2}$ dependence to get the last line. As already noted,
the wavefront set of $G_{E}(y,x;m^{2})$ (and, hence, of $\partial_{m^{2}}G_{E}(y,x;m^{2})$)
is such that if we smear in $x$, we obtain a smooth function of $y$
and $m^{2}$. Therefore, for any test functions $f_{1}$ and $f_{2}$,
we have 
\begin{align}
\frac{\partial}{\partial m^{2}}G_{E}(f_{1},f_{2};m^{2}) & =\int_{\mathbb{R}^{D}}d^{D}y\delta^{D}(y,f_{1})\frac{\partial}{\partial m^{2}}G_{E}(y,f_{2};m^{2})\nonumber \\
 & =\int_{\mathbb{R}^{D}}d^{D}y\left[K_{y}G_{E}(y,f_{1};m^{2})\right]\frac{\partial}{\partial m^{2}}G_{E}(y,f_{2};m^{2})\nonumber \\
 & =\int_{\mathbb{R}^{D}}d^{D}yG_{E}(y,f_{1};m^{2})K_{y}\frac{\partial}{\partial m^{2}}G_{E}(y,f_{2};m^{2})\nonumber \\
 & =-\int_{\mathbb{R}^{D}}d^{D}yG_{E}(y,f_{1};m^{2})G_{E}(y,f_{2};m^{2}).\label{eq: G_E fe}
\end{align}
Here, in the third line, we integrated by parts twice, invoking the
fall-off behavior\footnote{We have restricted to the case $m^{2}>0$ here, but it is worth noting
that for $m^{2}=0$, the fall-off of $G_{E}$ is too slow to justify
the integration by parts.} of $G_{E}$ as $y\to\infty$. Equation (\ref{eq: G_E fe}) is just
the smeared form of eq.~(\ref{eq:G_E flow eq}). 
\end{proof}
As an immediate consequence of this lemma, we have 
\begin{thm}
\label{thm:vac NO flow rel gen n} The flow relation \eqref{eq:flow eq Eucl id coef}
holds for OPE coefficients $(C_{G}){}_{\phi\cdots\phi}^{I}(x_{1},\dots,x_{n})$
corresponding to Euclidean vacuum normal-ordered Wick fields. 
\end{thm}
\begin{proof}
\label{rem:proof of vacuum Euclidean C^phi...phi_I flow eq}To obtain
the flow equation (\ref{eq:flow eq Eucl id coef}), we apply $\partial_{m^{2}}$
to eq.~(\ref{eq:Schwinger functions}), and use eq.~(\ref{eq:G_E flow eq})
together with the fact that $(C_{G})_{\phi^{2}\phi\cdots\phi}^{I}=0$
when $n$ is odd and, when $n$ is even, 
\begin{align}
(C_{G})_{\phi^{2}\phi\cdots\phi}^{I}(y,x_{1},\dots,x_{n}) & \equiv\sum_{\pi}2G_{E}(y,x_{\pi(1)})G_{E}(y,x_{\pi(2)})G_{E}(x_{\pi(3)},x_{\pi(4)})\cdots G_{E}(x_{\pi(n-1)},x_{\pi(n)}),\label{eq:G NO C^phi^2 phi ... phi _I}
\end{align}
where the $\pi$-sum runs over the same permutations as in (\ref{eq:Schwinger functions}).
Equation (\ref{eq:flow eq Eucl id coef}) then follows by inspection. 
\end{proof}

\subsection{Hadamard normal ordering with an infrared cutoff \label{subsec:Eucl Had NO fe}}

We turn now to the modifications to the Euclidean flow relations that
arise when we consider the OPE coefficients corresponding to a Hadamard
normal-ordering prescription using a Euclidean invariant Hadamard
parametrix, $H_{E}(x_{1},x_{2};m^{2})$, that varies smoothly with
$m^{2}$ for all $m^{2}$, including $m^{2}=0$. That is, $H_{E}$
is required to satisfy, 
\begin{equation}
(-\partial^{2}+m^{2})H_{E}(x_{1},x_{2};m^{2})=\delta^{(D)}(x_{1},x_{2})+h_{E}(x_{1},x_{2};m^{2}),\label{eq:H_E as a Green's function parametrix}
\end{equation}
where $h_{E}(x_{1},x_{2};m^{2})$ is smooth in all of its variables
and symmetric in $(x_{1},x_{2})$. Clearly, the choice of $H_{E}$
is not unique, but any two such parametrices must differ from each
other by addition of a smooth, Euclidean-invariant function, $w(x_{1},x_{2};m^{2})$.
If we now try to repeat the calculation of eq.~(\ref{eq: G_E fe})
to obtain a flow relation for $H_{E}$, we will pick up extra terms
involving $h$. In addition, $H_{E}$ will not, in general, vanish
as $|x_{1}-x_{2}|\to\infty$, so we will not be able to carry out
the integration by parts of the third line of eq.~(\ref{eq: G_E fe})
in the preceding subsection. We can deal with the latter problem in
the following manner by introducing a cutoff function $\chi(x_{1},x_{2})$.
We take $\chi$ to be Euclidean invariant by choosing it to be of
the form, 
\begin{equation}
\chi(x_{1},x_{2};L)=\zeta\left(L^{-2}\sigma(x_{1},x_{2})\right),\label{eq:Euclidean chi}
\end{equation}
where $\sigma(x_{1},x_{2})$ is the squared geodesic distance between
$x_{1}$ and $x_{2}$, $L$ is an arbitrary length scale, and $\zeta(s)$
is a smooth function that is equal to one for $|s|\leq1$ and vanishes
for $|s|\geq2$. Let $z$ be an arbitrary point in $\mathbb{R}^{D}$
and let $\mathcal{B}_{z}$ denote the ball of radius $L$ centered
about $z$. Then for $x_{1},x_{2}\in\mathcal{B}_{z}$, we have the
identity, 
\begin{equation}
\frac{\partial}{\partial m^{2}}H_{E}(x_{1},x_{2};m^{2})=\int_{\mathbb{R}^{D}}d^{D}y\chi(y,z;L)\delta^{D}(y,x_{1})\frac{\partial}{\partial m^{2}}H_{E}(y,x_{2};m^{2}).
\end{equation}
Starting with this equation, we can now carry out all the steps of
eq.~(\ref{eq: G_E fe}) including the integration by parts, although
we now pick up additional terms where derivatives act on $\chi$.
The final result is, 
\begin{align}
 & \frac{\partial}{\partial m^{2}}H_{E}(x_{1},x_{2};m^{2})=-\int_{\mathbb{R}^{D}}d^{D}y\,\chi(y;z;L)H_{E}(y,x_{1};m^{2})H_{E}(y,x_{2};m^{2})+\label{eq:flow relation w btms euclid}\\
 & +\int_{\mathbb{R}^{D}}d^{D}y\,\partial_{\mu}^{(y)}\chi(y;z;L)\left[\partial_{(y)}^{\mu}H_{E}(y,x_{1};m^{2})\frac{\partial}{\partial m^{2}}H_{E}(y,x_{2};m^{2})-H_{E}(y,x_{1};m^{2})\partial_{(y)}^{\mu}\frac{\partial}{\partial m^{2}}H_{E}(y,x_{2};m^{2})\right]+\nonumber \\
 & +\int_{\mathbb{R}^{D}}d^{D}y\,\chi(y;z;L)\left[H_{E}(y,x_{1};m^{2})\frac{\partial}{\partial m^{2}}h_{E}(y,x_{2};m^{2})-h_{E}(y,x_{1};m^{2})\frac{\partial}{\partial m^{2}}H_{E}(y,x_{2};m^{2})\right].\nonumber 
\end{align}
The first term on the right side corresponds to the final line of
eq.~(\ref{eq: G_E fe}). The second line contains the terms where
derivatives from the integration by parts act on $\chi$, and the
third line contains the terms arising from the fact that $H_{E}$
is a parametrix rather than a Green's function.

Equation (\ref{eq:flow relation w btms euclid}) is unsatisfactory
as a flow equation since the second and third lines on the right side
contain the unknown quantities $\partial_{m^{2}}H_{E}$ and $\partial_{m^{2}}h_{E}$.
Nevertheless, the second and third lines must be smooth in $(x_{1},x_{2};m^{2})$
for $x_{1},x_{2}\in{\mathcal{B}}_{z}$. To see this, we note that
$H_{E}(y,x)$ can be singular only when $y=x$. However, $\partial_{\mu}^{(y)}\chi(y,z)$
is nonvanishing only for $y\notin{\mathcal{B}}_{z}$, so the second
line is smooth for $x_{1},x_{2}\in{\mathcal{B}}_{z}$. Since $h$
is smooth in all of its variables and $\chi$ is of compact support
in $y$, the third line must be smooth for all $x_{1},x_{2}$. Recall
from its definition (\ref{eq:H_E as a Green's function parametrix})
that $H_{E}$ is only uniquely determined up to the addition of a
smooth function, so we have some freedom to redefine $H_{E}$. Thus,
a possible way of dealing with the problematic terms in the second
and third lines of eq.~(\ref{eq:flow relation w btms euclid}) would
to simply drop these terms from the flow relations, leaving only the
first line\footnote{Since only the asymptotic behavior of the parametrix is relevant for
the OPE coefficients, we have replaced the equality symbol in (\ref{eq:flow relation w btms euclid})
with the weaker relation ``$\sim_{\delta}$'' which implies both
sides are asymptotically equivalent to an arbitrary scaling degree
$\delta$. }, 
\begin{equation}
\frac{\partial}{\partial m^{2}}H_{E}(x_{1},x_{2};m^{2})\sim_{\delta}-\int_{\mathbb{R}^{D}}d^{D}y\,\chi(y;z;L)H_{E}(y,x_{1};m^{2})H_{E}(y,x_{2};m^{2}).\label{eq:1st line of cutoff Euclidean flow eq}
\end{equation}
In other words, one might attempt to use the freedom in the choice
of $H_{E}$ to work with flow relation (\ref{eq:1st line of cutoff Euclidean flow eq})
rather than (\ref{eq:flow relation w btms euclid}). Indeed, this
is a simple analog of the procedure used in \cite[Section V]{HH_Associativity}
to deal with the infrared difficulties in their Euclidean flow relations
for $\lambda\phi^{4}$-theory. Note that since $(C_{H}){}_{\phi\phi}^{I}=H_{E}$
and $(C_{H}){}_{\phi^{2}\phi\phi}^{I}(y,x_{1},x_{2};z)=2H_{E}(y,x_{1})H_{E}(y,x_{2})$,
we see that (\ref{eq:1st line of cutoff Euclidean flow eq}) is equivalent
to relation (\ref{eq:naive flow eq for Wick id coef}) discussed in
the introduction for the special case of $n=2$, with $\chi$ is chosen
to be a sharp step function instead of a smooth function.

However, for Klein-Gordon theory, the flow relation (\ref{eq:naive flow eq for Wick id coef})
would give rise to OPE coefficients that are incompatible with the
scaling axiom W7. Namely, in order to satisfy this axiom, $H_{E}$
must have scaling behavior given by eq.~(\ref{eq:scaling H'}) under
the simultaneous rescaling $(\delta_{ab},m^{2})\to(\lambda^{-2}\delta_{ab},\lambda^{2}m^{2})$.
Here we are working in a fixed global inertial coordinate system defined
with respect to metric $\delta_{ab}$. Hence, with the coordinate
basis held fixed, rescaling the metric $\delta_{ab}\to\lambda^{-2}\delta_{ab}$
is equivalent to rescaling the metric coordinate components as $\delta_{\mu\nu}\to\lambda^{-2}\delta_{\mu\nu}$
and the volume element as $d^{D}y\to\lambda^{-D}d^{D}y$. However,
under the rescaling 
\begin{equation}
(\delta_{\mu\nu},d^{D}y,m^{2})\to(\lambda^{-2}\delta_{\mu\nu},\lambda^{-D}d^{D}y,\lambda^{2}m^{2}),\label{eq:rescaling of Eucl metric volume element and squ mass}
\end{equation}
we find the quantity 
\begin{equation}
\Omega(x_{1},x_{2};z;m^{2};L)\equiv-\int_{\mathbb{R}^{D}}d^{D}y\,\chi(y;z;L)H_{E}(y,x_{1};m^{2})H_{E}(y,x_{2};m^{2}),\label{eq:Omega E}
\end{equation}
appearing on the right side of (\ref{eq:naive flow eq for Wick id coef})
does not scale almost homogeneously for any fixed power of $\lambda$
on account of the fact that\textemdash due to the presence of the
length scale $L$\textemdash $\chi$ scales as 
\begin{equation}
\chi[\lambda^{-2}\delta_{\mu\nu}](x_{1},x_{2};L)=\zeta\left((\lambda L)^{-2}\sigma[\delta_{\mu\nu}](x_{1},x_{2})\right).\label{eq:rescaling of chi w/rt metric}
\end{equation}
rather than homogeneously. It follows that eq.(\ref{eq:naive flow eq for Wick id coef})
is incompatible with the scaling behavior (\ref{eq:scaling H'}) of
$H_{E}$, and any prescription for defining Wick monomials based on
its solutions would fail to satisfy axiom W7.

Although $\Omega$, as defined in (\ref{eq:Omega E}), does not scale
almost homogeneously under \eqref{eq:rescaling of Eucl metric volume element and squ mass},
it \emph{does} transform almost homogeneously with an overall factor
of $\lambda^{D-4}$ under the simultaneous rescaling 
\begin{equation}
(\delta_{\mu\nu},d^{D}y,m^{2},L)\to(\lambda^{-2}\delta_{\mu\nu},\lambda^{-D}d^{D}y,\lambda^{2}m^{2},\lambda^{-1}L),\label{eq:rescaling of Eucl metric volume element squ mass and L}
\end{equation}
since 
\begin{equation}
\chi[\lambda^{-2}\delta_{\mu\nu}](x_{1},x_{2};\lambda^{-1}L)=\chi[\delta_{\mu\nu}](x_{1},x_{2};L).\label{eq:chi scaling under metric and L}
\end{equation}
It follows that we will obtain a satisfactory flow relation if we
can replace the flow relation (\ref{eq:1st line of cutoff Euclidean flow eq}),
i.e., 
\begin{equation}
\frac{\partial}{\partial m^{2}}H_{E}(x_{1},x_{2};m^{2})\sim_{\delta}\Omega(x_{1},x_{2};z;m^{2};L)\label{flowOmega}
\end{equation}
with the modified flow relation 
\begin{equation}
\frac{\partial}{\partial m^{2}}H_{E}(x_{1},x_{2};m^{2})\sim_{\delta}\widetilde{\Omega}_{\delta}(x_{1},x_{2};z;m^{2};L).\label{eq:Wick-compatible H_E f-eq}
\end{equation}
where $\widetilde{\Omega}_{\delta}(x_{1},x_{2};z;m^{2};L)$ satisfies
the following two properties: 
\begin{enumerate}
\item $\widetilde{\Omega}_{\delta}$ is an Euclidean-invariant distribution,
symmetric in $(x_{1},x_{2})$, and depending smoothly on $m^{2}$
such that for any $(x_{1},x_{2})\in\mathcal{B}_{z}$, the distribution
$\widetilde{\Omega}_{\delta}$ differs from $\Omega$ by at most a
smooth function in $(x_{1},x_{2})$ which scales almost homogeneously
under \eqref{eq:rescaling of Eucl metric volume element squ mass and L}
with an overall factor of $\lambda^{(D-4)}$. 
\item To scaling degree $\delta$, 
\begin{equation}
\frac{\partial}{\partial L}\widetilde{\Omega}_{\delta}(x_{1},x_{2};z;m^{2};L)\sim_{\delta}0.\label{eq:prop 2 Omega tilde}
\end{equation}
\end{enumerate}
Given the previously-described scaling behavior of $\Omega$, the
first property implies $\widetilde{\Omega}_{\delta}$ is required
to scale almost homogeneously under \eqref{eq:rescaling of Eucl metric volume element squ mass and L}
with an overall factor of $\lambda^{(D-4)}$. However, since the second
property requires $\widetilde{\Omega}_{\delta}$ to be independent
of $L$ up to asymptotic degree $\delta$, it follows immediately
that $\widetilde{\Omega}_{\delta}$ must, in fact, scale almost homogeneously
under just \eqref{eq:rescaling of Eucl metric volume element and squ mass}
up to asymptotic degree $\delta$. Note we make no demand that $\widetilde{\Omega}_{\delta}$
be independent of $L$ at asymptotic orders higher than the chosen
$\delta$ since this is not relevant to the OPE coefficients. Together,
the two properties above therefore formalize the notion that $\widetilde{\Omega}_{\delta}$
must contain the same $L$-independent distributional behavior in
$(x_{1},x_{2})$ as $\Omega$ and simultaneously scale almost homogeneously
with respect to \eqref{eq:rescaling of Eucl metric volume element and squ mass}
up to any arbitrary, but fixed, asymptotic degree. In odd spacetime
dimensions, any $\widetilde{\Omega}_{\delta}$ satisfying the two
properties described below eq.~\eqref{eq:Wick-compatible H_E f-eq}
is necessarily unique up to scaling degree $\delta$. In even spacetime
dimensions, any two $\widetilde{\Omega}_{\delta}$ satisfying the
described properties may differ, to asymptotic degree $\delta$, by
only a smooth function of the form, $m^{(D-4)}f(m^{2}\sigma(x_{1},x_{2}))$.
Our task is now to find $\widetilde{\Omega}_{\delta}(x_{1},x_{2};z;m^{2};L)$
satisfying the above two properties.

Since $L$ enters $\Omega$ only through the cutoff function $\chi$,
it follows that, 
\begin{equation}
\frac{\partial}{\partial L}\Omega(x_{1},x_{2};z;m^{2};L)=-\int_{\mathbb{R}^{D}}d^{D}y\,\frac{\partial}{\partial L}\chi(y;z;L)H_{E}(y,x_{1};m^{2})H_{E}(y,x_{2};m^{2}).\label{eq:L-variation of Omega}
\end{equation}
From the definition of the cutoff function (\ref{eq:Euclidean chi}),
we observe $\partial_{L}\chi(y;z;L)=0$ for any $y\in\mathcal{B}_{z}$.
However, since $H_{E}(y,x)$ is singular only when $y=x$, it follows
immediately that (\ref{eq:L-variation of Omega}) is, in fact, a smooth
function of $(x_{1},x_{2})$ in the neighborhood $\mathcal{B}_{z}$
containing $z$. If the $L$-dependence of the smooth function of
$(x_{1},x_{2})$ appearing on the right side of eq.(\ref{eq:L-variation of Omega})
were integrable in $L$ on the interval $[0,L]$, we could obtain
the desired $\widetilde{\Omega}_{\delta}$ by simply subtracting $\int_{0}^{L}$
of the right side of eq.(\ref{eq:L-variation of Omega}) from $\Omega$.
However, the right side of eq.(\ref{eq:L-variation of Omega}) is
not integrable in $L$ on the interval $[0,L]$. Nevertheless, the
singular behavior in $L$ of the right side of eq.(\ref{eq:L-variation of Omega})
can be characterized as follows.

The quantity $\Omega$ scales almost homogeneously with an overall
factor of $\lambda^{D-4}$ under the simultaneous rescaling \eqref{eq:rescaling of Eucl metric volume element squ mass and L}.
It follows that the quantities 
\begin{equation}
\left[\partial_{\gamma_{1}}^{(x_{1})}\partial_{\gamma_{2}}^{(x_{2})}\frac{\partial}{\partial L}\Omega(x_{1},x_{2};z;L)\right]_{x_{1},x_{2}=z},\label{eq:Taylor coef d/dL Omega}
\end{equation}
appearing in the Taylor expansion of $\partial\Omega/\partial L$
scale almost homogeneously with an overall factor of $\lambda^{(D-3+\Gamma)}$,
with $\Gamma\equiv|\gamma_{1}|+|\gamma_{2}|$, under $(m^{2},L)\to(\lambda^{2}m^{2},\lambda^{-1}L)$
with the metric components and volume element held fixed. The smoothness
of $\Omega$ in $m^{2}$ then implies that any divergent dependence
of (\ref{eq:Taylor coef d/dL Omega}) on $L$ (as $L\to0^{+}$) is
expressible as a finite linear combination of terms of the form, 
\begin{equation}
L^{-\Delta}\log^{N}L,\qquad\text{for integer \ensuremath{\Delta\le(D-3+\Gamma)} and \ensuremath{N\in\mathbb{N}_{0}}.}\label{eq:non-int L terms}
\end{equation}
For $\Delta\ge1$, this gives rise to non-integrable divergences in
$L$ in any neighborhood of $L=0$. Our procedure for eliminating
these non-integrable terms is to apply a differential operator $\mathfrak{L}[L]$
to $\Omega$ that annihilates these terms but leaves the $L$-independent
parts of $\Omega$ untouched. The remaining $L$ dependence can then
be eliminated by following the strategy indicated in the previous
paragraph.

The desired operator $\mathfrak{L}[L]$ is constructed as follows.
Define first the family of differential operators, 
\begin{equation}
\mathcal{L}_{\Delta}^{N}[L]\equiv\left(1+\Delta^{-1}L\frac{\partial}{\partial L}\right)^{N}=\Delta^{-1}L^{-\Delta}\frac{\partial^{N}}{\partial(\log L)^{N}}L^{\Delta},\qquad\Delta\ne0,N\in\mathbb{N}.\label{eq:L_Delta Delta ne 0}
\end{equation}
These operators are designed so as to act trivially on $L$-independent terms and annihilate terms of the form \eqref{eq:non-int L terms} with the same $\Delta$-value and lower $N$-values.  When acting on a term of the form \eqref{eq:non-int L terms} with a different $\Delta$-value, the operator $\mathcal{L}_{\Delta}^{N}$ leaves the leading $L$-behavior unchanged in the sense of eq.~\eqref{eq:script L acting on L terms} below. We define also the operators,
\begin{equation}
\mathcal{L}_{0}^{N}[L]\equiv\prod_{k=1}^{N}\left(1-k^{-1}L\log L\frac{\partial}{\partial L}\right)=\prod_{k=1}^{N}\left(1-k^{-1}\log L\frac{\partial}{\partial(\log L)}\right),\label{eq:L_0}
\end{equation}
whose definition is unambiguous because the commutator between every
$k$-factor vanishes.  The $\mathcal{L}_0^N$ operator is designed to annihilate terms of the form \eqref{eq:non-int L terms} with $\Delta =0$ and lower $N$-values.  When $\mathcal{L}_0^N$ acts on a term of the form \eqref{eq:non-int L terms} with $\Delta\ne0$, it produces terms of the same form and $\Delta$-value (but generally increases the $N$-value). We define $\mathfrak{L}[L]$ by\footnote{The product over $\mathcal{L}_{\Delta}$-operators with different
$\Delta$-values is needed to account for the dependence of the Wick
OPE coefficients on the dimensionful parameter $m^{2}$. In a theory
without dimensionful parameters, we could eliminate the $L$-dependence
of a flow relation by simply using the operator $\mathfrak{L}=\mathcal{L}_{\Delta}$
with $\Delta$ corresponding to the conformal scaling dimension of
the flow relation. (Note that $\Delta$ would generally depend on
the renormalized coupling parameter in an interacting theory.)} 
\begin{equation}
\mathfrak{L}\equiv\mathcal{L}_{0}^{N}\prod_{\Delta=1}^{D-4+\delta}\mathcal{L}_{\Delta}^{N},\label{eq:full L op}
\end{equation}
for any $N>2$. Note that $[\mathcal{L}_{\Delta}^{N},\mathcal{L}_{\Delta'}^{N'}]=0$
for any $\Delta,\Delta'\ne0$, so the order of composition between
these operators does not matter. However, $[\mathcal{L}_{0}^{N},\mathcal{L}_{\Delta\ne0}^{N'}]\ne0$,
so the order of composition for $\mathcal{L}_{0}^{N}$ relative to
the other operators $\mathcal{L}_{\Delta\ne0}^{N}$ \emph{does} matter.
Note that (\ref{eq:full L op}) scales almost homogeneously with an
overall factor of $\lambda^{0}$ under $L\to\lambda^{-1}L$ since
it is composed of operators (\ref{eq:L_Delta Delta ne 0}) and (\ref{eq:L_0})
with this property. By expanding out the product of operators in (\ref{eq:full L op}),
note also that $\mathfrak{L}[L]$ can be rewritten in the general
form, 
\begin{equation}
\mathfrak{L}[L]=1+\sum_{\Delta=1}^{N(D-3+\delta)}\sum_{n=0}^{N}c_{n,\Delta}L^{\Delta}\log^{n}L\frac{\partial^{\Delta}}{\partial L^{\Delta}},\label{eq:full L op expanded out}
\end{equation}
where $c_{n,\Delta}$ are $L$-independent numerical coefficients.
Hence, 
\begin{align}
\mathfrak{L}[L]\Omega(x_{1},x_{2};z;L)-\Omega(x_{1},x_{2};z;L) & =\sum_{\Delta=1}^{N(D-3+\delta)}\sum_{n=0}^{N}c_{n,\Delta}L^{\Delta}\log^{n}L\frac{\partial^{\Delta}}{\partial L^{\Delta}}\Omega(x_{1},x_{2};z;L).\label{eq:frakL Omega minus Omega-1}
\end{align}
Recalling $\partial_{L}\Omega$ is smooth in $(x_{1},x_{2},m^{2})$
in a neighborhood containing $z\in\mathbb{R}^{D}$, it follows the
right-hand side of (\ref{eq:frakL Omega minus Omega-1}) is also smooth
in $(x_{1},x_{2};m^{2})$ because every term involves at least one
$L$-derivative of $\Omega$. Every term on the right-hand side of
(\ref{eq:frakL Omega minus Omega-1}) clearly scales almost homogeneously
with an overall factor of $\lambda^{(D-4)}$.

We now define $\widetilde{\Omega}_{\delta}$ by 
\begin{equation}
\widetilde{\Omega}_{\delta}(x_{1},x_{2};z;L)\equiv\mathfrak{L}[L]\Omega(x_{1},x_{2};z;L)-\sum_{|\gamma_{1}|+|\gamma_{2}|\le\delta}\frac{1}{\gamma_{1}!\gamma_{2}!}b_{\gamma_{1}\gamma_{2}}^{H}(L)(x_{1}-z)^{\gamma_{1}}(x_{2}-z)^{\gamma_{2}},\label{eq:Omega E tilde}
\end{equation}
where 
\begin{equation}
b_{\gamma_{1}\gamma_{2}}^{H}(L)\equiv b_{\gamma_{1}\gamma_{2}}^{0}+\int_{0}^{L}dL'\left[\partial_{\gamma_{1}}^{(x_{1})}\partial_{\gamma_{2}}^{(x_{2})}\frac{\partial}{\partial L'}\left(\mathfrak{L}[L']\Omega(x_{1},x_{2};z;L')\right)\right]_{x_{1},x_{2}=z},\label{eq:b_gamma1gamma2}
\end{equation}
where $b_{\gamma_{1}\gamma_{2}}^{0}$ corresponds to the inherent
ambiguities in our prescription discussed underneath eq.~\eqref{eq:b_C(L)}
below. That $b_{\gamma_{1}\gamma_{2}}^{H}(L)$ is well defined is
a consequence of the following proposition: 
\begin{prop}
\label{prop: integrability at L=0}For any $N>2$, the $L$-dependent
(Euclidean-covariant) tensors, 
\begin{equation}
\left[\partial_{\gamma_{1}}^{(x_{1})}\partial_{\gamma_{2}}^{(x_{2})}\frac{\partial}{\partial L}\left(\mathfrak{L}_{\delta}[L]\Omega(x_{1},x_{2};z;L)\right)\right]_{x_{1},x_{2}=z},\label{eq:Taylor coef of d/dL frakL Omega}
\end{equation}
are integrable in $L$ on a finite interval containing $L=0$ for
any $|\gamma_{1}|+|\gamma_{2}|\le\delta$ 
\end{prop}
\begin{proof}
It is useful to first commute $\partial_{L}$ past the operator $\mathfrak{L}_{\delta}[L]$
in (\ref{eq:Taylor coef of d/dL frakL Omega}). To do this, we note,
\[
\partial_{L}\mathcal{L}_{\Delta}^{N}=\begin{cases}
\Delta^{-N}(\Delta+1)^{N}\mathcal{L}_{(\Delta+1)}^{N}\partial_{L} & \Delta>0\\
(-1)^{N}(N!)^{-1}\log^{N}(L)\mathcal{L}_{1}^{N}\partial_{L} & \Delta=0
\end{cases},
\]
and, therefore, 
\begin{equation}
\partial_{L}\mathfrak{L}\propto\log^{N}(L)\left(\prod_{\Delta=1}^{D-3+\delta}\mathcal{L}_{\Delta}^{N}\right)\partial_{L}.\label{eq:d_L frakL}
\end{equation}
Plugging this back into (\ref{eq:Taylor coef of d/dL frakL Omega})
and noting the smoothness of $\partial_{L}\Omega$ in $(x_{1},x_{2})$,
we obtain, 
\begin{equation}
\text{formula }\eqref{eq:Taylor coef of d/dL frakL Omega}\propto\log^{N}(L)\left(\prod_{\Delta=1}^{D-3+\delta}\mathcal{L}_{\Delta}^{N}[L]\right)\left[\partial_{\gamma_{1}}^{(x_{1})}\partial_{\gamma_{2}}^{(x_{2})}\frac{\partial}{\partial L}\Omega(x_{1},x_{2};z;L)\right]_{x_{1},x_{2}=z}.\label{eq:L op acting on Taylor coef of d/dL Omega}
\end{equation}
Noting that 
\[
\mathcal{L}_{\Delta'}^{N'}(L^{-\Delta}\log^{N}L)=\begin{cases}
0 & \Delta'=\Delta,N'>N\\
\mathcal{O}(L^{-\Delta}\log^{N}L) & \Delta'\ne\Delta,N'>0
\end{cases}, \label{eq:script L acting on L terms}
\]
we see that, for any $N'>N$ and all $\Gamma\le\delta$, all non-integrable
terms of the form (\ref{eq:non-int L terms}) are annihilated by the
string of operators $\mathcal{L}_{1}^{N'}\mathcal{L}_{2}^{N'}\cdots\mathcal{L}_{(D-3+\delta)}^{N'}$
appearing in (\ref{eq:Taylor coef d/dL Omega}), and that any integrable
terms of the form (\ref{eq:non-int L terms}) remain integrable after
application of $\mathcal{L}_{1}^{N'}\mathcal{L}_{2}^{N'}\cdots\mathcal{L}_{(D-3+\delta)}^{N'}$.
Noting that $H_{E}$ contains at most one power of the logarithm and
$\Omega$ depends quadratically on $H_{E}$, we conclude the right-hand
side of (\ref{eq:L op acting on Taylor coef of d/dL Omega}) and,
thus, the Taylor coefficients (\ref{eq:Taylor coef of d/dL frakL Omega})
must be integrable on an interval containing $L=0$ for any $N>2$
as we desired to show. 
\end{proof}
Note the translational symmetry of $\Omega$ implies $b_{\gamma_{1}\gamma_{2}}^{H}$
are independent of $z$ and rotational symmetry of $\Omega$ implies
$b_{\gamma_{1}\gamma_{2}}^{H}$ are composed of products of the Euclidean
metric\footnote{Note that $\Omega$ is invariant under the full orthogonal group including
improper rotations. Although the Levi-Civita symbols $\epsilon_{\mu_{1}\cdots\mu_{n}}$
are invariant under proper rotations, the Euclidean metric is the
only tensor invariant under all $R\in O(D)$. } and, thus, vanish unless $|\gamma_{1}|,|\gamma_{2}|$ are even. Recalling
the definition (\ref{eq:Omega E}) of $\Omega$ and the fact that
$(C_{H}){}_{\phi\phi}^{(\partial_{\gamma_{1}}\phi\partial_{\gamma_{2}}\phi)}(x_{1},x_{2};z)=(x_{1}-z)^{\gamma_{1}}(x_{2}-z)^{\gamma_{2}}/(\gamma_{1}!\gamma_{2}!)$,
we note that the $\widetilde{\Omega}_{\delta}$ defined in (\ref{eq:Omega E tilde})
is identical to the right-hand side of the $L$-independent flow relation
(\ref{eq:L-indep Euclidean fe}) claimed in the introduction for the
special case that $n=2$. i.e., $b_{C}^{H}$ in formula (\ref{eq:L-indep Euclidean fe})
is given explicitly by (\ref{eq:b_gamma1gamma2}) for $[C]_{\phi}=2$
and vanishes otherwise.

The required $L$-independence (\ref{eq:prop 2 Omega tilde}) of $\widetilde{\Omega}_{\delta}$
is verified by differentiating (\ref{eq:Omega E tilde}), 
\begin{align}
 & \frac{\partial}{\partial L}\widetilde{\Omega}_{\delta}(x_{1},x_{2};z;L)\nonumber \\
 & =\frac{\partial}{\partial L}\left[\mathfrak{L}[L]\Omega(x_{1},x_{2};z;L)\right]-\sum_{|\gamma_{1}|+|\gamma_{2}|\le\delta}\frac{1}{\gamma_{1}!\gamma_{2}!}\frac{\partial}{\partial L}b_{\gamma_{1}\gamma_{2}}^{H}(L)(x_{1}-z)^{\gamma_{1}}(x_{2}-z)^{\gamma_{2}}\\
 & \sim_{\delta}\sum_{|\gamma_{1}|+|\gamma_{2}|\le\delta}\frac{1}{\gamma_{1}!\gamma_{2}!}(x_{1}-z)^{\gamma_{1}}(x_{2}-z)^{\gamma_{2}}\left[\left[\partial_{\gamma_{1}}^{(x_{1})}\partial_{\gamma_{2}}^{(x_{2})}\frac{\partial}{\partial L}\left(\mathfrak{L}[L]\Omega(x_{1},x_{2};z;L)\right)\right]_{x_{1},x_{2}=z}-\frac{\partial}{\partial L}b_{\gamma_{1}\gamma_{2}}^{H}(L)\right]\nonumber \\
 & \sim_{\delta}0,\nonumber 
\end{align}
where in going to the third line we have used the smoothness of $\partial_{L}(\mathfrak{L}\Omega)$
in $(x_{1},x_{2})$ to Taylor expand the first term in the second
line around $x_{1},x_{2}=z$. The final line then follows from the
definition (\ref{eq:b_gamma1gamma2}) of $b_{\gamma_{1}\gamma_{2}}^{H}$
and the fundamental theorem of calculus. Thus, our construction (\ref{eq:Omega E tilde})
of $\widetilde{\Omega}_{\delta}$ complies with the required properties. 
\begin{rem}
For any cutoff function $\chi$ of the form (\ref{eq:Euclidean chi}),
it follows from (\ref{eq:full L op expanded out}) that $\mathfrak{L}\chi$
is also a cutoff function of the form (\ref{eq:Euclidean chi}). That
is, $\mathfrak{L}$ acts as a map from the set of all cutoff functions
to the subset of cutoff functions such that, to scaling degree $\delta$,
the asymptotic expansion of $\partial_{L}(\Omega[\mathfrak{L}\chi])$
diverges, at worst, logarithmically as $L\to0^{+}$. 
\end{rem}
It is worth noting that, using the formulas for the Hadamard-normal
ordered coefficients 
\begin{align}
(C_{H})_{\phi^{2}\phi\phi}^{I}(y,x_{1},x_{2};z) & =2H_{E}(y,x_{1})H_{E}(y,x_{2})\label{eq:C_phi^2 phi phi ^I}\\
(C_{H}){}_{\phi^{2}(\partial_{\gamma_{1}}\phi\cdots\partial_{\gamma_{k}}\phi)}^{I}(y,z;z) & =\begin{cases}
2\partial_{\gamma_{1}}^{(z)}H_{E}(y,z)\partial_{\gamma_{2}}^{(z)}H_{E}(y,z) & k=2\\
0 & \text{otherwise}
\end{cases}\\
(C_{H})_{\phi\phi}^{(\partial_{\gamma_{1}}\phi\partial_{\gamma_{2}}\phi)}(x_{1},x_{2};z) & =\frac{1}{\gamma_{1}!\gamma_{2}!}(x_{1}-z)^{(\gamma_{1}}(x_{2}-z)^{\gamma_{2})}\label{eq:C_phi phi ^C}
\end{align}
the flow relation (\ref{eq:Wick-compatible H_E f-eq}) can be written
equivalently as, 
\begin{align}
 & \frac{\partial}{\partial m^{2}}(C_{H}){}_{\phi\phi}^{I}(x_{1},x_{2};z)\label{eq:L-ind Euclidean fe n=2}\\
 & \sim_{\delta}-\frac{1}{2}\int d^{D}y\mathfrak{L}[L]\chi(y,z;L)(C_{H}){}_{\phi^{2}\phi\phi}^{I}(y,x_{1},x_{2};z)-\sum_{[C]\le\delta+2}b_{C}^{H}(L)(C_{H}){}_{\phi\phi}^{C}(x_{1},x_{2};z),\nonumber 
\end{align}
where, for $L>0$, 
\begin{align}
b_{C}^{H}(L) & =b_{C}^{0}-\frac{1}{2}\int_{0}^{L}dL'\int d^{D}y\frac{\partial}{\partial L'}\left(\mathfrak{L}[L']\chi(y,z;L')\right)(C_{H}){}_{\phi^{2}C}^{I}(y,z;z).\label{eq:b_C(L)}
\end{align}
The only ambiguities in our construction arise from a limited choice
for the value of the $L$-independent $b_{C}^{0}$-tensors when $[C]_{\phi}=2$.
We have $b_{C}^{0}=0$ unless $[C]_{\phi}=2$. The $b_{C}^{0}$ are
required to depend smoothly on $(\delta_{\mu\nu},m^2)$ and scale exactly homogeneously under $(\delta_{\mu\nu},m^{2})\to(\lambda^{-2}\delta_{\mu\nu},\lambda^{2}m^{2})$
with an overall factor of $\lambda^{(D-4)}$. This implies that $b_{C}^{0}$
must vanish identically when $D$ is odd. In even spacetime dimensions,
these ambiguities correspond to the freedom to choose the Taylor coefficients
of a smooth, Euclidean-invariant function which depends smoothly on
$(\delta_{\mu\nu},m^{2})$ and scales exactly homogeneously under
$(\delta_{\mu\nu},m^{2})\to(\lambda^{-2}\delta_{\mu\nu},\lambda^{2}m^{2})$
with an overall factor of $\lambda^{(D-4)}$. For any fixed cutoff
function $\chi$ and any choice of Euclidean-invariant Hadamard parametrix
$H_{E}$ which scales almost homogeneously, one can choose $b_{\gamma_{1}\gamma_{2}}^{0}$
such that OPE coefficients obtained via Hadamard normal ordering satisfy
(\ref{eq:L-ind Euclidean fe n=2}). Conversely, for any fixed
$\chi$ and admissible choice of $b_{\gamma_{1}\gamma_{2}}^{0}$,
one can find an $H_{E}$ such that the Hadamard-normal-ordered OPE
coefficients satisfy (\ref{eq:L-ind Euclidean fe n=2}). Hence,
the ambiguity in our construction of the $L$-independent flow relation
(\ref{eq:L-ind Euclidean fe n=2}) is in a 1-1 correspondence
with the inherent freedom to choose a Hadamard parametrix for defining
normal-ordered Wick fields compatible with axioms W1-W8. 
\begin{rem}
\label{rem:conservation in Eucl space}In flat space and all dimensions
$D\ge2$, we note the conservation axiom $\conservation$ places no
constraints on the ambiguities in $H_{E}$ and, thus, does not require
any further modifications to the flow relation \eqref{eq:L-ind Euclidean fe n=2}.
In particular, although $H_{E}$ is not an exact Greens function of
the Euclidean Klein-Gordon operator \eqref{eq:Euclidean K}, it does
automatically satisfy the Euclidean version of the conservation constraint
\eqref{eq:conservation constr for H_F}: 
\begin{equation}
\nabla_{\mu}^{(x_{1})}h_{E}(x_{1},x_{2})|_{x_{1},x_{2}=z}=\left[\nabla_{\mu}^{(x_{1})}K_{x_{2}}H_{E}(x_{1},x_{2})-\nabla_{\mu}^{(x_{1})}\delta(x_{1},x_{2})\right]_{x_{1},x_{2}=z}=0,\label{eq:Euclidean conservation constraint}
\end{equation}
where we recall the smooth function $h_E$ was defined via \eqref{eq:H_E as a Green's function parametrix}.
Because $H_{E}$ is required to be invariant under the inhomogeneous
orthogonal group, $\nabla_{\mu}^{(x_{1})}h_{E}(x_{1},x_{2})|_{x_{1},x_{2}=z}$
must be invariant under rotations about the point $z$. However, since
there does not exist a rotationally-invariant $D$-vector, we conclude
$\nabla_{\mu}^{(x_{1})}h_{E}(x_{1},x_{2})|_{x_{1},x_{2}=z}$ identically
vanishes in flat Euclidean space for any dimension, including $D=2$. 
\end{rem}
By the same reasoning used in the proof of Theorem \ref{thm:vac NO flow rel gen n},
the flow relation (\ref{eq:L-ind Euclidean fe n=2}) for $(C_{H}){}_{\phi\phi}^{I}$
straightforwardly implies flow relations for $(C_{H})_{\phi\cdots\phi}^{I}$
as expressed in the following theorem: 
\begin{thm}
\label{thm:Eucl fe for H-NO id coef}For any Hadamard parametrix satisfying
\eqref{eq:L-ind Euclidean fe n=2}, the corresponding Hadamard
normal-ordered coefficients $(C_{H})_{\phi\cdots\phi}^{I}$ satisfy
the flow relation: 
\begin{align}
 & \frac{\partial}{\partial m^{2}}(C_{H}){}_{\phi\cdots\phi}^{I}(x_{1},\dots,x_{n};z)\label{eq:L-ind Euclidean H-NO fe gen n}\\
 & \approx-\frac{1}{2}\int d^{D}y\mathfrak{L}[L]\chi(y,z;L)(C_{H}){}_{\phi^{2}\phi\cdots\phi}^{I}(y,x_{1},\dots,x_{n};z)-\sum_{C}b_{C}^{H}(L)(C_{H}){}_{\phi\cdots\phi}^{C}(x_{1},\dots,x_{n};z),\nonumber 
\end{align}
where $b_{C}^{H}(L)$ is again given by \eqref{eq:b_C(L)} with the
same constraints on $b_{C}^{0}$ as stated below \eqref{eq:b_C(L)}. 
\end{thm}
Finally, the results of Subsection \ref{subsec:gen Wick OPE coef}
can be used to obtain the flow relations for $C_{\phi\cdots\phi}^{I}$
for an arbitrary prescription for Wick monomials satisfying W1-W8.
We obtain 
\begin{align}
 & \frac{\partial}{\partial m^{2}}C{}_{\phi\cdots\phi}^{I}(x_{1},\dots,x_{n};z)\label{eq:Eucl L-indep fe with F_k terms}\\
 & \approx-\frac{1}{2}\int d^{D}y\mathfrak{L}[L]\chi(y,z;L)C{}_{\phi^{2}\phi\cdots\phi}^{I}(y,x_{1},\dots,x_{n};z)-\sum_{C}b_{C}(L)C{}_{\phi\cdots\phi}^{C}(x_{1},\dots,x_{n};z)+\text{``\ensuremath{F_{k}}-terms'',}\nonumber 
\end{align}
where 
\[
b_{C}(L)\equiv\delta_{n,2}b_{C}^{0}-\frac{1}{2}\int_{0}^{L}dL'\int d^{D}y\frac{\partial}{\partial L'}\left(\mathfrak{L}[L']\chi(y,z;L')\right)C{}_{\phi^{2}C}^{I}(y,z;z),
\]
and $\text{``\ensuremath{F_{k}}-terms''}$ denotes terms that contain
at least one factor of $F_{k}$ (for $k\le n$). By the discussion
in Subsection \ref{subsec:gen Wick OPE coef} below eq.~\eqref{eq:gen linear C ident coef-1},
$F_{j}$ can, in turn, be written purely in terms of OPE coefficients
of the form $C_{\phi\cdots\phi}^{I}(x_{1},\dots,x_{i};z)$ such that
$i\le j$. In this way, all terms on the right side of (\ref{eq:Eucl L-indep fe with F_k terms})
are expressible entirely in terms of OPE coefficients and the cutoff
function $\chi$, and \eqref{eq:Eucl L-indep fe with F_k terms} yields
the flow relation for the OPE coefficients corresponding to an arbitrary
prescription for the Wick fields compatible with the axioms W1-W8.
Note, in contrast to $b_{C}^{H}(L)$ given in \eqref{eq:b_C(L)},
here $b_{C}(L)$ can be nonzero when $[C]_{\phi}\ne2$ since, for
prescriptions not given via normal ordering, $C_{\phi^{2}C}^{I}$
is generally nonzero when $[C]_{\phi}\ne2$.

\section{Flow relations for OPE coefficients in Minkowski spacetime\label{sec:Minkowski-flow-relations}}

We turn, now, to the derivation of flow relations for OPE coefficients
in Minkowski spacetime $(\mathbb{R}^{D},\eta_{ab})$. As can be seen
from the derivation of the Euclidean flow relations in the preceding
section, it is essential that the two-point OPE coefficient for which
we are obtaining a flow relation be a Green's parametrix for the wave
equation. Consequently, we do not believe it is possible to obtain
a flow relation for the Lorentzian $C_{\phi\phi}^{I}$, since it does
not have this property. However, as we shall show, a flow relation
for $C_{T_{0}\{\phi\phi\}}^{I}$ can be obtained, where $T_{0}$ denotes
the unextended time-ordered-product.

In the Minkowski case, if we choose $C_{T_{0}\{\phi\phi\}}^{I}$ to
be the exact Feynman propagator for $m^{2}>0$, the spacetime integral
that would appear in the flow relation will not converge, so we would
need to introduce a cutoff function even in this case. Therefore,
in contrast to the Euclidean case, there is no advantage in initially
working with the exact Feynman propagator as compared with a Poincare-invariant
Feynman parametrix that is smooth at $m^{2}=0$. As we shall see,
a new difficulty arises from a cutoff in the Minkowski case in that
there does not exist a nontrivial function of compact support that
is Lorentz invariant. Consequently, in the Minkowski case, the introduction
of a cutoff spoils the Poincare invariance of the flow relations.
Nevertheless, we shall show that counterterms can be introduced into
the flow relations so as to restore Poincare invariance. The presence
of the cutoff function in the flow relations also spoils their scaling
behavior. However, this can be fixed using the same procedure as developed
for the Euclidean flow relations. Thus, we will, in the end, obtain
entirely satisfactory flow relations for the OPE coefficients of unextended
time-ordered-products in Minkowski spacetime (see Theorem \ref{thm:Minkowski flow rel}).
These flow relations will be unique up to modifications of the counterterms
that correspond to the ambiguities in the definitions of the Wick
monomials themselves.

The requirement W1 that the Wick monomials be locally and covariantly
defined implies that, in Minkowski spacetime, the Wick monomials must
be Poincare covariant \cite{HW_Axiomatic_QFTCS}. Thus, in a Hadamard
normal-ordering prescription, we must use a Poincare-invariant Hadamard
parametrix. Since, in this subsection, we will want to include the
case $m^{2}=0$, we will not use the usual choice $\langle\phi(x_{1})\phi(x_{2})\rangle_{\text{vac}}$\textemdash which
fails to be smooth in $m^{2}$ at $m^{2}=0$\textemdash but rather
will take $H(x_{1},x_{2};m^{2})$ to be given by eq.~(\ref{eq:H}),
with $\ell$ fixed (i.e., independent of $m^{2}$).

The starting point for our derivation of Euclidean flow relations
in the preceding section was the preliminary flow-like equation (\ref{eq:flow relation w btms euclid})
for the Euclidean Hadamard parametrix $H_{E}(x_{1},x_{2};m^{2})$.
The key ingredients that went into the derivation of this equation
were (i) that $H_{E}$ is a fundamental solution (\ref{eq:H_E as a Green's function parametrix})
of the Klein-Gordon operator up to smooth remainder and (ii) for any
test function $f$, $H_{E}(y,f)$ is smooth in $y$. In Minkowski
spacetime, the OPE coefficient $(C_{H}){}_{\phi\phi}^{I}=H(x_{1},x_{2})$
will \emph{not} be a Green's parametrix, i.e., it will satisfy $K_{x_{1}}(C_{H}){}_{\phi\phi}^{I}(x_{1},x_{2})={\rm smooth}$
rather than $K_{x_{1}}(C_{H}){}_{\phi\phi}^{I}(x_{1},x_{2})=(\delta(x_{1},x_{2})+{\rm smooth})$.
Consequently, the analog of condition (i) will not be satisfied and
we cannot expect to obtain flow relations for the ordinary OPE coefficients.
However, condition (i) does hold for the Feynman parametrix $H_{F}(x_{1},x_{2};m^{2})$
given by eq.~(\ref{eq:H_F def}). Such a parametrix satisfies, 
\begin{equation}
(-\eta^{ab}\partial_{a}\partial_{b}+m^{2})H_{F}(x_{1},x_{2};m^{2})=-i\delta^{(D)}(x_{1},x_{2})+h(x_{1},x_{2};m^{2})\label{feynpara}
\end{equation}
where $h$ is a smooth function of its arguments. As with the Euclidean
parametrix, any two Feynman parametrices $H_{F}$ and $H_{F}'$ satisfying
(\ref{feynpara}) can differ by a Poincare invariant smooth function
of $(x_{1},x_{2})$. Since $(C_{H}){}_{T_{0}\{\phi\phi\}}^{I}=H_{F}(x_{1},x_{2})$,
it might be expected that flow relations will hold for the OPE coefficients
of time-ordered products\footnote{Indeed, this also could be anticipated from the fact that a Wick rotation
from Euclidean space to Minkowski spacetime will take the Euclidean
Green's function $G_{E}$ to the Feynman propagator $G_{F}$.}. As we shall see below, flow relations do indeed hold for the OPE
coefficients of time-ordered products.

Condition (ii) also holds for $H_{F}(x_{1},x_{2};m^{2})$. Indeed,
for any translation invariant distribution ${\mathcal{D}}(x_{1},x_{2})$
on $\mathbb{R}^{D}\times\mathbb{R}^{D}$ and any test function $f$
on $\mathbb{R}^{D}$, we have that ${\mathcal{D}}(x_{1},f)$ is smooth
in $x_{1}$. Namely, if we define new variables $X_{1}=x_{1}+x_{2}$
and $X_{2}=x_{1}-x_{2}$, then, by translation invariance, $\mathcal{D}$
cannot depend on $X_{1}$, so the elements of its wavefront set must
be of the form $(X_{1},0;X_{2},K_{2})$ with $K_{2}\neq0$. Therefore,
in terms of the original variables $(x_{1},x_{2})$, the elements
of $\text{WF}[{\mathcal{D}}]$ must be of the form $(x_{1},k_{1};x_{2},-k_{1})$
with $k_{1}\neq0$. The wavefront set calculus rules then immediately
imply that ${\mathcal{D}}(x_{1},f)$ is smooth for any test function
$f$. 
\begin{rem}
\label{rem: ext of T0[C^I_phi phi]} Since the unextended
time-ordered products are only defined away from all diagonals, applying
the Klein-Gordon operator to $(C_{H})_{T_{0}\{\phi\phi\}}^{I}=T_{0}\{H(x_{1},x_{2})\}$
will yield a distribution that is a priori only defined when $x_{1}\ne x_{2}$
and, thus, the OPE coefficient $(C_{H})_{T_{0}\{\phi\phi\}}^{I}$
is itself not actually a Green's function satisfying \eqref{feynpara}.
Nevertheless, as discussed in Remark \ref{rem:ext of T0 H} below
\prettyref{thm:explicit TO coef}, the extension of $T_{0}\{H(x_{1},x_{2})\}$
to $x_{1}=x_{2}$ is uniquely given by the Feynman parametrix $H_{F}=H-i\Delta^{\text{adv}}$.
Hence, whenever we need to use the identity \eqref{feynpara} in what
follows below, we may, without introducing any new ambiguities, first
extend $(C_{H})_{T_{0}\{\phi\phi\}}^{I}$ to its diagonal $x_{1}=x_{2}$
and then subsequently apply the Green's function identity \eqref{feynpara}
for the Feynman parametrix. As we will see, this is sufficient to
derive all the flow relations for the time-ordered Wick OPE coefficients
of the form $C_{T_{0}\{\phi\cdots\phi\}}^{I}$. As discussed in the
introduction and Subsection \ref{subsec:TOP}, unique extensions of
the OPE coefficients appearing inside the integral on the right-hand
side of the flow relations are only possible, in general, to the ``partial
diagonals'', where the integration variable $y$ coincides with only
a single $x_{i}$-spacetime variable, so we will continue to write
all OPE coefficients appearing in the flow relations with the unextended
time-ordering symbol $T_{0}$ rather than $T$, with the understanding
that (unique) extensions to the appropriate partial diagonals with
$y$ are necessary for evaluating the $y$-integral. See Remark \ref{rem:unique ext of Mink flow rel}
below \prettyref{thm:Minkowski flow rel} for further discussion regarding
the extension of the OPE coefficients appearing in the Minkowski flow
relations. 
\end{rem}
Since conditions (i) and (ii) hold for $H_{F}$, we can directly parallel
the derivation of the key preliminary Euclidean flow-like equation
(\ref{eq:prelim H_F flow eq}) for $H_{E}$ to obtain a flow-like
relation for $H_{F}(x_{1},x_{2};m^{2})$ by introducing a cutoff function
$\chi(y,z;L)$ defined such that $\chi=1$ for $y$ in some compact
neighborhood, ${\mathcal{B}}_{1}$, of $z$ and $\chi=0$ outside
of some larger compact neighborhood, ${\mathcal{B}}_{2}$, of $z$.
We again denote by $L$ the arbitrary length scale which is required
to define a spacetime cutoff. Then, for $x_{1},x_{2}\in{\mathcal{B}}_{1}$,
we similarly obtain, 
\begin{align}
 & \frac{\partial}{\partial m^{2}}H_{F}(x_{1},x_{2};m^{2})=-i\int_{{\mathcal{B}}_{2}}d^{D}y\,\chi(y;z;L)H_{F}(y,x_{1};m^{2})H_{F}(y,x_{2};m^{2})+\label{eq:prelim H_F flow eq}\\
 & +i\int_{{\mathcal{B}}_{2}\backslash{\mathcal{B}}_{1}}d^{D}y\,\partial_{\mu}^{(y)}\chi(y;z;L)\left[\partial_{(y)}^{\mu}H_{F}(y,x_{1};m^{2})\frac{\partial}{\partial m^{2}}H_{F}(y,x_{2};m^{2})-H_{F}(y,x_{1};m^{2})\partial_{(y)}^{\mu}\frac{\partial}{\partial m^{2}}H_{F}(y,x_{2};m^{2})\right]+\nonumber \\
 & +i\int_{{\mathcal{B}}_{2}}d^{D}y\,\chi(y;z;L)\left[H_{F}(y,x_{1};m^{2})\frac{\partial}{\partial m^{2}}h(y,x_{2};m^{2})-h(y,x_{1};m^{2})\frac{\partial}{\partial m^{2}}H_{F}(y,x_{2};m^{2})\right],\nonumber 
\end{align}
where $h$ is defined via eq.~(\ref{feynpara}). Note that the factor
of $\partial_{\mu}^{(y)}\chi(y;z;L)$ appearing in the third line
has support only on $\mathcal{B}_{2}\backslash\mathcal{B}_{1}$ because
we require $\chi(y;z;L)=1$ for $y\in\mathcal{B}_{1}$. eq.~(\ref{eq:prelim H_F flow eq})
is identical to (\ref{eq:flow relation w btms euclid}) modulo the
substitutions $H_{E}\to iH_{F}$ and $h_{E}\to ih$.

As in the Euclidean formula (\ref{eq:flow relation w btms euclid}),
the third line is automatically smooth on account of the smoothness
of $h$ and the compact-support of $\chi$. Similarly, in the second
line, the differentiated cutoff function $\partial_{\mu}^{(y)}\chi(y;z;L)$
is only nonzero when $y\in{\mathcal{B}}_{2}\backslash{\mathcal{B}}_{1}$
and thus vanishes when $y=x_{1},x_{2}$ if $x_{1},x_{2}\in\mathcal{B}_{1}$.
However, whereas the Euclidean parametrix $H_{E}(y,x)$ is singular
only when $y=x$, the singular support of the Feynman parametrix $H_{F}(y,x)$
includes all $(y,x)$ such that $y$ and $x$ can be connected by
a null geodesic. Thus, the integrand in the second line of (\ref{eq:prelim H_F flow eq})
will be singular even for $y\in\mathcal{B}_{2}\backslash\mathcal{B}_{1}$
whenever $y$ is lightlike separated from either or both $(x_{1},x_{2})$.
Therefore, it is not at all obvious that the integral will yield a
smooth function. However, since the partial $m^{2}$-derivative does
not alter the wavefront set of $H_{F}$, the terms in the second line
of (\ref{eq:prelim H_F flow eq}) will be smooth if and only if the
quantity, 
\begin{equation}
\Theta[\chi,H_{F}](x_{1},x_{2};z;m^{2})\equiv\int_{{\mathcal{B}}_{2}\backslash{\mathcal{B}}_{1}}d^{D}y\,\partial_{\mu}^{(y)}\chi(y,z;L)H_{F}(y,x_{1};m^{2})\partial_{(y)}^{\mu}H_{F}(y,x_{2};m^{2})\label{eq:smooth boundary integral}
\end{equation}
is smooth. The following proposition establishes smoothness of this
quantity: 
\begin{prop}
\label{lem:smoothness of bt} For $x_{1},x_{2}\in{\mathcal{B}}_{1}(z)$,
the quantity $\Theta$ defined by (\ref{eq:smooth boundary integral})
is a $C^{\infty}$ function of $(x_{1},x_{2})$. 
\end{prop}
\begin{proof}
A generalized function is smooth if and only if its wavefront set
is the empty set. We show the wavefront set of the generalized function
(\ref{eq:smooth boundary integral}) is contained in the empty set
when $x_{1},x_{2}\in\mathcal{B}_{1}(z)$ and, thus, $\Theta(x_{1},x_{2};z)$
must be smooth. Note first the wavefront set of a Feynman parametrix
is, 
\begin{align}
\text{WF}[H_{F}] & =\text{WF}[\delta]\cup\left\{ (x_{1},k_{1};x_{2},k_{2})\in\times^{2}(T^{*}\mathbb{R}^{D}\backslash Z^{*}\mathbb{R}^{D})|\ x_{1}\ne x_{2},(x_{1},k_{1})\sim(x_{2},-k_{2}),\right.\nonumber \\
 & \qquad\left.k_{1}\in\dot{V}_{x_{1}}^{+}\text{ if }x_{1}\in J^{+}(x_{2}),k_{1}\in\dot{V}_{x_{1}}^{-}\text{ if }x_{1}\in J^{-}(x_{2})\right\} ,\label{eq:Feyn WF}
\end{align}
where we recall the notation: $\dot{V}_{x}^{\pm}$ denotes, respectively
the boundary of the future/past lightcone at $x$; $(x,k)\sim(y,p)$
iff points $x$ and $y$ may be joined by a null geodesic $\gamma$
such that $k$ and $p$ are cotangent and coparallel to $\gamma$;
and $Z^{*}\mathbb{R}^{D}$ denotes the zero section of cotangent bundle
$T^{*}\mathbb{R}^{D}$. Recall the wavefront set of the $\delta$-distribution
was given in (\ref{eq:WF delta}).

We write $\mathcal{B}\equiv\mathcal{B}_{2}\backslash\mathcal{B}_{1}$.
Theorem 8.2.14 of \cite{Hormander_book} immediately implies the wavefront
set of the bi-distribution (\ref{eq:smooth boundary integral}) is
bounded by the union of three sets, 
\begin{align}
\text{WF}[\Theta] & \subseteq\left(\text{WF}'[H_{F}]\circ\text{WF}[H_{F}]\right)\cup\left(\text{WF}_{\mathcal{B}}[H_{F}]\times(\mathcal{B})\times\{0\}\right)\cup\left((\mathcal{B})\times\{0\}\times\text{WF}_{\mathcal{B}}[H_{F}]\right).\label{eq:WF(Omega)}
\end{align}
Here the notation is defined as follows: For any $u\in\mathcal{D}'(\mathbb{R}^{D}\times\mathbb{R}^{D})$,
\begin{align}
\text{WF}'[u] & \equiv\left\{ (x,k;y,p)|(x,k;y,-p)\in\text{WF}[u]\right\} \\
\text{WF}_{\mathcal{B}}[u] & \equiv\left\{ (x,k)|(x,k;y,0)\in\text{WF}[u]\text{ for some \ensuremath{y\in}\ensuremath{\mathcal{B}}}\right\} .
\end{align}
For any $u,v\in\mathcal{D}'(\mathbb{R}^{D}\times\mathbb{R}^{D})$,
the composition of wavefront sets $\text{WF}'[u]$ and $\text{WF}[v]$
goes as, 
\begin{align}
\text{WF}'[u]\circ\text{WF}[v] & \equiv\left\{ (x_{1},k_{1};x_{2},k_{2})|(y,p;x_{1},k_{1})\in\text{WF}'[u]\text{ and }(y,p;x_{2},k_{2})\in\text{WF}[v],\right.\nonumber \\
 & \hphantom{\equiv\left\{ (x_{1},k_{1};x_{2},k_{2})|\;\right.}\left.\text{for some }(y,p)\in\left(\mathcal{B}\times\mathbb{R}^{D}\backslash\{0\}\right)\right\} .\label{eq:WF' circ WF}
\end{align}
The form of the Feynman wavefront set (\ref{eq:Feyn WF}) immediately
implies that\footnote{In fact, eq.~(\ref{eq:WF_B(HF)}) would hold if $H_{F}$ was replaced
with any bi-distribution whose wavefront set contains only covectors
such that $k_{1}=-k_{2}$. Hence, by the discussion above, it holds
also for all translationally-invariant bi-distributions.}, 
\begin{equation}
\text{WF}_{\mathcal{B}}[H_{F}]\subset\emptyset,\label{eq:WF_B(HF)}
\end{equation}
so nontrivial contributions to the right-hand side of (\ref{eq:WF(Omega)})
could only potentially come from the set $\text{WF}'[H_{F}]\circ\text{WF}[H_{F}]$.
We show now this set is empty. Note, for $y\in\mathcal{B}$, we have
$(y,p;x_{1},k_{1})\in\text{WF}'[H_{F}]$ and $(y,p;x_{2},k_{2})\in\text{WF}[H_{F}]$
only if all three spacetime points $(y,x_{1},x_{2})$ reside on the
same null geodesic. Furthermore, when $x_{1},x_{2}\in\mathcal{B}_{1}$,
then any $y\in\mathcal{B}$ must be either to the future or to the
past of both $x_{1}$ and $x_{2}$. Consider first the case where
$y$ is to the future of $x_{1}$: By (\ref{eq:Feyn WF}), $(y,p;x_{1},k_{1})\in\text{WF}'[H_{F}]$
only if $p\in V_{y}^{-}$. However, when $y$ is to the future of
$x_{2}$, then $(y,p;x_{2},k_{2})\in\text{WF}[H_{F}]$ only if $p\in V_{y}^{+}$.
Since $V_{y}^{-}\cap V_{y}^{+}\subset\emptyset$, it follows that,
when $y$ is to the future of both points, there are no nontrivial
elements in (\ref{eq:WF' circ WF}). In the case where $y$ lies instead
to the past of both points, one arrives at the same conclusions only
with the roles of $V_{y}^{+}$ and $V_{y}^{-}$ swapped. Therefore,
when $x_{1},x_{2}\in\mathcal{B}_{1}$, 
\begin{equation}
\text{WF}'[H_{F}]\circ\text{WF}[H_{F}]\subset\emptyset,\label{eq:WF'(HF) comp WF(HF)}
\end{equation}
and, thus, (\ref{eq:WF(Omega)}) implies, 
\[
\text{WF}[\Theta]\subseteq\emptyset,
\]
which is what we sought to show. 
\end{proof}
\begin{rem}
The proof of Proposition \ref{lem:smoothness of bt} would \emph{not}
go through if the Feynman parametrix, $H_{F}$, was replaced by parametrices
for the advanced, $G_{A}$, or retarded, $G_{R}$, Green's functions.
In particular, one finds, $\text{WF}'[G_{A/R}]\circ\text{WF}[G_{A/R}]=\text{WF}[G_{A/R}]$,
respectively, so (\ref{eq:WF'(HF) comp WF(HF)}) would no longer hold.
Note also, despite its apparent similarity to $\Theta$, Proposition
\ref{lem:smoothness of bt} does \emph{not} apply to the integral
on the first line of (\ref{eq:prelim H_F flow eq}) which is \emph{not}
a smooth function in $(x_{1},x_{2})$. In particular, for the result
of Proposition \ref{lem:smoothness of bt}, it was critical that $y\notin\mathcal{B}_{1}$;
otherwise, it would be possible for $y$ to simultaneously lie to
the past of one point and to the future of the other, while being
an element of both $(y,p;x_{1},k_{1})\in\text{WF}'[H_{F}]$ and $(y,p;x_{2},k_{2})\in\text{WF}[H_{F}]$,
in which case, $\text{WF}'[H_{F}]\circ\text{WF}[H_{F}]=\text{WF}[H_{F}]\ne\emptyset$
and (\ref{eq:WF'(HF) comp WF(HF)}) no longer holds. 
\end{rem}
Since the second and third lines of (\ref{eq:prelim H_F flow eq})
are smooth, we may attempt to drop these terms and replace that flow
relation with 
\begin{equation}
\frac{\partial}{\partial m^{2}}H_{F}(x_{1},x_{2};m^{2})=\Omega_{M}(x_{1},x_{2};z;m^{2};L)\equiv-i\int_{{\mathcal{B}}_{2}}d^{D}y\,\chi(y;z;L)H_{F}(y,x_{1};m^{2})H_{F}(y,x_{2};m^{2}),\label{eq:Omega_M}
\end{equation}
As in the Euclidean case, this replacement will lead to difficulties
with scaling behavior under $(\eta_{ab},m^{2})\to(\lambda^{-2}\eta_{ab},\lambda^{2}m^{2})$.
(As previously mentioned, in a fixed global inertial coordinate system,
this is equivalent to rescaling $(\eta_{\mu\nu},d^{D}y,m^{2})\to(\lambda^{-2}\eta_{\mu\nu},\lambda^{-D}d^{D}y,\lambda^{2}m^{2})$.)
If this were the only difficulty with (\ref{eq:Omega_M}), it could
be dealt with in the same manner as in the Euclidean case. However,
a potentially much more serious difficulty arises from the fact that
(\ref{eq:Omega_M}) fails to be Poincare-invariant since there do
not exist Lorentz-invariant functions of compact support\footnote{Note the function used in Euclidean space, $\zeta(L^{-2}\sigma(y,z))$,
is Lorentz invariant but not compactly-supported in Minkowski spacetime,
since $\sigma(y,z)$ is zero on the boundary of the entire lightcone
of point $z$. }, 
\[
\chi(\Lambda y,\Lambda z)\ne\chi(y,z).
\]
Hence, for a Lorentzian metric, naively dropping the second and third
lines of (\ref{eq:prelim H_F flow eq}) would necessarily violate
the locality and covariance axiom W1, since this axiom implies Poincare
invariance in the case of flat spacetime.

It follows from the smoothness of the last two lines of (\ref{eq:prelim H_F flow eq})
for all $x_{1},x_{2}\in\mathcal{B}_{1}(z)$ that the failure of (\ref{eq:Omega_M})
to be Poincare-invariant on its own must then be given by a smooth
function of $(x_{1},x_{2})$. More precisely, for any $x_{1},x_{2}\in\mathcal{B}_{1}(z)$
and any Poincare transformation $P$ such that $Px_{1},Px_{2}\in\mathcal{B}_{1}(Pz)$,
the quantity 
\begin{equation}
\Omega_{M}(Px_{1},Px_{2};Pz)-\Omega_{M}(x_{1},x_{2};z)\label{eq:Omega_M Poincare failure is smooth}
\end{equation}
is smooth in $(x_{1},x_{2})$. Therefore, in parallel with our restoration
of desired scaling behavior in the Euclidean case, we will restore
Poincare invariance to the flow relation (\ref{eq:Omega_M}) if we
can replace $\Omega_{M}$ on the right-hand side of that equation
with a distribution $\widetilde{\Omega}_{M,\delta}$ which satisfies
the following two properties: 
\begin{enumerate}
\item For $(x_{1},x_{2})\in\mathcal{B}_{1}(z)$, 
\begin{equation}
\widetilde{\Omega}_{M,\delta}(x_{1},x_{2};z;m^{2};L)\equiv\Omega_{M}(x_{1},x_{2};z;m^{2};L)-\sum_{|\gamma_{1}|+|\gamma_{2}|\le\delta}\frac{1}{\gamma_{1}!\gamma_{2}!}a_{\gamma_{1}\gamma_{2}}(\chi)(x_{1}-z)^{\gamma_{1}}(x_{2}-z)^{\gamma_{2}},\label{eq:Omega_M tilde prop 1}
\end{equation}
where $a_{\gamma_{1}\gamma_{2}}=a_{\gamma_{2}\gamma_{1}}$ are constant
tensors that scale almost homogeneously under $(\eta_{ab},m^{2},L)\to(\lambda^{-2}\eta_{ab},\lambda^{2}m^{2},\lambda^{-1}L)$
with an overall factor of $\lambda^{(D-4)}$. 
\item To asymptotic degree $\delta$, $\widetilde{\Omega}_{M,\delta}$ is
asymptotically Poincare-invariant with respect to $(x_{1},x_{2},z)$.
That is, for any Poincare transformation $P$ such that $(Px_{1},Px_{2})\in\mathcal{B}_{1}(Pz)$,
\begin{equation}
\widetilde{\Omega}_{M,\delta}(Px_{1},Px_{2};Pz;m^{2};L)\sim_{\delta}\widetilde{\Omega}_{M,\delta}(x_{1},x_{2};z;m^{2};L).\label{eq:Omega_M tilde prop 2}
\end{equation}
\end{enumerate}
Note it is not required that $\widetilde{\Omega}_{M,\delta}$ be Poincare-invariant
at asymptotic degrees higher than $\delta$. Any two $\widetilde{\Omega}_{M,\delta}$
satisfying these properties may differ, to scaling degree $\delta$,
by at most a quantity of the form $L^{(D-4)}f(m^{2}\sigma(x_{1},x_{2}),L^{-2}\sigma(x_{1},x_{2}))$,
where $f$ is a smooth bi-variate function. Thus, the difference between
any two $a_{\gamma_{1}\gamma_{2}}$ and $a_{\gamma_{1}\gamma_{2}}'$
in (\ref{eq:Omega_M tilde prop 1}) is necessarily of the form, 
\[
a_{\gamma_{1}\gamma_{2}}-a_{\gamma_{1}\gamma_{2}}'=L^{(D-4)}\partial_{\gamma_{1}}^{(x_{1})}\partial_{\gamma_{2}}^{(x_{2})}f(m^{2}\sigma(x_{1},x_{2}),L^{-2}\sigma(x_{1},x_{2}))|_{x_{1},x_{2}=z}.
\]

If we can find a distribution $\widetilde{\Omega}_{M,\delta}$ satisfying
the above two properties, then the flow relation 
\begin{equation}
\frac{\partial}{\partial m^{2}}H_{F}(x_{1},x_{2};m^{2})\sim_{\delta}\widetilde{\Omega}_{M,\delta}(x_{1},x_{2};z;m^{2};L),\label{eq:invariant H_F flow rel with bad scaling}
\end{equation}
will be Poincare invariant. This flow relation still fails to scale
almost homogeneously with respect to the metric and $m^{2}$ due to
the dependence of $\widetilde{\Omega}_{M,\delta}$ on $L$. However,
the unwanted $L$-dependence can then be eliminated by the same procedure
as used in the Euclidean case treated in Subsection \ref{subsec:Eucl Had NO fe}.
Thus, we will be able to obtain satisfactory flow relation if we can
find a distribution $\widetilde{\Omega}_{M,\delta}$ satisfying the
above two properties. We turn now to the construction of the tensors
$a_{\gamma_{1}\gamma_{2}}$ in the definition (\ref{eq:Omega_M tilde prop 2})
of $\widetilde{\Omega}_{M}$ such that $\widetilde{\Omega}_{M,\delta}$
is Poincare invariant to scaling degree $\delta$ in the sense of
(\ref{eq:Omega_M tilde prop 2}).

Although we cannot choose the cutoff function $\chi(y,z)$ to be Lorentz
invariant, we can require that it be invariant under a simultaneous
translation of $(y,z)$. In particular, we can choose a global inertial
coordinate system on Minkowski spacetime and take $\chi$ to be given
by 
\begin{equation}
\chi(y,z;L;t^{\mu})=\zeta\left(L^{-2}\left(\eta_{\mu\nu}+2t_{\mu}t_{\nu}\right)(y-z)^{\mu}(y-z)^{\nu}\right),\label{eq:trans-inv Mink cutoff}
\end{equation}
where $t^{\mu}$ is proportional to the unit time vector field of
these coordinates but is required to remain unit normalized with respect
to the metric components under the rescaling $\eta_{\mu\nu}\to\lambda^{-2}\eta_{\mu\nu}$,
i.e., under this rescaling, it is required that $t_{\mu}\to\lambda^{-1}t_{\mu}$.
As in the Euclidean case, $\zeta$ is a test function and $\zeta(s)=1$
if $|s|\le1$ and $\zeta(z)=0$ if $|s|>2$. Note that $\eta_{\mu\nu}+2t_{\mu}t_{\nu}$
is a Riemannian metric with components $\text{diag}(+1,\dots,+1)$
in the chosen global inertial coordinates, so (\ref{eq:trans-inv Mink cutoff})
is supported on a $D$-dimensional coordinate ball of radius $2L$.
Equation~(\ref{eq:trans-inv Mink cutoff}) is manifestly translationally
invariant under a simultaneous translation of $(y,z)$. It is also
invariant under pure spatial rotations $(y,z)\to(Ry,Rz)$ since $(R^{-1}t)_{\mu}=t_{\mu}$,
but it is not invariant under Lorentz boosts. Note also the cutoff
(\ref{eq:trans-inv Mink cutoff}) is invariant under the rescaling
$(\eta_{ab},L)\to(\lambda^{-2}\eta_{ab},\lambda^{-1}L)$ with the
coordinate basis held fixed.

For any translationally-invariant $\chi$ and any Poincare transformation
$P$ composed of an arbitrary Lorentz transformation $\Lambda$ together
with an arbitrary translation, it follows that, 
\[
\Omega_{M}(Px_{1},Px_{2};Pz)=\Omega_{M}(\Lambda x_{1},\Lambda x_{2};\Lambda z),
\]
Plugging this into (\ref{eq:Omega_M tilde prop 2}) and using the
definition (\ref{eq:Omega_M tilde prop 1}) of $\widetilde{\Omega}_{M}$,
it follows that $\widetilde{\Omega}_{M}$ will be Poincare-invariant
to the required scaling degree if and only if $a_{\gamma_{1}\gamma_{2}}$
can be found such that, 
\begin{align}
 & \Omega_{M}(\Lambda x_{1},\Lambda x_{2};\Lambda z)-\Omega_{M}(x_{1},x_{2};z)\label{eq:Omega_M(Lambda)-Omega_M}\\
 & \qquad\sim_{\delta}\sum_{|\gamma_{1}|+|\gamma_{2}|\le\delta}\frac{1}{\gamma_{1}!\gamma_{2}!}(x_{1}-z)^{\gamma_{1}}(x_{2}-z)^{\gamma_{2}}\left(\Lambda_{\hphantom{\gamma_{1}'}\gamma_{1}}^{\gamma_{1}'}\Lambda_{\phantom{\gamma_{2}'}\gamma_{2}}^{\gamma_{2}'}-\delta_{\hphantom{\gamma_{1}'}\gamma_{1}}^{\gamma_{1}'}\delta_{\hphantom{\gamma_{2}'}\gamma_{2}}^{\gamma_{2}'}\right)a_{\gamma_{1}'\gamma_{2}'},\nonumber 
\end{align}
where $\Lambda_{\hphantom{\alpha'}\alpha}^{\alpha'}\equiv\Lambda_{\hphantom{\mu_{1}'}\mu_{1}}^{\mu_{1}'}\cdots\Lambda_{\hphantom{\mu_{|\alpha|}'}\mu_{|\alpha|}}^{\mu_{|\alpha|}'}$
with the convention $\Lambda_{\hphantom{\alpha'}\alpha}^{\alpha'}=1$
if $|\alpha|=0$. Since the left side of \ref{eq:Omega_M(Lambda)-Omega_M})
has been shown to be smooth in $(x_{1},x_{2})$, it is asymptotic
to its Taylor series. Hence, Taylor expanding the first line and equating
the coefficients of $(x_{1}-z)^{\gamma_{1}}(x_{2}-z)^{\gamma_{2}}$
appearing on both sides, we see that $a_{\gamma_{1}\gamma_{2}}$ must
satisfy 
\begin{equation}
\left[\vphantom{\int_{B}^{A}}\partial_{\gamma_{1}}^{(x_{1})}\partial_{\gamma_{2}}^{(x_{2})}\left[\vphantom{\sum}\Omega_{M}(\Lambda x_{1},\Lambda x_{2};\Lambda z)-\Omega_{M}(x_{1},x_{2};z)\right]\right]_{x_{1},x_{2}=z}=\left(\Lambda_{\hphantom{\gamma_{1}'}\gamma_{1}}^{\gamma_{1}'}\Lambda_{\phantom{\gamma_{2}'}\gamma_{2}}^{\gamma_{2}'}-\delta_{\hphantom{\gamma_{1}'}\gamma_{1}}^{\gamma_{1}'}\delta_{\hphantom{\gamma_{2}'}\gamma_{2}}^{\gamma_{2}'}\right)a{}_{\gamma_{1}'\gamma_{2}'},\label{eq:Taylor coef Omega_M(Lambda)-Omega_M}
\end{equation}
If $\Omega_{M}$ were itself a smooth function of $(x_{1},x_{2})$,
then we could trivially satisfy (\ref{eq:Taylor coef Omega_M(Lambda)-Omega_M})
by setting $a_{\gamma_{1}\gamma_{2}}$ equal to the Taylor coefficients
of $\Omega_{M}(x_{1},x_{2};z)$ evaluated at $x_{1},x_{2}=z$. However,
$\Omega_{M}$ is fundamentally distributional, so it is far from obvious
that there exist $\Lambda$-independent $a_{\gamma_{1}\gamma_{2}}$
satisfying (\ref{eq:Taylor coef Omega_M(Lambda)-Omega_M}).

In Appendix \ref{sec:K-G Lorentz cts} we show that (\ref{eq:Taylor coef Omega_M(Lambda)-Omega_M})
can always be solved and we obtain explicit solutions. First, we use
a cohomological argument to prove existence of solutions $a_{\gamma_{1}\gamma_{2}}$
to (\ref{eq:Taylor coef Omega_M(Lambda)-Omega_M}). We then obtain
the explicit solutions for $a_{\gamma_{1}\gamma_{2}}$ in the cases
of rank $r=1,2$, where $r\equiv|\gamma_{1}|+|\gamma_{2}|$. The $r=1$
solutions are 
\begin{equation}
a_{\{\mu\}\{0\}}=a_{\{0\}\{\mu\}}=-i\int d^{D}y\partial_{\mu}^{(y)}\chi(y,\vec{0})H_{F}(y,\vec{0})H_{F}(y,\vec{0}).\label{eq:a mu 0}
\end{equation}
and the $r=2$ solutions are\footnote{In eqs. \eqref{eq:a mu nu} and \eqref{eq:a mu nu 0}, it is understood
that the subtraction inside the integrand must be performed prior
to evaluating the integral, since the individual terms in the integrand
contain non-integrable divergences at $y=\vec{0}$, i.e., the integrand
is well-defined as a distribution in $y$ only when $y\ne\vec{0}$,
but its definition can be \emph{uniquely} extended to include the
origin.}
\begin{align}
a_{\{(\mu\}\{\nu)\}} & =-i\int d^{D}y\chi(y,\vec{0})\left[\partial_{\mu}H_{F}(y,\vec{0})\partial_{\nu}H_{F}(y,\vec{0})-\frac{1}{D}\eta_{\mu\nu}\partial_{\sigma}H_{F}(y,\vec{0})\partial^{\sigma}H_{F}(y,\vec{0})\right],\label{eq:a mu nu}
\end{align}
and 
\begin{align}
a_{\{(\mu\nu)\}\{0\}}=a_{\{0\}\{(\mu\nu)\}} & =-i\int d^{D}y\chi(y,\vec{0})\left[H_{F}(y,\vec{0})\partial_{\mu}\partial_{\nu}H_{F}(y,\vec{0})-\frac{1}{D}\eta_{\mu\nu}H_{F}(y,\vec{0})\partial^{2}H_{F}(y,\vec{0})\right].\label{eq:a mu nu 0}
\end{align}
Finally, we obtain the recursive solution (\ref{eq:a soln for r>2})
for $a_{\gamma_{1}\gamma_{2}}$ for all $r>2$.

With the above solution for $a_{\gamma_{1}\gamma_{2}}$, we obtain
$\widetilde{\Omega}_{M,\delta}$ satisfying (\ref{eq:Omega_M tilde prop 1})
and (\ref{eq:Omega_M tilde prop 2}). We thereby obtain the Poincare-invariant
flow relation (\ref{eq:invariant H_F flow rel with bad scaling}).
However, as in the Euclidean case, the flow relation (\ref{eq:invariant H_F flow rel with bad scaling})
is not compatible with the scaling behavior of the Wick monomials
required by the scaling axiom W7. Nevertheless, as in the Euclidean
case, we can obtain a flow relation that remains compatible with Poincare
invariance and satisfies the desired scaling behavior by replacing
$\widetilde{\Omega}_{M,\delta}$ on the right side of (\ref{eq:invariant H_F flow rel with bad scaling})
with 
\begin{equation}
\mathfrak{L}[L]\widetilde{\Omega}_{M}(x_{1},x_{2};z;L)-\sum_{|\gamma_{1}|+|\gamma_{2}|\le\delta}\frac{1}{\gamma_{1}!\gamma_{2}!}c_{\gamma_{1}\gamma_{2}}(L)(x_{1}-z)^{(\gamma_{1}}(x_{2}-z)^{\gamma_{2})},\label{eq:L-indep Omega_M tilde}
\end{equation}
where $\mathfrak{L}$ was defined by (\ref{eq:full L op}) and 
\begin{equation}
c_{\gamma_{1}\gamma_{2}}(L)\equiv\int_{0}^{L}dL'\left[\partial_{\gamma_{1}}^{(x_{1})}\partial_{\gamma_{2}}^{(x_{2})}\frac{\partial}{\partial L'}\left(\mathfrak{L}[L']\widetilde{\Omega}_{M}(x_{1},x_{2};z;L'\right)\right].\label{eq:c_gamma1gamma2 Minkowski}
\end{equation}
The distribution (\ref{eq:L-indep Omega_M tilde}) is Poincare-invariant
and is asymptotically independent of $L$ up to scaling degree $\delta$.
Moreover, the distribution (\ref{eq:L-indep Omega_M tilde}) differs
from $\Omega_{M}$ by a smooth function of $(x_{1},x_{2})$. Hence,
the distribution (\ref{eq:L-indep Omega_M tilde}) \emph{can} be used
in a flow relation for the Feynman parametrix $C_{T_{0}\{\phi\phi\}}^{I}=H_{F}$
which is compatible with all Wick axioms. Recalling the definition
(\ref{eq:Omega_M tilde prop 1}) of $\widetilde{\Omega}_{M}$ and
the explicit formulas (\ref{eq:C_phi^2 phi phi ^I})-(\ref{eq:C_phi phi ^C})
for the OPE coefficients, the flow relation with (\ref{eq:L-indep Omega_M tilde})
on the right-hand side can be written in the form: 
\begin{align}
 & \frac{\partial}{\partial m^{2}}(C_{H})_{T_{0}\{\phi\phi\}}^{I}(x_{1},x_{2};z)\label{eq:Minkowski fe T0 2pt id coef}\\
 & \sim_{\delta}-\frac{i}{2}\int d^{D}y\mathfrak{L}[L]\chi(y,z;L)(C_{H})_{T_{0}\{\phi^{2}\phi\phi\}}^{I}(y,x_{1},x_{2};z)-\sum_{[C]\le\delta+2}c_{C}(C_{H}){}_{T_{0}\{\phi\phi\}}^{C}(x_{1},x_{2};z),\nonumber 
\end{align}
where $c_{C}=0$ unless $[C]_{\phi}=2$, in which case it is given
by 
\begin{equation}
c_{C}(L)\equiv\mathfrak{L}[L]a_{C}(L)-\int_{0}^{L}dL'\frac{\partial}{\partial L'}\left[\frac{i}{2}\int d^{D}y\mathfrak{L}[L']\chi(y,\vec{0};L')(C_{H}){}_{T_{0}\{\phi^{2}C\}}^{I}(y,z;z)+\mathfrak{L}[L']a_{C}(L')\right],\label{eq:c_C Minkowski}
\end{equation}
for $L>0$. The tensors $a_{C}$ are also zero unless $[C]_{\phi}=2$,
in which case, they are inductively defined via (\ref{eq:a soln for r>2})
in terms of 
\begin{equation}
(B^{\kappa\rho})_{C}\equiv i\int d^{D}y\,y^{[\kappa}\partial^{\rho]}\chi(y,\vec{0})(C_{H}){}_{T_{0}\{\phi^{2}C\}}^{I}(y,\vec{0};\vec{0}).\label{eq:(B^kappa rho)_C}
\end{equation}
Note, by writing the $y$-integral in \eqref{eq:Minkowski fe T0 2pt id coef},
we have implicitly (uniquely) extended the OPE coefficient $(C_{H})_{T_{0}\{\phi^{2}\phi\phi\}}^{I}(y,x_{1},x_{2};z)=2H_{F}(y,x_{1})H_{F}(y,x_{2})$
to the partial diagonals $y=x_{1}$ and $y=x_{2}$ as justified in
Remark \ref{rem: ext of T0[C^I_phi phi]} above.

The inductive solution (\ref{eq:a soln for r>2}) determines $a_{C}$
up to Lorentz-invariant tensors of the correct rank which scale with
an overall factor of $\lambda^{(D-4)}$ under $(\eta_{ab},m^{2},L)\to(\lambda^{-2}\eta_{ab},\lambda^{2}m^{2},\lambda^{-2}L)$
and depend smoothly on $(\eta_{ab},m^{2})$. Although the inherent
ambiguities in $a_{C}$ may depend on $L$, the $\mathfrak{L}$-operator
and $L$-integral terms in \eqref{eq:c_C Minkowski} ensure that only
the $L$-independent parts of $a_{C}$ can contribute non-trivially
to $c_{C}$. Therefore, the only ambiguity in $c_{C}$ corresponds
to the choice of an $L$-independent tensor in $a_{C}$ that scales
with an overall factor of $\lambda^{(D-4)}$ under $(\eta_{ab},m^{2})\to(\lambda^{-2}\eta_{ab},\lambda^{2}m^{2})$.
In odd dimensions, there are no tensors that scale in this way and
depend smoothly on $(\eta_{ab},m^{2})$, so $a_{C}$ is unique. In
even dimensions, this ambiguity corresponds to freedom to choose the
Taylor coefficients of a Poincare-invariant smooth function in $(x_{1},x_{2},m^{2})$.
We note also, as discussed in Remark \ref{rem:conservation in Eucl space},
that the conservation constraint \eqref{eq:H conservation condition-1}
is automatically satisfied in flat spacetime.

The flow relation (\ref{eq:Minkowski fe T0 2pt id coef}) for the
(unextended) time-ordered OPE coefficient $(C_{H}){}_{T_{0}\{\phi\phi\}}^{I}$
is the Minkowski spacetime analogue of the Euclidean flow relation
(\ref{eq:L-ind Euclidean fe n=2}) for the ordinary OPE coefficient
$(C_{H}){}_{\phi\phi}^{I}$. In both cases, the inherent ambiguity
in the flow relation  corresponds to a smooth function that is invariant
under the respective isometry group.  By \prettyref{thm:explicit TO coef},
formulas for the (unextended) time-ordered OPE coefficients, $C_{T_{0}\{A_{1}\cdots A_{n}\}}^{B}\equiv T_{0}\{C_{A_{1}\cdots A_{n}}^{B}\}$,
for any given Wick prescription are obtained from formulas for the
corresponding non-time-ordered OPE coefficients, $C_{A_{1}\cdots A_{n}}^{B}$,
by simply replacing all occurrences of the Hadamard parametrix $H$
with its corresponding Feynman parametrix $H_{F}=H(x_{1},x_{2})-i\Delta^{\text{adv}}(x_{1},x_{2})$.
Hence, from the explicit formulas for the Hadamard normal-ordered
OPE coefficients (see \eqref{eq:Schwinger functions} and \eqref{eq:G NO C^phi^2 phi ... phi _I})
and the flow relation \eqref{eq:Minkowski fe T0 2pt id coef}, we
immediately obtain the following theorem giving the flow relations
for the (unextended) time-ordered OPE coefficients $(C_{H}){}_{T_{0}\{\phi\cdots\phi\}}^{I}(x_{1},\dots,x_{n};z)$. 
\begin{thm}
\label{thm:Minkowski flow rel}For any Hadamard parametrix satisfying
\eqref{eq:Minkowski fe T0 2pt id coef}, the corresponding OPE coefficients
$(C_{H}){}_{T_{0}\{\phi\cdots\phi\}}^{I}$ satisfy: 
\begin{align}
 & \frac{\partial}{\partial m^{2}}(C_{H}){}_{T_{0}\{\phi\cdots\phi\}}^{I}(x_{1},\dots,x_{n};z)\approx\label{eq:id coef fe Minkowski}\\
 & -\frac{i}{2}\int d^{D}y\mathfrak{L}[L]\chi(y,z;L)(C_{H}){}_{T_{0}\{\phi^{2}\phi\cdots\phi\}}^{I}(y,x_{1},\dots,x_{n};z)-\sum_{C}c_{C}(L)(C_{H}){}_{T_{0}\{\phi\cdots\phi\}}^{C}(x_{1},\dots,x_{n};z),\nonumber 
\end{align}
where $c_{C}$ is given by formula (\ref{eq:c_C Minkowski}) with
the same ambiguities arising from $a_{C}$. 
\end{thm}
Note that the inherent ambiguities in these flow relations are in
1-1 correspondence with the freedom to choose a Hadamard parametrix
whose corresponding Hadamard normal-ordered Wick fields are compatible
with axioms W1-W8. 
\begin{rem}
\label{rem:unique ext of Mink flow rel}As emphasized in Subsection
\ref{subsec:TOP}, the extension of $T_{0}\{\Phi_{A_{1}}(x_{1})\cdots\Phi_{A_{n}}(x_{n})\}$
to algebra-valued distributions defined on the diagonals generally
introduces additional ``contact-term'' ambiguities proportional
to $\delta$-distributions (and their distributional derivatives).
However, the scaling degree of $\delta^{(D)}(x_{1},\dots,x_{n})$
is $n\cdot D$, whereas by Theorem \ref{thm:existence Wick coef and associativity}
the scaling degree of the coefficients $C_{T_{0}\{\phi\cdots\phi\}}^{I}(x_{1},\dots,x_{n};z)$
appearing in the flow relation (\ref{eq:id coef fe Minkowski}) is
$n\cdot(D-2)/2$. Since $(D-2)/2$ is strictly less than $D$ for
$D\ge2$, it follows that there exists no contact terms with scaling
degree less than or equal to the scaling degree of $C_{T_{0}\{\phi\cdots\phi\}}^{I}$.
By the axioms for time-ordered products in \cite{HW_local_Wick_poly,HW_Conservation_Stress-energy},
this implies that the extension of the $C_{T_{0}\{\phi\cdots\phi\}}^{I}$
coefficients to the diagonals is unique and, therefore, it so happens
that we could replace $T_{0}$ with $T$ in formula (\ref{eq:id coef fe Minkowski})
without introducing additional contact term ambiguities. Note, however,
that this does not occur for the general unextended time-ordered Wick
coefficients $C_{T_{0}\{A_{1}\cdots A_{n}\}}^{B}$ nor, in general,
for the coefficients appearing in the flow relations (\ref{eq:HH fe})
of $\lambda\phi^{4}$-theory. 
\end{rem}
Relation \eqref{eq:id coef fe Minkowski} of Theorem \ref{thm:Minkowski flow rel}
applies to the time-ordered OPE coefficients for the Hadamard normal-ordered
Wick fields. However, following the steps outlined below \prettyref{thm:Eucl fe for H-NO id coef}
of the preceding section, one may straightforwardly obtain flow relations
for the time-ordered OPE coefficients corresponding to any prescription
for the Wick fields satisfying axioms W1-W8. These relations will
similarly take the same general form as the Hadamard normal-ordered
relation \eqref{eq:id coef fe Minkowski} except there will be additional
terms containing factors of $F_{k}$ (with $k\le n)$ as in relation
(\ref{eq:Eucl L-indep fe with F_k terms}) above.

Finally, we note that our derivation of the flow relation \eqref{eq:id coef fe Minkowski}
for general $n$ relied heavily on our knowledge of the explicit expressions
for the time-ordered OPE coefficients of Hadamard normal-ordered Wick
fields, since this knowledge enabled us to obtain \eqref{eq:id coef fe Minkowski}
from the flow relation \eqref{eq:Minkowski fe T0 2pt id coef} for
$n=2$ via inspection. However, if the OPE coefficients with $n>2$
had not been related in a simple, known manner to the $n=2$ OPE coefficients,
we would not have been able to construct covariance-restoring terms
for the $n>2$ case using the techniques described in this section.
In Appendix \ref{sec:model-indep counterterms}, we develop a general
method for constructing covariance-restoring counterterms based on
the model-independent associativity conditions that can be applied
to the $n>2$ case and show that this general algorithm reproduces
the results claimed here.

\section{Flow relations for OPE coefficients in curved spacetime \label{sec:Wick fe in CS}}

In this section, we obtain flow relations for the unextended time-ordered
Wick OPE coefficients in general globally-hyperbolic Lorentzian spacetimes
$(M,g_{ab})$ in any dimension $D\ge2$. As in the preceding Minkowski
section, we focus attention initially to the flow relation for the
time-ordered OPE coefficient $(C_{H})_{T_{0}\{\phi\phi\}}^{I}=H_{F}(x_{1},x_{2};m^{2};\xi)$,
since the flow relations for other time-ordered Wick OPE coefficients
may be straightforwardly obtained once the flow relation for $(C_{H})_{T_{0}\{\phi\phi\}}^{I}$
is known.

In curved spacetime, any Feynman parametrix $H_{F}$ used for the
construction of $(C_{H})_{T_{0}\{\phi\phi\}}^{I}$ is required to
be locally and covariantly defined and have (jointly) smooth dependence
on the coupling parameters $(m^{2},\xi)$. As already noted in Proposition
\ref{thm:explicit TO coef}, the relation between $H_{F}$ and a Hadamard
parametrix $H$ (see (\ref{eq:H})) is given by 
\begin{equation}
H_{F}(x_{1},x_{x})=H(x_{1},x_{2})-i\Delta^{\text{adv}}(x_{1},x_{2})
\end{equation}
with $\Delta^{\text{adv}}$ denoting the advanced Green's function.
Since the forms of $H$ and $H_{F}$ depend on the squared geodesic
distance function $\sigma(x_{1},x_{2})$, these parametrices are well
defined only in convex normal neighborhoods. The Feynman parametrix
is a fundamental solution to the Klein-Gordon equation 
\begin{equation}
(-g^{ab}\nabla_{a}\nabla_{b}+m^{2}+\xi R)H_{F}(x_{1},x_{2};m^{2};\xi)\approx-i\delta^{(D)}(x_{1},x_{2})+\text{smooth terms},\label{eq:K H_F CS}
\end{equation}
Furthermore, in curved spacetime, the wavefront set of $H_{F}$ continues
to be of the form $(x_{1},k;x_{2},-k)$ \cite{Rad_microlocal_Hadamard_cond}.
In particular, when smeared in either of its spacetime variables with
a test function $f$ of sufficiently small compact support, $H_{F}(y,f)$
is a smooth function in $y$ within a convex normal neighborhood of
the support of $f$.

The above properties of $H_{F}$ were all that were needed to obtain
the initial flow relation (\ref{eq:Omega_M}) in Minkowski spacetime,
so we can parallel these steps to derive a similar flow relation in
any any globally-hyperbolic curved spacetime $(M,g_{ab})$. To do
so, let $U_{z}\subset M$ be a convex normal neighborhood of the point
$z\in M$. It is convenient to work in a Riemannian normal coordinate
(RNC) system about $z$. A RNC system is constructed introducing an
orthonormal basis (i.e., ``tetrad'') for $T_{z}M$, 
\begin{equation}
\left\{ (e_{\mu})^{a}\in T_{z}M|\,\mu\in\{0,\dots,D-1\}\text{ and }g_{ab}(e_{\mu})^{a}(e_{\nu})^{b}=\eta_{\mu\nu}\right\} .\label{eq:tetrad}
\end{equation}
The tetrad allows us to identify $T_{z}M$ with $\mathbb{R}^{D}$.
We then use the exponential map\textemdash which maps $v^{a}\in T_{z}M$
into the point in $M$ lying at unit affine parameter along the geodesic
determined by $(z,v^{a})$\textemdash to provide a diffeomorphism
between $U_{z}$ and a neighborhood $\mathcal{U}_{0}$ of the origin
of $\mathbb{R}^{D}$. This correspondence provides coordinates $x^{\mu}$
on $U_{z}$. We denote by $t^{\mu}$ the RNC components of the timelike
vector at $z$ that is proportional to $(e_{0})^{\mu}$ but required
to remain unit-normalized with respect to the metric components under
the rescaling $g_{\mu\nu}\to\lambda^{-2}g_{\mu\nu}$, i.e., under
this rescaling, it is required that $t_{\mu}\to\lambda^{-1}t_{\mu}$.
Let $\zeta\in C_{0}^{\infty}(\mathbb{R})$ again denote a test function
that is equal to one for $|s|\leq1$ and vanishes for $|s|\geq2$.
We then define a cutoff function on $U_{z}$ via, 
\begin{equation}
\chi[g{}_{\mu\nu},t_{\mu},L](y;\vec{0})=\zeta\left(L^{-2}\left(g_{\mu\nu}(\vec{0})+2t_{\mu}t_{\nu}\right)y^{\mu}y^{\nu}\right),\label{eq:chi_C}
\end{equation}
where $L$ is chosen such that the coordinate ball of radius $2L$
lies within $U_{z}$. Here $y^{\mu}$ denotes the RNC values of $y$
and we have denoted $z$ by its RNC value $\vec{0}$. Note that the
cutoff function (\ref{eq:chi_C}) is invariant under the simultaneous
rescaling $(g_{ab},L)\to(\lambda^{-2}g_{ab},\lambda^{-1}L)$ with
the RNC coordinate basis held fixed.

With these definitions and constructions, we can now straightforwardly
generalize the derivation of (\ref{eq:Omega_M}) to curved spacetime.
We obtain 
\begin{equation}
\frac{\partial}{\partial m^{2}}H_{F}[g_{\mu\nu}](x_{1},x_{2})\approx\Omega_{C}[g_{\mu\nu},t_{\mu},L](x_{1},x_{2};\vec{0})+\text{terms smooth in \ensuremath{(x_{1},x_{2})}},\label{eq:m^2 flow eq CS}
\end{equation}
where 
\begin{align}
\Omega_{C}[g_{\mu\nu}, & t_{\mu},L](f_{1},f_{2};\vec{0})\equiv\label{eq:Omega_C}\\
 & -i\int_{\mathbb{R}^{D}}d^{D}y\sqrt{-g(y)}\,\chi[g_{\mu\nu},t_{\mu},L](y;\vec{0})\,H_{F}[g_{\mu\nu}](y,f_{1})H_{F}[g_{\mu\nu}](y,f_{2}).\nonumber 
\end{align}
In curved spacetime, the parameter $\xi$ enters the Klein-Gordon
equation in a nontrivial manner and we also seek a flow equation in
$\xi$. Using the fact that the commutator of the differential operator
$\partial_{\xi}\equiv\partial/\partial\xi$ with the Klein-Gordon
operator (\ref{eq:K-G op}) is given by 
\[
[K,\partial_{\xi}]=-RI.
\]
we can similarly derive the $\xi$-flow equation 
\begin{align}
 & \frac{\partial}{\partial\xi}H_{F}[g_{\mu\nu}](x_{1},x_{2};\xi)\label{eq:xi flow eq}\\
 & \qquad\approx-i\int_{\mathbb{R}^{D}}d^{D}y\sqrt{-g(y)}\chi(y,z)R(y)H_{F}[g_{\mu\nu}](y,x_{1};\xi)H_{F}[g_{\mu\nu}](y,x_{2};\xi)+\text{smooth}.\nonumber 
\end{align}
Note that the integral in the second line vanishes unless the scalar
curvature is nonzero. Since the analysis of the flow relations (\ref{eq:m^2 flow eq CS})
and (\ref{eq:xi flow eq}) are essentially identical, in the following
we will focus attention on only the $m^{2}$-flow relation (\ref{eq:m^2 flow eq CS}),
it being understood that (\ref{eq:xi flow eq}) can be analyzed in
a completely parallel manner, with the minor differences described
in Remark \ref{rem:xi flow rel} below Theorem \ref{thm:flow rel curved spacetime}.

If we attempt to drop the smooth terms and use (\ref{eq:m^2 flow eq CS})
as our flow equation we will encounter three major difficulties: (i)
Since the quantity $\Omega_{C}$ is defined in (\ref{eq:Omega_C})
by an integral over a finite spacetime region, $\Omega_{C}$ depends
nonlocally on the metric, which is not compatible with axiom W1. (ii)
On account of the presence of the cutoff function $\chi$, $\Omega_{C}$
is not covariantly defined, which also is not compatible with axiom
W1. (iii) On account of the cutoff scale $L$ present in $\chi$,
the scaling dependence of the OPE coefficients will not be compatible
with axiom W7. As we shall now show, these difficulties can be overcome
by suitably modifying the flow relation (\ref{eq:m^2 flow eq CS}).
Specifically, difficulty (i) can be overcome by replacing (\ref{eq:Omega_C})
with a similar expression involving the Taylor coefficients of the
metric in an expansion about $z$ rather than the metric itself. Difficulty
(ii) then can be overcome by a generalization of the procedure used
to restore Lorentz invariance in Minkowski spacetime. Finally, difficulty
(iii) can be overcome by the same procedure as used for the Euclidean
and Minkowski flow relations. We now discuss, in turn, these difficulties
and their resolutions.

\paragraph{(i) Locality.}

As already indicated above, the key idea needed to convert (\ref{eq:Omega_C})
into an expression that depends only on the metric in an arbitrarily
small neighborhood of $z$ is to replace the metric by its Taylor
approximation about $z$, carried to sufficiently high order. To scaling
degree $\delta$, the RNC components of the metric are asymptotically
equivalent to its Taylor polynomial about the origin, 
\begin{align}
g_{\mu\nu}(x)\sim_{\delta}g{}_{\mu\nu}^{(N)}(x) & \equiv\sum_{k=0}^{N}\frac{1}{k!}x^{\sigma_{1}}\cdots x^{\sigma_{k}}\left.\frac{\partial^{k}g_{\mu\nu}(x)}{\partial x^{\sigma_{1}}\cdots\partial x^{\sigma_{k}}}\right|_{x=\vec{0}}\label{eq:g^(N)_munu}\\
 & =\eta_{\mu\nu}+\frac{1}{3}R_{\mu\nu\kappa\rho}(\vec{0})x^{\kappa}x^{\rho}-\frac{1}{6}\nabla_{\sigma}R_{\mu\nu\kappa\rho}(\vec{0})x^{\kappa}x^{\rho}x^{\sigma}+\cdots,\nonumber 
\end{align}
provided that we take $N\ge\delta$. As indicated by the second line
of (\ref{eq:g^(N)_munu}), the Taylor coefficients are expressible
entirely in terms of the Riemannian curvature tensor and its totally-symmetric
covariant derivatives evaluated at the origin\footnote{This follows from a close relative \cite[see Lemma 2.1]{Wald_Iyer}
of the ``Thomas replacement theorem'' \cite{Thomas_Replacement}.}. For sufficiently large $x^{\mu}$, the Taylor polynomial $g_{\mu\nu}^{(N)}(x)$
need not define a Lorentz metric. However, we can choose $L$ sufficiently
small that $|g_{\mu\nu}^{(N)}-\eta_{\mu\nu}|\ll1$ within a coordinate
ball of radius $2L$, so that $g_{\mu\nu}^{(N)}(x)$ is a Lorentz
metric wherever $\chi$ is nonvanishing.

To proceed, we perform an expansion of $\Omega_{C}[g_{\mu\nu},t_{\mu},L]$
about $g_{\mu\nu}=\eta_{\mu\nu}$ as a power series in the (symmetrized)
covariant derivatives of the Riemann curvature tensor. This curvature
expansion as well as the precise bound on the scaling degree of its
non-smooth terms is derived in Appendix \ref{sec:proof of Omega_C curvature exp}.
This expansion also will be needed for our construction of covariance-restoring
counterterms below. The expansion takes the form\footnote{To avoid overly cumbersome notation involving multiple subscripts
on spacetime indices, we have implicitly re-used some Greek letters
in (\ref{eq:Omega_C curvature exp}), but the intended summations
should be clear from context.} 
\begin{align}
 & \Omega_{C}[g_{\mu\nu},t_{\mu},L](f_{1},f_{2};\vec{0})\label{eq:Omega_C curvature exp}\\
 & \qquad\sim_{\delta}\sum_{k=0}^{\delta+D-4}\sum_{\vec{p}_{k}}(\Omega_{\vec{p}}){}^{\{\mu\cdots\sigma_{k-2}\}}[\eta_{\mu\nu},t_{\mu},L](f_{1},f_{2};\vec{0})\prod_{j=0}^{k-2}\left[R_{\mu\nu\kappa\rho;(\sigma_{1}\cdots\sigma_{j})}(\vec{0})\right]^{p_{j}}+\text{smooth terms}.\nonumber 
\end{align}
Here we have defined, 
\begin{equation}
(\Omega_{\vec{p}}){}^{\{\mu\cdots\sigma_{k-2}\}}[\eta_{\mu\nu},t_{\mu},L](f_{1},f_{2};\vec{0})\equiv\left.\frac{\partial^{P}\Omega_{C}[g_{\mu\nu}^{(k)},t_{\mu},L](f_{1},f_{2};\vec{0})}{\partial^{p_{0}}R_{\mu\nu\kappa\rho}(\vec{0})\cdots\partial^{p_{k-2}}R_{\mu\nu\kappa\rho;(\sigma_{1}\cdots\sigma_{k-2})}(\vec{0})}\right|_{g_{\mu\nu}^{(k)}=\eta_{\mu\nu}},\label{eq:Omega_C^gamma}
\end{equation}
where $g_{\mu\nu}^{(k)}$ denotes the $k$th-order polynomial metric
(\ref{eq:g^(N)_munu}) computed from $g_{\mu\nu}$ and $P\equiv\sum_{j=0}^{k-2}p_{j}$.
In (\ref{eq:Omega_C curvature exp}) the $\vec{p}_{k}$-sum runs over
all non-negative integers $\vec{p}_{k}\equiv(p_{0},\dots,p_{k-2})$
such that 
\begin{equation}
2p_{0}+3p_{1}+\cdots+kp_{(k-2)}=k.\label{eq:p-sum bounds}
\end{equation}

Note (\ref{eq:Omega_C^gamma}) are tensor-valued distributions defined
on a neighborhood of the origin in \emph{flat} Minkowski spacetime,
$(\mathcal{N}_{0},\eta_{\mu\nu})$. Hence, all of the curvature dependence
of the explicit terms in the curvature expansion (\ref{eq:Omega_C curvature exp})
for $\Omega_{C}$ comes through a finite product of curvature tensors
evaluated at the origin. Note the derivatives in (\ref{eq:Omega_C^gamma})
with respect to curvature tensors are well-defined because the smeared
distribution $\Omega_{C}$ is a smooth function of the metric and
the polynomial metric $g_{\mu\nu}^{(k)}$ is a smooth function of
finitely-many curvature tensors evaluated at the origin: 
\begin{equation}
g_{\mu\nu}^{(k)}(x)=g_{\mu\nu}^{(k)}[x^{\sigma},\eta_{\mu\nu},R_{\mu\nu\kappa\rho}(\vec{0}),\nabla_{\sigma}R_{\mu\nu\kappa\rho}(\vec{0}),\dots,\nabla_{(\sigma_{1}}\cdots\nabla_{\sigma_{k-2})}R_{\mu\nu\kappa\rho}(\vec{0})].\label{eq:g^k functional dep on curvature}
\end{equation}

The result needed to effectively replace $g_{\mu\nu}$ with $g_{\mu\nu}^{(N)}$
in (\ref{eq:Omega_C}) is given in the following proposition: 
\begin{prop}
\label{Prop: Omega_C[g]=Omega_C[g^N]+smooth}
Let $H_{F}$ be a local and covariant Feynman parametrix which scales
almost homogeneously with an overall factor of $\lambda^{(D-2)}$
under $(g_{\mu\nu},m^{2})\to(\lambda^{-2}g_{\mu\nu},\lambda^{2}m^{2})$
and which depends smoothly on $m^{2}$. Let $\Omega_{C}$ be given
by (\ref{eq:Omega_C}). Then for all $\ensuremath{N\ge\delta+D-4}$,
we have 
\begin{equation}
\Omega_{C}[g_{\mu\nu},t_{\mu},L](x_{1},x_{2};\vec{0})\sim_{\delta}\Omega_{C}[g_{\mu\nu}^{(N)},t_{\mu},L](x_{1},x_{2};\vec{0})+\text{terms smooth in \ensuremath{(x_{1},x_{2})},}\label{eq:Omega_C[g]=Omega_C[g^N]+smooth}
\end{equation}
\end{prop}
\begin{proof}[Proof.]
The proposition can be deduced from the curvature expansion (\ref{eq:Omega_C curvature exp})
for $\Omega_{C}$ as follows: Note the maximum number of covariant
derivatives of $R_{\mu\nu\kappa\rho}$ appearing in the curvature
expansion (\ref{eq:Omega_C curvature exp}) is $\delta+D-6$. Consider
first the special case that $g_{\mu\nu}=g_{\mu\nu}^{(P)}$ for arbitrary,
but finite, integer $P$. We want to determine the smallest integer
$N<P$ such that the relation (\ref{eq:Omega_C[g]=Omega_C[g^N]+smooth})
of the proposition holds. Since both $g_{\mu\nu}$ and $g_{\mu\nu}^{(N)}$
are themselves polynomial metrics of the form (\ref{eq:g^(N)_munu}),
it follows immediately that their respective polynomial approximations,
$g_{\mu\nu}^{(k)}[g_{\mu\nu}]$ and $g_{\mu\nu}^{(k)}[g_{\mu\nu}^{(N)}]$,
are identical for any $k\le N$ and, thus, all the coefficients (\ref{eq:Omega_C^gamma})
of the curvature expansion computed from their respective polynomial
approximations are identical so long as the number of covariant derivatives
in $g_{\mu\nu}^{(N)}$ is at least $\delta+D-6$ (i.e., if $N-2\ge\delta+D-6$).
Since their respective curvature expansions (\ref{eq:Omega_C curvature exp})
are thus identical for $N\ge\delta+D-4$, this then implies the claimed
relation (\ref{eq:Omega_C[g]=Omega_C[g^N]+smooth})
holds for the special case that $g_{\mu\nu}=g_{\mu\nu}^{(P)}$. To
extend the proof of relation (\ref{eq:Omega_C[g]=Omega_C[g^N]+smooth})
to arbitrary smooth metrics $g_{\mu\nu}$, we use the fact \cite[see proof of Theorem 4.1]{HW_Existence_TOP}
that it is always possible to define a 1-parameter family of metrics
$h_{\mu\nu}(x;p)$ which depend smoothly on $p$ in a neighborhood
of $p=0$ and such that: i) For any fixed $p\ne0$, $h_{\mu\nu}(x;p)$
is a polynomial metric of finite order and ii) $h_{\mu\nu}(x;p=0)=g_{\mu\nu}$.
The proposition has already been established for $h_{\mu\nu}(x;p)$
when $p\ne0$ since these are polynomial metrics, so compatibility
with the smoothness axiom $\spectrum$ then implies the proposition
must hold also for $p=0$. 
\end{proof}
Our provisional proposal is to replace (\ref{eq:m^2 flow eq CS})
with 
\begin{equation}
\frac{\partial}{\partial m^{2}}H_{F}[g_{\mu\nu}](x_{1},x_{2})\sim_{\delta}\Omega_{C}[g_{\mu\nu}^{(N)},t_{\mu},L](x_{1},x_{2};\vec{0}),\label{eq:preliminary local fe}
\end{equation}
where $N\ge\delta+D-4$. The distribution $\Omega_{C}[g_{\mu\nu}^{(N)},t_{\mu},L]$
appearing on the right-hand side of (\ref{eq:preliminary local fe})
is manifestly local with respect to the original spacetime $(M,g_{ab})$,
since the only dependence of $g_{\mu\nu}^{(N)}$ on the spacetime
curvature comes through a finite number of local curvature tensors
evaluated at $z$ (see (\ref{eq:g^(N)_munu})). Thus, the flow equation
(\ref{eq:preliminary local fe}) is now local in the spacetime metric.
However, it fails to be covariant. We turn now to making a further
modification of (\ref{eq:preliminary local fe}) to restore covariance.

\paragraph*{(ii) Covariance.}

The distribution $\Omega_{C}[g_{\mu\nu}^{(N)},t_{\mu},L]$ appearing
in (\ref{eq:preliminary local fe}) fails to be covariant because
the cutoff function, $\chi$, depends upon a choice of the unit timelike
co-vector $t_{\mu}$ at $z$, which is not determined by the metric.
However, any two normalized timelike co-vectors $t_{\mu}$ and $t_{\mu}'$
at $z$ are related via a restricted Lorentz transformation $\Lambda\in\mathcal{L}_{+}^{\uparrow}$,
\[
t_{\mu}'=\Lambda_{\ \mu}^{\nu}t_{\nu}=(\Lambda^{-1})_{\mu}^{\hphantom{\mu}\nu}t_{\nu}\equiv(\Lambda^{-1}t)_{\mu}.
\]
Thus, in order to obtain a covariant flow relation, we seek to modify
the flow relations by the addition of smooth locally-constructed ``counterterms''
that compensate for the failure of $\Omega_{C}[g_{\mu\nu}^{(N)},t_{\mu},L]$
to be invariant under Lorentz transformations of $t_{\mu}$.

The dependence of $\Omega_{C}[g_{\mu\nu}^{(N)},t_{\mu},L]$ on Lorentz
transformations of $t_{\mu}$ is quantified by the distribution, 
\begin{align}
 & Q_{C}[g_{\mu\nu}^{(N)},t_{\mu},L](x_{1},x_{2};\vec{0};\Lambda^{-1})\label{eq:Q_C}\\
 & \equiv\Omega_{C}[g_{\mu\nu}^{(N)},(\Lambda^{-1}t)_{\mu},L](x_{1},x_{2};\vec{0})-\Omega_{C}[g_{\mu\nu}^{(N)},t{}_{\mu},L](x_{1},x_{2};\vec{0})\nonumber \\
 & =-i\int d^{D}y\sqrt{-g^{(N)}(y)}\left[\chi[g_{\mu\nu}^{(N)},t{}_{\mu},L](\Lambda y;\vec{0})-\chi[g_{\mu\nu}^{(N)},t_{\mu},L](y;\vec{0})\right]H_{F}[g{}_{\mu\nu}^{(N)}](y,x_{1})H_{F}[g{}_{\mu\nu}^{(N)}](y,x_{2}).\nonumber 
\end{align}
By the same arguments as given for quantity $Q_{M}$ in Minkowski
spacetime (see eq.~\eqref{eq:def Q func}), the quantity $Q_{C}$
has smooth dependence on $(x_{1},x_{2})$. In Minkowski spacetime,
the Taylor coefficients, $\boldsymbol{Q}_{M}(\Lambda^{-1})\equiv\partial_{\gamma_{1}}^{(x_{1})}\partial_{\gamma_{2}}^{(x_{2})}Q_{M}(x_{1},x_{2};z;\Lambda^{-1})|_{x_{1},x_{2}=z}$,
of $Q_{M}(x_{1},x_{2};z;\Lambda^{-1})$ were shown to satisfy (\ref{eq:Q(Lambda1Lambda2)-Q(Lambda1)}).
The existence of the desired counterterms in the flow relations was
then established by cohomological arguments. However, in curved spacetime,
the Taylor coefficients of $Q_{C}$ do not satisfy (\ref{eq:Q(Lambda1Lambda2)-Q(Lambda1)})
for the simple reason that the curved metric $g_{\mu\nu}^{(N)}$ given
by (\ref{eq:g^(N)_munu}) is not invariant under Lorentz transformations.

Nevertheless, we can use the curvature expansion (\ref{eq:Omega_C curvature exp})
for $Q_{C}[g_{\mu\nu}^{(N)},t_{\mu},L]$ and consider the behavior
under Lorentz transformations of the coefficients $(\Omega_{\vec{p}}){}^{\{\mu\cdots\sigma_{k-2}\}}[\eta_{\mu\nu},t_{\mu},L]$
appearing in that expansion (see (\ref{eq:Omega_C^gamma})). We write
\begin{equation}
(Q_{\vec{p}}){}^{\gamma}[\eta_{\mu\nu},t_{\mu},L,\Lambda^{-1}]\equiv(\Omega_{\vec{p}}){}^{\gamma}[\eta_{\mu\nu},(\Lambda^{-1}t){}_{\mu},L]-(\Omega_{\vec{p}}){}^{\gamma}[\eta_{\mu\nu},t_{\mu},L],\label{eq:Q^gamma}
\end{equation}
where we use the multi-index notation $\gamma\equiv\{\mu\cdots\sigma_{k-2}\}$.
For notational convenience, we will suppress the $p$-subscripts in
the following and write the left side of (\ref{eq:Q^gamma}) simply
as $Q^{\gamma}$. Since $Q^{\gamma}$ is smooth, its asymptotic behavior
is determined by its Taylor coefficients, 
\begin{equation}
Q_{\ \gamma_{1}\gamma_{2}}^{\gamma}(\Lambda^{-1})\equiv\partial_{\gamma_{1}}^{(x_{1})}\partial_{\gamma_{2}}^{(x_{2})}Q^{\gamma}(x_{1},x_{2};\vec{0};\Lambda^{-1})|_{x_{1},x_{2}=\vec{0}},\label{eq:Q^gamma3_(gamma1 gamma2)}
\end{equation}
The crucial point is that the Taylor coefficients (\ref{eq:Q^gamma3_(gamma1 gamma2)})
depend only on $\eta_{\mu\nu}$, not the spacetime metric $g_{\mu\nu}^{(N)}$\textemdash all
of the dependence on the spacetime metric in the curvature expansion
(\ref{eq:Omega_C curvature exp}) appears in the curvature factors,
not in $(\Omega_{\vec{p}})$. Consequently, we obtain, 
\begin{align}
 & Q_{\hphantom{\gamma_{3}}\gamma_{1}\gamma_{2}}^{\gamma_{3}}(\Lambda_{1}\Lambda_{2})-Q_{\hphantom{\gamma_{3}}\gamma_{1}\gamma_{2}}^{\gamma_{3}}(\Lambda_{1})\nonumber \\
 & =\left[\partial_{\gamma_{1}}^{(x_{1})}\partial_{\gamma_{2}}^{(x_{2})}\left[\Omega^{\gamma_{3}}[(\Lambda_{1}\Lambda_{2}t)_{\mu}](x_{1},x_{2};\vec{0})-\Omega^{\gamma_{3}}[(\Lambda_{1}t)_{\mu}](x_{1},x_{2};\vec{0})\right]\right]_{x_{1},x_{2}=\vec{0}}\nonumber \\
 & =\Lambda_{\hphantom{\gamma_{3}'}\gamma_{3}'}^{\gamma_{3}}\left[\partial_{\gamma_{1}}^{(x_{1})}\partial_{\gamma_{2}}^{(x_{2})}\left[\Omega^{\gamma_{3}'}[(\Lambda_{2}t)_{\mu}](\Lambda_{1}^{-1}x_{1},\Lambda_{1}^{-1}x_{2};\vec{0})-\Omega^{\gamma_{3}'}[t{}_{\mu}](\Lambda_{1}^{-1}x_{1},\Lambda_{1}^{-1}x_{2};\vec{0})\right]\right]_{x_{1},x_{2}=\vec{0}}\nonumber \\
 & =(\Lambda_{1}){}_{\hphantom{\gamma_{3}}\gamma_{3}'}^{\gamma_{3}}(\Lambda_{1}){}_{\gamma_{1}}^{\hphantom{\gamma_{1}}\gamma_{1}'}(\Lambda_{1}){}_{\gamma_{2}}^{\hphantom{\gamma_{2}}\gamma_{2}'}Q_{\hphantom{\gamma_{3}'}\gamma_{1}'\gamma_{2}'}^{\gamma_{3}'}(\Lambda_{2}),\label{eq:coho id Q_C with indices}
\end{align}
where the second equility follows from the identity: 
\[
\text{\ensuremath{\Omega}}^{\gamma}[g_{\mu\nu}^{(N)},(\Lambda^{-1}t)_{\mu},L](x_{1},x_{2};\vec{0})=(\Lambda^{-1})_{\hphantom{\gamma_{3}}\gamma'}^{\gamma}\Omega^{\gamma'}[g_{\mu\nu}^{(N)},t_{\mu},L](\Lambda x_{1},\Lambda x_{2};\vec{0}),
\]
where we have used the fact that $H_{F}[(\Lambda g)_{\mu\nu}^{(N)}](\Lambda y,\Lambda x)=H_{F}[g_{\mu\nu}^{(N)}](y,x)$,
with 
\[
(\Lambda g)_{\mu\nu}^{(N)}(x)\equiv\Lambda_{\mu_{1}}^{\hphantom{\mu_{1}}\nu_{1}}\Lambda_{\mu_{2}}^{\hphantom{\mu_{2}}\nu_{2}}g_{\nu_{1}\nu_{2}}^{(N)}(\Lambda^{-1}x).
\]
Equation (\ref{eq:coho id Q_C with indices}) is a close analogue
of the equation (\ref{eq:Q(Lambda1Lambda2)-Q(Lambda1)}). Writing
$\boldsymbol{Q}_{C}(\Lambda)\equiv(Q_{\vec{p}}){}_{\hphantom{\gamma_{3}}\gamma_{1}\gamma_{2}}^{\gamma_{3}}(\Lambda)$,
we see that (\ref{eq:coho id Q_C with indices}) corresponds to the
cohomological identity 
\begin{equation}
\boldsymbol{Q}_{C}(\Lambda_{1})+D(\Lambda_{1})\boldsymbol{Q}_{C}(\Lambda_{2})-\boldsymbol{Q}_{C}(\Lambda_{1}\Lambda_{2})=0,\label{eq:Q_C coho id}
\end{equation}
see \eqref{eq:d^1Q=0}. By the same arguments as given in Proposition
\ref{Prop: existence of a} of Appendix \ref{sec:K-G Lorentz cts},
it follows that there exist tensors $\boldsymbol{a}\equiv(a_{\vec{p}})^{\gamma_{3}}{}_{\gamma_{1}\gamma_{2}}$
such that 
\begin{equation}
\boldsymbol{Q}_{C}(\Lambda)=\left(D(\Lambda)-\mathbb{I}\right)\boldsymbol{a},\label{eq:Q_C=(Lambda-I)b}
\end{equation}

We now can restore covariance to the curved spacetime flow equations
in close parallel with the procedure we used to restore Lorentz covariance
to the Minkowski flow equations. Let $\boldsymbol{a}\equiv(a_{\vec{p}})^{\gamma_{3}}{}_{\gamma_{1}\gamma_{2}}$
denote the solutions to (\ref{eq:Q_C=(Lambda-I)b}). Let 
\[
a_{\gamma_{1}\gamma_{2}}\equiv\sum_{k=0}^{\delta+D-4}\sum_{\vec{p}}\sum_{\gamma_{3}}(R^{\vec{p}})_{\gamma_{3}}(\vec{0})(a_{\vec{p}})^{\gamma_{3}}{}_{\gamma_{1}\gamma_{2}},
\]
with the $\vec{p}=(p_{0},p_{1},\dots,p_{k-2})$ sum running over (\ref{eq:p-sum bounds}).
Here we have abbreviated the product of curvature tensors appearing
in the curvature expansion by writing 
\begin{equation}
(R^{\vec{p}}){}_{\{\mu\nu\kappa\rho;(\sigma_{1}\cdots\sigma_{j})\}}(\vec{0})\equiv\prod_{j=0}^{k-2}\left[R_{\mu\nu\kappa\rho;(\sigma_{1}\cdots\sigma_{j})}(\vec{0})\right]^{p_{j}}.\label{eq:R^p}
\end{equation}
Now replace $\Omega_{C}[g_{\mu\nu}^{(N)},t_{\mu},L](x_{1},x_{2};\vec{0})$
in (\ref{eq:preliminary local fe}) with 
\begin{align}
 & \widetilde{\Omega}_{C}[g_{\mu\nu}^{(N)},t_{\mu},L](x_{1},x_{2};\vec{0})\label{eq:Omega_C tilde}\\
 & \equiv\Omega_{C}[g_{\mu\nu}^{(N)},t_{\mu},L](x_{1},x_{2};\vec{0})-\sum_{|\gamma_{1}|+|\gamma_{2}|\le\delta}\frac{1}{\gamma_{1}!\gamma_{2}!}a{}_{\gamma_{1}\gamma_{2}}[g_{\mu\nu}^{(N)},t_{\mu},L]x_{1}^{(\gamma_{1}}x_{2}^{\gamma_{2})},\nonumber 
\end{align}
Then, to scaling degree $\delta$, $\widetilde{\Omega}_{C}$ is independent
of the choice of unit-normalized timelike $t_{\mu}$ and differs from
$\Omega_{C}$ by a smooth function of $(x_{1},x_{2})$ with the same
scaling behavior as $\Omega_{C}$. Thus, the flow relation 
\begin{equation}
\frac{\partial}{\partial m^{2}}H_{F}[g_{\mu\nu}](x_{1},x_{2})\sim_{\delta}\widetilde{\Omega}_{C}[g_{\mu\nu}^{(N)},t_{\mu},L](x_{1},x_{2};\vec{0}),\label{eq:preliminary local fe cov}
\end{equation}
is both local and covariant in the metric. However, it does not have
the required scaling behavior, so we will make a further modification
to this flow relation in the ``scaling'' paragraph below.

Finally, we note that we can obtain a recursive formula for $\boldsymbol{a}$
by the same procedure as in the Minkowski case discussed in Appendix
\ref{sec:K-G Lorentz cts}. Define 
\begin{equation}
\boldsymbol{B}^{\kappa\rho}\equiv(B^{\kappa\rho})_{\hphantom{\gamma_{3}}\gamma_{1}\gamma_{2}}^{\gamma_{3}}\equiv-2\Omega{}_{\hphantom{\gamma_{3}}\gamma_{1}\gamma_{2}}^{\gamma_{3}}[\chi=y^{[\kappa}\partial^{\rho]}\chi],\label{eq:B^kappa rho CS}
\end{equation}
with $\Omega{}_{\hphantom{\gamma_{3}}\gamma_{1}\gamma_{2}}^{\gamma_{3}}[\chi=y^{[\kappa}\partial^{\rho]}\chi]$
denoting the Taylor coefficients of the smooth function, 
\begin{equation}
\Omega{}_{\hphantom{\gamma_{3}}}^{\gamma_{3}}[\chi=y^{[\kappa}\partial^{\rho]}\chi](x_{1},x_{2};\vec{0}).\label{eq:Omega^gamma3_(gamma1 gamma2)(chi^sigma2_sigma1)}
\end{equation}
(Smoothness of (\ref{eq:Omega^gamma3_(gamma1 gamma2)(chi^sigma2_sigma1)})
in $(x_{1},x_{2})$ is guaranteed by the fact that $\partial_{(y)}^{\rho}\chi(y,\vec{0})$
vanishes in a neighborhood of $y=\vec{0}$.) Then for any infinitesimal
restricted Lorentz transformation $\Lambda_{\theta}$, we have 
\begin{equation}
\boldsymbol{Q}_{C}(\Lambda_{\theta})=-\theta_{\kappa\rho}\boldsymbol{B}^{\kappa\rho}+\mathcal{O}(\theta^{2}).\label{eq:Q_C(Lambda_theta)}
\end{equation}
The analysis of Appendix \ref{sec:K-G Lorentz cts} then implies that
\begin{equation}
\boldsymbol{a}=\sum_{\substack{j=1\\
\widetilde{c}_{j}\ne0
}
}^{k}\frac{1}{\widetilde{c}_{j}}\mathbb{E}_{j}\left(-\mathbb{L}_{\kappa\rho}\boldsymbol{B}^{\kappa\rho}+4\sum_{i<j\le n}\eta_{\mu_{i}\mu_{j}}\text{tr}_{ij}\boldsymbol{a}\right),\label{eq:ind soln for b}
\end{equation}
with the notation defined in Appendix \ref{sec:K-G Lorentz cts},
where we have lowered all indices on the tensors so that all tensors
in (\ref{eq:ind soln for b}) are of type $(0,|\gamma_{1}|+|\gamma_{2}|+|\gamma_{3}|)$.
As explained in Appendix \ref{sec:K-G Lorentz cts}, equation (\ref{eq:ind soln for b})
determines higher rank coefficients $(a_{\vec{p}})^{\gamma_{3}}{}_{\gamma_{1}\gamma_{2}}$
inductively in terms of the equations for the lowest nontrivial ranks
with a given symmetry. When $\vec{p}=\vec{0}$, the coefficients $(a_{\vec{p}})^{\gamma_{3}}{}_{\gamma_{1}\gamma_{2}}$
coincide with those appearing in the Minkowski flow relations: i.e.,
$(a_{\vec{0}})^{\gamma_{3}}{}_{\gamma_{1}\gamma_{2}}=a_{\gamma_{1}\gamma_{2}}$,
whose rank $r\equiv|\gamma_{1}|+|\gamma_{2}|=0,1,2$ cases were stated
explicitly in the appendix. When $\vec{p}\ne\vec{0}$, the explicit
lower-rank cases can be straightforwardly obtained using the methods
of the appendix. For this purpose, it is worth noting the $(a_{\vec{p}})^{\gamma_{3}}{}_{\gamma_{1}\gamma_{2}}$
coefficients have the same symmetries as the Minkowski coefficients
$a_{\gamma_{1}\gamma_{2}}$ in the lower multi-indices (and their
respective spacetime indices $\gamma_{1}=\{\mu_{1}\cdots\mu_{p}\}$,
$\gamma_{2}=\{\nu_{1}\cdots\nu_{q}\}$). However, the symmetries of
the upper spacetime indices in $(a_{\vec{p}})^{\gamma_{3}}{}_{\gamma_{1}\gamma_{2}}$
are dictated by the curvature tensors \eqref{eq:R^p}.

Under the rescaling
\begin{equation}
(\eta_{\mu\nu},d^{D}y,m^{2},L)\to(\lambda^{-2}\eta_{\mu\nu},\lambda^{-D}d^{D}y,\lambda^{2}m^{2},\lambda^{-1}L),\label{eq:rescaling in tangent space}
\end{equation}
the inductive solutions (\ref{eq:ind soln for b}) for $(a_{\vec{p}_{k}}){}_{\hphantom{\gamma_{3}}\gamma_{1}\gamma_{2}}^{\gamma_{3}}$
will scale in the manner required for $\widetilde{\Omega}_{C}$ to
have the same scaling behavior as $\Omega_{C}$.

\paragraph*{(iii) Scaling}

The flow equation (\ref{eq:preliminary local fe cov}) is local and
covariant and scales almost homogeneously with the correct power of
$\lambda$ under $(g_{\mu\nu},m^{2},L)\to(\lambda^{-2}g_{\mu\nu},\lambda^{2}m^{2},\lambda^{-1}L)$.
However, on account of the nontrivial $L$ dependence, we do not have
the required almost homogeneous scaling under $(g_{\mu\nu},m^{2})\to(\lambda^{-2}g_{\mu\nu},\lambda^{2}m^{2})$.
This is the same difficulty as occurred in the Euclidean and Minkowski
cases, and it can be overcome by further modifying the flow relation in the same manner as for those cases. Specifically, we replace the flow
relation (\ref{eq:preliminary local fe cov}) with 
\begin{equation}
\frac{\partial}{\partial m^{2}}H_{F}[g_{\mu\nu}](x_{1},x_{2})\sim_{\delta}\mathfrak{L}[L]\widetilde{\Omega}_{C}(x_{1},x_{2};\vec{0};L)-\sum_{|\gamma_{1}|+|\gamma_{2}|\le\delta}\frac{1}{\gamma_{1}!\gamma_{2}!}c_{\gamma_{1}\gamma_{2}}(L)x_{1}^{(\gamma_{1}}x_{2}^{\gamma_{2})},\label{eq:frakL Omega_C tilde-cts}
\end{equation}
where $\mathfrak{L}$ was defined by (\ref{eq:full L op}) and where,
for $L>0$, 
\begin{align}
c_{\gamma_{1}\gamma_{2}}(L) & \equiv\sum_{k}\sum_{\vec{p}_{k}}(R^{\vec{p}}){}_{\gamma_{3}}(\vec{0})\left[\vphantom{\int_{0}^{L}}\mathfrak{L}[L](a_{\vec{p}})_{\hphantom{\gamma_{3}}\gamma_{1}\gamma_{2}}^{\gamma_{3}}(L)\right.+\label{eq:c_gamma1 gamma2 CS}\\
 & +\left.\int_{0}^{L}dL'\left[\Omega{}_{\hphantom{\gamma_{3}}\gamma_{1}\gamma_{2}}^{\gamma_{3}}\left[\chi=\partial_{L'}\!\left(\mathfrak{L}[L']\chi\right)\right]-\frac{\partial}{\partial L'}\left(\mathfrak{L}[L'](a_{\vec{p}}){}_{\hphantom{\gamma_{3}}\gamma_{1}\gamma_{2}}^{\gamma_{3}}(L')\right)\right]\right]\nonumber 
\end{align}
The flow relation (\ref{eq:frakL Omega_C tilde-cts}) is local and
covariant and has the proper scaling under $(g_{\mu\nu},m^{2})\to(\lambda^{-2}g_{\mu\nu},\lambda^{2}m^{2})$.

Using the definition (\ref{eq:Omega_C tilde}) of $\widetilde{\Omega}_{C}$
and the relations (\ref{eq:C_phi^2 phi phi ^I})-(\ref{eq:C_phi phi ^C}),
we may rewrite (\ref{eq:frakL Omega_C tilde-cts}) in terms of the
OPE coefficients: 
\begin{align}
\frac{\partial}{\partial m^{2}}(C_{H}){}_{T_{0}\{\phi\phi\}}^{I}(x_{1},x_{2};\vec{0})\sim_{\delta} & -\frac{i}{2}\int\sqrt{-g^{(N)}(y)}\mathfrak{L}[L]\chi(y,\vec{0};L)(C_{H}){}_{T_{0}\{\phi^{2}\phi\phi\}}^{I}(y,x_{1},x_{2};\vec{0})+\nonumber \\
 & -\sum_{[C]\le\delta+2}c_{C}[g_{\mu\nu}^{(N)},t_{\mu},L](C_{H})_{T_{0}\{\phi\phi\}}^{C}(x_{1},x_{2};\vec{0}),\label{eq:2-pt id coef fe CS}
\end{align}
where $N\ge\delta+D-4$. Here we have 
\begin{align}
 & c_{C}\equiv\sum_{k}\sum_{\vec{p}_{k}}(R^{\vec{p}}){}_{\gamma}(\vec{0})\left[\mathfrak{L}[L](a_{\vec{p}})_{\hphantom{\gamma}C}^{\gamma}(L)-\int_{0}^{L}dL'\frac{\partial}{\partial L'}\left(\mathfrak{L}[L'](a_{\vec{p}}){}_{\hphantom{\gamma}C}^{\gamma}(L')\right)\right]+\label{eq:c_C CS}\\
 & -\frac{i}{2}\sum_{k}\sum_{\vec{p}_{k}}(R^{\vec{p}}){}_{\gamma}(\vec{0})\left[\frac{\partial^{P}}{(\partial^{\vec{p}}R)_{\gamma}(\vec{0})}\int_{0}^{L}dL'\int d^{D}y\sqrt{-g^{(k)}(y)}\frac{\partial}{\partial L'}\left(\mathfrak{L}[L']\chi(y,\vec{0};L')\right)C_{T_{0}\{\phi^{2}C\}}^{I}(y,\vec{0};\vec{0})\right]_{g_{\mu\nu}^{(k)}=\eta_{\mu\nu}},\nonumber 
\end{align}
with the $k$ and $\vec{p}_{k}$ sums taken as in the curvature expansion
(\ref{eq:Omega_C curvature exp}) and we have abbreviated, 
\[
(\partial^{\vec{p}}R)_{\{\mu\cdots\sigma_{k-2}\}}\equiv\partial^{p_{0}}R_{\mu\nu\kappa\rho}\partial^{p_{1}}R_{\mu\nu\kappa\rho;\sigma}\cdots\partial^{p_{k-2}}R_{\mu\nu\kappa\rho;(\sigma_{1}\cdots\sigma_{k-2})}.
\]
It is required that $(a_{\vec{p}})_{\hphantom{\gamma}C}^{\gamma}=0$
unless $[C]_{\phi}=2$. For $[C]_{\phi}=2$, the tensors $(a_{\vec{p}})_{\hphantom{\gamma}C}^{\gamma}$
are given via the inductive formula \eqref{eq:ind soln for b} with,
\[
(B_{\vec{p}_{k}}^{\kappa\rho})_{\hphantom{\gamma}C}^{\gamma}\equiv i\left[\frac{\partial^{P}}{(\partial^{\vec{p}}R)_{\gamma}(\vec{0})}\int d^{D}y\sqrt{-g^{(k)}(y)}y^{[\kappa}\partial^{\rho]}\chi(y,\vec{0};L)(C_{H}){}_{T_{0}\{\phi^{2}C\}}^{I}(y,\vec{0};\vec{0})\right]_{g_{\mu\nu}^{(k)}=\eta_{\mu\nu}}.
\]
Formula (\ref{eq:ind soln for b}) determines $(a_{\vec{p}})_{\hphantom{\gamma}C}^{\gamma}$
up to Lorentz-invariant tensors that depend smoothly on $(\eta_{\mu\nu},m^{2},\xi)$
and that scale with the same overall factor of $\lambda$ as $\Omega_{C}^{\gamma}$
under the rescaling \eqref{eq:rescaling in tangent space}. As in
the Minkowski case, the $\mathfrak{L}$-operator and $L$-integral
terms in \eqref{eq:c_C CS} ensure that only the $L$-independent
terms in $(a_{\vec{p}})_{\hphantom{\gamma}C}^{\gamma}$ can contribute
to $c_{C}$. Therefore, the only ambiguity in $c_{C}$ corresponds
to the choice of an $L$-independent Lorentz-invariant tensor in $(a_{\vec{p}})_{\hphantom{\gamma}C}^{\gamma}$
that scales with the same overall factor of $\lambda$ as $\Omega_{C}^{\gamma}$
under $(\eta_{\mu\nu},m^{2})\to(\lambda^{-2}\eta_{\mu\nu},\lambda^{2}m^{2})$.
In odd dimensions, there are no tensors that scale in this way and
depend smoothly on $(\eta_{\mu\nu},m^{2})$ so $(a_{\vec{p}})_{\hphantom{\gamma}C}^{\gamma}$
is unique and, thus, $c_{C}$ has no ambiguities. In even dimensions,
$(a_{\vec{p}})_{\hphantom{\gamma}C}^{\gamma}$ is not unique and this
yields the freedom in $c_{C}$ to choose a local and covariant smooth
function in $(x_{1},x_{2},g_{ab}^{(N)},m^{2},\xi)$. In curved spacetime,
compatibility with the Leibniz axiom $\Leibniz$ places additional
constraints on the allowed choices of $c_{C}$ and, in even dimensional
curved spacetimes with $D>2$, there is an additional constraint coming
from the conservation axiom $\conservation$. These constraints can
always be (non-uniquely) satisfied and, for $c_{C}$ satisfying these
conditions, the remaining ambiguities in (\ref{eq:2-pt id coef fe CS})
are in 1-1 correspondence with the freedom to choose a Hadamard parametrix
whose corresponding Hadamard normal-ordered Wick fields are compatible
with axioms W1-W8.

By the same reasoning that led to Theorems \ref{thm:Eucl fe for H-NO id coef}
and \ref{thm:Minkowski flow rel}, the flow relation eq.(\ref{eq:2-pt id coef fe CS})
together with the explicit formulas for the unextended time-ordered
OPE coefficients of the Hadamard normal-ordered Wick fields imply
flow relations for $(C_{H})_{T_{0}\{\phi\cdots\phi\}}^{I}$, as expressed
by the following theorem: 
\begin{thm}
\label{thm:flow rel curved spacetime} For any construction of the
Wick monomials by Hadamard normal ordering, we have 
\begin{align}
\frac{\partial}{\partial m^{2}}(C_{H}){}_{T_{0}\{\phi\cdots\phi\}}^{I}(x_{1},\dots,x_{n};\vec{0})\approx & -\frac{i}{2}\int d^{D}y\sqrt{-g^{(N)}(y)}\mathfrak{L}[L]\chi(y,\vec{0};L)(C_{H})_{T_{0}\{\phi^{2}\phi\cdots\phi\}}^{I}(y,x_{1},\dots,x_{n};\vec{0})+\nonumber \\
 & -\sum_{C}c_{C}[g_{\mu\nu}^{(N)},t_{\mu},L](C_{H})_{T_{0}\{\phi\cdots\phi\}}^{C}(x_{1},\dots,x_{n};\vec{0}),\label{eq:gen id coef fe CS}
\end{align}
with $c_{C}$ defined in (\ref{eq:c_C CS}). 
\end{thm}
As was the case in the flat spacetime case, the ambiguities in these
flow relations are in 1-1 correspondence with the freedom to choose
$H$. The flow relations for general prescriptions for the Wick fields
may straightforwardly obtained from \eqref{eq:gen id coef fe CS}
in the manner discussed below \prettyref{thm:Eucl fe for H-NO id coef}. 
\begin{rem}
\label{rem:xi flow rel}The derivation of $L$-independent local and
covariant flow relations with respect to the curvature-coupling parameter
$\xi$ proceeds essentially identically as the one presented here
for $m^{2}$. The $\xi$ flow relations are of the same form as (\ref{eq:gen id coef fe CS})
with the substitutions $m^{2}\to\xi$ and $d^{D}y\to d^{D}yR(y)$.
Of course, locality requires the Ricci scalar curvature, $R$, must
be computed using the polynomial metric $g_{\mu\nu}^{(N)}$ rather
than $g_{\mu\nu}$. Note also $\xi$ is dimensionless and the Ricci
scalar curvature scales as $R[\lambda^{-2}g_{\mu\nu}^{(N)}]=\lambda^{2}R[g_{\mu\nu}^{(N)}]$
so the $\xi$ flow relations scale with an overall extra power of
$\lambda^{2}$ relative to the $m^{2}$ flow relations. 
\end{rem}

\paragraph{Acknowledgments}

We wish to thank Stefan Hollands for many helpful discussions and
for suggesting the use of cohomological arguments to obtain counterterms
for restoring Lorentz covariance. This research was supported in part
by NSF grant PHY-2105878 to the University of Chicago.  M.\:\nolinebreak[4]G.\;\nolinebreak[4]Klehfoth acknowledges support from the National Science Foundation Graduate Research Fellowship under Grant Nos. DGE-1144082 and DGE-1746045.

\appendix

\section{Existence of Hadamard parametrix satisfying the conservation constraint
\label{sec:existence of Q }}

In this Appendix, we prove that there exists $Q(x_{1},x_{2})$ satisfying
(\ref{eq:nabla K H equals -nabla K Q}) for $D>2$. Abbreviate $Q_{0}(y)\equiv Q(y,y)$
and 
\[
Q_{ab}(y)\equiv\left[\nabla_{a}^{(x_{1})}\nabla_{b}^{(x_{2})}Q(x_{1},x_{2})\right]_{x_{1},x_{2}=y}.
\]
It is straightforward to show that 
\[
\left[\nabla_{b}^{(x_{1})}K_{x_{2}}Q(x_{1},x_{2})\right]{}_{x_{1},x_{2}=y}=-\nabla^{a}Q_{ba}+\frac{1}{2}\nabla_{b}Q_{\hphantom{a}a}^{a}+\frac{1}{2}(m^{2}+\xi R)\nabla_{b}Q_{0},
\]
with $Q_{\hphantom{a}a}^{a}\equiv g^{ab}Q_{ab}$. Hence, the conservation
condition (\ref{eq:nabla K H equals -nabla K Q}) is equivalent to:
\begin{equation}
-\nabla^{a}Q_{ba}+\frac{1}{2}\nabla_{b}Q_{\hphantom{a}a}^{a}+\frac{1}{2}(m^{2}+\xi R)\nabla_{b}Q_{0}=-\frac{D}{2(D+2)}\nabla_{b}^{(y)}\left[K_{x_{2}}H(x_{1},x_{2})\right]_{x_{1},x_{2}=y},\label{eq:conservation constraint on Q}
\end{equation}
where we have used (\ref{eq:nabla K H}). eq.~(\ref{eq:conservation constraint on Q})
is solved (non-uniquely) for $D>2$ by setting $Q_{0}(y)=0$ and 
\[
Q_{ab}(y)=-\frac{D}{D^{2}-4}g_{ab}\left[K_{x_{2}}H(x_{1},x_{2})\right]_{x_{1},x_{2}=y}.
\]

To see that there exists a smooth function $Q(x_{1},x_{2})$ with
these properties, we first note one can always obtain a smooth function
$f(x_{1},x_{2};y)$ with arbitrarily-specified covariant derivatives
evaluated at $x_{1},x_{2}=y$, by the construction described in the
proof of Proposition \ref{prop:Wick uniqueness Fn}. Thus, we may
arrange that $f(y,y;y)=0$ and 
\[
\nabla_{a}^{(x_{1})}\nabla_{b}^{(x_{2})}f(x_{1},x_{2};y)|_{x_{1},x_{2}=y}=-\frac{D}{D^{2}-4}g_{ab}\left[K_{x_{2}}H(x_{1},x_{2})\right]_{x_{1},x_{2}=y},
\]
while requiring $f$ and its derivatives at $x_{1},x_{2}=y$ to depend
smoothly on $(m^{2},\xi)$ and scale almost homogeneously. Moreover,
this construction implies 
\[
\nabla_{\alpha_{1}}^{(x_{1})}\nabla_{\alpha_{2}}^{(x_{2})}\nabla_{\beta}^{(y)}f(x_{1},x_{2};y)|_{x_{1},x_{2}=y}=0,\qquad\text{for all \ensuremath{|\beta|>0}},
\]
and, thus, the ``germ'' of $f$ at $x_{1},x_{2}=y$ is independent
of $y$. Hence, we may construct a $y$-independent smooth bi-variate
$Q$ satisfying (\ref{eq:nabla K H equals -nabla K Q}) and (\ref{eq:Q scaling})
which depends symmetrically on $(x_{1},x_{2})$ via: 
\[
Q(x_{1},x_{2})=\frac{1}{2}f(x_{1},x_{2};x_{1})+\frac{1}{2}f(x_{1},x_{2};x_{2}).
\]

\section{\label{sec:proofs for gen Wick coef subsec}Proofs for Subsection
\ref{subsec:gen Wick OPE coef}}

We collect here proofs to the theorem and propositions contained in
Subsection \ref{subsec:gen Wick OPE coef} 
\begin{proof}[\textbf{Sketch of proof for Theorem \ref{thm:existence Wick coef and associativity}}.]
The manipulations leading to (\ref{eq:gen Wick OPE coefficients})
establish OPEs are preserved under field redefinitions, so the existence
of the OPE for general Wick prescriptions follows from the existence
of an OPE for Hadamard normal-ordered Wick fields (see Theorem \ref{thm: Had NO OPE coef}
of Subsection \ref{subsec:gen Wick OPE coef}). Moreover, the scaling
degree of the OPE coefficients are unaffected by the field redefinitions.
We now argue the associativity conditions are also preserved under
field redefinitions. For notational simplicity, we give the argument
for an OPE involving three spacetime points with the merger tree $\mathcal{T}$
corresponding to $x_{1}$ and $x_{2}$ approaching each other faster
than $x_{3}$. The argument can then be straightforwardly generalized
to $n$-point OPEs with arbitrary merger trees. From (\ref{eq:gen Wick OPE coefficients}),
we have, 
\begin{align}
 & C_{A_{1}A_{2}A_{3}}^{B}(x_{1},x_{2},x_{3};z)\label{eq:C A1 A2 A3 B assoc 1}\\
 & \approx\sum_{C_{0},\dots,C_{3}}\Z{C_{0}}B(z)\Zinv{A_{1}}{C_{1}}(x_{1})\Zinv{A_{2}}{C_{2}}(x_{2})\Zinv{A_{3}}{C_{3}}(x_{3})(C_{H}){}_{C_{1}C_{2}C_{3}}^{C_{0}}(x_{1},x_{2},x_{3};z)\nonumber 
\end{align}
The associativity condition for Hadamard normal-ordered OPE coefficients
implies the coefficient in the second line can be expanded as 
\begin{align}
(C_{H}){}_{C_{1}C_{2}C_{3}}^{C_{0}}(x_{1},x_{2},x_{3};z) & \sim_{\mathcal{T},\delta}\sum_{D_{1}}(C_{H}){}_{C_{1}C_{2}}^{D_{1}}(x_{1},x_{2};z')(C_{H})_{D_{1}C_{3}}^{C_{0}}(z',x_{3};z)\label{eq:assoc proof intermed rel}\\
 & =\sum_{D_{1},D_{2}}\left[\sum_{E}\Z{D_{2}}E(z')\Zinv E{D_{1}}(z')\right](C_{H}){}_{C_{1}C_{2}}^{D_{2}}(x_{1},x_{2};z')(C_{H})_{D_{1}C_{3}}^{C_{0}}(z',x_{3};z),\nonumber 
\end{align}
where, in going to the second line, we have used the identity: 
\begin{equation}
\sum_{E}\Z{D_{2}}E(z')\Zinv E{D_{1}}(z')=\delta_{D_{2}}^{D_{1}},\label{eq:Kronecker id for assoc}
\end{equation}
Plugging (\ref{eq:Kronecker id for assoc}) back into (\ref{eq:assoc proof intermed rel})
and rearranging summations, we find then, 
\begin{align*}
C_{A_{1}A_{2}A_{3}}^{B}(x_{1},x_{2},x_{3};z)\sim_{\mathcal{T},\delta}\sum_{E} & \left[\sum_{D_{2},C_{1},C_{2}}\Z{D_{2}}E(z')\Zinv{A_{1}}{C_{1}}(x_{1})\Zinv{A_{2}}{C_{2}}(x_{2})(C_{H}){}_{C_{1}C_{2}}^{D_{2}}(x_{1},x_{2};z')\right]\times\\
\times & \left[\sum_{C_{0},D_{1},C_{3}}\Z{C_{0}}B(z)\Zinv E{D_{1}}(z')\Zinv{A_{3}}{C_{3}}(x_{3})(C_{H})_{D_{1}C_{3}}^{C_{0}}(z',x_{3};z)\right].
\end{align*}
By (\ref{eq:gen Wick OPE coefficients}), this is equivalent to, 
\[
C_{A_{1}A_{2}A_{3}}^{B}(x_{1},x_{2},x_{3};z)\sim_{\mathcal{T},\delta}\sum_{E}C{}_{A_{1}A_{2}}^{E}(x_{1},x_{2};z')C_{EA_{3}}^{B}(z_{1},x_{3};z).
\]
All other associativity conditions, including (\ref{eq:pertinent assoc cond for C_H}),
for general prescriptions of the Wick powers may simiarly be established
using the corresponding associativity conditions for Hadamard normal-ordered
OPE coefficients and the identity (\ref{eq:Kronecker id for assoc}). 
\end{proof}
\begin{proof}[\textbf{Sketch of proof for Proposition \ref{prop:gen Wick intermed coef}.}]
The proof makes use of the relationship (\ref{eq:gen Wick OPE coefficients})
between the general Wick OPE coefficients and the Hadamard normal-ordered
coefficients, the identity (\ref{eq:H intermed coef}) for the Hadamard
OPE coefficients established in Proposition \ref{Prop:intermed H-coef recursion rel},
and the recursion relation (\ref{eq:Z recursion rel-1}) satisfied
by the mixing matrix $\mathcal{Z}_{A}^{B}$. By (\ref{eq:gen Wick OPE coefficients})
and (\ref{eq:H intermed coef}), we have for any $p\le m\equiv[B]_{\phi}$,
\begin{align}
 & C_{A_{1}\cdots A_{n}}^{B}=\label{eq:parallel deriv for gen intermed coef}\\
 & \sum_{C_{1},\dots,C_{n}}(\mathcal{Z}^{-1})_{A_{1}}^{C_{1}}\cdots(\mathcal{Z}^{-1})_{A_{n}}^{C_{n}}\sum_{\{P_{1},P_{2}\}\in\mathcal{P}_{p}(\mathcal{S}_{C})}\sum_{k\ge m}\underset{{\textstyle \text{(a)}}}{\underbrace{\binom{k}{p}^{-1}\mathcal{Z}_{\gamma_{1}\cdots\gamma_{k}}^{B}(C_{H}){}_{C_{1}'\cdots C_{n}'}^{(\nabla_{\gamma_{1}}\phi\cdots\nabla_{\gamma_{p}}\phi)}(C_{H}){}_{C_{1}''\cdots C_{n}''}^{(\nabla_{\gamma_{(p+1)}}\phi\cdots\nabla_{\gamma_{k}}\phi)}}}\nonumber 
\end{align}
Now, inserting the recursion relation (\ref{eq:Z recursion rel-1})
for the mixing matrix, 
\begin{equation}
\mathcal{Z}_{\gamma_{1}\cdots\gamma_{k}}^{B}=\mathcal{Z}_{\gamma_{1}\cdots\gamma_{k}}^{\beta_{1}\cdots\beta_{m}}=\binom{k}{p}\binom{m}{p}^{-1}\delta_{(\gamma_{1}}^{\beta_{1}}\cdots\delta_{\gamma_{p}}^{\beta_{p}}\mathcal{Z}{}_{\gamma_{(p+1)}\cdots\gamma_{k})}^{\beta_{(p+1)}\cdots\beta_{m}},
\end{equation}
into the underbraced factor immediately yields: 
\begin{align}
{\textstyle \text{(a)}} & =\binom{m}{p}^{-1}(C_{H}){}_{C_{1}'\cdots C_{n}'}^{(\nabla_{\beta_{1}}\phi\cdots\nabla_{\beta_{p}}\phi)}\mathcal{Z}_{\gamma_{(p+1)}\cdots\gamma_{k}}^{\beta_{(p+1)}\cdots\beta_{m}}(C_{H}){}_{C_{1}''\cdots C_{n}''}^{(\nabla_{\gamma_{(p+1)}}\phi\cdots\nabla_{\gamma_{k}}\phi)}.
\end{align}
Plugging this back into (\ref{eq:parallel deriv for gen intermed coef})
gives, 
\begin{align}
 & C_{A_{1}\cdots A_{n}}^{B}=\label{eq:intermed coef deriv intermed rel}\\
 & \binom{m}{p}^{-1}\sum_{C_{0}}\mathcal{Z}_{C_{0}}^{\beta_{(p+1)}\cdots\beta_{m}}\sum_{C_{1},\dots,C_{n}}\underset{{\textstyle \text{(b)}}}{\underbrace{(\mathcal{Z}^{-1})_{A_{1}}^{C_{1}}\cdots(\mathcal{Z}^{-1})_{A_{n}}^{C_{n}}\sum_{\{P_{1},P_{2}\}\in\mathcal{P}_{p}(\mathcal{S}_{C})}(C_{H}){}_{C_{1}'\cdots C_{n}'}^{(\nabla_{\gamma_{1}}\phi\cdots\nabla_{\gamma_{p}}\phi)}(C_{H}){}_{C_{1}''\cdots C_{n}''}^{C_{0}}}}.\nonumber 
\end{align}

Recalling the definition of $\mathcal{P}_{p}(\mathcal{S})$ above
eq.~(\ref{eq:Phi_P1}), one may use the recursion relation (\ref{eq:Z recursion rel-1})
for the inverse mixing matrices in a similar manner (as with $\mathcal{Z}_{\gamma_{1}\cdots\gamma_{k}}^{B}$
in the ${\textstyle \text{(a)}}$-term above) to rewrite the underbraced
term: 
\begin{equation}
\text{(b)}=\sum_{\{P_{1},P_{2}\}\in\mathcal{P}_{p}(\mathcal{S}_{A})}\delta_{A_{1}'}^{C_{1}'}\cdots\delta_{A_{1}'}^{C_{n}'}(C_{H}){}_{C_{1}'\cdots C_{n}'}^{(\nabla_{\gamma_{1}}\phi\cdots\nabla_{\gamma_{p}}\phi)}(\mathcal{Z}^{-1})_{A_{1}''}^{C_{1}''}\cdots(\mathcal{Z}^{-1})_{A_{n}''}^{C_{n}''}(C_{H}){}_{C_{1}''\cdots C_{n}''}^{C_{0}},\label{eq:II}
\end{equation}
where we note that the sum in (\ref{eq:II}) is now taken over elements
of $\mathcal{P}_{p}(\mathcal{S}_{A})$ rather than $\mathcal{P}_{p}(\mathcal{S}_{C})$.
Finally, inserting this back into (\ref{eq:intermed coef deriv intermed rel})
yields, 
\begin{align*}
 & C_{A_{1}\cdots A_{n}}^{B}\\
 & =\binom{m}{p}^{-1}\sum_{\{P_{1},P_{2}\}\in\mathcal{P}_{p}(\mathcal{S}_{A})}(C_{H}){}_{A_{1}'\cdots A_{n}'}^{(\nabla_{\gamma_{1}}\phi\cdots\nabla_{\gamma_{p}}\phi)}\sum_{C_{0}'',C_{1}'',\dots,C_{n}''}(\mathcal{Z}^{-1})_{A_{1}''}^{C_{1}''}\cdots(\mathcal{Z}^{-1})_{A_{n}''}^{C_{n}''}\mathcal{Z}_{C_{0}}^{\beta_{(p+1)}\cdots\beta_{m}}(C_{H}){}_{C_{1}''\cdots C_{n}''}^{C_{0}},
\end{align*}
which, by eqs. (\ref{eq:gen Taylor coef}) and (\ref{eq:gen Wick OPE coefficients}),
is equivalent to formula (\ref{eq:H intermed coef}) with the $H$-subscripts
removed. 
\end{proof}
\begin{proof}[\textbf{Proof of Proposition \ref{prop:C^A1...An_B via point splitting}}]
The proof of (\ref{eq:C^A1...An_I point-splitting formula-1}) is
based on the associativity conditions (\ref{eq:pertinent assoc cond for C_H})
and the behavior of the Wick OPE coefficients $C_{\phi\cdots\phi}^{B}(x_{1},\dots,x_{n};z)$
on the total diagonal when $[B]_{\phi}=n$. As established in Theorem
\ref{thm:existence Wick coef and associativity}, the associativity
conditions hold for general prescriptions for the Wick powers. In
particular, for the class of merger trees $\mathcal{T}$ such that
$\vec{y}_{i}\to x_{i}$ at a faster rate than $x_{i}\to z$, we have,
c.f. formula (\ref{eq:pertinent assoc cond for C_H}), 
\begin{align}
 & C_{\phi\cdots\phi}^{B}(\vec{y}_{1},\dots,\vec{y}_{n};z)\sim_{\mathcal{T},\delta}\sum_{C_{1},\dots,C_{n}}C_{\phi\cdots\phi}^{C_{1}}(\vec{y}_{1};x_{1})\cdots C_{\phi\cdots\phi}^{C_{n}}(\vec{y}_{n};x_{n})C_{C_{1}\cdots C_{n}}^{B}(x_{1},\dots,x_{n};z),\label{eq:assoc C_phi...phi^B}
\end{align}
with the summations carried to a sufficiently high, but finite, order.
As we shall see, for our purposes, it is sufficient to include only
$[C_{i}]\le[A_{i}]$ for all $i$.

We note the OPE coefficients $C_{\phi\cdots\phi}^{C_{i}}(\vec{y}_{i};x_{i})$
vanish unless $[C_{i}]_{\phi}\le k_{i}\equiv[A_{i}]_{\phi}$ for all
$i$. It is useful to rearrange (\ref{eq:assoc C_phi...phi^B}), putting
all terms such that $[C_{i}]=k_{i}$ for all $i$ on one side: 
\begin{align}
 & \sum_{[C_{1}]_{\phi}=[A_{1}]_{\phi}}\cdots\sum_{[C_{n}]=[A_{n}]_{\phi}}C_{\phi\cdots\phi}^{C_{1}}(\vec{y}_{1};x_{1})\cdots C_{\phi\cdots\phi}^{C_{n}}(\vec{y}_{n};x_{n})C_{C_{1}\cdots C_{n}}^{B}(x_{1},\dots,x_{n};z)\label{eq:pt split deriv intermed rel}\\
 & \sim_{\mathcal{T},\delta}C_{\phi\cdots\phi}^{B}(\vec{y}_{1},\dots,\vec{y}_{n};z)-\sum_{[C_{1}]_{\phi}<[A_{1}]_{\phi}}\cdots\sum_{[C_{n}]_{\phi}<[A_{n}]_{\phi}}C_{\phi\cdots\phi}^{C_{1}}(\vec{y}_{1};x_{1})\cdots C_{\phi\cdots\phi}^{C_{n}}(\vec{y}_{n};x_{n})C_{C_{1}\cdots C_{n}}^{B}(x_{1},\dots,x_{n};z).\nonumber 
\end{align}
We now note the limiting behavior of the coefficients: 
\begin{equation}
\lim_{\vec{y}_{i}\to x_{i}}C_{\phi\cdots\phi}^{C_{i}}(\vec{y}_{i};x)=\begin{cases}
1 & [C_{i}]=[C_{i}]_{\phi}=k_{i}\\
0 & [C_{i}]>k_{i}
\end{cases},\label{eq:lim y->x C_phi...phi^C_i}
\end{equation}
The second case follows from the fact that $C_{\phi\cdots\phi}^{C_{i}}(\vec{y}_{i};x_{i})$
has negative scaling degree when $[C_{i}]>k_{i}$ by (\ref{eq:scaling degree OPE coef}).
The first case follows from the fact that, when $[C_{i}]_{\phi}=k_{i}$,
the $C_{\phi\cdots\phi}^{C_{i}}(\vec{y}_{i};x_{i})$ are given by
geometric factors (\ref{eq:gen Taylor coef}) and these factors satisfy:
\[
\lim_{y\to x}\T[\beta](y;x)=\begin{cases}
1 & |\text{\ensuremath{\beta}}|=0\\
0 & |\beta|>0
\end{cases},
\]
because $\lim_{y\to x}\nabla_{(y)}^{b}\sigma(y;x)=0$. Evaluating
the proposed limit of (\ref{eq:pt split deriv intermed rel}), using
(\ref{eq:lim y->x C_phi...phi^C_i}), we then find: 
\begin{align*}
C_{\phi^{k_{1}}\cdots\phi^{k_{n}}}^{B}(x_{1},\dots,x_{n};z) & =\lim_{\vec{y}_{1}\to x_{1}}\cdots\lim_{\vec{y}_{n}\to x_{n}}\left[C_{\phi\cdots\phi}^{B}(\vec{y}_{1},\dots,\vec{y}_{n};z)+\right.\\
 & -\sum_{\substack{[C_{1}]<[A_{1}]\\{}
[C_{1}]_{\phi}<[A_{1}]_{\phi}
}
}\cdots\sum_{\substack{[C_{n}]<[A_{n}]\\{}
[C_{n}]_{\phi}<[A_{n}]_{\phi}
}
}\left.C_{\phi\cdots\phi}^{C_{1}}(\vec{y}_{1};x_{1})\cdots C_{\phi\cdots\phi}^{C_{n}}(\vec{y}_{n};x_{n})C_{C_{1}\cdots C_{n}}^{B}(x_{1},\dots,x_{n};z)\right].
\end{align*}
This establishes formula (\ref{eq:C^A1...An_I point-splitting formula-1})
for the OPE coefficients involving products of Wick powers with no
derivatives. To obtain the general case, apply the derivative operator
$\nabla_{\alpha_{(1,1)}}^{y_{(1,1)}}\cdots\nabla_{\alpha_{(n,k_{n})}}^{y_{(n,k_{n})}}$
to both sides of relation (\ref{eq:pt split deriv intermed rel})
and take the limits $\vec{y}_{i}\to x_{i}$ for all $i$, using the
identity: 
\[
\lim_{\vec{y}_{i}\to x_{i}}\sum_{[C_{i}]_{\phi}=[A_{i}]_{\phi}}\nabla_{\alpha_{(i,1)}}^{(y_{i,1})}\cdots\nabla_{\alpha_{(i,k_{i})}}^{(y_{i,k_{i}})}C_{\phi\cdots\phi}^{C_{i}}(\vec{y}_{i};x_{i})C_{C_{1}\cdots C_{i}\cdots C_{n}}^{B}(x_{1},\dots,x_{n};z)=C_{C_{1}\cdots A_{i}\cdots C_{n}}^{B}(x_{1},\dots,x_{n};z),
\]
which follows, in turn, from the identity (\ref{eq:sum nabla_alpha T^beta = nabla_alpha-1})
for the covariant derivative acting on any scalar field.

Note, in our derivation, no assumption has been made about the rate
$x_{1},\dots,x_{n}$ approach each other in (\ref{eq:assoc C_phi...phi^B}),
so the resulting formula (\ref{eq:C^A1...An_I point-splitting formula-1})
is valid under arbitrary merger trees for these points. 
\end{proof}
\begin{proof}[\textbf{Proof of Proposition \ref{prop:induct construction of Wick fields via id OPE coefs}.}]
Using Wick's theorem \eqref{eq:Wicks thm CCR}, we find: 
\begin{align}
 & \phi(x_{1})\phi(x_{2})\cdots\phi(x_{n})\label{eq:Wick expansion of product of linear fields}\\
 & =\sum_{k=0}^{\lfloor n/2\rfloor}\sum_{\sigma_{k}}H(x_{\sigma(1)},x_{\sigma(2)})\cdots H(x_{\sigma(2k-1)},x_{\sigma(2k)}):\phi(x_{\sigma(2k+1)})\cdots\phi(x_{\sigma(n)}):_{H},\nonumber 
\end{align}
where $\sigma_{k}$ runs over the same permutations as in formula
\eqref{eq:gen linear C ident coef-1}. Putting all terms on the right-hand
side and smearing with the test distribution $t_{n+1}\in\mathcal{E}'(\times^{(n+1)}M,g_{ab})$
defined in eq.~\eqref{eq:t distribution} then yields: 
\begin{align}
0 & =\int_{z,x_{1},\dots,x_{n}}f^{\alpha_{1}\cdots\alpha_{n}}\delta(z,x_{1},\dots,x_{n})\nabla_{\alpha_{1}}^{(x_{1})}\cdots\nabla_{\alpha_{n}}^{(x_{n})}\left[\phi(x_{1})\cdots\phi(x_{n})-\sum_{k=0}^{\lfloor n/2\rfloor}\sum_{\sigma_{k}}H(x_{\sigma(1)},x_{\sigma(2)})\times\cdots\right.\nonumber \\
 & \hphantom{=}\left.\cdots\times H(x_{\sigma(2k-1)},x_{\sigma(2k)})\T[\beta_{2k+1}](x_{\sigma(2k+1)};z)\cdots\T[\beta_{n}](x_{\sigma(n)};z)(\nabla_{\beta_{2k+1}}\phi\cdots\nabla_{\beta_{n}}\phi)_{H}(z)\text{\ensuremath{\vphantom{\sum_{k=0}^{\lfloor n/2\rfloor}}}}\right],\label{eq:smeared zero}
\end{align}
with implied summations over $\beta$ multi-indices. Note only finitely-many
terms contribute non-trivially to the sum. In writing \eqref{eq:smeared zero},
we have used the definition \eqref{eq:H normal ordered} of the Hadamard
normal-ordered Wick fields. We may now use formula \eqref{eq:Phi^H = Z Phi}
to write $(\nabla_{\beta_{1}}\phi\cdots\nabla_{\beta_{m}}\phi)_{H}$
in terms of $(\nabla_{\beta_{1}}\phi\cdots\nabla_{\beta_{m}}\phi)$
and the smooth functions $F_{q\le m}$. Plugging this into \eqref{eq:smeared zero},
one can then use the explicit expression \eqref{eq:gen linear C ident coef-1}
for the Wick OPE coefficients $C_{\phi\cdots\phi}^{I}$ to write \eqref{eq:smeared zero}
as: 
\begin{align}
 & 0=\int_{z,x_{1},\dots,x_{n}}f^{\alpha_{1}\cdots\alpha_{n}}\delta(z,x_{1},\dots,x_{n})\nabla_{\alpha_{1}}^{(x_{1})}\cdots\nabla_{\alpha_{n}}^{(x_{n})}\left[\text{\ensuremath{\vphantom{\sum_{k=0}^{\lfloor n/2\rfloor}}}}\phi(x_{1})\cdots\phi(x_{n})+\right.\nonumber \\
 & \left.-\sum_{m\le n}\sum_{\pi\in\Pi_{m}}C_{\phi\cdots\phi}^{I}(x_{\pi(m+1)},\dots,x_{\pi(n)};z)\T[\beta_{1}](x_{\pi(1)};z)\cdots\T[\beta_{m}](x_{\pi(m)};z)(\nabla_{\beta_{1}}\phi\cdots\nabla_{\beta_{m}}\phi)(z)\right],\label{eq:smeared 0 (ii)}
\end{align}
with $\Pi_{m}$ and $\int_{z,x_{1},\dots,x_{n}}$ defined as in \eqref{eq:explicit Wick monomial for gen Rx}.
Note again there are implied finite sums over $\beta$ multi-indices.
Note the $m=n$ term in the sum yields: 
\[
-\delta(z,x_{1},\dots,x_{n})\nabla_{\alpha_{1}}^{(x_{1})}\T[\beta_{1}](x_{1};z)\cdots\nabla_{\alpha_{n}}^{(x_{n})}\T[\beta_{n}](x_{n};z)(\nabla_{\beta_{1}}\phi\cdots\nabla_{\beta_{n}}\phi)=-(\nabla_{\alpha_{1}}\phi\cdots\nabla_{\alpha_{m}}\phi),
\]
using the identity \eqref{eq:sum nabla_alpha T^beta = nabla_alpha-1}.
Moving this term to the left-hand side of \eqref{eq:smeared 0 (ii)},
then gives the equation \eqref{eq:explicit Wick monomial for gen Rx}
we sought to show. 
\end{proof}

\section{\label{sec:K-G Lorentz cts}Construction of \texorpdfstring{$a_{\gamma_{1}\gamma_{2}}$}{a\_\{\unichar{92}gamma\_1 \unichar{92}gamma\_2\}}
for Lorentz-covariance-restoring terms}

The goal of this appendix is construct $\Lambda$-independent $\boldsymbol{a}\equiv a_{\gamma_{1}\gamma_{2}}$
such that (\ref{eq:Taylor coef Omega_M(Lambda)-Omega_M}) holds for
any choice of cutoff function $\chi$, i.e., to construct $\boldsymbol{a}$
such that $\widetilde{\Omega}_{M}$ defined via (\ref{eq:Omega_M tilde prop 1})
is Lorentz-invariant. Our strategy will be to solve (\ref{eq:Taylor coef Omega_M(Lambda)-Omega_M})
inductively for infinitesimal Lorentz transformations, 
\begin{equation}
\Lambda_{\theta}=I+\frac{1}{2}\theta_{\kappa\rho}l^{\kappa\rho},\label{eq:Lambda_Theta-1}
\end{equation}
which generate the restricted Lorentz group. Here $\theta_{\kappa\rho}=\theta_{[\kappa\rho]}$
parameterize an arbitrary infinitesimal transformation and $l^{\kappa\rho}$
denote the Lorentz generators. Restoring indices, the generators are
given explicitly by, 
\begin{equation}
(l^{\kappa\rho})_{\ \nu}^{\mu}\equiv2\eta^{\mu[\kappa}\delta_{\hphantom{\rho]}\nu}^{\rho]},\label{eq:Lorentz generators-1}
\end{equation}
in the $(1/2,1/2)$ (i.e., vector) representation. We define 
\begin{equation}
Q_{M}(x_{1},x_{2};z;\Lambda^{-1})\equiv\Omega_{M}(\Lambda x_{1},\Lambda x_{2};\Lambda z)-\Omega_{M}(x_{1},x_{2};z),\label{eq:def Q func}
\end{equation}
and we denote the Taylor coefficients which appear on the left side
of (\ref{eq:Taylor coef Omega_M(Lambda)-Omega_M}) by 
\begin{equation}
\boldsymbol{Q}(\Lambda^{-1})\equiv Q_{\gamma_{1}\gamma_{2}}(\Lambda^{-1})\equiv\partial_{\gamma_{1}}^{(x_{1})}\partial_{\gamma_{2}}^{(x_{2})}Q_{M}(x_{1},x_{2};z;\Lambda^{-1})|_{x_{1},x_{2}=z}.\label{eq:Q gamma1 gamma2}
\end{equation}
Thus, $\boldsymbol{Q}$ is a spacetime tensor of the same rank $r=|\gamma_{1}|+|\gamma_{1}|$
as $\boldsymbol{a}$. Translation invariance implies $\boldsymbol{Q}(\Lambda^{-1})$
is independent of $z$. With this notation, the set of equations (\ref{eq:Taylor coef Omega_M(Lambda)-Omega_M})
that we wish to solve for $\boldsymbol{a}$ can be written as, 
\begin{equation}
\boldsymbol{Q}(\Lambda)=\left(D(\Lambda)-\mathbb{I}\right)\boldsymbol{a},\label{eq:Q eq (D(Lambda)-I)a}
\end{equation}
where $D(\Lambda)$ denotes the representation of the Lorentz group
on tensors of rank $r$. The kernel of the operator $D(\Lambda)-\mathbb{I}$
is comprised of Lorentz-invariant tensors, so (\ref{eq:Q eq (D(Lambda)-I)a})
determines $\boldsymbol{a}$ up to the addition of a Lorentz-invariant
tensor of rank $r$.

For the purposes of showing existence of a solution, $\boldsymbol{a}$,
to (\ref{eq:Q eq (D(Lambda)-I)a}), it is useful to have a manifestly
smooth expression for the function $Q_{M}(x_{1},x_{2};z;\Lambda^{-1})$
defined in (\ref{eq:def Q func}). From the definition (\ref{eq:Omega_M})
of $\Omega_{M}$, we have, 
\begin{align}
\Omega_{M}(\Lambda x_{1},\Lambda x_{2};\Lambda z) & =-i\int d^{D}y\,\chi(y,\Lambda z)H_{F}(y,\Lambda x_{1})H_{F}(y,\Lambda x_{2})\label{eq:Omega_M(Lambda)}\\
 & =-i\int d^{D}y'\,\chi(\Lambda y',\Lambda z)H_{F}(\Lambda y',\Lambda x_{1})H_{F}(\Lambda y',\Lambda x_{2})\nonumber \\
 & =-i\int d^{D}y'\,\chi(\Lambda y',\Lambda z)H_{F}(y',x_{1})H_{F}(y',x_{2}).\nonumber 
\end{align}
Here, the second equality was obtained by making a change of integration
variables $y\to y'=\Lambda^{-1}y$ and the final equality follows
from the Lorentz invariance of $H_{F}$. Plugging (\ref{eq:Omega_M(Lambda)})
into (\ref{eq:def Q func}) yields, 
\begin{align}
Q_{M}(x_{1},x_{2};z;\Lambda^{-1}) & =-i\int d^{D}y\,\left[\chi(\Lambda y,\Lambda z)-\chi(y,z)\right]H_{F}(y,x_{1})H_{F}(y,x_{2}).\label{eq:Q explicit expr}
\end{align}
Since for arbitrary, fixed $\Lambda$, we have $\chi(\Lambda y,\Lambda z)-\chi(y,z)=0$
when $y$ is sufficiently close to $z$, it follows from Proposition
\ref{lem:smoothness of bt} that $Q(x_{1},x_{2};z;\Lambda^{-1})$
is smooth in $(x_{1},x_{2})$ when these points are sufficiently close
to $z$. Evaluating the Taylor coefficients of (\ref{eq:Q explicit expr})
yields, 
\begin{align}
Q_{\gamma_{1}\gamma_{2}}(\Lambda^{-1}) & =(-1)^{(1+|\gamma_{1}|+|\gamma_{2}|)}i\int d^{D}y\,\left[\chi(\Lambda y,\vec{0})-\chi(y,\vec{0})\right]\partial_{\gamma_{1}}^{(y)}H_{F}(y,\vec{0})\partial_{\gamma_{2}}^{(y)}H_{F}(y,\vec{0}),\label{eq:Q_gamma1gamma2 explicit formula}
\end{align}
where the translation symmetry has been used to put $z$ at the origin.
Note $Q_{\gamma_{1}\gamma_{2}}$ are manifestly invariant under interchange
of multi-indices $Q_{\gamma_{1}\gamma_{2}}=Q_{\gamma_{2}\gamma_{1}}$
and symmetric within their respective spacetime indices $Q_{\{\mu_{1}\cdots\mu_{|\gamma_{1}|}\}\{\nu_{1}\cdots\nu_{|\gamma_{2}|}\}}=Q_{\{(\mu_{1}\cdots\mu_{|\gamma_{1}|})\}\{(\nu_{1}\cdots\nu_{|\gamma_{2}|})\}}$.
The following proposition establishes the existence of $\boldsymbol{a}$
satisfying (\ref{eq:Q eq (D(Lambda)-I)a}) by the same type of cohomology
argument as used to prove the existence of counterterms in Epstein-Glaser
renormalization \cite{pedagogical_remark_main_thm_ren_theory} : 
\begin{prop}
\label{Prop: existence of a}For any translation-invariant cutoff
function $\chi$ and any restricted Lorentz transformation $\Lambda$,
the tensors $\boldsymbol{Q}(\Lambda^{-1})$ defined in (\ref{eq:Q_gamma1gamma2 explicit formula})
are always of the form (\ref{eq:Q eq (D(Lambda)-I)a}) for some $\Lambda$-independent
tensors $\boldsymbol{a}$, which are uniquely determined modulo Lorentz-invariant
tensors of rank $r\equiv|\gamma_{1}|+|\gamma_{2}|$. 
\end{prop}
\begin{proof}
Using the explicit formula \eqref{eq:Q_gamma1gamma2 explicit formula},
we find: 
\begin{align}
 & Q_{\gamma_{1}\gamma_{2}}(\Lambda_{1}\Lambda_{2})-Q_{\gamma_{1}\gamma_{2}}(\Lambda_{1})\label{eq:Q(Lambda1Lambda2)-Q(Lambda1)}\\
 & =(-1)^{(1+|\gamma_{1}|+|\gamma_{2}|)}i\int d^{D}y\,\left[\chi(\Lambda_{2}^{-1}\Lambda_{1}^{-1}y,\vec{0})-\chi(\Lambda_{1}^{-1}y,\vec{0})\right]\partial_{\gamma_{1}}^{(y)}H_{F}(y,\vec{0})\partial_{\gamma_{2}}^{(y)}H_{F}(y,\vec{0})\nonumber \\
 & =(\Lambda_{1}^{-1})_{\hphantom{\gamma_{1}'}\gamma_{1}}^{\gamma_{1}'}(\Lambda_{1}^{-1})_{\hphantom{\gamma_{2}'}\gamma_{2}}^{\gamma_{2}'}(-1)^{(1+|\gamma_{1}|+|\gamma_{2}|)}i\int d^{D}y'\,\left[\chi(\Lambda_{2}^{-1}y',\vec{0})-\chi(y',\vec{0})\right]\partial_{\gamma_{1}'}^{(y)}H_{F}(y',\vec{0})\partial_{\gamma_{2}'}^{(y)}H_{F}(y',\vec{0})\nonumber \\
 & =(\Lambda_{1})_{\gamma_{1}}^{\hphantom{\gamma_{1}}\gamma_{1}'}(\Lambda_{1})_{\gamma_{2}}^{\hphantom{\gamma_{1}}\gamma_{2}'}Q_{\gamma_{1}'\gamma_{2}'}(\Lambda_{2}).\nonumber 
\end{align}
In going to the first equality, we note $(\Lambda_{1}\Lambda_{2})^{-1}=\Lambda_{2}^{-1}\Lambda_{1}^{-1}$.
The second equality follows from a change of integration variables
$y\to y'=\Lambda_{1}^{-1}y$, noting the parametrix is Lorentz invariant
and $\det\Lambda_{1}=1$ so $d^{D}y'=d^{D}y$.

Given (\ref{eq:Q(Lambda1Lambda2)-Q(Lambda1)}), eq.~(\ref{eq:Q eq (D(Lambda)-I)a})
can now be established via the following cohomological argument: Denote
the restricted Lorentz group $\mathcal{L}_{+}^{\uparrow}\equiv SO^{+}(1,D-1)$
and denote by $C^{n}(\mathcal{L}_{+}^{\uparrow})$ the set of all
tensors $\boldsymbol{T}\equiv T_{\alpha}(\Lambda_{1},\dots,\Lambda_{n})$
which depend continuously on $\Lambda$. For each $n\geq0$, we define
the ``coboundary operator'' $d^{n}:C^{n}(\mathcal{L}^{\uparrow})\to C^{n+1}(\mathcal{L}_{+}^{\uparrow})$
by\footnote{In cohomology theory, $C^{n}$ are known as the group of ``$n$-cochains''.
The sequence, $C^{0}\overset{d^{0}}{\longrightarrow}C^{1}\overset{d^{1}}{\longrightarrow}C^{2}\overset{d^{2}}{\longrightarrow}\cdots$,
generated by the coboundary operators $d^{n}$ is called a ``cochain
complex''. }, 
\begin{align}
(d^{n}\boldsymbol{T})(\Lambda_{1},\dots,\Lambda_{n+1}) & \equiv(-1)^{(n+1)}\boldsymbol{T}(\Lambda_{1},\dots,\Lambda_{n})+D(\Lambda_{1})\boldsymbol{T}(\Lambda_{2},\dots,\Lambda_{(n+1)})+\label{eq:coboundary op-1}\\
 & +\sum_{k=1}^{n}(-1)^{k}\boldsymbol{T}(\Lambda_{1},\dots,\Lambda_{(k-1)},\widehat{\Lambda_{k}}\Lambda_{k}\Lambda_{(k+1)},\Lambda_{(k+2)}\dots,\Lambda_{(n+1)}).\nonumber 
\end{align}
For any $\boldsymbol{T}\in C^{n}(\mathcal{L}_{+}^{\uparrow})$, it
follows from the definition (\ref{eq:coboundary op-1}) via a straightforward
computation that we have 
\begin{equation}
(d^{n+1}\circ d^{n}\boldsymbol{T})(\Lambda_{1},\dots,\Lambda_{n+2})=0.\label{eq:d^2=0-2}
\end{equation}
Hence, for any $\boldsymbol{T}$ such that, 
\begin{equation}
d^{n}\boldsymbol{T}=0,\label{eq:d^nT=0}
\end{equation}
it follows immediately from (\ref{eq:d^2=0-2}) that (\ref{eq:d^nT=0})
is satisfied by, 
\begin{equation}
\boldsymbol{T}=d^{n-1}\boldsymbol{S},\label{eq:T=d^(n-1)S}
\end{equation}
for tensor $\boldsymbol{S}=\boldsymbol{S}(\Lambda_{1},\dots,\Lambda_{n-1})$
with the same rank as $\boldsymbol{T}$. If the only solutions to
(\ref{eq:d^nT=0}) are of the form (\ref{eq:T=d^(n-1)S}),
then it is said that the ``$n$-th cohomology group'', $H^{n}(\mathcal{L}_{+}^{\uparrow})\equiv\ker d^{n}\backslash\im d^{n-1}$,
is empty. It has been proven \cite[Subsection 5.C]{Wigner} that the
first cohomology group $H^{1}(\mathcal{L}_{+}^{\uparrow})$ is empty.
However, by eq.(\ref{eq:Q(Lambda1Lambda2)-Q(Lambda1)}), we have 
\begin{equation}
0=(d^{1}\boldsymbol{\boldsymbol{Q}})(\Lambda_{1},\Lambda_{2})=\boldsymbol{Q}(\Lambda_{1})+D(\Lambda_{1})\boldsymbol{Q}(\Lambda_{2})-\boldsymbol{Q}(\Lambda_{1}\Lambda_{2}),\label{eq:d^1Q=0}
\end{equation}
Therefore, the \emph{only} tensors satisfying (\ref{eq:d^1Q=0})
are of the form 
\begin{equation}
\boldsymbol{Q}(\Lambda)=(d^{0}\boldsymbol{a})(\Lambda)=(D(\Lambda)-\mathbb{I})\boldsymbol{a},\label{eq:Q=d^0a}
\end{equation}
Thus, for $\boldsymbol{Q}$ given by (\ref{eq:Q_gamma1gamma2 explicit formula}),
there exists a solution $\boldsymbol{a}$ to (\ref{eq:Q eq (D(Lambda)-I)a}). 
\end{proof}
Although Proposition \ref{Prop: existence of a} establishes existence
of $\boldsymbol{a}$, we wish to obtain an explicit solution for $\boldsymbol{a}$
in order to write the flow relations in an explicit form. In the remainder
of this appendix, we derive an explicit solution for $\boldsymbol{a}$
for ranks $r=0,1,2$ and then obtain an inductive solution for $\boldsymbol{a}$
for $r>2$. Our analysis closely follows the approach taken by \cite[Subsection 3.3]{Prange}
in the context of the Epstein-Glaser renormalization scheme, while
generalizing to arbitrary spacetime dimension.

For $r=0$, $D(\Lambda)=1$ and thus \eqref{eq:Q=d^0a} implies
$Q_{\{0\}\{0\}}=0$ for any Lorentz-invariant scalar $a_{\{0\}\{0\}}$.
For $r=1$, we have $Q_{\{\mu\}\{0\}}=Q_{\{0\}\{\mu\}}$, so there
is only a single independent $\boldsymbol{Q}(\Lambda)$. The dependence
of $\boldsymbol{Q}$ on $\Lambda$ comes entirely through the cutoff
function $\chi$. Since we have 
\begin{equation}
\chi(\Lambda_{\theta}^{-1}y,\vec{0})-\chi(y,\vec{0})=-\frac{1}{2}\theta_{\kappa\rho}(l^{\kappa\rho})_{\ \nu}^{\mu}y^{\nu}\partial_{\mu}\chi(y,\vec{0})+\mathcal{O}(\theta^{2})
\end{equation}
it follows from (\ref{eq:Q_gamma1gamma2 explicit formula}) that at
leading order in $\theta$ we have 
\[
Q_{\{\mu\}\{0\}}(\Lambda_{\theta})=-\frac{1}{2}\theta_{\kappa\rho}(B^{\kappa\rho})_{\{\mu\}\{0\}},
\]
where 
\begin{align}
(B^{\kappa\rho})_{\{\mu\}\{0\}} & \equiv i(l^{\kappa\rho})_{\hphantom{\sigma_{1}}\sigma_{2}}^{\sigma_{1}}\int d^{D}y\,y^{\sigma_{2}}\partial_{\sigma_{1}}^{(y)}\chi(y,\vec{0})\partial_{\mu}^{(y)}H_{F}(y,\vec{0})H_{F}(y,\vec{0}).\label{eq:B r=1}
\end{align}
Note that $(B^{\kappa\rho})_{\{\mu\}\{0\}}$ is independent of $\theta_{\kappa\rho}$.
On the other hand, for $r=1$, to leading order in $\theta$ the right-hand
side of (\ref{eq:Q eq (D(Lambda)-I)a}) is simply, 
\begin{align}
\left(D(\Lambda_{\theta})-I\right)\boldsymbol{a} & =-\frac{1}{2}\theta_{\kappa\rho}l^{\kappa\rho}\boldsymbol{a}.
\end{align}
Hence, for $r=1$, to leading order in $\theta$ equation (\ref{eq:Q eq (D(Lambda)-I)a})
is equivalent to, 
\begin{equation}
\theta_{\kappa\rho}l^{\kappa\rho}\boldsymbol{a}=\theta_{\kappa\rho}\boldsymbol{B}^{\kappa\rho}.\label{eq:theta l a = theta B}
\end{equation}
eq.~(\ref{eq:theta l a = theta B}) will hold for all infinitesimal
$\theta_{\kappa\rho}$ if and only if, 
\begin{equation}
l^{\kappa\rho}\boldsymbol{a}=\boldsymbol{B}^{\kappa\rho},\label{eq:l a = B for r=1}
\end{equation}
for all $\kappa,\rho=0,1,\dots,D-1$. Contracting this equation with
$l_{\kappa\rho}$ and using the identity\footnote{The left-hand side of \eqref{eq:Casimir} is the quadratic Casimir
operator of the Lie algebra of the homogeneous Lorentz group.} 
\begin{equation}
l_{\kappa\rho}l^{\kappa\rho}=-2(D-1)I,\label{eq:Casimir}
\end{equation}
we obtain the explicit solution 
\begin{align}
a_{\{0\}\{\mu\}}=a_{\{\mu\}\{0\}} & =-\frac{1}{2(D-1)}(l_{\kappa\rho}B^{\kappa\rho})_{\{\mu\}\{0\}}=-i\int d^{D}y\,\partial_{\mu}^{(y)}\chi(y,\vec{0})H_{F}(y,\vec{0})H_{F}(y,\vec{0}),\label{eq:a sol r equals 1}
\end{align}
where we have used (\ref{eq:Lorentz generators-1}) and (\ref{eq:B r=1})
to obtain the rightmost equality.

We proceed now to $r=2$. There are two independent $\boldsymbol{Q}$
tensors of rank two and they are both symmetric in their spacetime
indices: $Q_{\{\mu\}\{\nu\}}=Q_{\{(\mu\}\{\nu)\}}$ and $Q_{\{\mu\nu\}\{0\}}=Q_{\{0\}\{\mu\nu\}}=Q_{\{0\}\{(\mu\nu)\}}$.
For $r=2$ and infinitesimal $\Lambda=\Lambda_{\theta}$, one now
finds (\ref{eq:Q eq (D(Lambda)-I)a}) takes the form 
\begin{equation}
\left(l^{\kappa\rho}\otimes I+I\otimes l^{\kappa\rho}\right)\boldsymbol{a}=\boldsymbol{B}^{\kappa\rho}.\label{eq:l a = B for r=2}
\end{equation}
Here $(B^{\kappa\rho})_{\{\mu\}\{\nu\}}$ and $(B^{\kappa\rho})_{\{\mu\nu\}\{0\}}=(B^{\kappa\rho})_{\{0\}\{\mu\nu\}}$
are defined by a rank 2 generalization of (\ref{eq:B r=1});
the general formula for $\boldsymbol{B}^{\kappa\rho}$ for arbitrary
rank is given in equation (\ref{eq:B^kappa rho for general r}) below.
Applying the operator $\left(l_{\kappa\rho}\otimes I+I\otimes l_{\kappa\rho}\right)$
to both sides of (\ref{eq:l a = B for r=2}) and contracting
over the $\kappa,\rho$ indices yields, 
\begin{equation}
-4(D-1)\boldsymbol{a}+2\left(l_{\kappa\rho}\otimes l^{\kappa\rho}\right)\boldsymbol{a}=\left(l_{\kappa\rho}\otimes I+I\otimes l_{\kappa\rho}\right)\boldsymbol{B}^{\kappa\rho}.\label{eq:r=2 intermed rel}
\end{equation}
Using the explicit expression (\ref{eq:Lorentz generators-1}) for
$l^{\kappa\rho}$, it is easily seen that for any rank two tensor
$\boldsymbol{T}\equiv T_{\mu\nu}$ we have 
\begin{equation}
\left(\left(l_{\kappa\rho}\otimes l^{\kappa\rho}\right)T\right)_{\mu\nu}=2\left(\text{tr}\left(\boldsymbol{T}\right)\eta_{\mu\nu}-T_{\nu\mu}\right).\label{eq:l tensor prod l}
\end{equation}
where $\text{tr}(\boldsymbol{T})\equiv\eta^{\mu\nu}T_{\mu\nu}$. Note
that the trace is a Lorentz scalar, so this term is automatically
Lorentz invariant. Substituting (\ref{eq:l tensor prod l}) into (\ref{eq:r=2 intermed rel})
and symmetrizing over $(\mu,\nu)$, we obtain 
\begin{align}
a_{\{(\mu\}\{\nu)\}} & =-\frac{1}{4D}\left((l_{\kappa\rho}\otimes I+I\otimes l_{\kappa\rho})B^{\kappa\rho}\right)_{\{(\mu_{1}\}\{\mu_{2})\}}\label{eq:r equals 2 soln 1}\\
 & =-i\int d^{D}y\chi(y,\vec{0})\left[\partial_{\mu}H_{F}(y,\vec{0})\partial_{\nu}H_{F}(y,\vec{0})-\frac{1}{D}\eta_{\mu\nu}\partial_{\sigma}H_{F}(y,\vec{0})\partial^{\sigma}H_{F}(y,\vec{0})\right]\nonumber 
\end{align}
where all derivatives are taken with respect to the spacetime point
$y$. Similarly, we find, 
\begin{align}
a_{\{(\mu\nu)\}\{0\}}=a_{\{0\}\{(\mu\nu)\}} & =-\frac{1}{4D}\left((l_{\kappa\rho}\otimes I+I\otimes l_{\kappa\rho})B^{\kappa\rho}\right)_{\{0\}\{(\mu_{1}\mu_{2})\}}\label{eq:r equals 2 soln 2}\\
 & =-i\int d^{D}y\chi(y,\vec{0})\left[H_{F}(y,\vec{0})\partial_{\mu}\partial_{\nu}H_{F}(y,\vec{0})-\frac{1}{D}\eta_{\mu\nu}H_{F}(y,\vec{0})\partial^{2}H_{F}(y,\vec{0})\right].\nonumber 
\end{align}
Thus, we have explicitly solved for $\boldsymbol{a}$ for all ranks
$r\leq2$.

We turn now to the derivation of an inductive solution to (\ref{eq:Q eq (D(Lambda)-I)a})
for $r>2$. For infinitesimal $\Lambda=\Lambda_{\theta}$ and to leading
order in $\theta$, eq.~(\ref{eq:Q eq (D(Lambda)-I)a}) now yields
\begin{equation}
\mathbb{L}^{\kappa\rho}\boldsymbol{a}=\boldsymbol{B}^{\kappa\rho},\label{eq:L a = B for r>2}
\end{equation}
where 
\begin{equation}
\mathbb{L}^{\kappa\rho}\equiv\left(l^{\kappa\rho}\otimes I\otimes\cdots\otimes I\right)+\left(I\otimes l^{\kappa\rho}\otimes I\otimes\cdots\otimes I\right)+\cdots+\left(I\otimes\cdots\otimes I\otimes l^{\kappa\rho}\right),\label{eq:L^kappa rho}
\end{equation}
and 
\begin{align}
(B^{\kappa\rho} & )_{\{\mu_{1}\cdots\mu_{|\gamma_{1}|}\}\{\nu_{1}\cdots\nu_{|\gamma_{2}|}\}}\label{eq:B^kappa rho for general r}\\
 & \equiv i(-1)^{(1+|\gamma_{1}|+|\gamma_{2}|)}(l^{\kappa\rho})_{\hphantom{\sigma_{1}}\sigma_{2}}^{\sigma_{1}}\int d^{D}y\,y^{\sigma_{2}}\partial_{\sigma_{1}}\chi(y,\vec{0})\partial_{\mu_{1}}\cdots\partial_{\mu_{|\gamma_{1}|}}H_{F}(y,\vec{0})\partial_{\nu_{1}}\cdots\partial_{\nu_{|\gamma_{2}|}}H_{F}(y,\vec{0}).\nonumber 
\end{align}
with all derivatives taken with respect to $y$. As in the $r=1,2$
cases, we solve (\ref{eq:L a = B for r>2}) by applying the
operator $\mathbb{L}_{\kappa\rho}$ to both sides and contracting
the $\kappa,\rho$-indices. We begin by noting that the operator we
obtain on the left-hand side, 
\begin{equation}
\mathbb{L}_{\kappa\rho}\mathbb{L}^{\kappa\rho},\label{eq:L^2 op}
\end{equation}
contains two types of terms: There are $r$ terms of the form, 
\[
I\otimes\cdots\otimes I\otimes l_{\kappa\rho}l^{\kappa\rho}\otimes I\otimes\cdots\otimes I=-2(D-1)I\otimes\cdots\otimes I,
\]
where we used (\ref{eq:Casimir}). Similarly, using (\ref{eq:l tensor prod l}),
the remaining $r(r-1)$ terms in (\ref{eq:L^2 op}) are of the form,
\[
I\otimes\cdots I\otimes\underset{i\text{-th slot}}{\underbrace{l_{\kappa\rho}}}\otimes I\cdots I\otimes\underset{j\text{-th slot}}{\underbrace{l^{\kappa\rho}}}\otimes I\cdots\otimes I=2(\eta_{\mu_{i}\mu_{j}}\text{tr}_{ij}-\mathbb{T}_{ij}),
\]
where $\text{tr}_{ij}$ and $\mathbb{T}_{ij}$ denote, respectively,
the trace over the $i,j$-th spacetime indices and the transposition
of the $i,j$-th indices, i.e., for any tensor $\boldsymbol{T}$ we
have 
\begin{align}
(\text{tr}_{ij}T)_{\mu_{1}\cdots\mu_{r}} & \equiv\eta^{\mu_{i}\mu_{j}}T_{\mu_{1}\cdots\mu_{r}}\qquad,\qquad(\mathbb{T}_{ij}T)_{\mu_{1}\cdots\mu_{r}}\equiv T_{\mu_{1}\cdots\widehat{\mu_{i}}\mu_{j}\cdots\widehat{\mu_{j}}\mu_{i}\cdots\mu_{r}}.
\end{align}
Altogether, therefore, we have 
\[
\mathbb{L}_{\kappa\rho}\mathbb{L}^{\kappa\rho}=-2r(D-1)\mathbb{I}+4\sum_{i<j\le r}(\eta_{\mu_{i}\mu_{j}}\text{tr}_{ij}-\mathbb{T}_{ij}),
\]
where $\mathbb{I}\equiv I^{\otimes r}$. Hence, multiplying both sides
of (\ref{eq:L a = B for r>2}) by $\mathbb{L}_{\kappa\rho}$
and contracting the $\kappa,\rho$-indices yields, 
\begin{equation}
2r(D-1)\boldsymbol{a}+4\sum_{i<j\le n}\mathbb{T}_{ij}\boldsymbol{a}=-\mathbb{L}_{\kappa\rho}\boldsymbol{B}^{\kappa\rho}+4\sum_{i<j\le n}\eta_{\mu_{i}\mu_{j}}\text{tr}_{ij}\boldsymbol{a}.\label{eq:L(L a = B)}
\end{equation}
Now, the trace of (\ref{eq:L a = B for r>2}) yields 
\[
\text{tr}_{ij}\left(\mathbb{L}_{(r)}^{\kappa\rho}\boldsymbol{a}\right)=\mathbb{L}_{(r-2)}^{\kappa\rho}\left(\text{tr}_{ij}\boldsymbol{a}\right)=\text{tr}_{ij}\boldsymbol{B}^{\kappa\rho},
\]
where we have inserted a subscript $(r)$ on $\mathbb{L}_{(r)}^{\kappa\rho}$
to indicate the rank of the operator (\ref{eq:L^kappa rho}) being
considered. Thus, $\text{tr}_{ij}\boldsymbol{a}$ satisfies an equation
of the same form as (\ref{eq:L a = B for r>2}) but for the
lower rank $r'=r-2$ and with $\boldsymbol{B}^{\prime\kappa\rho}=\text{tr}_{ij}\boldsymbol{B}^{\kappa\rho}$.
For example, this implies the trace of the $r=3$ tensor $a_{\{\mu_{1}\mu_{2}\}\{\nu_{1}\}}$
with respect to its two $\mu$-spacetime indices is given by: 
\[
\eta^{\mu_{1}\mu_{2}}a_{\{\mu_{1}\mu_{2}\}\{\nu_{1}\}}=-\frac{1}{2(D-1)}\eta^{\mu_{1}\mu_{2}}(l_{\kappa\rho}B^{\kappa\rho})_{\{\mu_{1}\mu_{2}\}\{\nu_{1}\}},
\]
which is obtained by replacing $(B^{\kappa\rho})_{\{0\}\{\nu_{1}\}}$
with $\eta^{\mu_{1}\mu_{2}}(B^{\kappa\rho})_{\{\mu_{1}\mu_{2}\}\{\nu_{1}\}}$
in the $r=1$ solution \eqref{eq:a sol r equals 1} for $a_{\{0\}\{\nu_{1}\}}$.
Thus, since we are obtaining solutions inductively in $r$ and have
already obtained explicit solutions for $r=1,2$, we may treat $\text{tr}_{ij}\boldsymbol{a}$
in (\ref{eq:L(L a = B)}) as ``known.''

Thus, it remains only to extract $\boldsymbol{a}$ from the combination
of components of $\boldsymbol{a}$ appearing on the left side of (\ref{eq:L(L a = B)}).
To do so, we note that the sum over all transpositions commutes with
any permutation. A standard result in the representation theory of
finite-dimensional groups implies the set of all elements that commute
with the group algebra of the symmetric group $S_{r}$ is spanned
by a complete set of orthogonal (idempotent) elements $\mathbb{E}_{i}$,
\begin{equation}
\mathbb{E}_{i}\mathbb{E}_{j}=\delta_{ij}\mathbb{E}_{i}\qquad,\qquad\sum_{i=1}^{k}\mathbb{E}_{i}=\mathbb{I},\label{eq:idempotent identities}
\end{equation}
where $k$ denotes the number of partitions of $r$. Hence, we may
expand the sum over transpositions appearing in (\ref{eq:L(L a = B)}),
\begin{equation}
\sum_{i<j\le r}\mathbb{T}_{ij}=\sum_{i=1}^{k}c_{i}\mathbb{E}_{i},\label{eq:expansion of transposition}
\end{equation}
for some real-valued coefficients $c_{i}$. Applying the operator
$\mathbb{E}_{j}$ to both sides of (\ref{eq:L(L a = B)}) and
using the orthogonality property (\ref{eq:idempotent identities}),
we obtain then, 
\begin{align}
\left(2r(D-1)+4c_{j}\right)\mathbb{E}_{j}\boldsymbol{a}=\mathbb{E}_{j}\left(-\mathbb{L}_{\kappa\rho}\boldsymbol{B}^{\kappa\rho}+4\sum_{i<j\le n}\eta_{\mu_{i}\mu_{j}}\text{tr}_{ij}\boldsymbol{a}\right) & .\label{eq:E L(L a = B)}
\end{align}
We abbreviate the numerical coefficients, 
\begin{equation}
\widetilde{c}_{j}\equiv2r(D-1)+4c_{j}.\label{eq:c tilde}
\end{equation}
For any $j$ such that $\widetilde{c}_{j}=0$, eq.~(\ref{eq:E L(L a = B)})
places no constraint on the corresponding $\mathbb{E}_{j}\boldsymbol{a}$
and, thus, this particular $\mathbb{E}_{j}\boldsymbol{a}$ must automatically
be composed of an Lorentz-invariant combination of the metric and
totally-antisymmetric $D$-dimensional tensor densities (i.e. ``Levi-Civita
symbols'' $\epsilon_{\mu_{1}\cdots\mu_{n}}$). For all $j$ such
that $\widetilde{c}_{j}\ne0$, we may divide (\ref{eq:E L(L a = B)})
through by $\widetilde{c}_{j}$ and use the completeness relation
(\ref{eq:idempotent identities}) to obtain the inductive solution,
\begin{equation}
\boldsymbol{a}=\sum_{\substack{j=1\\
\widetilde{c}_{j}\ne0
}
}^{k}\frac{1}{\widetilde{c}_{j}}\mathbb{E}_{j}\left(-\mathbb{L}_{\kappa\rho}\boldsymbol{B}^{\kappa\rho}+4\sum_{i<j\le n}\eta_{\mu_{i}\mu_{j}}\text{tr}_{ij}\boldsymbol{a}\right),\label{eq:a soln for r>2}
\end{equation}
modulo arbitrary Lorentz-invariant tensors which may be identified
with the value of the sum over the terms which are unconstrained by
(\ref{eq:E L(L a = B)}): 
\[
\sum_{\substack{j=1\\
\widetilde{c}_{j}=0
}
}^{k}\mathbb{E}_{j}\boldsymbol{a}=\text{Lorentz invariant tensor of rank \ensuremath{r}}.
\]
All quantities appearing in our inductive solution (\ref{eq:a soln for r>2})
for $\boldsymbol{a}$ have been explicitly defined here except for
the numerical coefficients $c_{j}$ and the idempotent elements $\mathbb{E}_{j}$
which may be constructed via standard methods from the representation
theory of the symmetric group (see \cite[see ``Appendix A: Representation of the symmetric groups'']{Prange}
and references therein). Note that the inductive solution, eq.~(\ref{eq:a soln for r>2}),
with $\boldsymbol{B}^{\kappa\rho}$ defined via 
\begin{equation}
\boldsymbol{Q}(\Lambda_{\theta})=-\frac{1}{2}\theta_{\kappa\rho}\boldsymbol{B}^{\kappa\rho}+\mathcal{O}(\theta^{2}),\label{eq:Q(Lambda theta) eq theta B plus order theta squ}
\end{equation}
holds for \emph{any} tensors $\boldsymbol{Q}(\Lambda)$ satisfying
(\ref{eq:Q eq (D(Lambda)-I)a}) not just those defined\footnote{In particular, the solution (\ref{eq:a soln for r>2}) for $\boldsymbol{a}$
holds when $\boldsymbol{Q}$ corresponds to the $\Lambda$-dependent
coefficients of the contact terms, 
\[
(Q_{A_{1}\cdots A_{n}})^{\alpha_{1}\cdots\alpha_{n}}(\Lambda^{-1})\partial_{\alpha_{1}}^{(x_{1})}\cdots\partial_{\alpha_{n}}^{(x_{n})}\delta^{(D)}(x_{1},\dots,x_{n}),
\]
that quantify the failure of the Epstein-Glaser renormalized (i.e.
``extended'') time-ordered products, $T\{\Phi_{A_{1}}(x_{1})\cdots\Phi_{A_{n}}(x_{n})\}$,
to be Lorentz-covariant. Hence, there is a close analogy between the
counterterms required to restore Lorentz-covariance in Epstein-Glaser
renormalization and our ``counterterms'' for the flow relations.
The primary difference is that that our counterterms are not proportional
to (differentiated) $\delta$-functions and, in the particular case
of the flow relation for $(C_{H})_{T_{0}\{\phi\phi\}}^{I}=H_{F}$,
they are actually smooth functions of the spacetime variables. } via (\ref{eq:Q gamma1 gamma2}). 
\begin{rem}
In the case where either $|\gamma_{1}|=0$ or $|\gamma_{2}|=0$, the
tensor $a_{\gamma_{1}\gamma_{2}}$ is totally symmetric in its spacetime
indices and a closed form solution to the induction equation (\ref{eq:a soln for r>2})
can be obtained (see the solution to the analogous problem in Epstein-Glaser
renormalization given in \cite[Section 3]{Bresser_Pinter_Prange}). 
\end{rem}
\begin{rem}
\label{rem:recovering r<=2 cases}When $r\le2$, the inductive
solution \eqref{eq:a soln for r>2} to \eqref{eq:Q eq (D(Lambda)-I)a}
reproduces the explicit solutions we obtained above. The $r=0,1$
cases are trivial. To verify the $r=2$ case, note the symmetric group
$S_{2}$ contains two elements: the identity $I$ and the transposition
$\mathbb{T}_{12}$. It is easily checked, in this case, that the idempotent
decomposition \eqref{eq:expansion of transposition} is satisfied
by, 
\begin{equation}
\mathbb{T}_{12}=\mathbb{S}-\mathbb{A},\label{eq:idemp exp of T_12}
\end{equation}
where $\mathbb{S}^{2}=\mathbb{S}$ and $\mathbb{A}^{2}=\mathbb{A}$
denote, respectively, the projector onto the symmetric part and the
anti-symmetric part of any tensor of rank 2. Note these projectors
are ``orthogonal'' in the sense that, for any tensor $\boldsymbol{T}$
of rank 2, $\mathbb{S}(\mathbb{A}\boldsymbol{T})=0=\mathbb{A}(\mathbb{S}\boldsymbol{T})$.
Moreover, they are ``complete'' in the sense that $\mathbb{S}+\mathbb{A}=I$.
Therefore, since they satisfy all the requisite properties, we may
identify these projectors with the idempotents \eqref{eq:idempotent identities}
for $r=2$. Denoting $\mathbb{E}_{1}=\mathbb{S}$ and $\mathbb{E}_{2}=\mathbb{A}$,
we simply read off the coefficients $c_{1}=1$ and $c_{2}=-1$ by
comparing \eqref{eq:idemp exp of T_12} with \eqref{eq:expansion of transposition}.
Hence, the formula \eqref{eq:c tilde} gives $\widetilde{c}_{1}=4D$
and $\widetilde{c}_{2}=4(D-2)$ in this case. Plugging these into
the general formula \eqref{eq:a soln for r>2} immediately yields:
for $D\ne2$, 
\begin{equation}
\boldsymbol{a}=-\frac{1}{4D}\left(\mathbb{S}+\frac{D}{D-2}\mathbb{A}\right)\left(l_{\kappa\rho}\otimes I+I\otimes l_{\kappa\rho}\right)\boldsymbol{B}^{\kappa\rho}+4\eta_{\mu_{1}\mu_{2}}\text{tr}_{12}\boldsymbol{a},\label{eq:r=2 a from gen soln}
\end{equation}
which is the most general rank 2 solution to \eqref{eq:Q eq (D(Lambda)-I)a}.
For $D=2$, we have $\widetilde{c}_{2}=0$, so the general formula
\eqref{eq:a soln for r>2} yields \eqref{eq:r=2 a from gen soln}
without the anti-symmetric term: note, in $D=2$, any anti-symmetric
tensor of type $(0,2)$ is proportional to the Levi-Civita symbol
$\epsilon_{\mu\nu}$ and, thus, is automatically invariant under restricted
Lorentz transformations. For our application in any dimension, only
the symmetric part of $\boldsymbol{a}$ is of interest. Note also
the trace of any rank two tensor is a Lorentz scalar. Hence, \eqref{eq:r=2 a from gen soln}
is consistent with the results given in eqs. \eqref{eq:r equals 2 soln 1}
and \eqref{eq:r equals 2 soln 2} above, i.e., 
\[
\mathbb{S}\boldsymbol{a}=-\frac{1}{4D}\mathbb{S}\left(l_{\kappa\rho}\otimes I+I\otimes l_{\kappa\rho}\right)\boldsymbol{B}^{\kappa\rho}+\text{ Lorentz-invariant tensor}.
\]
\end{rem}

\section{Curvature expansion of \texorpdfstring{$\Omega_{C}$}{\unichar{937}\_C} \label{sec:proof of Omega_C curvature exp}}

In this appendix, we derive the curvature expansion, eq.~(\ref{eq:Omega_C curvature exp}),
for $\Omega_{C}$. The derivation closely follows the approach of
\cite[Proof of Theorem 4.1]{HW_Existence_TOP} with modifications
to account for the non-local metric dependence of $\Omega_{C}$ and
its dependence on $t_{\mu}$. Let $g_{\mu\nu}$ denote the components
of the metric in RNC centered at $z\in M$. Let $S_{\lambda}:\mathbb{R}^{D}\to\mathbb{R}^{D}$
denote the diffeomorphism corresponding to re-scaling the Riemannian
normal coordinates $x^{\mu}\mapsto\lambda x^{\mu}$ by $\lambda\in[0,1]$.
Consider now the smooth 1-parameter family of smooth metrics defined
via, 
\begin{equation}
h_{\mu\nu}(x;\lambda)\equiv\lambda^{-2}(S_{\lambda}^{*}g)_{\mu\nu}(x)=g_{\mu\nu}(\lambda x).\label{eq:h_munu(lambda)}
\end{equation}
Note that $h_{\mu\nu}(\lambda)$ smoothly interpolates between the
flat spacetime metric, $\eta_{\mu\nu}$, at $\lambda=0$ and the original
curved metric, $g_{\mu\nu}$, at $\lambda=1$.

For any Feynman parametrix compatible with the joint smoothness axiom
$\spectrum$, the quantity, 
\[
\Omega_{C}[h_{\mu\nu}(\lambda),t_{\mu},L](f_{1},f_{2};\vec{0}),
\]
defined via (\ref{eq:Omega_C}) is smooth in $\lambda$. Hence, by
Taylor's theorem with remainder, for any nonnegative integer $n$,
we have 
\begin{equation}
\Omega_{C}[g_{\mu\nu},t_{\mu},L](f_{1},f_{2};\vec{0})=\sum_{k=0}^{n}\frac{1}{k!}\left[\frac{d^{k}}{d\lambda^{k}}\Omega_{C}[h_{\mu\nu}(\lambda),t_{\mu},L](f_{1},f_{2};\vec{0})\right]_{\lambda=0}+R_{n}(f_{1},f_{2};\vec{0}),\label{eq:scaling expansion Omega_C tilde-1-1}
\end{equation}
where the Taylor remainder is given by 
\begin{equation}
R_{n}(f_{1},f_{2};\vec{0})\equiv\frac{1}{n!}\int_{0}^{1}d\lambda(1-\lambda)^{n}\frac{d^{(n+1)}}{d\lambda^{(n+1)}}\Omega_{C}[h_{\mu\nu}(\lambda),t_{\mu},L](f_{1},f_{2};\vec{0}).\label{eq:R_n-1-1}
\end{equation}

We now show that, modulo smooth terms, the remainder (\ref{eq:R_n-1-1})
is of scaling degree $(n-D+5)$ and, thus, the non-smooth behavior
of $\Omega_{C}$ is entirely contained (up to scaling degree $\delta$)
in the finite $k$-sum of (\ref{eq:scaling expansion Omega_C tilde-1-1})
for $n\ge\delta+D-4$. We have 
\begin{equation}
\left(S_{s}^{*}\Omega_{C}[h_{\mu\nu}(\lambda),t_{\mu},L]\right)=\Omega_{C}[(S_{s}^{*}h)_{\mu\nu}(\lambda),(S_{s}^{*}t)_{\mu},L]=\Omega_{C}[s^{2}h{}_{\mu\nu}(s\lambda),st{}_{\mu},L].\label{eq:pullback of Omega_C via S_s}
\end{equation}
where the first equality follows directly from the definition (\ref{eq:Omega_C})
of $\Omega_{C}$ and the second equality follows from the definition
(\ref{eq:h_munu(lambda)}) of $h_{\mu\nu}$, 
\[
(S_{s}^{*}h)_{\mu\nu}(x;\lambda)=\lambda^{-2}(S_{s}^{*}\circ S_{\lambda}^{*}g)_{\mu\nu}(x)=s^{2}(S_{s\lambda}^{*}g)_{\mu\nu}(x)=s^{2}h_{\mu\nu}(x;s\lambda).
\]
On the other hand, since $\chi(y,\vec{0};sL)-\chi(y,\vec{0};L)$ vanishes
in a neighborhood of the origin, $y=\vec{0}$, it follows from the
same wavefront set arguments used in Proposition \ref{lem:smoothness of bt}
that for any $s\in(0,1]$, we have 
\begin{equation}
\Omega_{C}[h_{\mu\nu}(\lambda),t_{\mu},L]=\Omega_{C}[h_{\mu\nu}(\lambda),t_{\mu},sL]+\text{smooth terms}.\label{eq:Omega_C[L]=Omega_C[sL]+smooth}
\end{equation}
Plugging (\ref{eq:Omega_C[L]=Omega_C[sL]+smooth})
into (\ref{eq:pullback of Omega_C via S_s}) yields, 
\[
\left(S_{s}^{*}\Omega_{C}[h_{\mu\nu}(\lambda),t_{\mu},L]\right)=\Omega_{C}[s^{2}h{}_{\mu\nu}(s\lambda),st{}_{\mu},sL]+\text{smooth terms}.
\]
Plugging this back into the remainder (\ref{eq:R_n-1-1}), we find
modulo smooth terms, 
\begin{align}
(S_{s}^{*}R_{n})(f_{1},f_{2};\vec{0}) & =\frac{1}{n!}\int_{0}^{1}d\lambda(1-\lambda)^{n}\frac{d^{(n+1)}}{d\lambda^{(n+1)}}\Omega_{C}[s^{2}h{}_{\mu\nu}(s\lambda),st_{\mu},sL](f_{1},f_{2};\vec{0})\label{eq:estimate for pullback of R_n via S_s}\\
 & =s^{(n+1)}\frac{1}{n!}\int_{0}^{1}d\lambda(1-\lambda)^{n}\left[\frac{\partial^{(n+1)}}{\partial q{}^{(n+1)}}\Omega_{C}[s^{2}h{}_{\mu\nu}(q),st_{\mu},sL](f_{1},f_{2};\vec{0})\right]_{q=s\lambda}\nonumber 
\end{align}
However, from the almost homogeneous scaling behavior of the Feynman
parametrix and its smoothness in $m^{2}$ together with the invariance
of the cutoff function (\ref{eq:chi_C}) under the simultaneous rescaling
$(g_{\mu\nu},t_{\mu},L)\to(s^{2}g{}_{\mu\nu},st_{\mu},sL)$, it follows
that for any $q\in[0,1]$, we have 
\begin{align}
\Omega_{C}[s^{2}h{}_{\mu\nu}(q),st_{\mu},sL](f_{1},f_{2};\vec{0}) & =\mathcal{O}\left(s^{(-D+4)}\right),
\end{align}
Consequently, we find modulo smooth terms 
\[
(S_{s}^{*}R_{n})(f_{1},f_{2};\vec{0})=\mathcal{O}(s^{(n+5-D)})
\]
which implies that the scaling degree of any non-smooth contributions
to $R_{n}(f_{1},f_{2};\vec{0})$ must be at least $n+5-D$.

Thus, we have shown that 
\begin{equation}
\Omega_{C}[g_{\mu\nu},t_{\mu},L](f_{1},f_{2};\vec{0})\sim_{\delta}\sum_{k=0}^{\delta-D+4}\frac{1}{k!}\left[\frac{d^{k}}{d\lambda^{k}}\Omega_{C}[\lambda^{-2}(S_{\lambda}^{*}g)_{\mu\nu},t_{\mu},L](f_{1},f_{2};\vec{0})\right]_{\lambda=0}+\text{smooth}.\label{eq:Omega_C Taylor exp mod smooth}
\end{equation}
We now rewrite (\ref{eq:Omega_C Taylor exp mod smooth}) in the form
of the claimed curvature expansion (\ref{eq:Omega_C curvature exp})
for the special case that the metric has polynomial dependence on
the coordinates, $g_{\mu\nu}=g_{\mu\nu}^{(P)}$. Since we have 
\[
\lambda^{-2}(S_{\lambda}^{*}g)_{\mu\nu}^{(P)}(x)=g_{\mu\nu}^{(P)}[x^{\sigma},\eta_{\mu\nu},\lambda^{2}R_{\mu\nu\kappa\rho}(\vec{0}),\lambda^{3}\nabla_{\sigma}R_{\mu\nu\kappa\rho}(\vec{0}),\dots,\lambda^{P}\nabla_{(\sigma_{1}}\cdots\nabla_{\sigma_{(P-2)})}R_{\mu\nu\kappa\rho}(\vec{0})],
\]
it follows that 
\[
\Omega_{C}[\lambda^{-2}(S_{\lambda}^{*}g)_{\mu\nu},t_{\mu},L]=\Omega_{C}[\eta_{\mu\nu},\lambda^{2}R_{\mu\nu\kappa\rho}(\vec{0}),\dots,\lambda^{P}\nabla_{(\sigma_{1}}\cdots\nabla_{\sigma_{(P-2)})}R_{\mu\nu\kappa\rho}(\vec{0}),t_{\mu},L].
\]
For any smooth function of the form $f=f(\lambda^{2}\alpha_{0},\lambda^{3}\alpha_{1},\dots,\lambda^{P}\alpha_{P-2})$,
a straightforward application of the multi-variate chain rule yields,
\[
\left.\frac{d^{k}f}{d\lambda^{k}}\right|_{\lambda=0}=\sum_{2p_{0}+3p_{1}+\cdots+kp_{(k-2)}=k}k!\,\alpha_{0}^{p_{0}}\cdots\alpha_{(k-2)}^{p_{(k-2)}}\,\left.\frac{\partial^{(p_{0}+\cdots+p_{(k-2)})}f(\alpha_{0},\dots,\alpha_{(P-2)})}{\partial^{p_{0}}\alpha_{0}\cdots\partial^{p_{k-2}}\alpha_{(k-2)}}\right|_{\alpha_{0},\dots,\alpha_{(P-2)}=0}.
\]
Using this formula to evaluate the terms in the $k$-sum of (\ref{eq:Omega_C Taylor exp mod smooth})
then yields the claimed curvature expansion (\ref{eq:Omega_C curvature exp}).
The result can then be extended to general smooth $g_{\mu\nu}$ via
compatibility with axiom $\spectrum$, using the same argument as
in the proof of Proposition \ref{Prop: Omega_C[g]=Omega_C[g^N]+smooth}.

\section{Construction of covariance-restoring counterterms based on general
associativity conditions\label{sec:model-indep counterterms}}

The purpose of this appendix is to develop an algorithm for constructing
covariance-restoring counterterms without relying on explicit formulas
for the OPE coefficients or any other special model-dependent properties.
This algorithm is based on the general associativity properties of
OPE coefficients and, thus, should be applicable to flow relations
for any renormalizable Lorentzian quantum field theory. At the end
of the appendix, we show this algorithm reproduces the counterterms
derived in Section \ref{sec:Minkowski-flow-relations} for the Klein-Gordon
OPE coefficients $C_{T_{0}\{\phi\cdots\phi\}}^{I}$ and we will use
the algorithm to generate counterterms for the flow relations of $\lambda\phi^{4}$-theory.
For simplicity, we give a derivation for Lorentz-covariance restoring
counterterms in flat spacetime; however, the derivation can be generalized
to curved spacetimes using the approach developed in Section \ref{sec:Wick fe in CS}.

Consider a theory arising from a Lagrangian with a self-interaction
term $\zeta\Phi_{V}$, where $\zeta$ denotes the coupling parameter.
(Note that for power-counting renormalizable theories, the engineering
dimension of $\Phi_{V}$ must be less than or equal to the spacetime
dimension.) For example, for $\lambda\phi^{4}$-theory we have $\zeta=\lambda$
and $\Phi_{V}=\phi^{4}/4!$. Consider the OPE coefficients arising
from products $\Phi_{A_{1}}(x_{1})\dots\Phi_{A_{n}}(x_{n})$, where
the fields $\Phi_{A_{i}}$ are of arbitrary tensorial (or spinorial)
type. We assume that the Lorentzian OPE coefficients $C_{T_{0}\{A_{1},\dots,A_{n}\}}^{B}$
have been found to satisfy a flow relation of the form 
\begin{align}
 & \frac{\partial}{\partial\zeta}C_{T_{0}\{A_{1},\dots,A_{n}\}}^{B}(x_{1},\dots,x_{n};z)\label{eq:gen preliminary Minkowski flow rel}\\
 & \approx-i\int d^{D}y\chi(y,z;L)\Omega{}_{T_{0}\{VA_{1}\cdots A_{n}\}}^{B}(y,x_{1},\dots,x_{n};z)+\text{covariance-restoring counterterms,}\nonumber 
\end{align}
where $\chi(y,z;L)$ is a suitable translationally-invariant cutoff
function (see (\ref{eq:trans-inv Mink cutoff})) and the quantity
$\Omega{}_{T_{0}\{VA_{1}\cdots A_{n}\}}^{B}(y,x_{1},\dots,x_{n};z)$
is given in terms of OPE coefficients by a formula of the general
form 
\begin{align}
\Omega{}_{T_{0}\{VA_{1}\cdots A_{n}\}}^{B}(y,x_{1},\dots,x_{n};z) & =C_{T_{0}\{VA_{1}\cdots A_{n}\}}^{B}(y,x_{1},\dots,x_{n};z)+\label{eq:Omega gen}\\
 & -\sum_{i=1}^{n}\sum_{[C]\le[A_{i}]+[V]-D}C_{T_{0}\{VA_{i}\}}^{C}(y,x_{i};x_{i})C_{T_{0}\{A_{1}\cdots\widehat{A_{i}}C\cdots A_{n}\}}^{B}(x_{1},\dots,x_{n};z)+\nonumber \\
 & -\sum_{[C]<[B]-[V]+D}C_{T_{0}\{A_{1}\cdots A_{n}\}}^{C}(x_{1},\dots,x_{n};z)C_{T_{0}\{VC\}}^{B}(y,z;z),\nonumber 
\end{align}
where $D$ denotes the spacetime dimension. For Klein-Gordon theory
($\zeta=m^{2}$ and $\Phi_{V}=\phi^{2}/2$), eq.~\eqref{eq:Omega gen}
corresponds to the flow relation \eqref{eq:Lorentz inv id fe} for
the Wick OPE coefficient $C_{T_{0}\{\phi\cdots\phi\}}^{I}$ (where
only the first line of eq.~\eqref{eq:Omega gen} contributes in this
case). For $4$-dimensional $\lambda\phi^{4}$-theory ($\zeta=\lambda$
and $\Phi_{V}=\phi^{4}/4!$), eq.~\eqref{eq:Omega gen} corresponds
to the Wick-rotated integrand of the Euclidean Holland and Hollands
flow equation \eqref{eq:HH fe}. For $4$-dimensional Yang-Mills gauge
theories, eq.~\eqref{eq:Omega gen} coincides with the Wick-rotated
integrand of the Euclidean flow relations given in \cite[Theorem 4]{Frob_Holland_Yang-Mills}.
Thus, eq.~\eqref{eq:Omega gen} encompasses all of these cases. Our
aim is to explicitly obtain the covariance restoring counterterms
in eq.~(\ref{eq:gen preliminary Minkowski flow rel}).

Note that the individual terms in the sum for $\Omega{}_{T_{0}\{VA_{1}\cdots A_{n}\}}^{B}$
are well-defined as distributions in spacetime variables $y,x_{1},\dots,x_{n}$
only away from all diagonals, i.e., where none of the spacetime events
coincide. However, assuming the OPE coefficients satisfy the associativity
and scaling axioms postulated in \cite{HW_Axiomatic_QFTCS}, then
the scaling degree of $\Omega{}_{T_{0}\{VA_{1}\cdots A_{n}\}}^{B}$
on any partial diagonal involving $y$ and one other spacetime event
$x_{i}$ is guaranteed to be strictly less than the spacetime dimension
$D$. It follows then that $\Omega{}_{T_{0}\{VA_{1}\cdots A_{n}\}}^{B}$
can be \emph{uniquely} extended to a distribution on these partial
diagonals involving $y$, so the integral in (\ref{eq:gen preliminary Minkowski flow rel})
is well defined (even though individual terms in $\Omega{}_{T_{0}\{VA_{1}\cdots A_{n}\}}^{B}$
generally contain non-integrable divergences at $y=x_{i}$ for $i=1,\dots,n$).

The failure of the integral in (\ref{eq:gen preliminary Minkowski flow rel})
by itself to be covariant under Lorentz transformation $\Lambda$
is characterized by the nonvanishing of the quantity 
\begin{equation}
-i\int d^{D}y\left[\chi(\Lambda y,\Lambda z;L)-\chi(y,z;L)\right]\Omega{}_{T_{0}\{VD_{1}\cdots D_{n}\}}^{E}(y,x_{1},\dots,x_{n};z),\label{eq:non-cov of int Omega tilde}
\end{equation}
Since $\chi(y,z;L)=1$ in an open neighborhood of $z$, if the spacetime
events $x_{i}$ are sufficiently near to $z$, then we have 
\[
\chi(\Lambda x_{i},\Lambda z;L)=\chi(x_{i},z;L),\qquad\text{for all \ensuremath{i=1,\dots,n}.}
\]
It then follows that the integrand in \eqref{eq:non-cov of int Omega tilde}
vanishes as $y$ approaches the partial diagonals $y=x_{i}\ne x_{j}$.
Consequently, unlike the integral in (\ref{eq:gen preliminary Minkowski flow rel}),
the expression (\ref{eq:non-cov of int Omega tilde}) is well defined
for each of the individual terms in the sum defining $\Omega{}_{T_{0}\{VA_{1}\cdots A_{n}\}}^{B}$.
Note the $y$-dependence of $\Omega{}_{T_{0}\{VA_{1}\cdots A_{n}\}}^{B}$
is isolated within terms of the form 
\begin{equation}
C_{T_{0}\{VD_{1}\cdots D_{n}\}}^{E}.\label{eq:C^E_VD1...Dn}
\end{equation}
Specifically, the $y$-dependence of first line of \eqref{eq:Omega gen}
appears in $C_{T_{0}\{VA_{1}\cdots A_{n}\}}^{B}$. The $y$-dependence
of the second line of \eqref{eq:Omega gen} appears in $C_{T_{0}\{VA_{i}\}}^{C}$.
Finally, the $y$-dependence of the third line of \eqref{eq:Omega gen}
appears in $C_{T_{0}\{VC\}}^{B}$. It follows that the non-covariance
of \eqref{eq:non-cov of int Omega tilde} is quantified by integrals
of the form: 
\begin{align}
 & \Upsilon_{T_{0}\{D_{1}\cdots D_{n}\}}^{E}(x_{1},\dots,x_{n};z;\Lambda)\label{eq:Upsilon def}\\
 & \equiv-i\int d^{D}y\left[\chi(\Lambda y,\Lambda z;L)-\chi(y,z;L)\right]C_{T_{0}\{VD_{1}\cdots D_{n}\}}^{E}(y,x_{1},\dots,x_{n};z).\nonumber 
\end{align}
Our task is now to show that the non-covariance of these terms can
be compensated by counterterms and thereby to construct the ``covariance-restoring
counterterms'' for the flow relation \eqref{eq:gen preliminary Minkowski flow rel}.

The integrand of \eqref{eq:Upsilon def} is nonvanishing only when
$y$ lies outside an open neighborhood of $z$. The associativity
condition (see eq.~\eqref{eq:pertinent assoc cond for C_H}) implies
that, for any merger tree $\mathcal{T}$ such that $x_{1},\dots,x_{n}$
approach an auxiliary point $z'$ faster than $z'$ and $y$ approach
$z$, we have 
\begin{equation}
C_{T_{0}\{VD_{1}\cdots D_{n}\}}^{E}(y,x_{1},\dots,x_{n};z)\sim_{\mathcal{T},\delta}\sum_{C}C_{T_{0}\{D_{1}\cdots D_{n}\}}^{C}(x_{1},\dots,x_{n};z')C_{T_{0}\{VC\}}^{E}(y,z';z),\label{eq:assoc exp of C(y,x1,...xn;z)}
\end{equation}
where both sides are viewed as distributions in $(y,x_{1},\dots,x_{n},z')$
but with the left-hand side having trivial dependence on the auxiliary
point $z'$. Plugging \eqref{eq:assoc exp of C(y,x1,...xn;z)} into
\eqref{eq:Upsilon def} yields, 
\begin{align}
 & \Upsilon_{T_{0}\{D_{1}\cdots D_{n}\}}^{E}(x_{1},\dots,x_{n};z;\Lambda)\label{eq:assoc exp of Upsilon}\\
 & \sim_{\mathcal{T}',\delta}-i\sum_{C}C_{T_{0}\{D_{1}\cdots D_{n}\}}^{C}(x_{1},\dots,x_{n};z')\int d^{D}y\left[\chi(\Lambda y,\Lambda z;L)-\chi(y,z;L)\right]C_{T_{0}\{VC\}}^{E}(y,z';z),\nonumber 
\end{align}
where $\mathcal{T}'$ denotes any merger tree with $(x_{1},\dots,x_{n})$
approaching $z'$ faster than $z'$ approaches $z$. Assuming the
OPE coefficient $C_{VC}^{E}$ satisfies the general microlocal spectrum
condition stated in \cite{HW_Axiomatic_QFTCS}, then all elements
$(y,k_{1},z',k_{2},z,k_{3})\in(T^{*}M)^{3}$ in the wavefront set
of $C_{T_{0}\{VC\}}^{E}(y,z';z)$ will be such that $k_{1}=-k_{2}$
and $k_{3}=\vec{0}$. It follows then from a straightforward application
of \cite[Theorem 8.2.12]{Hormander_book} that the dependence of \eqref{eq:assoc exp of Upsilon}
on $(z',z)$ is, in fact, smooth and, thus, we may set $z'=z$ : 
\begin{align}
\Upsilon_{T_{0}\{D_{1}\cdots D_{n}\}}^{E}(x_{1},\dots,x_{n};z;\Lambda) & \approx\sum_{C}Q_{T_{0}\{VC\}}^{E}(\Lambda^{-1})C_{T_{0}\{D_{1}\cdots D_{n}\}}^{C}(x_{1},\dots,x_{n};z),\label{eq:Upsilon assoc exp in terms of Q}
\end{align}
where $Q_{T_{0}\{VC\}}^{E}(\Lambda^{-1})$ is given by 
\begin{equation}
Q_{T_{0}\{VC\}}^{E}(\Lambda^{-1})\equiv-i\int d^{D}y\left[\chi(\Lambda y,\vec{0};L)-\chi(y,\vec{0};L)\right]C_{T_{0}\{VC\}}^{E}(y,\vec{0};\vec{0}),\label{eq:gen Q def}
\end{equation}
where translation invariance was used to set $z=\vec{0}$. Thus, $Q_{T_{0}\{VC\}}^{E}(\Lambda^{-1})$
is independent of spacetime point $z$. Note that no assumption has
been made on how quickly events $x_{1},\dots,x_{n}$ approach $z$
relative to each other, so \eqref{eq:Upsilon assoc exp in terms of Q}
is valid for all merger trees involving the events $x_{1},\dots,x_{n}$.
Hence, we simply use the notation ``$\approx$'' that was introduced
in the paragraph surrounding eq.~\eqref{eq:OPE precise asymp rel}.

We now show that \eqref{eq:gen Q def} satisfies a cohomological identity
that enables us to obtain the desired counterterms. Let $\Lambda_{1}$
and $\Lambda_{2}$ be Lorentz transformations. Then we have 
\begin{align}
 & Q_{T_{0}\{VC\}}^{E}(\Lambda_{1}\Lambda_{2})-Q_{T_{0}\{VC\}}^{E}(\Lambda_{1})\nonumber \\
 & =-i\int d^{D}y\left[\chi(\Lambda_{2}^{-1}\Lambda_{1}^{-1}y,\vec{0};L)-\chi(\Lambda_{1}^{-1}y,\vec{0};L)\right]C_{T_{0}\{VC\}}^{E}(y,\vec{0};\vec{0})\nonumber \\
 & =-i\int d^{D}y'\left[\chi(\Lambda_{2}^{-1}y',\vec{0};L)-\chi(y',\vec{0};L)\right]C_{T_{0}\{VC\}}^{E}(\Lambda_{1}y',\vec{0};\vec{0})\nonumber \\
 & =-i\int d^{D}y'\left[\chi(\Lambda_{2}^{-1}y',\vec{0};L)-\chi(y',\vec{0};L)\right]\sum_{A,B}D_{A}^{E}(\Lambda_{1})D_{C}^{B}(\Lambda_{1}^{-1})C_{T_{0}\{VB\}}^{A}(y',\vec{0};\vec{0})\nonumber \\
 & =\sum_{A,B}D_{A}^{E}(\Lambda_{1})D_{C}^{B}(\Lambda_{1}^{-1})Q_{T_{0}\{VB\}}^{A}(\Lambda_{2}),\label{eq:coho id for gen Q}
\end{align}
where the second equality follows from a change of integration variables
$y\to y'=\Lambda_{1}^{-1}y$ and third equality follows from the Lorentz
covariance of the OPE coefficients (where we recall that $\Phi_{V}$
is a Lorentz scalar). Denoting $\boldsymbol{Q}\equiv Q_{T_{0}\{VC\}}^{E}$
and suppressing field indices, eq.~\eqref{eq:coho id for gen Q}
is equivalent to: 
\[
0=(d^{1}\boldsymbol{Q})(\Lambda_{1},\Lambda_{2})=\boldsymbol{Q}(\Lambda_{1})+D(\Lambda_{1})\boldsymbol{Q}(\Lambda_{2})-\boldsymbol{Q}(\Lambda_{1}\Lambda_{2}),
\]
which is the cohomological identity \eqref{eq:d^1Q=0}. As established
in Proposition \ref{Prop: existence of a}, this identity implies
there exists $\boldsymbol{a}\equiv a_{T_{0}\{VC\}}^{B}$ such that:
\[
\boldsymbol{Q}(\Lambda)=(d^{0}\boldsymbol{a})(\Lambda)=(D(\Lambda)-I)\boldsymbol{a}.
\]
For tensor-valued\footnote{A formula analogous to \eqref{eq:a for gen Omega} can be obtained
using the methods of Appendix \ref{sec:K-G Lorentz cts} when $\boldsymbol{Q}$
is spinor-valued (see also \cite[Section 4]{Prange}).} $\boldsymbol{Q}$, the results of Appendix \ref{sec:K-G Lorentz cts}
imply the $\boldsymbol{a}$ can be inductively constructed (modulo
Lorentz-invariant tensors) from 
\begin{equation}
\boldsymbol{a}=\sum_{\substack{j=1\\
\widetilde{c}_{j}\ne0
}
}^{k}\frac{1}{\widetilde{c}_{j}}\mathbb{E}_{j}\left(-\mathbb{L}_{\kappa\rho}\boldsymbol{B}^{\kappa\rho}+4\sum_{i<j\le n}\eta_{\mu_{i}\mu_{j}}\text{tr}_{ij}\boldsymbol{a}\right),\label{eq:a for gen Omega}
\end{equation}
with 
\[
\boldsymbol{B}^{\kappa\rho}\equiv(B^{\kappa\rho})_{T_{0}\{VC\}}^{E}=2i\int d^{D}yy^{[\kappa}\partial^{\rho]}\chi(y,\vec{0})C_{T_{0}\{VC\}}^{E}(y,\vec{0};\vec{0}).
\]

By reasoning analogous to the arguments of Section \ref{sec:Minkowski-flow-relations},
we obtain counterterms that ensure the Lorentz-covariance of the flow
relation \eqref{eq:gen preliminary Minkowski flow rel} by making
the following substitution in every appearance of $C_{T_{0}\{VD_{1}\cdots D_{n}\}}^{E}$
in $\Omega_{T_{0}\{VA_{1}\cdots A_{n}\}}^{B}$: 
\begin{equation}
C_{T_{0}\{VD_{1}\cdots D_{n}\}}^{E}\to C_{T_{0}\{VD_{1}\cdots D_{n}\}}^{E}(y,x_{1},\dots,x_{n};z)-\frac{1}{\mathcal{V}}\sum_{C}a_{T_{0}\{VC\}}^{E}C_{T_{0}\{D_{1}\cdots D_{n}\}}^{C}(x_{1},\dots,x_{n};z),\label{eq:gen substitution of counterterms}
\end{equation}
where we have written 
\[
\mathcal{V}\equiv\int d^{D}y\chi(y,\vec{0};\vec{0}).
\]
It is understood the $C$-sum in \eqref{eq:gen substitution of counterterms}
is carried to sufficiently-large engineering dimension $[C]$ to achieve
whatever asymptotic precision is desired from the flow relation. The
substitution rule (\ref{eq:gen substitution of counterterms}) is
the key result of this Appendix. We now illustrate it by applying
it to the cases of the massive Klein-Gordon field and 4-dimensional
$\lambda\phi^{4}$-theory.

For the case of the flow relations for $C_{T_{0}\{\phi\cdots\phi\}}^{I}$
obtained in this paper for the massive Klein-Gordon field, we have
$\Phi_{V}=\phi^{2}/2$, $\zeta=m^{2}$, and \eqref{eq:Omega gen}
reduces to: 
\[
\Omega_{T_{0}\{(\phi^{2}/2)\phi\cdots\phi\}}^{I}(y,x_{1},\dots,x_{n};z)=\frac{1}{2}C_{T_{0}\{\phi^{2}\phi\cdots\phi\}}^{I}(y,x_{1},\dots,x_{n};z).
\]
Our algorithm instructs us to make the substitution \eqref{eq:gen substitution of counterterms}
in $\Omega_{T_{0}\{(\phi^{2}/2)\phi\cdots\phi\}}^{I}$. Plugging the
result of this substitution into \eqref{eq:gen preliminary Minkowski flow rel}
yields the flow relation: 
\begin{align}
 & \frac{\partial}{\partial m^{2}}C_{T_{0}\{\phi\cdots\phi\}}^{I}(x_{1},\dots,x_{n};z)\label{eq:K-G flow rel generated by gen assoc algorithm}\\
 & \approx-\frac{i}{2}\int d^{D}y\chi(y,z;L)C_{T_{0}\{\phi^{2}\phi\cdots\phi\}}^{I}(y,x_{1},\dots,x_{n};z)-\sum_{C}a_{T_{0}\{(\phi^{2}/2)C\}}^{I}C_{T_{0}\{A_{1}\cdots A_{n}\}}^{C}(x_{1},\dots,x_{n};z)\nonumber 
\end{align}
where $a_{T_{0}\{(\phi^{2}/2)C\}}^{I}$ is given recursively by \eqref{eq:a for gen Omega}
with 
\begin{equation}
(B^{\kappa\rho})_{T_{0}\{(\phi^{2}/2)C\}}^{I}=i\int d^{D}yy^{[\kappa}\partial^{\rho]}\chi(y,\vec{0})C_{T_{0}\{\phi^{2}C\}}^{I}(y,\vec{0};\vec{0}).\label{eq:K-G B tensor from assoc algorithm}
\end{equation}
Comparing \eqref{eq:K-G flow rel generated by gen assoc algorithm}
with \eqref{eq:id coef fe Minkowski} of Theorem \ref{thm:Minkowski flow rel}
and \eqref{eq:K-G B tensor from assoc algorithm} with \eqref{eq:(B^kappa rho)_C},
we find that the substitution (\ref{eq:gen substitution of counterterms})
reproduces the covariance-restoring counterterms obtained in Section
\ref{sec:Minkowski-flow-relations} for the flow relations of the
Klein-Gordon OPE coefficients $C_{T_{0}\{\phi\cdots\phi\}}^{I}$.

For $\lambda\phi^{4}$-theory, we have $\Phi_{V}=\phi^{4}/4!$ and
$\zeta=\lambda$. Our algorithm instructs us to make the following
substitutions in the formula \eqref{eq:Omega gen} for $\Omega_{T_{0}\{(\phi^{4}/4!)A_{1}\cdots A_{n}\}}^{B}$:
\begin{align}
C_{T_{0}\{\phi^{4}A_{1}\cdots A_{n}\}}^{B} & \to C_{T_{0}\{\phi^{4}A_{1}\cdots A_{n}\}}^{B}(y,x_{1},\dots,x_{n};z)-\frac{1}{\mathcal{V}}\sum_{C}a_{T_{0}\{\phi^{4}C\}}^{B}C_{T_{0}\{A_{1}\cdots A_{n}\}}^{C}(x_{1},\dots,x_{n};z)\label{eq:lambda phi 4 sub 1}\\
C_{T_{0}\{\phi^{4}A_{i}\}}^{C} & \to C_{T_{0}\{\phi^{4}A_{i}\}}^{C}(y,x_{i};z)-\frac{1}{\mathcal{V}}\sum_{D}a_{T_{0}\{\phi^{4}D\}}^{C}C_{T_{0}\{A_{i}\}}^{D}(x_{i};z)\label{eq:lambda phi 4 sub 2}\\
C_{T_{0}\{\phi^{4}C\}}^{B} & \to C_{T_{0}\{\phi^{4}C\}}^{B}(y,z;z)-\frac{1}{\mathcal{V}}a_{T_{0}\{\phi^{4}C\}}^{B},\label{eq:lambda phi 4 sub 3}
\end{align}
where $a_{T_{0}\{\phi^{4}C\}}^{B}$ is given inductively by \eqref{eq:a for gen Omega}
in terms of 
\[
(B^{\kappa\rho})_{T_{0}\{\phi^{4}C\}}^{B}=2i\int d^{D}yy^{[\kappa}\partial^{\rho]}\chi(y,\vec{0})C_{T_{0}\{\phi^{4}C\}}^{B}(y,\vec{0};\vec{0}).
\]
Note that the OPE coefficient $C_{T_{0}\{A\}}^{B}=C_{A}^{B}$ appearing
on the right side of (\ref{eq:lambda phi 4 sub 2}) is given by 
\begin{equation}
C_{A}^{B}(x;z)=\partial_{\alpha_{1}}^{(x)}(x-z)^{(\beta_{1}}\cdots\partial_{\alpha_{k}}^{(x)}(x-z)^{\beta_{k}}(x-z)^{\beta_{k+1}}\cdots(x-z)^{\beta_{m})},\label{eq:1 point OPE}
\end{equation}
where $A=\alpha_{1}\cdots\alpha_{k}$, $B=\beta_{1}\cdots\beta_{m}$
and $C_{A}^{B}=0$ if $k>m$. Making the substitutions \eqref{eq:lambda phi 4 sub 1}-\eqref{eq:lambda phi 4 sub 3}
in $\Omega_{T_{0}\{(\phi^{4}/4!)A_{1}\cdots A_{n}\}}^{B}$ and plugging
this back into the flow relation \eqref{eq:gen preliminary Minkowski flow rel},
we obtain 
\begin{align}
 & \frac{\partial}{\partial\lambda}C_{T_{0}\{A_{1}\cdots A_{n}\}}^{B}(x_{1},\dots,x_{n};z)\approx\nonumber \\
 & -\frac{1}{4!}\int d^{4}y\chi(y,z;L)\left[\vphantom{\int_{\epsilon}}C_{T_{0}\{\phi^{4}A_{1}\cdots A_{n}\}}^{B}(y,x_{1},\dots,x_{n};z)+\right.\nonumber \\
 & -\sum_{i=1}^{n}\sum_{[C]\le[A_{i}]}\left[C_{T_{0}\{\phi^{4}A_{i}\}}^{C}(y,x_{i};x_{i})-\frac{1}{\mathcal{V}}\right.\sum_{[D]}\left.a_{T_{0}\{\phi^{4}D\}}^{C}C_{T_{0}\{A_{i}\}}^{D}(x_{i};z)\right]C_{T_{0}\{A_{1}\cdots\widehat{A_{i}}C\cdots A_{n}\}}^{B}(x_{1},\dots,x_{n};z)+\nonumber \\
 & -\left[\vphantom{C_{T_{0}\{\phi^{4}\}}^{B}}\right.\sum_{[C]<[B]}C_{T_{0}\{\phi^{4}C\}}^{B}(y,z;z)-\frac{1}{\mathcal{V}}\sum_{[C]\ge[B]}\left.\vphantom{C_{T_{0}\{\phi^{4}\}}^{B}}a_{T_{0}\{\phi^{4}C\}}^{B}\right]\left.\vphantom{\int_{\epsilon}}C_{T_{0}\{A_{1}\cdots A_{n}\}}^{C}(x_{1},\dots,x_{n};z)\right].
\end{align}
The generalization of this relation to curved spacetime was already
given in \eqref{eq:lambda phi 4 Mink fe} of the Introduction.

 \bibliographystyle{hunsrt}
\phantomsection\addcontentsline{toc}{section}{\refname}\bibliography{citations}

\begin{thebibliography}{10}

\bibitem{Wilson_nonlagrangian_current_algebra}
Kenneth~G. Wilson.
\newblock {Nonlagrangian models of current algebra}.
\newblock {\em Phys. Rev.}, 179:1499--1512, 1969.

\bibitem{Wilson_Zimmermann}
K.~G. Wilson and W.~Zimmermann.
\newblock {Operator product expansions and composite field operators in the
  general framework of quantum field theory}.
\newblock {\em Commun. Math. Phys.}, 24:87--106, 1972.

\bibitem{Zimmermann_normal_products_perturabative_OPE}
Wolfhart Zimmermann.
\newblock {Normal products and the short distance expansion in the perturbation
  theory of renormalizable interactions}.
\newblock {\em Annals Phys.}, 77:570--601, 1973.

\bibitem{Hollands_perturbative_OPE_CS}
Stefan Hollands.
\newblock {The Operator product expansion for perturbative quantum field theory
  in curved spacetime}.
\newblock {\em Commun. Math. Phys.}, 273:1--36, 2007, gr-qc/0605072.

\bibitem{Bostelmann_OPE_via_phase_space}
Henning Bostelmann.
\newblock {Operator product expansions as a consequence of phase space
  properties}.
\newblock {\em J. Math. Phys.}, 46:082304, 2005, math-ph/0502004.

\bibitem{Bostelmann_phase_space_short_distance_QFT}
Henning Bostelmann.
\newblock {Phase space properties and the short distance structure in quantum
  field theory}.
\newblock {\em J. Math. Phys.}, 46:052301, 2005, math-ph/0409070.

\bibitem{Fredenhagen_conformal_Haag_Kastler_nets_pointlike_fields_existence_OPE}
K.~Fredenhagen and M.~Jorss.
\newblock {Conformal Haag-Kastler nets, point - like localized fields and the
  existence of operator product expansions}.
\newblock {\em Commun. Math. Phys.}, 176:541--554, 1996.

\bibitem{Fredenhagen_local_algebra_and_pointlike_fields}
Klaus Fredenhagen and Joachim Hertel.
\newblock {Local Algebras of Observables and Point - Like Localized Fields}.
\newblock {\em Commun. Math. Phys.}, 80:555, 1981.

\bibitem{HW_Axiomatic_QFTCS}
Stefan Hollands and Robert~M. Wald.
\newblock {Axiomatic quantum field theory in curved spacetime}.
\newblock {\em Commun. Math. Phys.}, 293:85--125, 2010, 0803.2003.

\bibitem{HW_OPE_Dark_Energy}
Stefan Hollands and Robert~M. Wald.
\newblock {Quantum field theory in curved spacetime, the operator product
  expansion, and dark energy}.
\newblock {\em Gen. Rel. Grav.}, 40:2051--2059, 2008, 0805.3419.

\bibitem{Hollands_Action_Principle}
Stefan Hollands.
\newblock {Action principle for OPE}.
\newblock {\em Nucl. Phys. B}, 926:614--638, 2018, 1710.05601.

\bibitem{HH_Recursive}
Jan Holland and Stefan Hollands.
\newblock {Recursive construction of operator product expansion coefficients}.
\newblock {\em Commun. Math. Phys.}, 336(3):1555--1606, 2015, 1401.3144.

\bibitem{Frob_Holland_Yang-Mills}
Markus~B. Fr\"ob and Jan Holland.
\newblock {All-order existence of and recursion relations for the operator
  product expansion in Yang-Mills theory}.
\newblock March 2016, 1603.08012.

\bibitem{Frob_Riemannian}
Markus~B. Fr\"ob.
\newblock {Recursive construction of the operator product expansion in curved
  space}.
\newblock {\em JHEP}, 02:195, 2021, 2007.15668.

\bibitem{pedagogical_remark_main_thm_ren_theory}
G.~Popineau and R.~Stora.
\newblock {A pedagogical remark on the main theorem of perturbative
  renormalization theory}.
\newblock {\em Nucl. Phys. B}, 912:70--78, 2016.

\bibitem{Bresser_Pinter_Prange}
K.~Bresser, G.~Pinter, and D.~Prange.
\newblock {Lorentz invariant renormalization in causal perturbation theory}.
\newblock 3 1999, hep-th/9903266.

\bibitem{Prange}
Dirk Prange.
\newblock {Lorentz covariance in Epstein-Glaser renormalization}.
\newblock 4 1999, hep-th/9904136.

\bibitem{Wald_GR_text}
Robert~M. Wald.
\newblock {\em {General Relativity}}.
\newblock Chicago Univ. Pr., Chicago, USA, 1984.

\bibitem{Friedlander}
F.~G. Friedlander.
\newblock {\em {The Wave Equation on a Curved Space-Time}}.
\newblock Cambridge University Press, 3 2010.

\bibitem{Haag_Kastler}
Rudolf Haag and Daniel Kastler.
\newblock {An Algebraic approach to quantum field theory}.
\newblock {\em J. Math. Phys.}, 5:848--861, 1964.

\bibitem{Haag_book}
Rudolf Haag.
\newblock {\em {Local Quantum Physics: Fields, Particles, Algebras}}.
\newblock Theoretical and Mathematical Physics. Springer Berlin, Heidelberg,
  second revised and enlarged edition, 1996.

\bibitem{HW_local_Wick_poly}
Stefan Hollands and Robert~M. Wald.
\newblock {Local Wick polynomials and time ordered products of quantum fields
  in curved space-time}.
\newblock {\em Commun. Math. Phys.}, 223:289--326, 2001, gr-qc/0103074.

\bibitem{BFK_muSC_Wick_polynomials}
R.~Brunetti, K.~Fredenhagen, and M.~Kohler.
\newblock {The Microlocal spectrum condition and Wick polynomials of free
  fields on curved space-times}.
\newblock {\em Commun. Math. Phys.}, 180:633--652, 1996, gr-qc/9510056.

\bibitem{BF_muA_interacting_QFTs}
Romeo Brunetti and Klaus Fredenhagen.
\newblock {Microlocal analysis and interacting quantum field theories:
  Renormalization on physical backgrounds}.
\newblock {\em Commun. Math. Phys.}, 208:623--661, 2000, math-ph/9903028.

\bibitem{DF_pert_algebraicQFT_deform_quant}
Michael Duetsch and Klaus Fredenhagen.
\newblock {Perturbative algebraic field theory, and deformation quantization}.
\newblock {\em Fields Inst. Commun.}, 30:151--160, 2001, hep-th/0101079.

\bibitem{DF_algebraicQFT_loop_expansion}
M.~Duetsch and K.~Fredenhagen.
\newblock {Algebraic quantum field theory, perturbation theory, and the loop
  expansion}.
\newblock {\em Commun. Math. Phys.}, 219:5--30, 2001, hep-th/0001129.

\bibitem{Rad_microlocal_Hadamard_cond}
M.~J. Radzikowski.
\newblock {Micro-local approach to the Hadamard condition in quantum field
  theory on curved space-time}.
\newblock {\em Commun. Math. Phys.}, 179:529--553, 1996.

\bibitem{Hollands_Ruan}
Stefan Hollands and Weihua Ruan.
\newblock {The State space of perturbative quantum field theory in curved
  space-times}.
\newblock {\em Annales Henri Poincare}, 3:635--657, 2002, gr-qc/0108032.

\bibitem{HW_Conservation_Stress-energy}
Stefan Hollands and Robert~M. Wald.
\newblock {Conservation of the stress tensor in interacting quantum field
  theory in curved spacetimes}.
\newblock {\em Rev. Math. Phys.}, 17:227--312, 2005, gr-qc/0404074.

\bibitem{KM_analytic_dep_unnecessary}
Igor Khavkine and Valter Moretti.
\newblock {Analytic Dependence is an Unnecessary Requirement in Renormalization
  of Locally Covariant QFT}.
\newblock {\em Commun. Math. Phys.}, 344(2):581--620, 2016, 1411.1302.

\bibitem{KM_Wick_poly}
Igor Khavkine, Alberto Melati, and Valter Moretti.
\newblock {Wick Polynomials of Locally Covariant Boson Fields}.
\newblock {\em Annales Henri Poincare}, 20(3):929--1002, 2019, 1710.01937.

\bibitem{Moretti_comments_stress-energy}
Valter Moretti.
\newblock {Comments on the stress energy tensor operator in curved space-time}.
\newblock {\em Commun. Math. Phys.}, 232:189--221, 2003, gr-qc/0109048.

\bibitem{Barvinsky_Vilkovisky_covariant_Taylor}
A.~O. Barvinsky and G.~A. Vilkovisky.
\newblock {The Generalized Schwinger-Dewitt Technique in Gauge Theories and
  Quantum Gravity}.
\newblock {\em Phys. Rept.}, 119:1--74, 1985.

\bibitem{Hormander_book}
Lars H\"ormander.
\newblock {\em {Analysis of Linear Partial Differential Operators I}}.
\newblock Grundlehren der mathematischen Wissenschaften 256. Springer Berlin,
  Heidelberg, second edition, 1990.

\bibitem{Steinmann}
Othmar Steinmann.
\newblock {\em {Perturbation Expansions in Axiomatic Field Theory}}.
\newblock Springer Berlin, Heidelberg, 1971.

\bibitem{Sanders_Thesis}
Ko~{Sanders}.
\newblock {\em {Aspects of locally covariant quantum field theory}}.
\newblock PhD thesis, University of York, UK, September 2008, 0809.4828.

\bibitem{HH_Associativity}
Jan Holland and Stefan Hollands.
\newblock {Associativity of the operator product expansion}.
\newblock {\em J. Math. Phys.}, 56(12):122303, 2015, 1507.07730.

\bibitem{Caianiello_textbook}
E.~R. Caianiello.
\newblock {\em {Combinatorics and renormalization in quantum field theory}},
  volume~38.
\newblock Benjamin, Reading, 1973.

\bibitem{HW_Existence_TOP}
Stefan Hollands and Robert~M. Wald.
\newblock {Existence of local covariant time ordered products of quantum fields
  in curved space-time}.
\newblock {\em Commun. Math. Phys.}, 231:309--345, 2002, gr-qc/0111108.

\bibitem{Riemannian_Wick_algebra}
Claudio Dappiaggi, Nicol\`o Drago, and Paolo Rinaldi.
\newblock {The algebra of Wick polynomials of a scalar field on a Riemannian
  manifold}.
\newblock {\em Rev. Math. Phys.}, 32(08):2050023, 2020, 1903.01258.

\bibitem{Wald_Iyer}
Vivek Iyer and Robert~M. Wald.
\newblock {Some properties of Noether charge and a proposal for dynamical black
  hole entropy}.
\newblock {\em Phys. Rev. D}, 50:846--864, 1994, gr-qc/9403028.

\bibitem{Thomas_Replacement}
Tracey~Yerkes Thomas.
\newblock {\em {The Differential Invariants of Generalized Spaces}}.
\newblock Cambridge University Press, 1934.

\bibitem{Wigner}
Eugene~P. Wigner.
\newblock {On Unitary Representations of the Inhomogeneous Lorentz Group}.
\newblock {\em Annals Math.}, 40:149--204, 1939.

\end{thebibliography}
 
\end{document}